\documentclass[a4paper, oneside, 12pt, openany]{book}
\pdfoutput=1

\usepackage{marvosym}
\usepackage[utf8]{inputenc}
\usepackage{packages}
\usepackage{macros}
\usepackage{unicode}

\numberwithin{equation}{section}
\numberwithin{table}{section}
\makeatletter
\let\c@equation\c@table
\makeatother

\newbool{shade-envs}
\setboolean{shade-envs}{true}

\ifthenelse{\boolean{shade-envs}}{%
  \colorlet{ShadeOfPurple}{blue!5!white}
  \colorlet{ShadeOfYellow}{yellow!8!white}
  \colorlet{ShadeOfGreen} {green!7!white}
  \colorlet{ShadeOfBrown} {brown!10!white}
  \colorlet{ShadeOfGray}  {gray!10!white}
}
{%
  \declaretheoremstyle[
      spaceabove=6pt,
      spacebelow=6pt,
      bodyfont=\normalfont,
      qed=\(\lozenge\)
  ]{definitionwithbox}
  \declaretheoremstyle[
      headfont=\itshape,
      bodyfont=\normalfont,
      qed=\(\lozenge\)
      ]{remarkwithbox}
}

\ifthenelse{\boolean{shade-envs}}{%
  \declaretheorem[sibling=equation]{theorem}
  \declaretheorem[sibling=theorem]{lemma,proposition,corollary,question,conjecture}
  \declaretheorem[sibling=theorem,style=definition]{definition,attempt}
  \declaretheorem[sibling=theorem,style=definition]{example}
  \declaretheorem[sibling=theorem,style=remark]{remark,assumption}
  \tcbset{shadedenv/.style={
      colback={#1},
      frame hidden,
      enhanced,
      breakable,
      boxsep=0pt,
      left=2mm,
      right=2mm,
      add to width=1.1mm,
      enlarge left by=-0.6mm}
  }
  \tcolorboxenvironment{theorem}    {shadedenv={ShadeOfPurple}}
  \tcolorboxenvironment{lemma}      {shadedenv={ShadeOfPurple}}
  \tcolorboxenvironment{proposition}{shadedenv={ShadeOfPurple}}
  \tcolorboxenvironment{question}   {shadedenv={ShadeOfGray}}  \tcolorboxenvironment{conjecture}   {shadedenv={ShadeOfGray}}
  \tcolorboxenvironment{corollary}  {shadedenv={ShadeOfPurple}}
  \tcolorboxenvironment{definition} {shadedenv={ShadeOfYellow}}
  \tcolorboxenvironment{attempt}    {shadedenv={ShadeOfYellow}}
  \tcolorboxenvironment{example}    {shadedenv={ShadeOfGreen}}
  \tcolorboxenvironment{problem}    {shadedenv={ShadeOfBrown}}
  \tcolorboxenvironment{remark}     {shadedenv={ShadeOfGray}}
  \tcolorboxenvironment{assumption}    {shadedenv={ShadeOfGray}}
  \tcolorboxenvironment{proof}      {shadedenv={ShadeOfGreen}}
  \tcolorboxenvironment{agda}    {shadedenv={ShadeOfPurple}}
}{
  \declaretheorem[sibling=equation]{theorem}
  \declaretheorem[sibling=theorem]{lemma,proposition,corollary,question,conjecture}
  \declaretheorem[sibling=theorem,style=definitionwithbox]{definition,attempt}
  \declaretheorem[sibling=theorem,style=definitionwithbox]{example}
  \declaretheorem[sibling=theorem,style=remarkwithbox]{remark,problem,assumption}
  }

\AtBeginEnvironment{proof}{}{}{}

\pgfplotsset{compat=1.18} 

\begin{document}

\frontmatter

\newgeometry{total={180mm,267mm}} 
\begin{titlepage}
\begin{center}
  \vspace*{\stretch{0.5}}

  \large 

  {\textsc{\Huge{Exact Real Search:} \\ \Large{Formalised Optimisation and Regression in \\ Constructive Univalent Mathematics}}\par}

  \vspace{\stretch{0.2}}

  by

  \vspace{\stretch{0.2}}

  {\huge\textsc{Todd Waugh Ambridge}}

  \vspace{\stretch{0.5}}

  A thesis submitted to the University of Birmingham for the degree of\\
  \textsc{Doctor of Philosophy}

  \vfill

  \flushright
  School of Computer Science \\
  College of Engineering and Physical Sciences \\
  University of Birmingham \\
  14\textsuperscript{th} July 2023

\end{center}
\end{titlepage}
\restoregeometry%

\chapter{Abstract}

The real numbers are important in both mathematics and computation theory.
Computationally, real numbers can be represented in several ways; most commonly using inexact floating-point data-types, but also using exact arbitrary-precision data-types which satisfy the expected mathematical properties of the reals.
This thesis is concerned with formalising properties of certain types for exact real arithmetic, as well as utilising them computationally for the purposes of search, optimisation and regression.

We develop, in a constructive and univalent type-theoretic foundation of mathematics, a formalised framework for performing search, optimisation and regression on a wide class of types.
This framework utilises \MartinEscardo’s prior work on searchable types, along with a convenient version of ultrametric spaces --- which we call closeness spaces --- in order to consistently search certain infinite types using the functional programming language and proof assistant \textsc{Agda}.

We formally define and prove the convergence properties of type-theoretic variants of global optimisation and parametric regression, problems related to search from the literature of analysis.
As we work in a constructive setting, these convergence theorems yield computational algorithms for correct optimisation and regression on the types of our framework.

Importantly, we can instantiate our framework on data-types from the literature of exact real arithmetic.
The types for representing real numbers that we use are the ternary signed-digit encodings and a simplified version of Hans-J. Boehm's functional encodings.
Furthermore, we contribute to the extensive work on ternary signed-digits by formally verifying the definition of certain exact real arithmetic operations using the \Escardo-Simpson interval object specification of compact intervals.

With this instantiation, we are able to perform our variants of search, optimisation and regression on representations of the real numbers. These three processes comprise our framework of \emph{exact real search}; we close the thesis by providing some computational examples of this framework in practice.
\newenvironment{dedication}
    {\vspace{6ex}\begin{quotation}\begin{center}\begin{em}}
    {\par\end{em}\end{center}\end{quotation}}

\begin{dedication}
\ \\ I dedicate this thesis to my father, Michael Andrew Ambridge. \\ I miss the chats that we could have had about Computer Science.
\end{dedication}
\chapter{Acknowledgements}

The writing of this thesis could not have been done without significant academic, technical, motivational and emotional input from a large number of people.
This short note cannot capture everyone that I am thankful for, and so I would simply like to thank every colleague, friend and family member who has supported me on this academic journey, which I have found both terribly challenging and hugely rewarding.

Most of all, I would first like to thank my supervisors, Professors Dan Ghica and \MartinEscardo. 
To Dan, I thank you for our many enjoyable and enlightening conversations, for always having faith in me and my work, and for always being a motivating and encouraging supervisor.
To \Martin, I thank you for educating me in constructive mathematics and type theory, for being rigorous in your expectations of me and my work, and for going completely above-and-beyond in your role as a second supervisor.
I look forward to continuing to work with both of you, as a fellow teacher and researcher, in the future.

I also thank John Longley and Ulrich Berger for agreeing to examine this thesis and for giving me an interesting, rigorous and engaging viva. Your comments have improved this thesis tenfold, and for that I am exceedingly grateful.

I next thank three peers of mine over the past four years: George Kaye, Tom de Jong and Andrew Sneap. 
George and I have been friends since our undergraduate days, and embarking on parallel academic journeys has meant I am a constant companion of his wit, charm and pedantry --- of which he has all three in great measure.
Tom has supported me ever since we started our Ph.D.s at the same time, and I'm very happy that we became great friends --- and his patience and thoughtfulness in answering my countless questions has been a tremendous help.
Andrew and I have collaborated throughout the latter half of my Ph.D., and I have found our collaboration immensely rewarding --- I have enjoyed my transition from his co-supervisor, to his colleague, to his friend.
I also want to thank these three for their contributions to this thesis: Tom for his fantastic template (sorry for butchering it with better line spacing), George for much-needed technical support and Andrew for the time we have spent in \textsc{Java} together.

Further, I thank the University of Birmingham's School of Computer Science \emph{Theory Group} for the continual help and support, and moreover the sense of community and belonging over the past four years.
The ill-named Lab Lunch has proved a supportive setting for workshopping ideas and getting feedback, but the help has been more keenly felt by the fact I can call many of you friends.
In particular, I would like to thank Ayberk Tosun, Paul Levy, Sonia Marin, Eric Finster, Anupam Das, Paul Taylor and Achim Jung.

I would further like to thank my office-mates Qamar Natsheh and Mubashir Ali, for their kind words and support especially in these last few weeks. 
In the wider academic community, I would like to thank the CCC research community, especially Alex Simpson with whom I had some very interesting discussions.
I also give my thanks to the DO\"UVK community at the School of Computer Science for their camaraderie --- thank you Jacqui, Matthew, Jon, Anna, Bruno, Tobias, Charlotte, Charlie and Yan.

Outside of my studies, I thank all of my friends who have supported me: especially my brotherhood of school friends, Callam, James, Ollie, Owen, Will and Todd, and the friends I have made since moving to Birmingham, Freddie, Fran, Charlie, Cathy, Adam and Tom. Thank you also to my family: to Jody, to Colin, to Nana Vee, and to my lovely girls Cocoa, Lola and Sookie.

There are two people left to thank, who have been there in my best and worst moments, and who have helped more than I can express in words.
To Alice, I say thank you for your encouragement and support, for your love and companionship, for putting up with me and for being my best friend --- I look forward to the next chapter of our lives together, where hopefully the word ``thesis'' does not feature so much.
Finally, to my mum, I simply say thank you for everything you have ever given me, which is everything I have.

\setcounter{tocdepth}{2}
\renewcommand*\contentsname{Table of Contents}
\tableofcontents

\mainmatter%

\chapter{Introduction}\label{chap:introduction}

The real numbers are a fundamental structure in a variety of fields such as real analysis, calculus, optimisation theory and regression analysis~\cite{GlobalIntroduction,YanBook}.
The reals are also important computationally; the re-burgeoning field of machine learning, for example, is heavily dependent on floating-point representations of real numbers, wherein they are used to model continuous parameters of artificial neural networks~\cite{NN}.
Arithmetic on floating-point numbers is incredibly, and increasingly, efficient~\cite{BFloat16}, but not without fault: floating-point methods represent real numbers as one of finitely-many dyadic rationals\footnote{A dyadic rational is a rational number of form \(\frac{n}{2^i}\).}, leading to representation and calculation errors ~\cite{FloatErrors,Floats}.
There are, however, alternative approaches to real number computation, such as interval arithmetic --- wherein functions are evaluated by their behaviour on intervals of the real numbers~\cite{Interval} --- or arbitrary-precision `\emph{exact real}' computation~\cite{Boehm86}.

Using exact real computation, any (computable) real number can be represented to any degree of precision, and manipulated correctly with respect to the represented real number~\cite{Turing,Boehm86}.
This means that, in situations where ``floating-point algorithms can lead to completely erroneous results ... exact real number computation provides guaranteed correctness"~\cite{Simpson}.
The applications of exact real computation have thus far largely focused on arithmetic; indeed, perhaps the most omnipresent use of exact reals is in the built-in \textsc{Android} smartphone calculator app~\cite{Boehm17}. In this thesis, we investigate another application of exact real computation: the construction of formally verified algorithms for search, optimisation and regression.

Search has previously been performed on exact reals. For example, Simpson uses an earlier algorithm by Berger to search for the global minimum \emph{value} of a real-valued function in the compact interval \([-1,1]\)~\cite{Simpson,BergerThesis}.
The representation of \([-1,1]\) utilised by Simpson is the \emph{ternary signed-digit encodings}, which has been extensively explored in the literature of exact real arithmetic~\cite{Plume98,Gianantonio93,BergerCoinductive}.
More recently, \Escardo~ implemented search algorithms using the ternary signed-digits in \textsc{Haskell} \cite{Escardo11fun}, as part of his wider work on \emph{searchable sets}~\cite{Escardo08}.

In this thesis, we use searchable types in order to construct algorithms for global optimisation and parametric regression on the ternary-signed digits.
Global optimisation is the problem of finding a global minimum \emph{argument} to a (usually real-valued) function in a compact interval, while parametric regression is the problem of finding parameters that fit a (usually real-valued) model function to some reference data~\cite{GlobalIntroduction,YanBook}.

In order to build formally verified algorithms, we will work in a \emph{constructive} foundation for performing mathematics and computation in tandem. 
When working constructively, a mathematician cannot rely on the law of excluded middle, proofs by contradiction or principles of omniscience\footnote{Such as LPO, WLPO or LLPO, which are used at various points in this thesis to show that something is contructively invalid.}~\cite{Bishop}. Instead, in order to show the truth of a claim, a constructivist must show exactly how it holds. In the case where the claim relates to an algorithm, they must genuinely construct that algorithm.
A constructive approach to mathematics, therefore, provides additional challenges; but there are also advantages: one could extract the algorithm from the constructive proof, and use it computationally with the knowledge it satisfies the claim.
For this purpose, we develop our work entirely within the functional programming language, and proof assistant, \textsc{Agda}.
By working in \textsc{Agda} --- which itself is based on \MartinLof~'s constructive type theory (MLTT) \cite{Agda,MLTTProgBook} --- we effectively \emph{program} formalised type-theoretic mathematical definitions, statements and proofs, and can immediately extract and run algorithms based on these structures' computational content.
Our constructive, type-theoretic formalisation is defined within \Escardo~'s \textsc{Agda} library for univalent mathematics \textsc{TypeTopology}~\cite{TypeTopology}. 

We first develop a framework for search, optimisation and regression on a wide class of types, using the work on searchable types already defined by \Escardo in \textsc{TypeTopology}~\cite{CompactTypes}.
Then, we formalise the ternary-signed digit encodings and verify many of their exact real arithmetic operations in \textsc{Agda}; similar to di Gianantonio's formal verification of the ternary signed-digits in \textsc{Coq}~\cite{Di07}. 
For this purpose, we first formalise a specification of \([-1,1]\), namely the \Escardo~-Simpson interval object~\cite{EscardoSimpson}.
We instantiate our general framework on the ternary signed-digits, allowing us to extract algorithms for search, optimisation and regression on this representation.
Following this primary instantiation, in order to investigate the framework's applicability (and whether or not in can lead us towards efficient practical implementations) we utilise another type for representing exact real numbers: a simplified definition of the Boehm encodings, which are used today --- nearly thirty-years on from their original definition --- in the \textsc{Android} calculator app~\cite{Boehm90s,BoehmAPI}.

\section{Thesis outline and key contributions}
\label{sec:outline}

In \cref{chap:mltt}, we introduce the constructive and univalent type-theoretic foundation of mathematics in which we build our formal framework for search, optimisation and regression. 
We will establish the notation of the thesis informally, based on the syntax of \textsc{Agda}.

In \cref{chap:searchable}, we review and define the two key mathematical concepts of the thesis: \emph{searchability} and \emph{uniform continuity}.
We review the literature on \emph{searchable types}, and their current status in constructive type theory as implemented by \MartinEscardo, and contribute an extension of his work in order to be able to consistently search certain infinite types in \textsc{Agda} by introducing uniform continuity on a convenient version of ultrametric spaces that we call \emph{closeness spaces}.
We give a version of the totally bounded property for closeness spaces, and show that a variety of types yield closeness spaces.
The key technical contribution of this section is the formalised result which shows these \emph{uniformly continuously searchable types} are closed under countable products (\cref{thm:tychonoff}).

In \cref{chap:generalised}, we use uniformly continuously searchable closeness spaces to define our formal convergence properties of global optimisation and parametric regression on a wide class of types.
For this purpose, we introduce \emph{approximate linear preorders}, which approximately order elements of closeness spaces.
The key contribution of this section --- the statement of the type-theoretic variants of global optimisation (\cref{th:min}) and parametric regression (\cref{reg:min,th:perfect,th:imp}) --- is methodological rather than technical, as the proofs of their convergence follow naturally from the concepts we have introduced.

In \cref{chap:reals}, we review and define within \textsc{Agda} two types for representing real numbers: ternary signed-digit encodings and ternary Boehm encodings.
On the former, we formally verify exact real arithmetic operations (namely, negation, binary and infinitary midpoint and multiplication) using the \Escardo~-Simpson interval object specification of closed intervals --- which we also review and formalise in this section.
On the latter, we define the type in \textsc{Agda}, prove the correctness of its structure and show how it yields representations of compact intervals that we can then use for search.
The key technical contributions of this section are:
\begin{itemize}
\item The \textsc{Agda} formalisation of the \Escardo-Simpson interval object specification of closed intervals (\cref{sec:interval-object}),
\item The \textsc{Agda} formalisation of the ternary signed-digit encodings and their aforementioned arithmetic operations (\cref{sec:signed-digits}),
\item The formal verification in \textsc{Agda} that these operations for exact real arithmetic on the ternary-signed digits are correct with respect to the specification of those operations on the interval object (\cref{lem:neg-realise,cor:mid-realise,thm:bigmid-realiser,thm:mul-realiser}).
\end{itemize}

In \cref{chap:exact-real-search}, we bring our formal framework full-circle by instantiating it on these two types for representing real numbers. 
Example evaluations of algorithms for search, optimisation and regression --- either extracted from \textsc{Agda} or implemented in \textsc{Java} --- are then given to show the use of the framework in practice.
A contribution of this section is the formal result that the arithmetic operations we define on the ternary signed-digit encodings are uniformly continuous (\cref{sec:K-ucont}).

Finally, in \cref{chap:conclusion}, by way of conclusion we discuss some further avenues for this line of work. 

\subsection{Reading the formal proofs of this thesis}
\label{sec:symbols}

A chief contribution of this thesis is that most of its mathematical content --- both from the literature and our own contributions --- is formalised in \textsc{Agda} within the library \textsc{TypeTopology}~\cite{TypeTopology}.
We describe the constructive and univalent philosophy of \textsc{TypeTopology} in the intoduction to \cref{chap:mltt}.

The reader is invited to explore our \textsc{Agda} formalisation by `clicking' on the symbols at the top left of each mathematical environment. 
This will take them directly to the \textsc{Agda} function which formalises that definition or proof.
The different files of the library are described in \cref{appendix:agda}.

For the purpose of clarity, we use three different symbols:
\begin{itemize}
\item The \emph{library book} symbol \includesvg[height=9pt]{symbols/book.svg} denotes a statement we are recalling from the literature and which we have \emph{not} formalised, nor does a formalisation appear within \textsc{TypeTopology},
\item The \emph{topological donut} \includesvg[height=9pt]{symbols/mug.svg} symbol denotes a statement formalised within \textsc{TypeTopology}, usually by \MartinEscardo, but sometimes by another collaborator,
\item The \emph{rune of Gandalf} \includesvg[height=9pt]{symbols/gandalf.svg} symbol denotes a statement formalised for this thesis by the author. Sometimes, this will be a repetition from the literature or \textsc{TypeTopology}, but in other cases this will be one of our main contributions, as outlined in \cref{sec:outline}.
\end{itemize}
\chapter{Constructive Univalent Type Theory via \textsc{Agda}}
\label{chap:mltt}

We wish to perform mathematics in a way that supports computer programming by default, in order to extract computational content from our mathematical proofs.
As discussed in \cref{chap:introduction}, this means that we will be utilising a \emph{constructive} approach to mathematics. 
But we want to go further; we do not wish to just support programming but to actively program formal mathematics.
For this purpose, we work formally in a constructive and univalent foundation of mathematics: the variation of \MartinLof~constructive type theory (MLTT) provided by the functional programming language and proof assistant \textsc{Agda}~\cite{MLIntro,Agda}.
More specifically, we work in \textsc{TypeTopology}, an \textsc{Agda} library by \MartinEscardo~and a growing number of collaborators interested in formalising both new and previous theorems in univalent mathematics~\cite{TypeTopology}. 
This thesis is second only to Tom de Jong's recent thesis on \emph{Domain Theory in Constructive and Predicative Univalent Foundations} in having the majority of its results formalised within \textsc{TypeTopology}~\cite{TDJ}.

The philosophy of this thesis is aligned with that of \textsc{TypeTopology}, the full extent of which is given on the library's webpage~\cite{TypeTopology}.
In particular, we work in a small version of MLTT --- which we introduce in \cref{sec:mltt} --- and we adopt the univalent approach to mathematics introduced by Voevodsky and popularised by \emph{The HoTT Book}~\cite{HoTTBook}. This latter point means that, even when not invoking the univalence axiom itself, we utilise the terminology and the perspective of univalent mathematics, which we detail in \cref{sec:univalent}.
Our framework is hence compatible with other formalisations of univalent mathematics, such as the \textsc{UniMath} library for \textsc{Coq}~\cite{UniMath} --- unlike \textsc{UniMath} however, we do not assume the propositional resizing axiom. When we \emph{do} use axioms such as function extensionality or univalence, we make them explicit parameters to those proofs or modules which use them.
Finally, we restrict ourselves to those features of \textsc{Agda} which allow us to remain consistent, and avoid inconsistent assumptions such as the `type-in-type' axiom (which assumes the type of all types is an element of itself).

We begin this chapter with a brief introduction to MLTT, before recalling concepts of this theory and explaining how they are written in general \textsc{Agda}. 
We then introduce the aspects of univalent mathematics that we utilise, before recalling some fundamental concepts that are used throughout the following chapters of the thesis.

\section{A brief introduction to type theory}
\label{sec:intro-type}

A \emph{type theory} is a set of rules for formally reasoning about the behaviour of a \emph{system of terms} such as a logical calculus, foundation of mathematics, philosophical theory, or programming language~\cite{PierceBook}.
Informally, the rules assign to each term a \emph{type} that is used to determine the \emph{behaviour} of such terms, such as which methods of the theory can manipulate them. 
These methods, such as \emph{functions}, are themselves terms of more complicated types formed from simpler types using \emph{type families}.

Early type theories were developed by Russell in the 1900s as foundations of mathematics alternative to those built within naive set theory, which he famously showed yielded inherent paradoxes~\cite{PM}.
Soon thereafter, Church's \emph{simply typed lambda-calculus} utilised Ramsey's \emph{simple theory of types} to ensure that each term of the calculus is well-typed~\cite{ChurchSTLC}.
Later, in 1972, Per \MartinLof~introduced MLTT as a ``full scale system for formalising intuitionistic mathematics"~\cite{MLIntro}.
\MartinLof's intention to fully-formalise Bishop-style constructive mathematics (discussed in~\cref{chap:introduction}) had clear motivations: the internal computation rules of type theories are happily married to the ability to extract computational content from constructive proofs.

In MLTT, a mathematical proposition is interpreted as a type. 
An \emph{element} (i.e.\ a \emph{program}) of such a type is considered a proof of the proposition, and can therefore be computed to yield a constructive witness of that proof.
These types are built from type families which interpret the connectives of intuitionistic logic, and are often \emph{dependent} on types that model more fundamental mathematical structures (such as the natural numbers) or, indeed, other propositions.
This is known as the \emph{propositions as types} interpretation~\cite{HoTTBook}.
This interpretation means that two proofs \(\ty{a_1,a_2}{A}\) of the same mathematical statement may be very different programs, and this difference is often relevant to proving later corollaries of that statement --- this is called \emph{proof relevance}.

The most significant dependent type family of MLTT is the \emph{identity type family}, which is used to interpret statements about term equality. We can say that two elements \(\ty{a_1,a_2}{A}\) are equal if we can construct an element of type \((a_1 = a_2)\)~\cite{MLIntro}.
This makes our foundation of mathematics very rich: we can begin to reason directly about the equality of mathematical objects/programs in the way that a mathematician/programmer would naturally do so.
Matching equality in MLTT more and more closely to a natural notion of mathematical equality is the primary concern of the field of \emph{univalent type theory}~\cite{HoTTBook}; an area that this thesis operates within, and which we shall return to later in \cref{sec:univalent}.

\section{\textsc{Agda} notation for constructive type theory}
\label{sec:mltt}

There are already a variety of formal introductions to MLTT's types and type families, from various perspectives and levels of introduction (for example, see~\cite{MLNotes,Gambino} or the Appendices of~\cite{HoTTBook}).
We assume the reader has some familiarity with type theory, and in this section instead recall the constructions and concepts of MLTT and show how to write them in \textsc{Agda}.
In the remainder of the thesis, however, we will write mathematics informally, in the style of \emph{The HoTT Book}~\cite{HoTTBook}.
This section can therefore be thought of as the documentation required to be able to read this thesis' \textsc{Agda} formalisation.

\subsection{Type universes}
\label{sec:universe}

Recall that a type universe is a type whose elements are themselves types~\cite{MLIntro}. MLTT utilises a countable sequence of \emph{type universes}, \(\U_0 : \U_1 : \U_2, ...\); the integer indices of these type universes are sometimes called \emph{universe levels}~\cite{EscardoUniverses}.

In \textsc{Agda}, type universes are explicitly stratified into countably-many universe levels. There is therefore a lowest level \texttt{lzero} and a successor operation \texttt{lsuc}. This successor function forms a semilattice, whose join is given as the binary operation \verb|_|$\ssqcup$\verb|_|.

\begin{agda}
lzero : Level
lsuc  : (l : Level) $\sto$ Level
_$\ssqcup$_    : (l$_1$ l$_2$ : Level) $\sto$ Level
\end{agda}

\noindent
As \texttt{Level} is a semilattice, the supremum function is associative (\(u \ \ssqcup \ (v \ \ssqcup \ w) \equiv (u \ \ssqcup \ v) \ \ssqcup \ w\)), commutative (\(u \ \ssqcup \ v \equiv v \ \ssqcup \ u\)) and idempotent (\(u \ \ssqcup \ u \equiv u\)) --- furthermore, the successor operation is a homomorphism (\(\texttt{lsuc} \ (u \ \ssqcup \ v) \equiv \texttt{lsuc} \ u \ \ssqcup \ \texttt{lsuc} \ v\)). Note that the preceding equations are \emph{definitional equalities} and not equalities induced by the identity type former (introduced in \cref{sec:id}).

Each universe level \verb|i| is mapped to a type universe \texttt{Set\textsubscript{i}} such that \(\ty{\texttt{Set\textsubscript{i}}}{\texttt{Set\textsubscript{lsuc i}}}\), meaning there is indeed, as in MLTT, a countable sequence of type universes.
Furthermore, there is a type universe \texttt{Set\(\omega\)} that is larger than every type universe, but is not part of the hierarchy --- for example, the expression \((\ty{i}{\texttt{Level}}) \to \texttt{Set}_i\) has type \texttt{Set\(\omega\)}~\cite{AgdaTypeUniverses}.

The notation of type universes in \textsc{TypeTopology} aligns the above closer to that of MLTT~\cite{TypeTopologyUniverses}:

\begin{agda}
open import Agda.Primitive public
  using (_$\ssqcup$_)
  renaming (lzero to $\sU_0$
          ; lsuc to $\_^+$
          ; Level to Universe
          ; Set$\omega$ to $\sU_\omega$
          )
\end{agda}

\noindent
\verb|Set| is renamed to \verb|Type|, and universes are mapped to type universes using the `dot function'.
We prefer this notation as not all types are sets in MLTT, and furthermore the very small difference in syntax between universes \(\U\) and type universes \(\U^\cdot\) allows the reader of a \textsc{TypeTopology} file to align these concepts in their mind, as we do when working in MLTT~\cite{mgs}.

Formally, we follow these conventions of \textsc{TypeTopology}; informally, universes are implicit in most of our work. Where necessary, we use \(\U,\V,...\) to range over type universes.

\subsection{Function and \texorpdfstring{\(\Pi\)}{Π}-types}

\emph{Functions} are built-in to \textsc{Agda}, and interpret statements that are universally quantified over a type.
We give the curried projection functions for every type (in every universe) as example function definitions:

\begin{agda}
fst     : {X : $\sUd$} {Y : $\sVd$} $\to$ X $\to$ Y $\to$ X
fst {$\sU$} {$\sV$} {X} {Y} x y = x

snd     : {X : $\sUd$} {Y : $\sVd$} $\to$ X $\to$ Y $\to$ Y
snd {$\sU$} {$\sV$} {X} {Y} x y = y
\end{agda}

\noindent
Note that implicit arguments (i.e.\ the types and the universes they are elements of) are given in braces in the type definition, and can be accessed using braces in the function definition.

\emph{Dependent functions} (elements of \(\Pi\)-types), where the type of the output depends on that of the input, are also built-in using the same syntax.
In \textsc{TypeTopology}, we match this syntax to that of MLTT by defining a universe-valued dependent function \(\Pi\) as below (although, for brevity of code, we almost always utilise the usual \textsc{Agda} notation of \(\Pi\)-types):

\begin{agda}
$\Pi$ : {X : $\sUd$} (Y : X $\sto$ $\sVd$) $\sto$ ${\sU \ \ssqcup \ \sV}^\cdot$
$\Pi$ {$\sU$} {$\sV$} {X} Y = (x : X) $\sto$ Y x
\end{agda}

\noindent
Elements of \(\Pi\)-types \(\ty{f}{\Pity{x}{X}{Y \ x}}\) are therefore dependent functions that, on input of \(\ty{x}{X}\), output \(\ty{f(x)}{Y \ x}\) --- usual functions are just dependent functions where the type family \(Y\) is constant. 

\subsection{Natural number and integer types}

The \emph{natural numbers} are defined inductively using the \verb|data| keyword:

\begin{agda}
data $\N$ : ${\sU_0}^\cdot$ where
 zero : $\N$
 succ : $\N$ $\sto$ $\N$
\end{agda}

A type defined by \verb|data| is given by a list of its \emph{constructors}; here, \(\N\) is defined using the Peano inductive definition.
Note, therefore, that the constructor \(\ty{\mathsf{succ}}{\N \to \N}\) is itself an \textsc{Agda} function from natural numbers to natural numbers.
We will write naturals in the expected shorthand way, e.g.\ \(0 := \mathsf{zero}\) and \(n +1 := \mathsf{succ} \ n\). 

In \textsc{Agda}, functions on inductive types can be defined by \emph{pattern matching}, i.e.\ by considering the function's output for each of its constructors.
As an example, we define addition on the natural numbers as an infix operation by pattern matching:

\begin{agda}
_+$\N$_ : $\N$ $\sto$ $\N$ $\sto$ $\N$
0      +$\N$ y = y
succ n +$\N$ y = succ (n +$\N$ y)
\end{agda}

The \emph{integers} \(\Z\) are also defined using \verb|data|. The two constructors are for negative and non-negative numbers:

\begin{agda}
data $\Z$ : ${\sU_0}^\cdot$ where
 pos     : $\N$ $\sto$ $\Z$
 negsucc : $\N$ $\sto$ $\Z$
\end{agda}

\noindent
Again, we will write integers in the expected shorthand way; e.g.\ \(0 := \pos 0\), and \(-6 := \negsuc 5\).

\subsection{Unit, empty and negated types}

Recall, by the propositions-as-types interpretation, that any \emph{pointed} type (i.e.\ a type that we can exhibit an element of) represents the trivially true proposition.
The \emph{unit type} \(\1\), which has one constructor and which we often use for representing truth, is defined as follows:

\begin{agda}
data $\1$ {$\sU$} : $\sUd$ where 
  $\sstar$ : $\1$
\end{agda}

\noindent
The \emph{empty type} \(\0\), which has no constructors, can also be defined using \verb|data|:

\begin{agda}
data $\0$ {$\sU$} : $\sUd$ where 
\end{agda}

\noindent
Note that there are unit and empty types in every universe.

Recall that negation is interpreted in MLTT by empty-valued functions. We therefore define the type family \(\neg X\) given any type \(X\) as follows:

\begin{agda}
$\neg$ : (X : $\sUd$) $\sto$ $\sUd$
$\neg$ {$\sU$} X = X $\sto$ $\0$
\end{agda}

\noindent
Exhibiting an element of \(\neg X\) therefore amounts to defining a function that proves any element of \(X\) leads to false.

In the opposite direction, by the \emph{principle of explosion} any proof of false itself implies anything.
This can be defined in \textsc{Agda} as a dependent function by pattern matching:

\begin{agda}
$\0$-elim : {X : $\sUd$} $\sto$ $\0$ $\sto$ $X$
$\0$-elim {$\sU$} {$X$} ()
\end{agda}

The special pattern \verb|()| is used here by \textsc{Agda} to denote that there is no pattern for that element (i.e.\ the element of type \(\0\)), and thus we do not have to define the function for that case.

\subsection{Disjoint union of types}

Given two types, we can form their \emph{disjoint union type}, which has one constructor for each side of the disjunction:

\begin{agda}
data _+_ {$\sU$} {$\sV$} (X : $\sUd$) (Y : $\sVd$) : ${\sU \ \ssqcup \ \sV}^\cdot$ where
 inl : X $\sto$ X + Y
 inr : Y $\sto$ X + Y
\end{agda}

These types represent disjoint statements in intuitionistic logic, as we must provide a witness as to which side of the disjunction holds. 
Furthermore, the form of disjunction that disjoint union types represent is \emph{inclusive}, and yields a direct answer as to which of the two statements holds --- even if both could feasibly hold, the element can only reduce to one side of the disjunction.

Disjoint unions are therefore in general \emph{structure} and not \emph{property}.
Indeed, we often use them in the definition of mathematical structures; as an example, the three-point type can be given by a disjunction of the unit type:

\begin{agda}
$\3$ : $\sUzd$
$\3$ = $\1$ + $\1$ + $\1$
\end{agda}

We will return to discussion of properties versus structure in \cref{sec:props}.

\subsection{Binary product and \texorpdfstring{\(\Sigma\)}{Σ}-types}
\label{sec:sigma}

Given two types, we can form their \emph{binary product type}, whose elements are pairs. 
We define these as non-dependent versions of \emph{dependent pairs} (elements of \(\Sigma\)-types), where the type of the second projection of the pair depends on that of the first.
\(\Sigma\)-types are defined coinductively using the \verb|record| keyword:

\begin{agda}
record $\Sigma$ {$\sU$} {$\sV$} {X : $\sUd$} (Y : X $\sto$ $\sVd$) : ${\sU \ \ssqcup \ \sV}^\cdot$ where
  constructor _,_
  field
   pr$_1$ : X
   pr$_2$ : Y pr$_1$
\end{agda}

\noindent
Then, binary product types can be easily be defined as following type family:

\begin{agda}
_$\sx$_ : $\sUd$ $\sto$ $\sVd$ $\sto$ ${\sU \ \ssqcup \ \sV}^\cdot$
X $\sx$ Y = $\Sigma$ {$\sU$} {$\sV$} {X} ($\lambda$x.Y)
\end{agda}

A type defined by \verb|record| is given by a list of its \emph{fields}; in this way, records are \textsc{Agda}'s built-in version of \(\Sigma\)-types.

Recall that, as mathematical statements, \(\Sigma\)-types interpret constructive existence: proving there is an element \(\ty{a}{A}\) which satisfies the statement \(\ty{B(a)}{\V}\), for a type \(\ty{A}{\U}\) and type family \(\ty{B}{A \to \V}\), is performed by constructing an element of type \(\sigmaty{a}{A}{B(a)}\).
In general, this element can be defined in multiple ways, yielding (by proof relevance) multiple proofs of existence.

\(\Sigma\)-types are therefore in general \emph{structure} and not \emph{property}. They are used to define both collections and the more restricted notion of subtypes. As an example, we give below the collection of even natural numbers \(\N_e\):

\begin{agda}
is-even : $\N$ $\sto$ $\sUzd$
is-even 0 = $\1$
is-even 1 = $\0$
is-even (succ (succ n)) = is-even n

$\N\textsubscript{e}$ : $\sUzd$
$\N\textsubscript{e}$ = $\Sigma$ n : $\N$ , is-even n
\end{agda}

\noindent
An element of \(\N_e\) is a dependent pair: a number \(\ty{n}{\N}\) and a proof of evenness of type \(\mathsf{is{\hy}even} \ n\).
Later, in \cref{sec:props} we will see that \(\N_e\) is in fact a subtype of \(\N\): informally, this is because for any given integer \(n\) there is only one proof of evenness (i.e.\ there is only one element of the type \(\mathsf{is{\hy}even} \ n\)).

\subsection{Identity types}
\label{sec:id}

Every type has an associated family of \emph{identity types}.
Recall that there is only a single constructor for identity types in MLTT --- so too in \textsc{Agda}:

\begin{agda}
data Id {$\sU$} (X : $\sUd$) : X $\sto$ X $\sto$ $\sUd$
  refl : (x : X) $\sto$ Id X x x
\end{agda}

The function \verb|refl| allows us only to construct elements of the type \verb|Id X x x|, which we write in \textsc{TypeTopology} as \verb|x = x| for any \verb|x : X|\footnote{Note that we actually use a Unicode `long equals' symbol in \textsc{TypeTopology}, as \(=\) is reserved in \textsc{Agda} for definitions.}, leaving the type of \verb|x| implicit.

By default in \textsc{Agda}, it is the case therefore that Streicher's K axiom is \emph{assumed}: \(\mathsf{refl} \ x\) is the only element of \(x = x\)~\cite{Streicher}.
However, as this cannot be \emph{proved} in MLTT (or \textsc{Agda}), we disable its use in \textsc{TypeTopology}.
But, it \emph{is} the case that \MartinLof's induction rule for identity types only requires the consideration of elements constructed by \(\mathsf{refl}\)~\cite{MLNotes}. \MartinLof~called this rule \(J\)~\cite{mgs}:

\begin{agda}
J : (X : $\sUd$) (A : (x y : X) $\sto$ x = y $\sto$ $\sVd$)
  $\sto$ ((x : X) $\sto$ A x x (refl x))
  $\sto$ (x y : X) (p : x = y) $\sto$ A x y p
J {$\sU$} {$\sV$} X A f x x (refl x) = f x
\end{agda}

\noindent
From the above, we see that \(J\) is defined in \textsc{Agda} by pattern matching on the element of the identity type \(\ty{p}{x = y}\).
Once the identity type is pattern matched, \(x\) and \(y\) are aligned definitionally in \textsc{Agda} and can be used interchangeably.

Elements of other identity types must be derived from the expected properties of identities,
which are defined in \textsc{Agda} as functions either by using \(J\) or by directly pattern matching.
We prove the four key identity rules --- symmetry, transitivity, function application, and transport --- by pattern matching.

\begin{agda}
sym : {X : $\sUd$} {Y : $\sVd$} (x y : X) $\sto$ x = y $\sto$ y = x 
sym {$\sU$} {$\sV$} {X} {Y} x x (refl x) = refl x

trans : {X : $\sUd$} {Y : $\sVd$} (x y z : X) $\sto$ x = y $\sto$ y = z $\sto$ x = z 
trans {$\sU$} {$\sV$} {X} {Y} x x x (refl x) = refl x

ap : {X : $\sUd$} {Y : $\sVd$} (f : X $\sto$ Y) (x y : X) $\sto$ x = y $\sto$ f x = f y 
ap {$\sU$} {$\sV$} {X} {Y} f x x (refl x) = refl (f x)

transport : {X : $\sUd$} (A : X $\sto$ $\sVd$) (x y : X) $\sto$ x = y $\sto$ A x $\sto$ A y 
transport {$\sU$} {$\sV$} {X} A x x (refl x) = id
\end{agda}

\noindent
We can consider all of the above functions as \emph{proofs} about identities.
Note that when working informally in this thesis we do not detail when we use these rules.

\section{Univalent mathematics in \textsc{TypeTopology}}
\label{sec:univalent}

Now that we have recalled the key constructions of MLTT using \textsc{Agda}, we change to a more informal approach to mathematics for the remainder of the thesis.
This approach will be a mixture of English statements and proofs, peppered with mathematical notation that resembles \textsc{Agda} code.
If the reader wishes to read the formalisation of any particular definition or proof, they can be easily accessed by the method described in \cref{sec:symbols} (i.e.\ by simply clicking the symbol at the top of the environment).

In this section, we recall principles of the burgeoning field of univalent foundations of mathematics --- which is sometimes called \emph{univalent type theory} or \emph{homotopy type theory} --- from the point of view of their definition in \textsc{TypeTopology}.
In order to get the reader used to our informal notation, we spell out many of the recalled definitions and proofs of this section in English and \textsc{Agda}-like syntax.

Four principles that we will introduce in this section that are worth mentioning in advance are \emph{function extensionality} (labelled \axioms{f}), \emph{propositional extensionality} (labelled \axioms{p}), \emph{univalence} (labelled \axioms{u}) and \emph{propositional truncation} (labelled \axioms{t}). These are all axioms in MLTT, and thus we want to be especially clear about where they are used in this thesis. If a proof uses one of these four axioms, we will mark it with its corresponding label --- for example, a proof labelled \axioms{fp} invokes both function extensionality and propositional extensionality.

\subsection{Function extensionality}

Two (possibly dependent) functions are equivalent if they are behaviourally indistinguishable; i.e.\ if they are \emph{pointwise-equal},

\begin{definition}
\mbox{}
\typetop{MLTT}{Id}{_\urlsim_}
Two functions \(\ty{f,g}{\pity{x}{X}{Y(x)}}\) are \emph{pointwise-equal} \(f \sim g\) if \(f(x) = g(x)\) holds for every \(\ty{x}{X}\):
\[ f \sim g := \pitye{x}{X}{f(x) = g(x)}. \]
\end{definition}

\noindent
Of course, equal functions are pointwise equal.

\begin{lemma}
\label{lem:pe-refl}
Every function is pointwise-equal to itself.
\end{lemma}
\begin{proof}
The proof term is constructed immediately the identity type's constructor (i.e. reflexivity):
\begin{alignat*}{3}
\simrefl &: \Pity{f}{\pitye{x}{X}{Y(x)}}{f \sim f} ,\\
\simrefl & (f,x) := \mathsf{refl} \ (f(x)).
\end{alignat*}
\end{proof}

\begin{lemma}
\typetop{UF}{Base}{happly}
Given functions \(\ty{f,g}{\pitye{x}{X}{Y(x)}}\) such that \(f = g\), then \(f \sim g\).
\end{lemma}
\begin{proof}
We utilise the given equality \(\ty{e}{f = g}\) to build the function \\ \[\ty{\mathsf{transport} (\lambda(\ty{h}{\Pitye{x}{X}{Y(x)}}).f \sim h,e)}{f \sim f \to f \sim g}.\]
Thus, by applying \cref{lem:pe-refl}, the result follows:
\begin{alignat*}{3}
\mathsf{happly} &: \Pity{f,g}{\pitye{x}{X}{Y(x)}}{(f = g) \to (f \sim g)} ,\\
\mathsf{happly} & (f,g,e,x) := \mathsf{transport} (\lambda(\ty{h}{\Pitye{x}{X}{Y(x)}}).f \sim h,e,\simrefl(X)).
\end{alignat*}
\end{proof}

There is no way to prove (for non-trivial types) that pointwise-equal functions are equal in MLTT.
In order to be able to treat equivalent functions as equal functions, we can add an extra axiom called \emph{(naive) function extensionality} (later, in \cref{def:funext} we will introduce the non-naive version).

\begin{definition}
\label{def:naive-fe}
\typetop{UF}{FunExt}{naive-funext}
We say that \emph{naive function extensionality} holds for given universes \(\U,\V\) when, given \(\ty{X}{\U}\), \(\ty{Y}{X \to \V}\) and functions \(\ty{f,g}{\pitye{x}{X}{Y(x)}}\), pointwise-equality implies equality:
\[ \mathsf{naive{\hy}funext}(\U,\V) := \pitye{X}{\U} \pitye{Y}{X \to \V} \pity{f,g}{\pitye{x}{X}{Y(x)}}{f \sim g \to f = g} \]
\end{definition}

Although we accept it as a mathematical statement and would like to use it when reasoning about certain functions, function extensionality is independent of MLTT; meaning it is consistent with the type theory but cannot be proved by it.
As it is consistent, one can add it as an axiom to the theory and use it in their proofs --- in \textsc{Agda}, this amounts to postulating the existence of a proof term \(\left(\ty{\mathsf{fe}}{\mathsf{naive{\hy}funext}}\right)\). If one does this, however, one must be careful about where the postulated term is used. Postulated terms have no computational interpretation; while they may help us to write a proof, we will be unable to extract any data from said proof. In this thesis, function extensionality is sometimes assumed in order to show that our results are correct with regards to certain specifications --- but we never assume it when computing those results themselves.

\subsection{Propositions and unique proofs}
\label{sec:props}

Under the propositions as types interpretation described in \cref{sec:intro-type}, the types of MLTT can be viewed as mathematical statements, with elements of those types viewed as proofs of the statement. Hence, by exhibiting different elements of the same type, the truth of a statement can be given in multiple ways.
Univalent type theory distinguishes statements that can be true in multiple ways and those that can only be true in one way, calling the latter \emph{(h)-propositions}.
A type interprets a proposition if it can only be true in one way; i.e.\ if every element of the type is identical. These types are sometimes referred to as \emph{subsingletons}, though we simply call them \emph{propositions}. 
If we agree that propositions are the only types that truly interpret a mathematical proposition, then we are aligning ourselves with the \emph{propositions as some types} interpretation of dependent type theory~\cite{PropAsSome}.

\begin{definition}
\label{def:proposition}
\typetop{UF}{Subsingletons}{is-prop}
A type \(X\) is a \emph{proposition} if all of its elements are equal:
\[ \mathsf{is{\hy}prop}(X) := \pity{x,y}{X}{x = y} .\]
\end{definition}

\begin{lemma}
\typetop{UF}{Subsingletons}{\urlone-is-prop}
\(\1\) is a proposition.
\end{lemma}
\begin{proof}
\(\1\) only has one element, \(\ty{\star}{\1}\); thus, it is trivial to prove that any two \(\ty{x,y}{\1}\) are equal, as \(x = \star = y\).
\end{proof}

\begin{lemma}
\typetop{UF}{Subsingletons}{\urlzero-is-prop}
\label{lem:0-prop}
\(\0\) is a proposition.
\end{lemma}
\begin{proof}
\(\0\) has no elements; thus, given any two \(\ty{a,b}{\0}\) we can immediately complete the vacuous proof using \(\0\mathsf{{\hy}elim}\).
\end{proof}

\noindent
But it is not the case that these are the only propositions. For every proposition \(X\), it is independent that \((X \simeq \1) + (X \simeq \0)\) --- a proof of this would be logically equivalent to the classical law of excluded middle for propositions.

\(+\)-types that interpret exclusive-or statements preserve the proposition property, meaning the type is now a \emph{property} rather than \emph{structure}.

\begin{lemma} \label{lem:plus-prop}
\typetop{UF}{Subsingletons}{\urlplus-is-prop}
If \(\ty{X}{\U}\) and \(\ty{Y}{\V}\) are propositions, and \(\neg (X \x Y)\), then \(\ty{X + Y}{\U \sqcup \V}\) is a proposition.
\end{lemma}
\begin{proof}
\begin{alignat*}{3}
\plusprop &: \isprop{X} \to \isprop{Y} &&\to \neg (X \x Y) \to \isprop{(X + Y)} \\
\plusprop &(p,q,f,\inl \ x_1, \inl \ x_2) &&:= \mathsf{ap}(\inl,p(x_1,x_2)) \\
\plusprop &(p,q,f,\inl \ x, \inr \ y) &&:= \zelim{(f(x,y))} \\
\plusprop &(p,q,f,\inr \ y, \inl \ x) &&:= \zelim{(f(x,y))} \\
\plusprop &(p,q,f,\inr \ y_1, \inr \ y_2) &&:= \mathsf{ap}(\inr,q(y_1,y_2))
\end{alignat*}
\end{proof}

\noindent
\(\Pi\)-types, \(\neg\)-types and \(\Sigma\)-types also preserve the proposition property.

\begin{lemma} \label{lem:pi-prop}
\typetop{UF}{Subsingletons-FunExt}{\urlPi-is-prop}
Given \(\ty{X}{\U}\) and \(\ty{Y}{X \to \V}\) where every \(\ty{Y(x)}{\V}\) is a proposition, then \(\ty{\Pi Y}{\U \sqcup \V}\) is a proposition.
\end{lemma}
\begin{proof} \axioms{f}
We need to show that given \(\ty{f,g}{\prod Y}\) we have \(f = g\). 
We have \(f \sim g\) because given any \(\ty{x}{X}\), \(Y(x)\) is a proposition (so therefore \(f(x) = g(x)\)).
The result immediately follows by function extensionality.
\end{proof}

\begin{corollary} \label{lem:neg-prop}
\typetop{UF}{Subsingletons-FunExt}{negations-are-props}
The negation of every type is a proposition.
\end{corollary}
\begin{proof} \axioms{f}
By \cref{lem:pi-prop,lem:0-prop}.
\end{proof}

\begin{lemma} \label{lem:sigma-prop}
\typetop{UF}{Subsingletons}{\urlSigma-is-prop}
Given a proposition \(\ty{X}{\U}\) and \(\ty{Y}{X \to \V}\) where every \(\ty{Y(x)}{\V}\) is a proposition, then \(\ty{\Sigma Y}{\U \sqcup \V}\) is a proposition.
\end{lemma}
\begin{proof}[Proof (Sketch).]
The idea is that, given \(\ty{(a,p)}{\Sigmatye{a}{X}{Y(a)}}\) and \(\ty{(b,q)}{\Sigmatye{b}{X}{Y(b)}}\), we want to show that conclude \((a,p) = (b,q)\) by first showing \(a = b\) and \(p = q\) --- but the first and last equations here do not type check, as the identity type requires both arguments to be of the same type.
As \(X\) is a proposition, however, we have a proof that \(a = b\), and therefore \(Y(a)\) and \(Y(b)\) \emph{are} indeed the same type (we will arbitrarily choose to call this type \(Y(a)\)).
The result then follows by the fact that \(Y(a)\) is a proposition and therefore \(p = q\).
\end{proof}

If a type \(P\) is a proposition, any proof \(\ty{p}{P}\) is \emph{unique}. 
For example, the proof that an integer is even (defined in \cref{sec:sigma}) is unique, because each of its cases is a proposition.
Hence, \(\N_e\) (which we also defined in \cref{sec:sigma}) is a \emph{subtype} of \(\Z\); i.e.\ it is a unique collection of elements of \(\N\). 
Related to this, we introduce the idea of \emph{embeddings}, which we recall are functions \(\ty{f}{X \to Y}\) such that each \(\ty{y}{Y}\) is reached by at most one \(\ty{x}{X}\).

\begin{definition}
\label{def:embedding}
\typetop{UF}{Embeddings}{is-embedding}
Given types \(X\) and \(Y\), a function \(\ty{f}{X \to Y}\) is an \emph{embedding} if each fiber of \(f\) is a proposition. We write this as follows:
\[ \mathsf{is{\hy}embedding}(f) := \pitye{y}{Y}{\mathsf{is{\hy}prop}\left(\sigmaty{x}{X}{f x = y}\right)} .\]
\end{definition}

\noindent
As an example, the function \(\ty{\mathsf{pr}_1}{\sigmatye{x}{X}{Y(x)} \to X}\) is an embedding if \(Y\) is a proposition; also, \(\ty{\mathsf{inl}}{X \to X + Y}\) and \(\ty{\mathsf{inr}}{Y \to X + Y}\) are embeddings.

The subtype that encapsulates all interpretations of truth values in a given universe is that universe's \emph{type of truth values}.

\begin{definition}
\typetop{UF}{SubtypeClassifier}{\urlOmega}
The \emph{type of truth values} for a universe \(\U\) is the type,
\[ \Omega_\U := \sigmaty{P}{\U}{\mathsf{is{\hy}prop} \ P} .\]
\end{definition}

\noindent
Note that we usually leave the particular universe implicit, and simply write \(\Omega\).
In this thesis, we use truth-valued functions often, without explicitly proving (via the above lemmas) that the values of the function are propositions.
In the formalisation, this requirement adds additional, but usually straightforward, complexity to the proofs.

\subsection{Propositional extensionality}

Logical equivalence between types is defined in the usual way.

\begin{definition}
\typetop{Notation}{General}{_\urlLeftrightarrow_}
Two types \(X\) and \(Y\) are \emph{logically equivalent} \(X \leftrightarrow Y\) if \(X \rightarrow Y\) and \(Y \rightarrow X\).
\end{definition}

\noindent
However, given our discussions in the previous subsection, we would like a distinction of this notion for propositions.

\begin{definition}
A proposition \(\ty{(P,p)}{\Omega_\U}\) implies a proposition \(\ty{(Q,q)}{\Omega_\V}\), written \((P,p) \Rightarrow (Q,q)\), if \(P \to Q\).
\end{definition}

\begin{definition}
\label{def:prop-equiv}
Two propositions \(\ty{P}{\Omega_\U}\) and \(\ty{Q}{\Omega_\V}\) are \emph{propositionally equivalent} \(P \Leftrightarrow Q\) if \(P \Rightarrow Q\) and \(Q \Rightarrow P\).
\end{definition}

\noindent
With this definition, which we renamed \emph{propositional equivalence}, we now correctly model logical equivalences for types that interpret propositions.

As with function extensionality, we can assume that this implies identification.

\begin{definition}
\typetop{UF}{Subsingletons}{propext}
We say that \emph{propositional extensionality} holds for a given universe \(\U\) when, given \(\ty{P,Q}{\Omega_\U}\), propositional equivalence \(P \Leftrightarrow Q\) implies equality \(P = Q\),
\[ \mathsf{propext}(\U) := \pity{P,Q}{\Omega_\U}{P \Leftrightarrow Q \to P = Q} \]
\end{definition}

\subsection{Type equivalences}
\label{sec:equiv}

Informally, and in general, two mathematical or computational structures are considered to be \emph{equivalent} if they can be interchanged with each other in any context.
Indeed, in (homotopy) type theory, two types are equivalent if their elements can be interchanged within a program without changing that program's behaviour.
An early attempt at formalising this idea type-theoretically results in \emph{quasi-inverses}.

\begin{definition}
\typetop{UF}{Equiv}{qinv}
Given types \(X\) and \(Y\), a function \(\ty{f}{X \to Y}\) is a \emph{quasi-inverse} if there is a function \(\ty{g}{Y \to X}\) such that either composition of \(f\) and \(g\) is the identity function:
\[\mathsf{is{\hy}qinv}(f) := \sigmaty{g}{Y \to X}{(f \circ g \sim \mathsf{id}_Y) \x (g \circ f \sim \mathsf{id}_X)}.\]
\end{definition}

\begin{definition}
Types \(X\) and \(Y\) are \emph{quasi-invertible} if there is a quasi-inverse from \(X\) to \(Y\):
\[ \sigmatye{f}{X \to Y}{\mathsf{is{\hy}qinv}(f)} .\]
\end{definition}

\noindent
A trivial example of a quasi-inverse is the identity function, meaning that every type is quasi-invertible to itself.

While quasi-inverses indeed align with the notion of equivalence in some fields of mathematics (such as the notion of an isomorphism in category theory), there is an outstanding issue for univalent type theory: it is not the case that \(\mathsf{is{\hy}qinv}(f)\) is a proposition for every \(f\), meaning that he proof that a function is a quasi-inverse may not be unique~\cite{HoTTBook}. 
However, the statement that a function is an equivalence should be a property rather than structure.

Univalent mathematics has fixed this: there is a more general notion of equivalence that \emph{does} yield propositions.

\begin{definition}
\typetop{UF}{Equiv}{is-equiv}
Given types \(X\) and \(Y\), a function \(\ty{f}{X \to Y}\) is an \emph{equivalence} if there are functions \(\ty{g,h}{Y \to X}\) such that \(f \circ g\) and \(h \circ f\) are the identity function:
\[\mathsf{is{\hy}equiv}(f) := \sigmatye{g}{Y \to X}{(f \circ g \sim \mathsf{id}_Y) \x \sigmatye{h}{Y \to X}(h \circ f \sim \mathsf{id}_X)}.\]
\end{definition}

\begin{definition}
\label{def:simeq}
\typetop{UF}{Equiv}{_\urlsimeq_}
Types \(X\) and \(Y\) are \emph{equivalent} \(\ty{X \simeq Y}{\U \sqcup \V}\) if,
\[ X \simeq Y := \sigmatye{f}{X \to Y}{\mathsf{is{\hy}equiv}(f)} .\]
\end{definition}

The fact that \(\mathsf{is{\hy}qinv}(f)\) is \emph{not} in general a proposition, while \(\mathsf{is{\hy}equiv}(f)\) \emph{is}, may at first be unintuitive. As an example of the former see Theorem 4.1.3 of \cite{HoTTBook}, which takes \(X := \Sigmatye{A}{\U}{\| A \simeq \2 \|}\)\footnote{The propositional truncation map \(\|-\|\) is explained in \cref{sec:prop-trunc}.} and proves that \(\ty{\mathsf{id}_X}{X \to X}\) is a quasi-inverse in two unequal ways. The (rather involved) proof of the latter follows from function extensionality --- it appears as Theorem 4.3.2 of the same reference.

Of course, every quasi-inverse is trivially an equivalence.
We also note here that any type that is equivalent to a proposition is a proposition.

\begin{lemma} \label{lem:equiv-prop}
\typetop{UF}{Equiv}{equiv-to-prop}
Given \(\ty{X,Y}{\U}\) such that \(X \simeq Y\), if \(X\) is a proposition then so is \(Y\).
\end{lemma}
\begin{proof}
Recall from \cref{def:proposition} that \(X\) being a proposition means that given any \(\ty{x_1,x_2}{X}\) we have \(x_1 = x_2\); we wish to show the same for the equivalent type \(Y\). Given \(\ty{y_1,y_2}{Y}\), we use the equivalence \(\ty{f}{X \to Y}\) and its left-inverse \(\ty{g}{Y \to X}\). We therefore have elements \(\ty{g(y_1),g(y_2)}{X}\) which satisfy \(g(y_1) = g(y_2)\), and therefore (by the \(\mathsf{ap}\) function defined at the end of \cref{sec:id}) also \(f(g(y_1)) = f(g(y_2))\). Using the fact that \(f \circ g \sim id_Y\), we can conclude \(y_1 = f(g(y_1)) = f(g(y_2)) = y_2\).
\end{proof}

When working informally, mathematicians will often conflate equivalence and equality; for example, we may say ``the reals are the unique complete ordered field" and forget the qualifier ``up to isomorphism". However, in MLTT, these concepts are distinct.
It is obvious that equality implies equivalence:

\begin{corollary}
\label{cor:eq-refl}
\typetop{UF}{Equiv}{\urlsimeq-refl}
Every type is equivalent to itself.
\end{corollary}
\begin{proof}[Proof (Sketch).]
By using the identity function, which is trivially an equivalence. This proof yields the function \(\ty{\simeqrefl}{\pity{X}{\U}{X \simeq X}}\) in \textsc{TypeTopology}.
\end{proof}

\begin{corollary}
\label{cor:id-to-eq}
Given two types \(X\) and \(Y\) such that \(X = Y\), then \(X \simeq Y\).
\end{corollary}
\begin{proof}
We utilise the given equality \(\ty{e}{X = Y}\) to build the function \\ \(\ty{\mathsf{transport} (\lambda(\ty{Z}{\V}).X \simeq Z,e)}{X \simeq X \to X \simeq Y}\).
Thus, by applying \cref{cor:eq-refl}, the result follows,
\begin{alignat*}{3}
\mathsf{id{\hy}to{\hy}equiv} &: \Pity{X,Y}{\U}{(X = Y) \to (X \simeq Y)} ,\\
\mathsf{id{\hy}to{\hy}equiv} & (X,Y,e) := \mathsf{transport} (\lambda(\ty{Z}{\V}).X \simeq Z,e, \simeqrefl(X)).
\end{alignat*}
\end{proof}

\noindent
But it is not the case that equivalence implies equality --- this is another independent proposition in our type theory.

\subsection{Univalence}

A foundational aim of univalent mathematics is to bring equivalence and equality into alignment. The axiom that allows us to prove that equivalence implies equality is called \emph{univalence}. Univalence is due to Vladimir Voevodsky, who desired a foundation of mathematics in which all objects ``are invariant under equivalence of structures"~\cite{Univalence}.

\begin{definition} \label{def:univalent}
\typetop{UF}{Univalence}{is-univalent}
We say that a given universe \(\U\) is \emph{univalent} when, given \(\ty{X,Y}{\U}\), the function \(\ty{\mathsf{id{\hy}to{\hy}equiv}(X,Y)}{X = Y \to X \simeq Y}\) (defined in \cref{cor:id-to-eq}) is an equivalence,
\[ \ty{\mathsf{is{\hy}univalent}(\U)}{\pitye{X,Y}{\U}{\mathsf{is{\hy}equiv}(\mathsf{id{\hy}to{\hy}equiv}(X,Y))}} .\]
\end{definition}

\begin{corollary}
\label{cor:univalence}
\typetop{UF}{Univalence}{eqtoid}
If\ \(\U\) is univalent, then all types \(\ty{X,Y}{\U}\) that are equivalent \(X \simeq Y\) are equal \(X = Y\).
\end{corollary}
\begin{proof}[Proof (Sketch).] \axioms{u}
Because, by univalence, the function \(\ty{\mathsf{id{\hy}to{\hy}equiv}(X,Y)}{X = Y \to X \simeq Y}\) is an equivalence, there is also a function \(\ty{\mathsf{equiv{\hy}to{\hy}id}}{X \simeq Y \to X = Y}\).
\end{proof}

\noindent
We use equivalences throughout this thesis, but only invoke univalence itself once (in \cref{thm:io-sip}).
We do, however, use both function extensionality and propositional extensionality, as previously discussed.

We now state function extensionality in its non-naive form, which reflects our earlier comments on equivalence in \cref{sec:equiv}.

\begin{definition}
\label{def:funext}
\typetop{UF}{FunExt}{funext}
We say that \emph{function extensionality} holds for given universes \(\U,\V\) when, given \(\ty{X}{\U}\), \(\ty{Y}{X \to \V}\) and functions \(\ty{f,g}{\pity{x}{X}{Y(x)}}\), the function \(\ty{\mathsf{happly}(f,g)}{f = g \to f \sim g}\) is an equivalence,
\[ \mathsf{funext}(\U,\V) := \pitye{X}{\U} \pitye{Y}{X \to \V}\pity{f,g}{\pity{x}{X}{Y(x)}}{\mathsf{is{\hy}equiv}(\mathsf{happly}(f,g))} .\]
\end{definition}

\begin{corollary}
If function extensionality holds for \(\U,\V\) then naive function extensionality holds for \(\U,\V\).
\end{corollary}
\begin{proof}[Proof (Sketch).] \axioms{f}
Because, by function extensionality, the function \(\ty{\mathsf{happly}(f,g)}{f = g \to f \sim g}\) is an equivalence, there is also a function \(\ty{\mathsf{naive{\hy}fe}{f \sim g \to f = g}}\) which has the same type as naive function extensionality (\cref{def:naive-fe}).
\end{proof}

Univalence truly captures that equivalence and identity are equivalent: indeed, it generalises the two more specific extensionality principles.

\begin{corollary}
If \(\U\) and \(\V\) are univalent, function extensionality holds for \(\U,\V\).
\end{corollary}

\noindent
Voevodsky's original proof of this corollary is too detailed to reproduce here: we invite the interested reader to see~\cite{UnivalenceFunExt}.

\begin{corollary}
If \(\U\) is univalent, propositional extensionality holds for \(\U\).
\end{corollary}
\begin{proof}[Proof (Sketch).] \axioms{u}
If two propositions \(\ty{P,Q}{\Omega_\U}\) are propositionally equivalent \(P \Leftrightarrow Q\), then they are immediately equivalent \(P \simeq Q\). Thus, by \cref{cor:univalence} of univalence, they are equal \(P = Q\).
\end{proof}

\subsection{Propositional truncation}
\label{sec:prop-trunc}

Under the propositions as (some) types interpretation, we have shown the relationship between proposition types and mathematical propositions.
However, with respect to the logical interpretation of the type theory, there are two outstanding problems:
\(\Sigma\)-types represent constructive existence and \(+\)-types tell us exactly which of our premises holds. 
Both of these type families do not in general interpret propositions (see \cref{lem:sigma-prop,lem:plus-prop}) and, moreover, their interpretations as statements do not match the `traditional' logical perspective of existence and disjunction.

This is resolved in univalent mathematics by a further axiom called \emph{propositional truncation}, which effectively forces a structure to become a property. In \textsc{TypeTopology}, propositional truncations are axiomatised in the following way.

\begin{definition}
\label{def:prop-trunc}
\typetop{UF}{PropTrunc}{propositional-truncations-exist}
We say that a given universe \(\U\) has \emph{propositional truncations} when there is a type truncation function \(\ty{\| - \|}{\U \to \Omega_\U}\) and an element truncation function \(\ty{| - |}{X \to \| X \|}\), such that for every function \(\ty{f_X}{X \to P}\), where \(P\) is a proposition, there exists the truncated function \(\ty{f_{\| X \|}}{\| X \| \to P}\).
\end{definition}

Note that although the behaviour of the truncated function \(f_{\| X \|}\) is not specified in the above axiomatisation, it indeed behaves correctly with respect to \(f_X\) and \(| - |\) due to the fact that \(P\) is a proposition. By this, we mean that the commutative diagram for truncations shown in \cref{fig:truncate-commute} commutes.

\begin{figure}[ht]
\centering
\[\begin{tikzcd}
	X &&& {\| X \|} \\
	\\
	&&& P
	\arrow["{|-|}", from=1-1, to=1-4]
	\arrow["{f_{\| X \|}}", from=1-4, to=3-4]
	\arrow["{f_X}"', from=1-1, to=3-4]
\end{tikzcd}\]
\caption{Commutative diagram illustrating \cref{lem:truncate-commute}.}
\label{fig:truncate-commute}
\end{figure}

\begin{lemma}
\label{lem:truncate-commute}
Given any type \(\ty{X}{\U}\) in a universe which has propositional truncations and a function \(\ty{f_X}{X \to P}\), where \(P\) is a proposition, the truncated function \(\ty{f_{\| X \|}}{\| X \| \to P}\) is such that for all \(\ty{x}{X}\) we have \(f_{X}(x) = f_{\| X \|}(| x |)\).
\end{lemma}
\begin{proof} \axioms{t}
The proof is immediate because \(f_X(x)\) and \(f_{\| X \|}(| x |)\) are both of type \(P\), which is a proposition. Recall from \cref{def:proposition} that meaning that all elements of a proposition are equal.
\end{proof}

We can truncate \(\Sigma\)-types to give types for interpreting traditional existence, which do not carry a constructive witness of the statement.

\begin{definition}
\typetop{UF}{PropTrunc}{PropositionalTruncation.\urlExists}
Given a type \(\ty{X}{\U}\) and a type family \(\ty{Y}{X \to \V}\) we define the \emph{traditional existence type}: 
\[ \existstye{x}{X}{Y(x)} := \| \sigmatye{x}{X}{Y(x)}\| .\]
\end{definition}

\noindent
Furthermore, we \(+\)-types can be truncated to give traditional disjunction types, which do not label which said of their disjunction holds. 

We do not use traditional disjunciton types in this thesis, though, as previously mentioned, we do use propositional truncation; in particular for traditional existence types.

\subsection{Homotopy sets and beyond}

A type \(\ty{X}{\U}\) yields identities \(\ty{x =_X y}{\U}\).
A recognition of core importance to homotopy type theory is that \(x =_X y\) \emph{itself} yields identities \(\ty{p =_{x =_X y} q}{\U}\). This means, of course, that that type also yields identities \(\ty{\alpha \ _{p =_{x =_X y} q} \ \beta}{\U}\) --- and so on ad infinitum. 
This recognition illuminates the concept of \emph{homotopy-levels} (or h-levels), which assigns to some types a number that refers to how complex their yielded identity types are.

A type with \emph{exactly} one element (a `\emph{contractible} type') has h-level \(0\), a type with \emph{at most} one element (i.e.\ a proposition) has h-level \(1\), a type whose identities have at most one element has h-level \(2\), a type whose identities' identities have at most one element has h-level \(3\), and --- again --- so on and so forth.
This can be formally expressed by defining the following recursive function --- which first requires us to formally define the notion of a contractible type.

\begin{definition}
\typetop{UF}{Subsingletons}{is-contr}
A type \(\ty{X}{\U}\) is \emph{contractible} if there is some \(\ty{c}{X}\) such that for all \(\ty{x}{X}\) we have \(c = x\):
\vspace{-0.5em}
\begin{alignat*}{3}
\mathsf{is{\hy}contractible} &\ X && := \Sigmatye{c}{X}\Pitye{x}{X} \ c = x.
\end{alignat*}
\end{definition}

\begin{definition}
\typetop{MGS}{hlevels}{_is-of-hlevel_}
A type \(\ty{X}{\U}\) has \emph{h-level} \(\ty{n}{\N}\) if \(\mathsf{has{\hy}h{\hy}level}(X,n)\), defined below, holds.
\vspace{-0.5em}
\begin{alignat*}{3}
\mathsf{has{\hy}h{\hy}level} & (X,0) && := \mathsf{is{\hy}contractible} \ X, \\
\mathsf{has{\hy}h{\hy}level} & (X,n+1) && := \Pitye{x,y}{X} \ \mathsf{has{\hy}h{\hy}level}(x =_X y,n).
\end{alignat*}
\end{definition}

Types with h-level \(2\) are also called \emph{sets} --- they represent types that interpret classical sets, where elements can only be equal in one way.

\begin{definition} \label{def:hset}
\typetop{UF}{Sets}{is-set}
A type \(\ty{X}{\U}\) is a \emph{set} if, for any elements \(\ty{x,y}{X}\), we have \(\isprop (x = y)\).
\end{definition}

\noindent
From there, types with h-level \(n+2\) are \(n\)-groupoids, adopting terminology from homotopy theory~\cite{HoTTBook}.

In this thesis, we often utilise the fact that certain types are propositions or sets, but do not utilise higher types in this way.

\section{Fundamental concepts for this thesis}

To close this chapter, we introduce for reference a number of additional concepts that are used throughout the thesis.

\subsection{Decidability and discreteness}

We introduce three notions of decidability: decidability of \emph{types}, decidability of \emph{type-valued functions} and decidable \emph{equality}.

\begin{definition}
\typetop{MLTT}{Negation}{is-decidable}
A type is \emph{decidable} if we can show it is either pointed or its negation is pointed:
\[\decidable{X} := X + \neg X .\]
\end{definition}

\begin{definition}
\typetop{NotionsOfDecidability}{Complemented}{is-complemented}
Given a type \(\ty{X}{\U}\) and type-valued function \(\ty{B}{X \to \V}\) the \(\Pi\)-type \(\ty{\prod B}{\U \sqcup \V}\) is \emph{decidable} (or \emph{complemented}) if \(\ty{B(a)}{\U}\) is decidable for every \(\ty{a}{A}\), i.e.\ if we have, \[\mathsf{complemented} \ B := \pity{a}{A}{\decidable{B(a)}}.\]
\end{definition}

We cannot in general decide whether two elements of a type are equal; types for which we can are called \emph{discrete}. 

\begin{definition}
\typetop{UF}{DiscreteAndSeparated}{is-discrete}
A type \(X\) is \emph{discrete} if it has decidabile equality:
\[ \mathsf{discrete}(X) :=  \pity{x,y}{X}{\mathsf{decidable} \ (x = y)} .\]
\end{definition}

\noindent
The empty type, singleton type, natural numbers and integers are all discrete. Furthermore, given two discrete types \(X\) and \(Y\), both \(X + Y\) and \(X \times Y\) are discrete.

\subsection{Finite types}
\label{sec:finite-linear-ordered}

There are various notions of finiteness in \textsc{TypeTopology}, and more broadly in constructive type theory.
In this thesis, we work with only those types that are \emph{finite linearly ordered}.

\mbox{}
\begin{definition}
\label{def:fin}
\typetop{Fin}{Type}{Fin}
The type family \(\ty{\F}{\N \to \U}\) is defined by induction:
\begin{alignat*}{2}
\F &: \N \to \U, \\
\F &(0) &&:= \0, \\
\F &(n + 1) &&:= \plus{\F(n)}{\1}.
\end{alignat*}
\end{definition}

\begin{definition}
\typetop{Fin}{Bishop}{finite-linear-order}
A type \(F\) is \emph{finite linearly ordered} if there is some \(\ty{n}{\N}\) such that \(F \simeq \F(n)\).
\end{definition}

Note that this is structure and not property --- it defines the collection of finite linear orders on \(F\).

\begin{definition}
\typetop{Fin}{Bishop}{finiteness.is-finite}
A type \(F\) is \emph{finite} if there is some \(\ty{n}{\N}\) such that \(\| F \simeq \F(n) \|\).
\end{definition}

\noindent
This is now a property. Every finite linearly ordered type is trivially finite, but univalence implies that not every finite type is finite linearly ordered --- hence, it is not provable to provide an order for a general finite type~\cite{FinKuratowski}.

However, in this thesis we only utilise finite linearly ordered types for a notion of finite, and so we abuse terminology and call these types `finite' throughout.

We recall two easy lemmas concerning these finite types:

\begin{lemma}
\thesislit{2}{Finite}{finite-is-discrete}
\label{lem:fin-discrete}
Every finite linearly ordered type \(F\) is discrete.
\end{lemma}
\begin{proof}[Proof (Sketch.)]
By induction on the \(\ty{n}{\N}\) such that \(F \simeq \mathsf{Fin}(n)\). In the base case where \(n := 0\), then \(F \simeq \0\) which is vacuously discrete. In the inductive case where \(n := n' + 1\), for some \(\ty{n'}{\N}\), then \(F \simeq \mathsf{Fin}(n') + \1\) --- as both sides of this are discrete (the left-hand side by induction and right-hand side by the fact that \(\1\) is trivially discrete), the overall \(+\)-type is discrete.
\end{proof}

\begin{lemma}
\thesislit{2}{Finite}{finite-is-set}
\label{lem:fin-set}
Every finite linearly ordered type is a set.
\end{lemma}
\begin{proof}[Proof (Sketch.)]
Similarly to the above we proceed by induction and use that \(\0\) and \(\1\) are both sets and that the disjoint union of sets is also a set.
\end{proof}

\subsection{Vectors and sequences}
\label{sec:sequences}

We often require types for products beyond the binary case.
For finitary products we now introduce \emph{vectors}, and for infinitary products we introduce \emph{sequences}. 
Sequences are of particular importance to this work, as the representations of real numbers we use later are sequence types. 

We first give both non-dependent and dependent sequences.
Given a type \(\ty{X}{\U}\), the type of \emph{(non-dependent) sequences} of elements of \(X\) is \(\seq X := (\N \to X)\). Meanwhile, given an \(\N\)-indexed family of types \(\ty{Y}{\N \to \U}\), the type of \emph{dependent sequences} of elements of \(Y(0),Y(1),Y(2),...\) respectively is the \(\Pi\)-type \(\ty{\Pi Y}{\U}\).
Of course, the non-dependent version is an instance of the dependent version by setting \(Y\) as the constant function \(\ty{\left(\lambda (\ty{-}{\N}).X\right)}{\N \to \U}\).

For vectors, we formally define the dependent version.

\begin{definition}
\typetop{MLTT}{SpartanList}{vec}
Given a number \(\ty{n}{\N}\) and a finite family of types \(\ty{X}{\F(n) \to \U}\), the type of \emph{\(n\)-size (dependent) vectors} \(\mathsf{Vec}(n,X)\) is defined inductively using binary products.
\begin{alignat*}{3}
\mathsf{Vec} &: \Pitye{n}{\N}{(\F(n) \to \U) \to \U}, \\
\mathsf{Vec} &(0,X) && := \1, \\
\mathsf{Vec} &(n+1,X) && := X_0 \x \mathsf{Vec}(n,\lambda i.X_{i+1}).
\end{alignat*}
\end{definition}

\noindent
As with dependent sequences, each point of a dependent vector can have a different type. Again, if the given type family \(Y\) is constant on a type \(X\), then we have defined a non-dependent vector of elements of \(X\).
We usually notate \(n\)-size non-dependent vectors as \(\ty{\{x_0,...,x_{n-1}\}}{X^n}\).

In the rest of the thesis we use the following (overloaded) notation for functions that operate on both vectors and sequences. Given a vector or sequence \(xs\), we write:
\begin{itemize}
\item \(\alpha_n\) for the \(n\)\textsuperscript{th} element of the vector/sequence,
\item \(\mathsf{head} \ xs\) for the first element of the vector/sequence \(\alpha_0\),
\item \(\mathsf{tail} \ xs\) for the vector/sequence with the head dropped \(\lambda n.xs_{n+1}\),
\item \((x :: xs)\) for the vector/sequence prepended with an element \(x\),
\item \(\mathsf{map}(f,xs)\) for the vector/sequence \(f(\mathsf{head} \ xs) :: \mathsf{map}(f,\mathsf{tail} \ xs)\),
\item \(\mathsf{repeat} \ x\) for the sequence \(\left(\lambda n.x\right)\).
\end{itemize}
Note that these are not \emph{just} notation but functions in our formalisation.
For example, the following function \(\mathsf{zipWith}\) is the extension of \(\mathsf{map}\) to binary functions:

\begin{definition}
\label{def:zipWith}
\thesislit{2}{Sequences}{zipWith}
Given types \(X\), \(Y\) and \(Z\), we define the \(\mathsf{zipWith}\) functional that canonically lifts functions of type \((X \to Y \to Z)\) to those of type \((\seq X \to \seq Y \to \seq Z)\):
\begin{alignat*}{3}
\mathsf{zipWith} &: (X \to Y \to Z) \to (\seq X \to \seq Y \to \seq Z), \\
\mathsf{zipWith} &(f,\alpha,\beta) := \lambda n.f(\alpha_n,\beta_n).
\end{alignat*}
\end{definition}

An important notion in our work is that of the \emph{equality of sequence prefixes}, which is decidable for any length of prefix if the types of the sequence are discrete.

\begin{definition}
\thesislit{2}{Sequences}{_\urlsim\urlsuperscriptn_}
Given an \(\N\)-indexed family of types \(\ty{X}{\N \to \U}\), two sequences \(\ty{\alpha,\beta}{\Pi X}\) have \emph{equal \(n\)-prefixes} if they are equal at every point below \(n\):
\begin{alignat*}{3}
\ &\sim^n \ : &&\N \to \Pi X \to \Pi X \to \U, \\
\alpha &\sim^n \beta &&:= \Pity{i}{\N}{i < n \to \alpha_i = \beta_i}
\end{alignat*}
\end{definition}

\begin{lemma}
\label{lem:sim-decidable}
\thesislit{2}{Sequences}{\urlsim\urlsuperscriptn-decidable}
Given an \(\N\)-indexed family of discrete types \(\ty{X}{\N \to \U}\), given any sequences \(\ty{\alpha,\beta}{\Pi X}\) we can decide whether \(\alpha \sim^n \beta\) for any prefix length \(\ty{n}{\N}\).
\end{lemma}
\begin{proof}
In the base case, \(\alpha \sim^0 \beta\) is always satisfied (as there is no natural number below \(0\)).
In the inductive case, we decide whether \(\alpha \sim^{n+1} \beta\) by deciding whether \(\alpha \sim^n \beta\) (by induction) and whether \(\alpha_n = \beta_n\) (by discreteness of \(X_n\)).
\end{proof}

\begin{lemma}
\label{lem:sim-eq}
Given an \(\N\)-indexed family of discrete types \(\ty{X}{\N \to \U}\), given any sequences \(\ty{\alpha,\beta}{\Pi X}\) the type \(\alpha \sim^n \beta\), for any prefix length \(\ty{n}{\N}\), is an equivalence relation (i.e.\ it is reflexive, symmetric and transitive).
\end{lemma}
\begin{proof}
By those same properties (given in \cref{sec:id}) on the identity type \(\alpha_i = \beta_i\) for all \(i < \beta\).
\end{proof}
\chapter{Searchability and Continuity}
\label{chap:searchable}

In \cref{chap:mltt}, we outlined our constructive and univalent foundation of mathematics that we will work within in this thesis.
Within our \textsc{Agda} framework, we have begun to define a wide collection of mathematical structures, properties and propositions.
Two structures of chief importance to this thesis are \emph{searchability} and \emph{total boundedness}: related concepts which are explored in this chapter in a generalised format, employed in \cref{chap:generalised} for defining convergent optimisation and regression theorems, and equipped in \cref{chap:exact-real-search} for computation of such search, optimisation and regression procedures on types for for representing real numbers.

In order to be able to express searchability on infinite types (such as those used for representing real numbers), we need to define another key property in our framework --- \emph{continuity} on functions and predicates. In \cref{sec:cspace}, we define a convenient variant of continuity on \emph{closeness spaces}; a structure that we also explore in this chapter which also allows us to define total boundedness.
By the end of this chapter, we will have two methods for constructive infinite search algorithms: by using the total boundedness property (\cref{thm:tb-csearch}) or by using the \emph{Tychonoff theorem} for searchable types (\cref{thm:tychonoff}).
In the examples given in \cref{chap:exact-real-search}, we find that searchers derived from the former theorem are usually better suited for practical purposes, though not in every case.

\section{Searchable Types}
\label{sec:search-finite}

\subsection{Background and motivation}

Mart\'{\i}n Escard\'o introduced the concept of \emph{searchable sets} in higher computability theory~\cite{Escardo08}.
Informally, these sets \(K\) are those for which we can establish a computable functional \(\mathcal{E}_K\), that we call a \emph{searcher}.
A searcher takes as input any Boolean-valued predicate \(\ty{p}{K \to \2}\) and returns a distinguished element \(\ty{\mathcal{E}_K(p)}{K}\) that satisfies the following \emph{search condition}: if there is at least one element \(\ty{k}{K}\) that satisfies \(p\) (i.e.\ \(p(k) = 1\)), then the distinguished element also satisfies \(p\).

Infinite searchers --- sometimes called `seemingly impossible functional programs' --- search elements of infinite spaces such as the type of binary sequences (i.e.\ elements of the \emph{Cantor space} \(\seq \2 := \N \to \2\))~\cite{Escardo12}.
The searcher on the Cantor space is originally due to Berger~\cite{BergerThesis}. The more general searcher on the product space \(\Pi T\) of infinitely-many searchable sets \(\ty{T_i}{\U}\) (where \(\ty{T}{\N \to \U}\)) is due to Escard\'o~\cite{Escardo08}.

Of particular relevance to this thesis, infinite search programs have been previously defined on representations of real numbers and used to perform analysis~\cite{Escardo11fun}. One example is that Escard\'o has used infinite searchers to find solutions \(x\) to equations \(f(x) = y\)~\cite{Escardo13}.
Another is that Simpson has used them to perform Riemann integration and to find the global maximum value of a continuous function on the type of ternary signed-digit encodings (introduced in \cref{sec:signed-digits})~\cite{Simpson}.
We will use searchable types for our own purposes of analysis (specifically, optimisation and regression) in \cref{chap:generalised}.

\subsection{Searchable types in MLTT}

For defining searchability in our \textsc{Agda} framework, there are a wealth of different options, many of which are already defined in \textsc{TypeTopology}~\cite{CompactTypes}.
We follow Escard\'o's type-theoretic variants in that we have no requirement for our \emph{type} \(\ty{K}{\U}\) to be a \emph{set} (\cref{def:hset}); further, instead of using \(\2\) as the domain of the searched predicates, we use a type of truth values \(\Omega\).
In this section, we recall definitions and proofs concerning searchable types --- the proof techniques are shown only to get the reader comfortable with these techniques, which will be useful when we introduce \emph{uniformly continuously searchable types} in \cref{sec:search-infinite}.

In order to be able to test our searched predicates \(\ty{p}{K \to \Omega}\), we require them to be decidable.
Note that every predicate in the original formulation was decidable, because every function with domain \(\2\) is trivially decidable --- indeed, the type of decidable predicates of type \(K \to \Omega\) is, by function and proposition extensionality, equivalent to the original type of predicates \(K \to \2\).

\begin{definition}
\thesislit{3}{SearchableTypes}{decidable-predicate}
The type of \emph{decidable predicates} on a type \(K\) is defined by, \[\mathsf{decidable{\hy}predicate}(K) := \sigmaty{p}{K \to \Omega}{\mathsf{complemented} \ p} .\]
\end{definition}

\noindent
We often leave the universe implicit --- but note that it can be different to the universe that the co-domain type lives in --- and will usually leave the witness of decidability implicit.

This gives us a definition of searchability in our type theory.

\begin{definition}
\label{def:searcher}
\typetop{TypeTopology}{CompactTypes}{is-compact\urlinterpunct'}
\thesislit{3}{SearchableTypes}{searchableE}
A function \(\ty{\mathcal{E}_K}{\mathsf{decidable{\hy}predicate}(K) \to K}\) is a \emph{searcher} on a given type \(K\) if, for all \(\ty{p}{\mathsf{decidable{\hy}predicate}(K)}\), it is the case that \(p(\mathcal{E}_K(p))\) holds if there is some element \(\ty{k}{K}\) such that \(p(k)\) holds:
\[ \mathsf{is{\hy}searcher}(\mathcal{E}) := \pitye{p}{\mathsf{decidable{\hy}predicate}(K)}{\left( \sigmatye{k}{K}{p(k)} \right) \to p(\mathcal{E}(p))} .\]
\end{definition}

\begin{definition}
\label{def:searchable}
\typetop{TypeTopology}{CompactTypes}{is-compact\urlinterpunct}
\thesislit{3}{SearchableTypes}{searchable}
A type \(K\) is \emph{searchable} if we can define a searcher on that type:
\[ \mathsf{searchable^\mathcal{E}}(K) := \sigmatye{\mathcal{E}_K}{\mathsf{decidable{\hy}predicate} \ K \to K}{\mathsf{is{\hy}searcher}(\mathcal{E})} .\]
We often use the following equivalent definition, which is more convenient:
\[ \mathsf{searchable}(K) := \pitye{p}{\mathsf{decidable{\hy}predicate}(K)}{ \sigmatye{k_0}{K}{\left( \sigmatye{k}{K}{p(k)} \right) \to p(k_0)}} .\]
\end{definition}

For non-trivial searchable types \(K\), the searcher is not unique.
As an example, consider two (informally defined) searchers for the \(\2\) type:
\begin{enumerate}
\item \(\mathcal{E}_\2(p) := \mathsf{if} \ p(\tru) \ \mathsf{return} \ \tru \mathsf{; return} \ \ff \mathsf{;} \)
\item \(\mathcal{E}_\2(p) := \mathsf{if} \ p(\tru) \ \mathsf{return} \ \ff \mathsf{; return} \ \tru \mathsf{;} \)
\end{enumerate}
\noindent
Although both of these searchers satisfy the search condition, they can return different elements given the same predicate.
Therefore, the type \((\mathsf{searchable} \ K)\) is not a subsingleton.

Note that we can always exhibit an element of a searchable type, and separately if the element returned by the searcher does not satisfy the searched predicate then no element does.

\begin{lemma}
\label{lem:searchable-pointed}
\typetop{TypeTopology}{CompactTypes}{Compact\urlinterpunct-gives-pointed}
\thesislit{3}{SearchableTypes}{searchable-pointed}
Every searchable type is pointed.
\end{lemma}
\begin{proof}
For the given searchable type \(K\), we define the constant predicate \(p^\top(k) := \top\), which every element satisfies.
We can then introduce the element \(\ty{\mathcal{E}_K(p^\top)}{K}\).
\end{proof}

\begin{lemma}
\label{lem:searchable-none}
\typetop{TypeTopology}{CompactTypes}{is-compact\urlinterpunct}
Given a searchable type \(K\) and any decidable predicate \(\ty{p}{\mathsf{decidable{\hy}predicate} \ K}\), if \(\neg p(\mathcal{E}_K(p))\) then for all \(\ty{k}{K}\) it is the case that \(\neg p(k)\).
\end{lemma}
\begin{proof}
Assuming \(\ty{z}{\neg p(\mathcal{E}_K(p))}\) then given any \(\ty{k}{K}\) we use the decidability of \(p\) to decide whether \(p(k) + \neg p(k)\). In the latter case, the result follows trivially. In the former case, the search condition (in \cref{def:searcher}) implies that, because there is an element satisfying the predicate, we have some \(\ty{q}{p(\mathcal{E}_K(p))}\); therefore, this case is impossible and eliminated by \(\mathsf{\0{\hy}elim}(z(q))\).
\end{proof}

\subsection{Searching finite types}

Every pointed finite linearly ordered type (as defined in \cref{sec:finite-linear-ordered}) is searchable. The searcher for a finite type can simply check each element in turn (by any particular search strategy, but we assume it checks them relative to the type's linear order), returning the first element that satisfies the predicate --- if the type is fully exhausted, then the searcher can return any element of the type as no such satisfying element exists.

First note that \(\1\) is trivially searchable, and that the disjoint sum type former preserves searchability.

\begin{remark}\label{lem:1-searchable}
\typetop{TypeTopology}{CompactTypes}{\urlone-is-Compact}
\thesislit{3}{SearchableTypes}{\urlone-searchable}
\(\1\) is searchable.
\end{remark}
\begin{proof}
The searcher always returns \(\ty{\star}{\1}\).
Whether or not \(p(\star)\), this satisfies the search condition.
\end{proof}

\begin{lemma}\label{lem:plus-searchable}
\typetop{TypeTopology}{CompactTypes}{\urlplus-is-Compact}
\thesislit{3}{SearchableTypes}{\urlplus-searchable}
If types \(K\) and \(J\) are searchable, then so is \(K + J\).
\end{lemma}
\begin{proof}
Given any predicate \(\ty{p}{{K + J} \to \Omega}\), we want to find some element \(\ty{{k\hspace{-0.05em}j}_0}{K + J}\) that satisfies \(p\), if such an element exists.
First, we define the predicates \(p_K(k) := p(\inl \ k)\) and \(p_J(j) := p(\inr \ j)\).

Then, we search \(K\) to return \(\ty{k_0 := \mathcal{E}_K(p_K)}{K}\). If \(p(\inl \ k_0)\), then we are done and return \({k\hspace{-0.05em}j}_0 := \inl \ k_0\).
Otherwise, we go on to search \(J\) and set \(j_0 := \mathcal{E}_J(p_J)\).
If \(p(\inr \ j_0)\), then we are done and return \({k\hspace{-0.05em}j}_0 := \inl \ k_0\); if not, then there is no element of \(K + J\) that satisfies \(p\), and so we can safely return any element.
\end{proof}

\noindent
The above allows us to show that every pointed type in the family \(\ty{\F}{\N \to \U}\) is a searchable type.

\begin{lemma}\label{lem:fp-searchable}
\thesislit{3}{SearchableTypes}{Fin-searchable}
Given any \(\ty{n}{\N}\), if the type \(\F(n)\) is pointed then it is searchable.
\end{lemma}
\begin{proof}
By induction on the given \(\ty{n}{\N}\).
The base case, where \(n := 0\), is vacuous as \(\F(0) := \0\) is not pointed.
The inductive case, where \(n := \suc{n'}\), requires showing \(\F(\suc{n'}) := \1 + \F(n')\) is searchable. As \(1\) is trivially searchable (see \cref{lem:1-searchable}) and \(\F(n')\) is searchable by the inductive hypothesis, the result follows by \cref{lem:plus-searchable}.
\end{proof}

\noindent
We next show that any type \(K\) can be searched using a searcher on an equivalent searchable type \(J\) --- i.e.\ equivalence preserves searchability.

\begin{lemma}\label{lem:equiv-searchable}
\thesislit{3}{SearchableTypes}{equivs-preserve-searchability}
Given types \(K\) and \(J\) such that \(K \simeq J\), if \(J\) is searchable then so is \(K\).
\end{lemma}
\begin{proof}
Because \(K \simeq J\), we have \(\ty{f}{K \to J}\), \(\ty{g}{J \to K}\) and \(\ty{h}{J \to K}\) such that \(f \circ g \sim \mathsf{id}_J\) and \(h \circ f \sim \mathsf{id}_K\).

Given any predicate \(\ty{p}{K \to \Omega}\), we want to find some element \(\ty{k_0}{K}\) that satisfies \(p\), if such an element exists.
We define \(\ty{p'}{J \to \Omega}\) as \(p'(j) := p(g(j))\) and search \(J\) to return \(\ty{j_0 := \mathcal{E}_{J}(p')}{J}\).
If \(p(g(j))\), then we are done and return \(k_0 := g(j)\).

Otherwise, \(p'\) is never satisfied (i.e.\ \(\neg p(g(j))\) for all \(\ty{j}{J}\)) and from this we can prove that \(p\) is never satisfied. Given \(\ty{k}{K}\) such that \(p(k)\), we would have \(\ty{f(k)}{J}\) such that \(p(k) := p(h(f(k)) := p'(f(k))\), which is a contradiction.
We have proved that \(p\) is never satisfied, and hence our searcher can return any element.
\end{proof}

\noindent
We can now show that any pointed finite linearly ordered type is a searchable type.

\begin{lemma} \label{lem:fin-searchable}
\thesislit{3}{SearchableTypes}{finite-searchable}
Every pointed finite linearly ordered type is searchable.
\end{lemma}
\begin{proof}
By \cref{lem:fp-searchable,lem:equiv-searchable}.
\end{proof}

Finally, we show that finite products preserve searchability.

\begin{lemma}\label{lem:prod-searchable}
\typetop{TypeTopology}{CompactTypes}{binary-Tychonoff}
\thesislit{3}{SearchableTypes}{\urlx-searchable}
If \(K\) and \(J\) are searchable, then so is \(K \x J\).
\end{lemma}
\begin{proof}
Given any predicate \(\ty{p}{{K \x J} \to \Omega}\), we want to find a pair of elements \(\ty{(k_0,j_0)}{K \x J}\) that satisfy \(p\), if such a pair exists.
We define a family of predicates on \(J\),
\begin{alignat*}{3}
p_J &: K \to \mathsf{decidable{\hy}predicate}(J), \\
p_J &(k) := \lambda j.p(k,j).
\end{alignat*}
and a predicate on \(K\), \[ p_K(k) := p(k,\mathcal{E}_J(p_J(k))), \] where \(\mathcal{E}_J\) is the searcher (see \cref{def:searcher}) on \(J\). 
By searching \(K\) for an answer to \(p_K\), we in turn search \(J\) for an answer to \(p_J(k)\) --- we name the former answer \(\ty{k_0}{K}\) and the latter, \(\ty{\mathcal{E}_J(p_J(k_0))}{Y}\), is dependent on the former.

\vspace{1em}
We now need to show that the search condition (\cref{def:searcher}) is satisfied; i.e.\ if there is \(\ty{(k,j)}{K \x J}\) such that \(p(k,j)\) then also \(p(k_0,\mathcal{E}_J(p_J(k_0)))\).
We use the search conditions of \(k_0\) and \(\mathcal{E}_J(p_J(k_0))\): 
\begin{enumerate}
\item If there is \(\ty{k'}{K}\) such that \(p_K(k') := p(k',\mathcal{E}_J(p_J(k')))\) then \(p_K(k_0) := p(k_0,\mathcal{E}_J(p_J(k_0)))\),
\item For any \(\ty{k^*}{K}\) if there is a \(\ty{j'}{J}\) such that \(p(k^*,j')\) then \(p(k^*,\mathcal{E}_J(p_J(k^*)))\).
\end{enumerate}

We have \(p(k,j)\) and so by (2) we have \(p(k,\mathcal{E}_J(p_J(k)))\), and by (1) we then have the conclusion.
\end{proof}

\subsection{Can we search infinite types?}
\label{sec:search-infinite-canwe}

Our intuition told us that all finite types are searchable, and likewise our intuition tells us that all infinite types are not searchable.
Indeed, searchability of the canonical countably infinite type \(\N\) is logically equivalent to the \emph{limited principle of omniscience} (LPO)~\cite{Troelstra}.

\begin{remark}
\label{remark:LPO}
Recall that LPO states that, given a binary sequence, we can decide that either all points of the sequence are \(0\) or there is an explicit index at which point the sequence is \(1\).
\begin{alignat*}{2}
\mathsf{LPO} &: \U,\\
\mathsf{LPO} &:= \Pitye{\alpha}{\seq \2}{\left( \Pity{n}{\N}{\alpha_n = 0} \right) + \left( \Sigmaty{n}{\N}{\alpha_n = 1} \right)}.
\end{alignat*}
\end{remark}

\begin{lemma}
\thesislit{3}{SearchableTypes}{\urlN-searchability-is-taboo}
The natural numbers are searchable if and only if LPO holds.
\end{lemma}
\begin{proof}
We first prove that the searchability of \(\N\) implies LPO.
Given a binary sequence \(\ty{\alpha}{\seq \2}\), we define the following predicate on natural numbers, which is true when the sequence at that index is \(1\):
\begin{alignat*}{2}
p &: \N \to \Omega, \\
p &(n) := \alpha_n = 1.
\end{alignat*}

\noindent
Note that this is indeed subsingleton-valued as \(\2\) is a set (see \cref{lem:fin-set}), and is decidable as \(\2\) is discrete (see \cref{lem:fin-discrete}).
Now, using the searcher \(\ty{\mathcal{E}_\N}{\mathsf{decidable{\hy}predicate} \ \N \to \N}\), we obtain the element \(\mathcal{E}_\N(p)\) and check whether it satisfies the predicate.
If \(p(\mathcal{E}_\N(p))\), then the right-hand side of LPO holds because \(\alpha_{\mathcal{E}_\N(p)} = 1\). If \(\neg p(\mathcal{E}_\N(p))\), then the left hand side holds because, by \cref{lem:searchable-none}, nothing satisfies the predicate --- no element of \(\alpha\) is \(1\) and therefore every element must be \(0\)).

\ \newline
To prove the opposite direction, we assume that LPO holds and want to show that we can find an element that satisfies any given \(\ty{(p,d)}{\mathsf{decidable{\hy}predicate} \ \N}\) if such an element exists.
Given such a predicate, we define the following binary sequence that is \(1\) whenever the index satisfies the predicate and \(0\) otherwise:
\begin{alignat*}{2}
\alpha' &: \Pity{n}{\N}{\mathsf{decidable}(p(n)) \to \2}, \\
\alpha' &(n,\inl \ q) := \ 1,\\
\alpha' &(n,\inr \ z) := \ 0,
\\ \vspace{1em}
\alpha \ & : \seq \2, \\
\alpha_n & := \alpha'(n,d(n)).
\end{alignat*}

The proof follows by applying LPO to \(\alpha\). In the case where \(\alpha_n = 0\) for all \(\ty{n}{\N}\), then the predicate is never satisfied and our searcher can return any element.
In the case where there is some \(\ty{n}{\N}\) such that \(\alpha_n = 1\), then we also have a proof of \(p(n)\), and the searcher can return \(\N\).
In both cases, the search condition holds.
\end{proof}

Some infinite types, however, go against our intuitions and \emph{are} searchable.
Specifically, \MartinEscardo~proved that infinitary products preserve searchability; i.e. given a type family \(\ty{T}{\N \to \U}\) of searchable types, the type \(\ty{\Pi T}{\U}\) is itself searchable~\cite{Escardo08}.
This result is the \emph{Tychonoff theorem} for searchable types\footnote{This name comes from the relationship between the concepts of searchable types and compact spaces in (synthetic) topology.
In topology, Tychonoff's theorem states that arbitrary products preserve compactness.} and it allows a wide range of types, such as the aforementioned Cantor space \(\seq \2\) and those types we use in \cref{chap:exact-real-search} for representing compact intervals of real numbers, to be searched.
We therefore require a formulation of the Tychonoff theorem in our framework for constructive type theory. However, it is not the case that we can simply translate the original proof directly into our setting.

\Escardo's original infinite search algorithms were written in PCF, a language which has access to \emph{general recursion}.
The proof that this infinite search returns an answer is written in the Scott model of PCF, a setting in which all PCF-definable functions are automatically \emph{continuous}. A corollary of this states that any decidable predicate \(\ty{p}{\Pi T \to \Omega}\) is \emph{uniformly continuous}, which means that only a finite amount of information about a particular search candidate \(\ty{x}{\Pi T}\) is required to determine whether or not \(p(x)\) holds.
The combination of general recursion and the assumption that all functions are continuous means that the infinite search algorithm always returns an answer.

In our setting, we deliberately do not assume that all functions are continuous (in particular, this makes our mathematics compatible with classical mathematics) and we do not have access to general recursion (indeed, the recursive functions of MLTT are \emph{primitive recursive})~\cite{MLTTProgBook}.
Therefore, we are required to use a different approach in our formulation of the Tychonoff theorem. This different approach is that we rephrase the notion of searchability to incorporate an \emph{explicit} notion of uniform continuity. The primitive-recursive infinite search algorithms themselves are then defined using the information provided by this notion.

This results in an alternative proof of the Tychonoff theorem for searchable types, which we give in \cref{sec:tychonoff}, that does not use general recursion and \emph{explicitly} assumes that the searched predicate is uniformly continuous.
Of course, this requires us first to actually define an explicit notion of continuity within our framework, which we aim to keep as general as possible for definition on a wide class of types; for this purpose, we introduce \emph{closeness spaces} in the next section.

\section{Closeness Spaces}
\label{sec:cspace}

Continuity, and the stronger notion of uniform continuity, are properties of functions \(\ty{f}{X \to Y}\). Informally, continuity says that determining \(f(x)\) to a requested precision \(\varepsilon\) depends on determining \(x\) to a certain precision \(\delta\), the \emph{modulus of continuity}, which is constructed from \(\varepsilon\) and \(x\). A \emph{modulus of uniform continuity}, meanwhile, is constructed only from \(\varepsilon\).

Continuity can be defined in a variety of ways; we may choose a particular definition based on the types \(X\) and \(Y\), the form of topological information we have about those type, and/or the types of the precisions \(\varepsilon\) and \(\delta\).
In this section, we motivate and define within our framework a structure we call \emph{closeness spaces}, which exhibit a constructively economical definition of continuity and uniform continuity that we find convenient for our work.

\subsection{Motivation via metric spaces}

In analysis, one common definition of continuity is between \emph{metric spaces}, wherein the distance between objects can be measured by a real-valued binary function called a \emph{metric}~\cite{Metric,PrinciplesMA}.
Although we have not yet discussed what a type for real numbers looks like in our framework (this will be done later, in \cref{chap:reals}), the below definitions are for illustrative purposes only, and so the reader can assume the type of positive reals \(\R_{\geq 0}\) in the below satisfies the expected properties.

\begin{definition}
\label{def:mspace}
\lit
A \emph{metric space} is a type \(X\) equipped with a \emph{metric} \(d : X \to X \to \R_{\geq 0}\) such that,
\begin{enumerate}
\item \(d(x,y) = 0 \leftrightarrow x = y\),
\item \(d(x,y) = d(y,x)\),
\item \(d(x,z) \leq d(x,y) + d(y,z)\).
\end{enumerate}
\end{definition}

For functions \(\ty{f}{X \to Y}\) on metric spaces, continuity says that for any \(\ty{\varepsilon}{\R_{\geq 0}}\) there exists a \(\ty{\delta}{\R_{\geq 0}}\) such that all elements \(\ty{x,y}{X}\) that are \(\delta\)-close to each other (i.e.\ \(d_X(x,y) < \delta\)) will be mapped to elements that are \(\varepsilon\)-close to each other (i.e.\ \(d_Y(f(x),f(y)) < \varepsilon\)).
This notion of closeness can be described by the following family of reflexive and symmetric binary relations \(\ty{C}{\R_{\geq 0} \to X \to X \to \Omega}\) where, for all \(\ty{\varepsilon}{\R_{\geq 0}}\) and \(\ty{x,y}{X}\), we say that \(x\) and \(y\) are \(\varepsilon\)-close if \(C_\varepsilon(x,y)\).

\begin{definition}
\label{def:metr-closerelation}
\lit
Given a metric space \(X\), the family of \emph{closeness relations} is defined as follows,
\begin{alignat*}{2}
C \ &: \R_{\geq 0} \to X \to X \to \Omega, \\
C_\varepsilon &(x,y) := d(x,y) < \varepsilon.
\end{alignat*}
\end{definition}

\begin{lemma}
\label{lem:mrefsym}
\lit
Given a metric space \(X\) and distance \(\ty{\varepsilon}{\R_{\geq 0}}\), the closeness relation \(C_\varepsilon\) is reflexive and symmetric.
\end{lemma}
\begin{proof}
Reflexivity follows from \cref{def:mspace}.1; i.e.\ \(C_\varepsilon(x,x) := \left( 0 \leq \varepsilon \right)\), which is immediately satisfied for all \(\ty{x}{X}\).
Symmetry follows from \cref{def:mspace}.2; i.e.\ \(C_\varepsilon(x,y) := \left( d(x,y) \leq \varepsilon \right) = \left( d(y,x) \leq \varepsilon \right) =: C_\varepsilon(y,x)\) for all \(\ty{x,y}{X}\).
\end{proof}

By strengthening the triangle inequality property (\cref{def:mspace}.3) of metric spaces we instead define ultrametric spaces, for which closeness relations are additionally transitive and, thus, are equivalence relations.

\begin{definition}
\label{def:umspace}
\lit
An \emph{ultrametric space} is a type \(X\) equipped with an \emph{ultrametric} \(d : X \to X \to \R_{\geq 0}\) such that,
\begin{enumerate}[(i)]
\item \(d(x,y) = 0 \leftrightarrow x = y\),
\item \(d(x,y) = d(y,x)\),
\item \(d(x,z) \leq \mathsf{max}(d(x,y),d(y,z))\).
\end{enumerate}
\end{definition}

\begin{lemma}
\lit
Every ultrametric space is a metric space.
\end{lemma}
\begin{proof}
The first two conditions are the same --- the third condition of ultrametric spaces implies that of metric spaces because clearly \(\mathsf{max}(d(x,y),d(y,z)) \leq d(x,y) + d(y,z)\).
\end{proof}

\begin{lemma}
\label{lem:umtrans}
\lit
Given an ultrametric space \(X\) and distance \(\ty{\varepsilon}{\R_{\geq 0}}\), the closeness relation \(C_\varepsilon\) is transitive.
\end{lemma}
\begin{proof}
Given \(C_\varepsilon(x,y) := d(x,y) \leq \varepsilon\) and \(C_\varepsilon(y,z) := d(y,z) \leq \varepsilon\) we have that \(\mathsf{max}(d(x,y),d(y,z)) \leq \varepsilon\).
By \cref{def:umspace}.(ii), \(d(x,z) \leq \mathsf{max}(d(x,y),d(y,z))\) and therefore by transitivity of the order \(d(x,z) \leq \varepsilon =: C_\varepsilon(x,z)\).
\end{proof}

\begin{corollary}
\label{cor:umeq}
\lit
Given an ultrametric space \(X\) and distance \(\ty{\varepsilon}{\R_{\geq 0}}\), the closeness relation \(C_\varepsilon\) is an equivalence relation.
\end{corollary}
\begin{proof}
By \cref{lem:mrefsym,lem:umtrans}.
\end{proof}

We now use these closeness relations to define continuity and uniform continuity on (ultra)metric spaces.

\begin{definition}
\label{def:metr-cont}
\lit
Given (ultra)metric spaces \(X\) and \(Y\), a function \(\ty{f}{X \to Y}\) is \emph{continuous} if for all \(\ty{x_1}{X}\) and \(\ty{\varepsilon}{\R_{\geq 0}}\) there is some \(\ty{\delta}{\R_{\geq 0}}\) such that elements that are \(\delta\)-close to \(x_1\) map to elements that are \(\varepsilon\)-close to \(f(x_1)\):
\begin{multline*}
\mathsf{metric{\hy}f{\hy}continuous}(f) := \\
\Pitye{\varepsilon}{\R_{\geq 0}}{\Pitye{x_1}{X}{\Sigmatye{\delta}{\R_{\geq 0}}{\Pity{x_2}{X}{C_\delta(x_1,x_2) \to C_\varepsilon(f(x_1),f(x_2))}}}}.
\end{multline*}
\end{definition}

\begin{definition}
\label{def:metr-ucont}
\lit
Given (ultra)metric spaces \(X\) and \(Y\), a function \(\ty{f}{X \to Y}\) is \emph{uniformly continuous} if for all \(\ty{\varepsilon}{\R_{\geq 0}}\) there is some \(\ty{\delta}{\R_{\geq 0}}\) such that elements that are \(\delta\)-close map to elements that are \(\varepsilon\)-close:
\begin{multline*}
\mathsf{metric{\hy}f{\hy}ucontinuous}(f) := \\
\Pitye{\varepsilon}{\R_{\geq 0}}{\Sigmatye{\delta}{\R_{\geq 0}}{\Pity{x_1,x_2}{X}{C_\delta(x_1,x_2) \to C_\varepsilon(f(x_1),f(x_2))}}}.
\end{multline*}
\end{definition}

\begin{lemma}
\lit
Every uniformly continuous function between metric spaces is continuous.
\end{lemma}
\begin{proof}
If \(\ty{\delta}{\R_{\geq 0}}\) is a modulus of uniform continuity for \(\ty{f}{X \to Y}\) (i.e. it depends only on \(\ty{\varepsilon}{\R_{\geq 0}}\) and not on any point \(\ty{x_1}{X}\)) then it is a modulus of continuity for any \(x_1\).
\end{proof}

\noindent
For the special case where the domain of the function is a truth value --- and hence the function is a truth-valued predicate --- the only notion of closeness we care about in the domain is logical equivalence.
Therefore, we specialise uniform continuity for predicates.

\begin{definition}
\label{def:metr-ucont-pred}
\lit
A predicate \(\ty{p}{X \to \Omega}\) on an (ultra)metric space \(X\) is \emph{uniformly continuous} if there is some \(\ty{\delta}{\R_{\geq 0}}\) such that sequences that are \(\delta\)-close give the same answer to the predicate:
\[ \mathsf{metric{\hy}p{\hy}ucontinuous}(p) := \Sigmatye{\delta}{\R_{\geq 0}}{\Pity{x_1,x_2}{X}{C_\delta(x_1,x_2) \to p(x_1) \Leftrightarrow p(x_2)}} .\]
\end{definition}

\subsection{Extended naturals and definition of closeness spaces}

It will be convenient for our purposes to avoid the use of real numbers at this point and formulate, primarily for the purposes of continuity, an alternative notion of ultrametric spaces which we call \emph{closeness spaces}.
In a closeness space, the ultrametric is replace by a different binary function called a \emph{closeness function}.
The idea is that while (ultra)metrics measure the distance between the given objects, closeness functions measure the closeness between them: two identical elements of a type have closeness infinity.
Thus, closeness functions give values on the type of extended natural numbers \(\Ni\) defined below.

We start this subsection by recalling the \textsc{TypeTopology} definition of extended naturals.

\begin{definition}
\typetop{CoNaturals}{GenericConvergentSequence}{\urlN\urlinfty}
The \emph{extended natural numbers} type \(\Ni\) is the type of decreasing binary sequences,
\[ \Ni  := \sigmaty{\alpha}{\seq \2}{\mathsf{is{\hy}decreasing} \ \alpha}, \]
where
\(\mathsf{is{\hy}decreasing}(\alpha) := \prod_{(i : \mathbb N)} (\alpha_i \geq \alpha_{i + 1})\).
\end{definition}

We can lift any natural \(\ty{n}{\N}\) by defining \(\ty{{\underline n}}{\Ni}\) as the sequence of \(n\)-many \(1\)s followed by infinitely-many \(0\)s. \(\ty{\infty}{\Ni}\) is therefore defined as the sequence of infinitely-many \(1\)s.
Given this, the partial order on the extended naturals --- which extends that on the naturals --- is given below, along with its minimum and maximum.

\begin{definition}
\label{def:Ni-order}
\typetop{CoNaturals}{GenericConvergentSequence}{_\urlpreceq\urlN\urlinfty_}
The partial order on the extended naturals is defined as follows, \[u \preceq v := \Pity{n}{\N}{u_n = 1 \to v_n = 1} .\]
\end{definition}

\begin{lemma}
\label{lem:Ni-is-preorder}
\typetop{Ordinals}{Notions}{\urlpreceq-refl}
\typetop{Ordinals}{Notions}{\urlpreceq-trans}
The partial order on extended naturals \(\ty{\preceq}{\Ni \to \Ni \to \Omega}\) is indeed a preorder (see the later \cref{def:preorder}); i.e.\ it is reflexive and transitive.
\end{lemma}
\begin{proof}
Both are straightforward: for reflexivity, clearly for any \(\ty{n}{\N}\) if \(u_n = 1\) then \(u_n = 1\); for transitivity, if we have \(u_n = 1 \to v_n = 1\) and also \(v_n = 1 \to w_n = 1\) then \(u_n = 1 \to w_n = 1\).
\end{proof}

\begin{lemma}
\label{lem:Ni-preorder-extend}
Given \(\ty{n,m}{\N}\) such that \(n \leq m\), it is the case that \(\underline n \preceq \underline m\).
\end{lemma}
\begin{proof}
Recall that \(\ty{\underline n, \underline m}{\Ni}\) are the sequences of \(n\)-many 1s and \(m\)-many 1s respectively; i.e.\ \(\underline n \sim^n 1\) and \(\underline m \sim^m 1\). Because \(n \leq m\), this means we also have \(\underline m \sim^n 1\); thus, by transitivity of \(\sim^n\) (\cref{lem:sim-eq}) we have \(\underline n \sim^n \underline m\), meaning that at every point \(\ty{i}{\N}\) where \({\underline n}_i = 1\) then also \({\underline m}_i = 1\).
\end{proof}

\begin{lemma}
\label{lem:preceq-bounds}
\typetop{CoNaturals}{GenericConvergentSequence}{\urlinfty-largest}
\typetop{CoNaturals}{GenericConvergentSequence}{Zero-smallest}
Given any \(\ty{u}{\Ni}\), we have that \({\underline 0} \preceq u\) and \(u \preceq \infty\).
\end{lemma}
\begin{proof}
For the former, \(\underline 0\) is never \(1\) and so, for every \(\ty{n}{\N}\), the antecedent of the implication \({\underline 0}_n = 1 \to u_n = 1\) is always empty, making the overall implication always true.
For the latter, \(\infty\) is always \(1\) and so the subsequent of \(u_n = 1 \to \infty_n = 1\) is always true, making the overall implication always true.
\end{proof}

\noindent
In general, the partial order on the extended naturals is not decidable (i.e.\ we cannot decide whether or not \(u \preceq v\) for any \(\ty{u,v}{\Ni}\)), but it is decidable when either side is a natural number. We give the left-hand side proof of this.

\begin{lemma}
\label{lem:preceq-left-decidable}
\thesislit{3}{ClosenessSpaces}{\urlpreceq-left-decidable}
Given any \(\ty{n}{\N}\) and \(\ty{v}{\Ni}\), it is the case that \({\underline n} \preceq v\) is decidable.
\end{lemma}
\begin{proof}
We proceed by induction on \(\ty{n}{\N}\). When \(n := 0\), then the result follows by \cref{lem:preceq-bounds}.

When \(n := \suc{n'}\), we want to check that for every \(\ty{i}{\N}\) such that \(\underline{\suc{n'}}_i = 1\) then also \(v_i = 1\).
It suffices to simply check whether \(v_{n'} = 1\), because we know \(\underline{\suc{n'}}_i = 1\) for all \(i \leq n'\) (by the decreasing property), and so the value of \(v_i\) only matters when \(i \leq n'\); i.e.\ if it is \(0\) at any of these points then \(\neg (\underline{\suc{n'}} \preceq v)\). By the decreasing property, if \(v_i = 0\) at any point \(i \leq n'\), then \(v_{n'} = 0\).

We proceed by checking whether \(v_{n'} = 1\) (which is decidable by \cref{lem:fin-discrete}).
If \(v_{n'} = 1\) then \(\underline{\suc{n'}} \preceq v\). On the other hand, if \(\neg(v_{n'} = 1)\) then \(v_{n'} = 0\) and therefore \(\neg (\underline{\suc{n'}} \preceq v)\).
\end{proof}

We can also define the minimum of two extended naturals as the extension of the minimum of two binary digits.

\begin{definition}
\typetop{CoNaturals}{GenericConvergentSequence}{min}
The \emph{minimum} of two extended naturals \(\mathsf{min}(u,v):\Ni\) is given via
\( \mathsf{min} := \lambda (\ty{n}{\N}).\mathsf{min}(u_n,v_n)\), which is clearly decreasing as its arguments are.
\end{definition}

We can now define closeness spaces which, by comparing the below to \cref{def:umspace}, the reader can see are a kind of dual of ultrametric spaces.

\begin{definition}
\label{def:cspace}
\thesislit{3}{ClosenessSpaces}{ClosenessSpace}
A \emph{closeness space} is a type \(X\) equipped with a \emph{closeness function} \(\ty{c}{\closeness X}\) such that,
\begin{enumerate}
\item \(c(x,y) = \infty \leftrightarrow {x = y}\),
\item \(c(x,y) = c(y,x)\),
\item \(\mathsf{min}(c(x,y),c(y,z)) \preceq c(x,z)\).
\end{enumerate}
\end{definition}

\begin{remark}[Connection between closeness and metric spaces]
\label{remark:closeness-metric}
Every closeness space \(X\) is informally an ultrametric space (and, hence, a metric space): using the closeness function \(\ty{c}{\closeness X}\), one informally defines the ultrametric \(\ty{d}{X \to X \to \R_{\geq 0}}\) as \(d(x,y) := 2^{-c(x,y)}\), with the usual convention that \(2^{-\infty} := 0\).
By this ultrametric, the distance between objects decreases towards zero as their closeness increases towards infinity.
\end{remark}

\subsection{Closeness relations and continuity}

The reformulated definitions of closeness relations (the reader can compare to \cref{def:metr-closerelation} on metric spaces) on closeness spaces are given below.

\begin{definition}
\mbox{}
\label{def:C}
\thesislit{3}{ClosenessSpaces}{C}\textsuperscript{fe}
Given a closeness space \(X\), the family of \emph{closeness relations} is defined as follows,
\begin{alignat*}{2}
C \ &: \N \to X \to X \to \Omega, \\
C_\varepsilon &(x,y) := {\underline \varepsilon} \preceq c(x,y).
\end{alignat*}
\end{definition}

\noindent
As before, for all \(\ty{\varepsilon}{\N}\) and \(\ty{x,y}{X}\), we say that \(x\) and \(y\) are \(\varepsilon\)-close if \(C_\varepsilon(x,y)\).
\noindent
We borrow further terminology from topology and say that if \(x\) and \(y\) are \(\varepsilon\)-close then they are in the same \emph{\(\varepsilon\)-neighbourhood}.

\begin{remark}
\label{remark:0-close}
Every element of a closeness space is \(0\)-close to every other element by \cref{lem:preceq-bounds}.
\end{remark}

Because closeness spaces are ultrametric spaces, closeness relations on them are equivalence relations.

\begin{lemma}
\label{lem:C-eq}
\thesislit{3}{ClosenessSpaces}{C-refl}
\thesislit{3}{ClosenessSpaces}{C-sym}
\thesislit{3}{ClosenessSpaces}{C-trans}
Given a closeness space \(X\) and precision \(\ty{\varepsilon}{\N}\), the closeness relation \(C_\varepsilon\) is an equivalence relation (reflexive, symmetric and transitive).
\end{lemma}
\begin{proof}
Reflexivity follows from \cref{def:cspace}.1 ; i.e.\ \(C_\varepsilon(x,x) := \left( {\underline \varepsilon} \preceq c(x,x) \right) = \left( {\underline \varepsilon} \preceq \infty \right) \), which follows for all \(\ty{x}{X}\) from \cref{lem:preceq-bounds}.
Symmetry follows from \cref{def:mspace}.2; i.e.\ \(C_\varepsilon(x,y) := \left( {\underline \varepsilon} \preceq c(x,y) \right) = \left( {\underline \varepsilon} \preceq c(y,x) \right) =: C_\varepsilon(y,x)\) for all \(\ty{x,y}{X}\).
Transitivity follows from \cref{def:mspace}.3; i.e.\ \(C_\varepsilon(x,z) := \left( {\underline \varepsilon} \preceq c(x,z) \right) \) is proved by first showing that \(\left( {\underline \varepsilon} \preceq \mathsf{min}(c(x,y),c(y,z)) \right)\), which is clearly the case by the assumptions \(C_\varepsilon(x,y) := \left( {\underline \varepsilon} \preceq c(x,y) \right)\) and \(C_\varepsilon(y,z) := \left( {\underline \varepsilon} \preceq c(y,z) \right)\) --- the result then follows by using the transitivity of \(\preceq\) (\cref{lem:Ni-is-preorder}).
\end{proof}

\noindent
Closeness functions are monotonically decreasing: if two elements are \(\varepsilon_2\)-close, then they are \(\varepsilon_1\)-close for all \(\varepsilon_1 \leq \varepsilon_2\).

\begin{corollary}
\label{cor:C-mono}
\thesislit{3}{ClosenessSpaces}{C-mono}
Given a closeness space \(X\), precisions \(\ty{\varepsilon_1,\varepsilon_2}{\N}\) such that \(\varepsilon_1 \leq \varepsilon_2\) and elements \(\ty{x,y}{X}\), if \(C_{\varepsilon_2}(x,y)\) then \(C_{\varepsilon_1}(x,y)\).
\end{corollary}
\begin{proof}
By \cref{lem:Ni-preorder-extend}.
\end{proof}

\noindent
Furthermore, we can always decide whether two elements are \(\varepsilon\)-close.

\begin{lemma}
\label{lem:C-dec}
\thesislit{3}{ClosenessSpaces}{C-decidable}
Given a closeness space \(X\), precision \(\ty{\varepsilon}{\N}\) and elements \(\ty{x,y}{X}\), it is the case that \(C_\varepsilon(x,y)\) is decidable.
\end{lemma}
\begin{proof}
By \cref{lem:preceq-left-decidable}.
\end{proof}

Critically, we now reformulate continuity (the reader can compare to \cref{def:metr-cont,def:metr-ucont,def:metr-ucont-pred} on metric spaces) for closeness spaces. This definition will be used throughout the rest of the thesis.

\begin{definition}
\label{def:cf-cont}
\thesislit{3}{ClosenessSpaces}{f-continuous}
Given closeness spaces \(X\) and \(Y\), a function \(\ty{f}{X \to Y}\) is \emph{continuous} if for all \(\ty{x_1}{X}\) and \(\ty{\varepsilon}{\N}\) there is some \(\ty{\delta}{\N}\) such that elements that are \(\delta\)-close to \(x_1\) map to elements that are \(\varepsilon\)-close to \(f(x_1)\):
\[\mathsf{f{\hy}continuous}(f) := \Pitye{\varepsilon}{\N}{\Pitye{x_1}{X}{\Sigmatye{\delta}{\N}{\Pity{x_2}{X}{C_\delta(x_1,x_2) \to C_\varepsilon(f(x_1),f(x_2))}}}}.
\]
\end{definition}

\begin{definition}
\label{def:clos-ucont}
\thesislit{3}{ClosenessSpaces}{f-ucontinuous}
Given closeness spaces \(X\) and \(Y\), a function \(\ty{f}{X \to Y}\) is \emph{uniformly continuous} if for all \(\ty{\varepsilon}{\N}\) there is some \(\ty{\delta}{\N}\) such that elements that are \(\delta\)-close map to elements that are \(\varepsilon\)-close:
\[\mathsf{f{\hy}ucontinuous}(f) :=
\Pitye{\varepsilon}{\N}{\Sigmatye{\delta}{\N}{\Pity{x_1,x_2}{X}{C_\delta(x_1,x_2) \to C_\varepsilon(f(x_1),f(x_2))}}}.\]
\end{definition}

\begin{lemma}
\thesislit{3}{ClosenessSpaces}{ucontinuous-continuous}
Every uniformly continuous function between closeness spaces is continuous.
\end{lemma}
\begin{proof}
If \(\ty{\delta}{\N}\) is a modulus of uniform continuity for \(\ty{f}{X \to Y}\) (i.e. it depends only on \(\ty{\varepsilon}{\N}\) and not on any point \(\ty{x_1}{X}\)) then it is a modulus of continuity for any \(x_1\).
\end{proof}

\begin{definition}
\label{def:clos-ucont-pred}
\thesislit{3}{ClosenessSpaces}{p-ucontinuous}
A predicate \(\ty{p}{X \to \Omega}\) on a closeness space \(X\) is \emph{uniformly continuous} if there is some \(\ty{\delta}{\N}\) such that sequences that are \(\delta\)-close give the same answer to the predicate:
\[ \mathsf{p{\hy}ucontinuous}(p) := \Sigmatye{\delta}{\N}{\Pity{x_1,x_2}{X}{C_\delta(x_1,x_2) \to p(x_1) \Leftrightarrow p(x_2)}} .\]
\end{definition}

As can be expected, the identity function is uniformly continuous and composition preserves uniform continuity.

\begin{lemma}
\label{lem:id-comp-ucont}
\thesislit{3}{ClosenessSpaces}{id-ucontinuous}
\thesislit{3}{ClosenessSpaces}{f-ucontinuous-comp}
The identity function is uniformly continuous and composition of functions preserves uniform continuity.
\end{lemma}
\begin{proof}
For the identity function, we need to show for all \(\ty{\varepsilon}{\N}\) there is \(\ty{\delta}{\N}\) that \(C_\delta(x_1,x_2)\) implies \(C_\varepsilon(\mathsf{id}(x_1),\mathsf{id}(x_2))\) --- thus we just set \(\delta := \varepsilon\).

For composition of uniformly continuous functions \(\ty{f}{X \to Y}\) and \(\ty{g}{Y \to Z}\) (i.e. for all \(\ty{\varepsilon}{\N}\) we have some \(\ty{\delta_f^\varepsilon,\delta_g^\varepsilon}{\N}\) such that \(C_{\delta_f^\varepsilon}(x_1,x_2) \to C_\varepsilon(f(x_1),f(x_2))\) and \(C_{\delta_g^\varepsilon}(y_1,y_2) \to C_\varepsilon(g(y_1),g(y_2))\)), we need to show that for all \(\ty{\varepsilon}{\N}\) there is some \(\ty{\delta}{\N}\) such that \(C_{\delta}(x_1,x_2) \to C_{\varepsilon}(g(f(x_1)),g(f(x_2)))\).
By setting
\(\delta := {\delta_f}^{{\delta_g}^\varepsilon}\), we find that \(C_{{\delta_f}^{{\delta_g}^\varepsilon}}(x_1,x_2) \to C_{{\delta_g}^\varepsilon}(f(x_1),f(x_2)) \to C_\varepsilon(f(g(x_1)),f(g(x_2))\).
\end{proof}

\noindent
Further, and importantly for the purpose of function search, we can compose uniformly continuous functions and predicates to form a new uniformly continuous predicate.

\begin{lemma}
\label{lem:f-p-ucont}
\thesislit{3}{ClosenessSpaces}{p-ucontinuous-comp}
Given closeness spaces \(X\) and \(Y\), a uniformly continuous function \(\ty{f}{X \to Y}\) and a uniformly continuous (and decidable) predicate \(\ty{p}{Y \to \Omega}\), the predicate \[ p_f(x) := p(f(x)), \] is uniformly continuous (and decidable).
\end{lemma}
\begin{proof}
The modulus of uniform continuity of the predicate \(\ty{p_f}{X \to Y}\) is the modulus of uniform continuity of \(f\) at point \(\delta\), where \(\delta\) is the modulus of uniform continuity of \(p\).
\end{proof}

\noindent
Finally, we note that closeness relations themselves yield uniformly continuous and decidable predicates.

\begin{lemma}
\label{lem:closeness-ucd-pred}
\thesislit{3}{ClosenessSpaces}{C-decidable-uc-predicate-l}
\thesislit{3}{ClosenessSpaces}{C-decidable-uc-predicate-r}
Given a closeness space \(X\) and precision \(\ty{\varepsilon}{\N}\), the predicates \[ p^y_l(x) := C_\varepsilon(x,y) \text{ and } p^y_r(x) := C_\varepsilon(y,x) \] are uniformly continuous and decidable for any \(\ty{y}{X}\).
\end{lemma}
\begin{proof}
By the decidability (\cref{lem:C-dec}) and transitivity (\cref{lem:C-eq}) of closeness relations.
\end{proof}

We have shown, therefore, that closeness spaces yield a formulation of continuity that is exactly related to continuity on metric spaces; but which we can manipulate more conveniently for the use cases of this thesis.

\subsection{Totally bounded closeness spaces}

We previously mentioned that searchable types relate to compact spaces --- recall from topology that a metric space is compact if it is totally bounded and complete, meaning that totally bounded spaces are generalisations of compact spaces~\cite{Sutherland}.
In this subsection, we will reformulate the idea of \emph{totally bounded metric spaces} for closeness spaces, which generalise the variant of searchable types we use to search infinite types in \cref{sec:search-infinite}.

In order to achieve this, we first define \emph{\(\varepsilon\)-nets} on closeness spaces. Recall that an \(\varepsilon\)-net on a metric space \(X\) is a finite subset \(\{x_0,...,x_{n-1}\}\) of \(X\) such that the union of the \(n\)-many open balls with these points at their center is \(X\) itself~\cite{Sutherland}.

\begin{definition}
\label{def:cover-original}
\lit
Given a closeness space \(X\), a finite linearly ordered type \(X'\) equipped with a function \(\ty{g}{X' \to X}\) is an \emph{\(\varepsilon\)-net of \(X\)} if for all \(\ty{x}{X}\) there is an element \(\ty{x'}{X'}\) such that \(C_\varepsilon(x,g(x')))\).
\end{definition}

\noindent
As our interpretation of the quantifier ``there is" in \cref{def:cover-original} is \(\Sigma\), rather than \(\exists\), we can use the following equivalent definition.

\begin{definition}
\label{def:cover}
\thesislit{3}{ClosenessSpaces}{_is_net-of_}
Given a closeness space \(X\), a finite linearly ordered type \(X'\) equipped with a function \(\ty{g}{X' \to X}\) is an \emph{\(\varepsilon\)-net of \(X\)} if there is a function \(\ty{h}{X \to X'}\) such that for all \(\ty{x}{X}\) we have \(C_\varepsilon(x,g(h(x)))\).
\end{definition}

By this definition, every element \(\ty{x}{X}\) of a closeness space \(X\) with an \(\varepsilon\)-net \((X',g)\) is represented by at least one point \(\ty{h(x)}{X'}\) of the net --- indeed, these points represent the whole \(\varepsilon\)-neighbourhood in which they lie.
Furthermore, because closeness relations are equivalence relations for closeness spaces, two elements \(\ty{x'_1,x'_2}{X'}\) of the net either represent the exact same \(\varepsilon\)-neighbourhood of \(X\) or completely disjoint\footnote{This means that each closeness space with an \(\varepsilon\)-net has a minimal \(\varepsilon\)-net obtained by discarding elements of the net that represent the same \(\varepsilon\)-neighbourhoods --- though, in this thesis, we never require the net to be minimal.}.
As an \(\varepsilon\)-net is finite, the existence of such a net implies that the number of \(\varepsilon\)-neighbourhoods of \(X\) is also finite.
Because every element of a closeness space is \(0\)-close to every other element, any pointed type (e.g.\ \(\1\)) is a \(0\)-net for any closeness space).

Reflecting the definition on metric spaces, a closeness space is totally bounded if there is an \(\varepsilon\)-net for every \(\ty{\varepsilon}{\N}\).

\begin{definition}
\label{def:totallybounded}
\thesislit{3}{ClosenessSpaces}{totally-bounded}
A closeness space \(X\) is \emph{totally bounded} if, for every precision \(\ty{\varepsilon}{\N}\), there is an \(\varepsilon\)-net of \(X\).
\end{definition}

\noindent
If a closeness space is totally bounded, it has a finite number of \(\varepsilon\)-neighbourhoods for every \(\ty{\varepsilon}{\N}\).

\subsection{Examples of (totally bounded) closeness spaces}
\label{sec:cspace-examples}

In this subsection, we give examples of how we can build a wide class of closeness spaces, many of which are additionally totally bounded (and, therefore, admit search).
Note that, unlike metric spaces, closeness spaces are only ever defined on `intensional' mathematical spaces, such as the Cantor space \(\seq \2\), and are not defined on `extensional' spaces such as the unit interval \([0,1]\).

In each case, we prove that the closeness space is indeed a closeness space (i.e.\ it satisfies the three properties of \cref{def:cspace}).
When proving a closeness function is totally bounded, the \(0\) case is left out due to its triviality.

\subsubsection{Discrete closeness spaces}

\begin{definition}
\label{def:discrete-closeness}
\thesislit{3}{ClosenessSpaces-Examples}{discrete-clofun'}
Given a discrete type \(X\), the \emph{discrete closeness function} \(\ty{c_X}{\closeness X}\) is defined by case analysis of \(d(x,y) : \mathsf{decidable}(x = y)\) for arguments \(\ty{x,y}{X}\),
\begin{itemize}
\item If \(x = y\) then \(c_X(x,y) := \infty\),
\item If \(x \neq y\) then \(c_X(x,y) := {\underline 0}\).
\end{itemize}
\end{definition}

\begin{lemma}
\label{lem:discrete-cspace}
\thesislit{3}{ClosenessSpaces-Examples}{D-ClosenessSpace}
Discrete types are closeness spaces by the discrete closeness function.
\end{lemma}
\begin{proof}
We prove that each of the three conditions of \cref{def:cspace} is satisfied:
\begin{enumerate}[(i)]
\item
\begin{enumerate}
    \item If \(x = y\), then \(c_X(x,y) := \infty\) by definition.
    \item If \(c_X(x,y) := \infty\), then we check whether or not \(x = y\). If \(\neg (x = y)\), then we get a contradiction (because \(c_X(x,y)_n := 1\) and \(c_X(x,y)_n := 0\) for every \(\ty{n}{\N}\)) and hence we can derive \(x = y\).
    Therefore, \(x = y\) in both cases.
\end{enumerate}
\item Because \(x = y \leftrightarrow y = x\) and \(\neg (x = y) \leftrightarrow \neg (y = x)\), \(c_X(x,y) = c_X(y,x)\).
\item We proceed by case analysis on \(d(x,y)\) and \(d(y,z)\):
    \begin{enumerate}
    \item In the case where both \(x = y\) and \(y = z\), then \(x = z\).
    Hence, \(\mathsf{min}(c_X(x,y),c_X(y,z)) \preceq c_X(x,z)\) reduces to \(\mathsf{min}(\infty,\infty) \preceq \infty\).
    This is trivially satisfied because, at each point \(\ty{n}{\N}\), it is the case that \(\mathsf{min}(\infty,\infty)_n := 1\) and \(\infty_n := 1\).
    This follows from \cref{lem:preceq-bounds}.
    \item Alternatively, if either \(x \neq y\) or \(y \neq z\), then \(\mathsf{min}(c_X(x,y),c_X(y,z)) \preceq c_X(x,z)\) reduces to \({\underline 0} \preceq c_X(x,z)\).
    This follows from \cref{lem:preceq-bounds}.
    \end{enumerate}
\end{enumerate}
\end{proof}

\begin{lemma}
\thesislit{3}{ClosenessSpaces-Examples}{finite-totally-bounded}
Given a discrete closeness space \(X\), if \(X\) is finite linearly ordered then it is totally bounded.
\end{lemma}
\begin{proof}
For any \(\varepsilon := \suc{\varepsilon'}\), the only \(\varepsilon\)-net is \(X\) itself, as for any \(\ty{x}{X}\) the only element \(\ty{y}{Y}\) that can satisfy \(C_\varepsilon(x,y)\) is such that \(y := x\); therefore \(\ty{g,h}{X \to X}\) should simply be the identity map.
\end{proof}

\subsubsection{Disjoint union of closeness spaces}

\begin{definition}
\label{def:+-closeness}
\thesislit{3}{ClosenessSpaces-Examples}{\urlplus-clofun'}
Given closeness spaces \(X\) and \(Y\), the \emph{disjoint union closeness function} is defined by,
\begin{alignat*}{3}
c_{X + Y} &: X + Y && \to X + Y \to \Ni, \\
c_{X + Y} &(\inl \ x_1, && \inl \ x_2) && := c_X(x_1,x_2), \\
c_{X + Y} &(\inr \ y_1, && \inr \ y_2) && := c_Y(y_1,y_2), \\
c_{X + Y} &(\inl \ x  , && \inr \ y) && := {\underline 0}, \\
c_{X + Y} &(\inr \ y  , && \inl \ x) && := {\underline 0}.
\end{alignat*}
\end{definition}

\begin{lemma}
\thesislit{3}{ClosenessSpaces-Examples}{\urlplus-ClosenessSpace}
Given closeness spaces \(X\) and \(Y\), the type \(X + Y\) is a closeness space by the disjoint union closeness function.
\end{lemma}
\begin{proof}
We prove that each of the three conditions of \cref{def:cspace} is satisfied:
\begin{enumerate}[(i)]
\item 
\begin{enumerate}
    \item To show that \(c_{X + Y}({xy}_1,{xy}_2) = \infty\) whenever \({xy}_1 = {xy}_2\), we first recognise that the equality means either \({xy}_1 := \mathsf{inl}\ x_1 = \mathsf{inl}\ x_2 =: {xy_2}\) (for \(\ty{x_1,x_2}{X}\)) or \({xy}_1 := \mathsf{inr}\ y_1 = \mathsf{inr}\ y_2 =: {xy_2}\) (for \(\ty{y_1,y_2}{Y}\)) holds. As \(\inl\) and \(\inr\) are both embeddings (\cref{def:embedding}), in the former case \(x_1 = x_2\) and in the latter \(y_1 = y_2\). The former case is reduced to showing \(c_X(x_1,x_2) = \infty\), and the latter is reduced to showing \(c_Y(y_1,y_2) = \infty\). In either case we achieve the result by the assumption that the underlying closeness functions satisfy condition (i).(a).
    \item To show that \({xy}_1 = {xy}_2\) given \(c_{X + Y}({xy}_1,{xy}_2) = \infty\), we proceed by induction on \({xy}_1,{xy}_2\). The latter two cases of the definition of the disjoint union closeness function (\cref{def:+-closeness}) are impossible, as \(\infty \neq \underline 0\), and thus either \({xy}_1 := \mathsf{inl}\ x_1\) and \({xy_2} := \mathsf{inl}\ x_2\) for \(\ty{x_1,x_2}{X}\) or \({xy}_1 := \mathsf{inr}\ y_1\) and \({xy_2} := \mathsf{inr}\ y_2\) for \(\ty{y_1,y_2}{X}\). In the former case we now have \(c_{X + Y}(\inl\ x_1, \inl\ x_2) := c_X(x_1,x_2) = \infty\) and want to show \(x_1 = x_2\), and in the latter we have \(c_{X + Y}(\inr\ y_1, \inr\ y_2) := c_Y(y_1,y_2) = \infty\) and want to show \(y_1 = y_2\).
    In either case we achieve the result by the assumption that the underlying closeness functions satisfy condition (i).(b).
\end{enumerate}
\item To show that \(c_{X+Y}\) is symmetric, we check each case of its definition (\cref{def:+-closeness}). The first two cases are each symmetric by the assumption that the underlying closeness functions satisfy condition (ii); the latter two cases are already symmetries of each other, and so they are also symmetric.
\item To show that \(\mathsf{min}(c_{X+Y}({xy}_1,{xy}_2),c_{X+Y}({xy}_2,{xy}_3)) \preceq c_{X+Y}({xy}_1,{xy}_3)\) holds, we proceed by induction on \({xy}_1,{xy}_2,{xy}_3\):
\begin{itemize}
    \item[] If all three elements are from the same side of the disjoint union, the result follows by the assumption that the underlying closeness functions satisfy condition (iii),
    \item[] If \({xy}_1\) and \({xy}_2\), or \({xy}_2\) and \({xy}_3\), are from different sides of the disjoint union, then the left side of the inequality reduces to \(\underline 0\) and we simply need to show \(\underline 0 \preceq c_{X+Y}({xy}_1,{xy}_3)\). This is by \cref{lem:preceq-bounds}.
    \item[] If \({xy}_1\) and \({xy}_3\) are from different sides of the disjoint union, then we need to show \(\mathsf{min}(c_{X+Y}({xy}_1,{xy}_2),c_{X+Y}({xy}_2,{xy}_3)) \preceq \underline 0\), which can only occur if \(\mathsf{min}(c_{X+Y}({xy}_1,{xy}_2),c_{X+Y}({xy}_2,{xy}_3)) \sim \underline 0\). This is indeed the case, as no matter which side \({xy}_2\) is from, it will differ from either \({xy}_1\) or \({xy}_3\), and so the left side of the inequality reduces to \(\underline 0\).
\end{itemize}
\end{enumerate}
\end{proof}

\begin{lemma}
\thesislit{3}{ClosenessSpaces-Examples}{\urlplus-totally-bounded}
Given totally bounded closeness spaces \(X\) and \(Y\), the disjoint union closeness space \(X + Y\) is totally bounded.
\end{lemma}
\begin{proof}
For a given \(\ty{\varepsilon}{\N}\), both \(X\) and \(Y\) have \(\varepsilon\)-nets; i.e. there are types \(X'\) and \(Y'\) and functions \(\ty{g_X}{X' \to X}\), \(\ty{g_Y}{Y' \to Y}\), \(\ty{h_X}{X \to X'}\), \(\ty{h_Y}{Y \to Y'}\) such that for all \(\ty{x}{X}\) and \(\ty{y}{Y}\) we have \(C_\varepsilon(x,g_X(h_X(x)))\) and \(C_\varepsilon(y,g_Y(h_Y(y)))\).
To show that \(X + Y\) has an \(\varepsilon\)-net, we define the following functions:
\begin{alignat*}{3}
g_{X + Y} &: {X' + Y'} \to {X + Y}, \\
g_{X + Y} &(\inl\ x') := \inl\ g_X(x') ,\\
g_{X + Y} &(\inr\ y') := \inl\ g_Y(y') ,\\
h_{X + Y} &: {X + Y} \to {X' + Y'}, \\
h_{X + Y} &(\inl\ x) := \inl\ h_X(x) ,\\
h_{X + Y} &(\inr\ y) := \inl\ h_Y(y).
\end{alignat*}
We then need to show that for all \(\ty{xy}{X + Y}\) we have \(C_\varepsilon(xy,g_{X + Y}(h_{X + Y}(xy))\). This follows by induction on \({xy}\):
\begin{itemize}
\item[] If \({xy} := \inl\ x\) for some \(\ty{x}{X}\) then \(C_\varepsilon(xy,g_{X + Y}(h_{X + Y}(xy)) := C_\varepsilon(x,g_{X}(h_{X}(x))\) (by definition of \(g_{X + Y}\), \(h_{X + Y}\) and \cref{def:+-closeness}), which we have by the fact \((X',g_X)\) is an \(\varepsilon\)-net for \(X\),
\item[] Similarly, if \({xy} := \inr\ y\) for some \(\ty{y}{Y}\) then \(C_\varepsilon(xy,g_{X + Y}(h_{X + Y}(xy)) := C_\varepsilon(y,g_{Y}(h_{Y}(y))\), which we have by the fact \((Y',g_Y)\) is an \(\varepsilon\)-net for \(Y\).
\end{itemize}
Therefore the type \(X' + Y'\), equipped with the functions we defined above, is an \(\varepsilon\)-net for \(X + Y\). 
\end{proof}

\subsubsection{Finite product closeness spaces}

\begin{definition}
\label{def:prod-closeness}
\thesislit{3}{ClosenessSpaces-Examples}{\urlx-clofun'}
Given closeness spaces \(X\) and \(Y\), the \emph{binary product closeness function} is defined by,
\begin{alignat*}{3}
c_{X \x Y} &: \closeness{X \x Y}, \\
c_{X \x Y} &((x_1,y_1),(x_2,y_2)) := \mathsf{min}(c_X(x_1,x_2),c_Y(y_1,y_2)).
\end{alignat*}
\end{definition}

\begin{lemma}
\label{lem:bin-cspace}
\thesislit{3}{ClosenessSpaces-Examples}{\urlx-ClosenessSpace}
Given closeness spaces \(X\) and \(Y\), the type \(X \x Y\) is a closeness space by the binary product closeness function.
\end{lemma}
\begin{proof} \axioms{f}
We prove that each of the three conditions of \cref{def:cspace} is satisfied:
\begin{enumerate}[(i)]
\item 
\begin{enumerate}
\item To show that \(c_{X \x Y}((x_1,y_1),(x_2,y_2)) = \infty\) whenever \((x_1,y_1) = (x_2,y_2)\), we first recognise that the equality means both \(x_1 = y_1\) and \(x_2 = y_2\) and therefore, by the assumption that the underlying closeness functions satisfy condition (i).(a), we have \(c_X(x_1,x_2) = \infty\) and \(c_Y(y_1,y_2) = \infty\). Thus, by definition of the binary product closeness function (\cref{def:prod-closeness}) we only need to show that \(\mathsf{min}(\infty,\infty) = \infty\), which is immediate as \(\lambda n.\mathsf{min}(1,1) = \lambda n.1\) (even without function extensionality).
\item To show that \((x_1,y_1) = (x_2,y_2)\) given \(c_{X \x Y}((x_1,y_1),(x_2,y_2)) = \infty\) then we first show that \(c_X(x_1,x_2) = \infty\) and \(c_Y(y_1,y_2) = \infty\) (these are because, by function extensionality, if \(\mathsf{min}(u,v) = \infty\) then \(u = \infty = v\)). We now use the assumption that the underlying closeness functions satisfy condition (i).(b) to give us \(x_1 = x_2\) and \(y_1 = y_2\), and thus the result follows immediately.
\end{enumerate}
\item The symmetry of \(c_{X\x Y}\) is immediate from the assumption that the underlying closeness functions satisfy condition (ii).
\item We want to show that \( \mathsf{min}(c_{X\x Y}((x_1,y_1),(x_2,y_2)),c_{X\x Y}((x_2,y_2),(x_3,y_3)))\) \(\preceq c_{X\x Y}((x_1,y_1),(x_3,y_3))\) \(:= \mathsf{min}(\mathsf{min}(a,c),\mathsf{min}(b,d)) \preceq \mathsf{min}(e,f)\) holds, where \(a := c_X(x_1,x_2)\), \(b := c_X(x_2,x_3)\), \(c := c_Y(y_1,y_2)\) \(d := c_Y(y_2,y_3)\), \(e := c_X(x_1,x_3)\) and \(f := c_Y(y_1,y_3)\).
Using the assumption that the underlying closeness functions satisfy condition (iii), we have \(\mathsf{min}(a,b) \preceq e\) and \(\mathsf{min}(c,d) \preceq f\), and therefore \(\mathsf{min}(\mathsf{min}(a,b),\mathsf{min}(c,d)) \preceq \mathsf{min}(e,f)\). The result then follows by a rearrangement of the arguments on the left-hand side of the inequality.~\qedhere
\end{enumerate}
\end{proof}

The following lemma, which follows easily from \cref{def:prod-closeness}, is useful when working with closeness relations of binary product closeness spaces:

\begin{lemma}
\label{lem:x-C}
\thesislit{3}{ClosenessSpaces-Examples}{\urlx-C-left}
\thesislit{3}{ClosenessSpaces-Examples}{\urlx-C-right}
\thesislit{3}{ClosenessSpaces-Examples}{\urlx-C-combine}
Given closeness spaces \(X\) and \(Y\), the binary product closeness space \(X \x Y\) is such that, for any \(\ty{\varepsilon}{\N}\) and \(\ty{(x_1,y_1),(x_2,y_2)}{X \x Y}\), we have \[C_\varepsilon((x_1,y_1),(x_2,y_2)) \Leftrightarrow C_\varepsilon(x_1,x_2) \x C_\varepsilon(y_1,y_2) .\]
\end{lemma}
\begin{proof}[Proof (Sketch).]
This is straightforward to show once we recognise, by using the definition of the binary product closeness function (\cref{def:prod-closeness}), that we only need to show that \(\underline \varepsilon \preceq \mathsf{min}(c_X(x_1,x_2),c_Y(y_1,y_2))\) holds if and only if both \(\underline \varepsilon \preceq c_X(x_1,x_2)\) and \(\underline \varepsilon \preceq c_Y(y_1,y_2)\) hold.
\end{proof}

\begin{lemma}
\label{lem:bin-totallybounded}
\thesislit{3}{ClosenessSpaces-Examples}{\urlx-totally-bounded}
Given totally bounded closeness spaces \(X\) and \(Y\), the binary product closeness space \(X \x Y\) is totally bounded.
\end{lemma}
\begin{proof}
For a given \(\ty{\varepsilon}{\N}\), both \(X\) and \(Y\) have \(\varepsilon\)-nets; i.e. there are types \(X'\) and \(Y'\) and functions \(\ty{g_X}{X' \to X}\), \(\ty{g_Y}{Y' \to Y}\), \(\ty{h_X}{X \to X'}\), \(\ty{h_Y}{Y \to Y'}\) such that for all \(\ty{x}{X}\) and \(\ty{y}{Y}\) we have \(C_\varepsilon(x,g_X(h_X(x)))\) and \(C_\varepsilon(y,g_Y(h_Y(y)))\).
To show that \(X \x Y\) has an \(\varepsilon\)-net, we define the following functions:
\begin{alignat*}{3}
g_{X \x Y} &: {X' \x Y'} \to {X \x Y}, \\
g_{X \x Y} &(x',y') := (g_X(x'),g_Y(y')) ,\\
h_{X \x Y} &: {X \x Y} \to {X' \x Y'}, \\
h_{X \x Y} &(x,y) := (h_X(x),h_Y(y)).
\end{alignat*}
We then need to show that for all \(\ty{(x,y)}{X \x Y}\) we have \(C_\varepsilon((x,y),g_{X \x Y}(h_{X \x Y}((x,y))))\) which, by definition of \(g_{X \x Y}\), \(h_{X \x Y}\) and \cref{def:prod-closeness}, reduces to \(\underline \varepsilon \preceq \mathsf{min}(c_X(x,g_X(h_X(x))),c_Y(y,g_Y(h_Y(y))))\).
This is by \cref{lem:x-C} and the fact that \((X',g_X)\) and \((Y',g_Y)\) are, respectively, \(\varepsilon\)-nets for \(X\) and \(Y\).
Therefore the type \(X' \x Y'\), equipped with the functions we defined above, is an \(\varepsilon\)-net for \(X + Y\). 
\end{proof}

We can use the binary product closeness function to define closeness functions for finite product types (i.e.\ dependent and non-dependent vectors).

\begin{corollary}
\label{cor:dvec-cspace}
Given an \((\ty{n}{\N})\)-size vector \(\ty{Y}{\F \ n \to \U}\) of closeness spaces, the type of \(n\)-size dependent vectors \(\F \ n \to Y_n\) is a closeness space.
\end{corollary}
\begin{proof} \axioms{f}
By induction and \cref{lem:bin-cspace}.
\end{proof}

\begin{corollary}
\label{cor:dvec-totallybounded}
Given an \((\ty{n}{\N})\)-size vector \(\ty{Y}{\F \ n \to \U}\) of totally bounded closeness spaces, the finite product closeness space of \(n\)-size dependent vectors \(\F \ n \to Y_n\) is totally bounded.
\end{corollary}
\begin{proof}
By induction and \cref{lem:bin-totallybounded}.
\end{proof}

\begin{corollary}
\label{cor:vec-cspace} \axioms{f}
\thesislit{3}{ClosenessSpaces-Examples}{Vec-ClosenessSpace}
Given a closeness space \(X\), the type of \(n\)-size vectors \(\F \ n \to X\) is a closeness space.
\end{corollary}
\begin{proof}
By \cref{cor:dvec-cspace}.
\end{proof}

\begin{corollary}
\label{cor:vec-totallybounded}
\thesislit{3}{ClosenessSpaces-Examples}{Vec-totally-bounded}
Given a totally bounded closeness space \(X\), the closeness space of \(n\)-size vectors \(\F \ n \to X\) is totally bounded.
\end{corollary}
\begin{proof}
By \cref{cor:dvec-totallybounded}.
\end{proof}

\subsubsection{Subtype closeness functions.}

\begin{definition}
\label{def:subtype-closeness}
\thesislit{3}{ClosenessSpaces-Examples}{\urlhook-clospace}
Given a type \(X\), closeness space \(Y\) and a function \(\ty{f}{X \to Y}\), the \emph{subtype closeness function} \(\ty{c_{X}}{\closeness{X}}\) is defined by,
\[ c_X(x_1,x_2) := c_Y(f(x_1),f(x_2)), \]
\end{definition}

\begin{lemma}
\label{lem:embedding-cspace}
\thesislit{3}{ClosenessSpaces-Examples}{\urlhook-clospace}
Given a type \(X\), closeness space \(Y\) and function \(\ty{f}{X \to Y}\), the type \(X\) is a closeness space by the subtype closeness function if \(f\) is an embedding (\cref{def:embedding}).
\end{lemma}
\begin{proof}
Conditions 2 and 3 of \cref{def:cspace} are immediate from the fact \(Y\) is a closeness space.
The same is true of one direction of the first condition (i.e.\ that \(x = y \to c_X(x,y) = \infty\)).

\vspace{1em}
For the other direction of the first condition, we have \(c_X(x,y) := c_Y(f(x),f(y)) = \infty\) and need to prove \(x = y\). 
By the same condition on \(Y\), we have \(f(x) = f(y)\), and so \(x = y\) follows from the fact \(f\) is an embedding.
\end{proof}

\begin{corollary}
\label{cor:sigma-cspace}
\thesislit{3}{ClosenessSpaces-Examples}{\urlhook-ClosenessSpace}
Given a closeness space \(X\) and truth-valued function \(\ty{P}{X \to \Omega}\), the type \(\sigmaty{x}{X}{P(x)}\) is a closeness space.
\end{corollary}
\begin{proof}
By \cref{lem:embedding-cspace} because \(\mathsf{pr}_1\) is an embedding.
\end{proof}

\begin{corollary}
\label{cor:Ni-cspace}
\thesislit{3}{ClosenessSpaces-Examples}{\urlN\urlinfty-ClosenessSpace}
\(\Ni\) is a closeness space.
\end{corollary}
\begin{proof} \axioms{f}
By \cref{cor:sigma-cspace} and the later \cref{cor:disseq-cspace} applied to \(\2\).
\end{proof}

\begin{corollary}
\label{cor:equiv-closeness}
\thesislit{3}{ClosenessSpaces-Examples}{\urlsimeq-ClosenessSpace}
Given a type \(X\) and closeness space \(Y\) such that \(X \simeq Y\), the type \(X\) is a closeness space.
\end{corollary}
\begin{proof}
By \cref{lem:embedding-cspace} because the equivalence \(\ty{f}{X \to Y}\) is an embedding.
\end{proof}

\begin{lemma}
\thesislit{3}{ClosenessSpaces-Examples}{\urlsimeq-totally-bounded}
Given a type \(X\) and totally bounded closeness space \(Y\) such that \(X \simeq Y\), the equivalent closeness space \(X\) is totally bounded.
\end{lemma}
\begin{proof}
For a given \(\ty{\varepsilon}{\N}\), both \(Y\) has an \(\varepsilon\)-nets; i.e. there is a type \(Y'\) and functions \(\ty{g_\varepsilon}{Y' \to Y}\) and \(\ty{h_\varepsilon}{Y \to Y'}\) such that for all \(\ty{y}{Y}\) we have \(C_\varepsilon(y,g_\varepsilon(h_\varepsilon(y)))\).
Furthermore, by \(X \simeq Y\) there is an equivalence \(\ty{f_\simeq}{X \to Y}\), and thus there is a function \(\ty{g_\simeq}{Y \to X}\) such that \(f_\simeq \circ g_\simeq \sim \mathsf{id}_Y\).
We will now show that \(Y'\), equipped with \(\ty{g := g_\simeq \circ g_\varepsilon}{Y' \to X}\), is an \(\varepsilon\)-net for \(X\): meaning there is some function \(\ty{h}{X \to Y'}\) such that for all \(\ty{x}{X}\) we have \(C_\varepsilon(f_\simeq(x),(f_\simeq \circ g \circ h)(x))\) (by \cref{def:subtype-closeness}). 

\vspace{1em}
This function is defined \(h := h_\varepsilon \circ f\), and we give \(C_\varepsilon(f_\simeq(x),(f_\simeq \circ g_\simeq \circ g_\varepsilon \circ h_\varepsilon \circ f_\simeq)(x))\) by transitivity of \(C_\varepsilon\) (\cref{lem:C-eq}) after first proving (i) \(C_\varepsilon(f_\simeq(x),(g_\varepsilon \circ h_\varepsilon \circ f_\simeq)(x))\) and (ii) \(C_\varepsilon((g_\varepsilon \circ h_\varepsilon \circ f_\simeq)(x),(f_\simeq \circ g_\simeq \circ g_\varepsilon \circ h_\varepsilon \circ f_\simeq)(x))\):
\begin{enumerate}[(i)]
\item holds because \((Y',g_\varepsilon)\) is an \(\varepsilon\)-net for \(Y\),
\item holds because \(f_\simeq \circ g_\simeq \sim \mathsf{id}_Y\), and therefore \((g_\varepsilon \circ h_\varepsilon \circ f_\simeq)(x) = (f_\simeq \circ g_\simeq \circ g_\varepsilon \circ h_\varepsilon \circ f_\simeq)(x)\), which in turn (by \cref{cor:equiv-closeness}) yields \(C_n((g_\varepsilon \circ h_\varepsilon \circ f_\simeq)(x),(f_\simeq \circ g_\simeq \circ g_\varepsilon \circ h_\varepsilon \circ f_\simeq)(x))\) for all \(\ty{n}{\N}\).
\end{enumerate}
\end{proof}

\subsubsection{Discrete-sequence closeness functions.}

Recalling the definition of sequence prefix equality from \cref{sec:sequences}, we define the closeness function on discrete sequence types.

\begin{definition}
\label{def:disseq-closeness}
\thesislit{3}{ClosenessSpaces-Examples}{discrete-seq-clofun}
Given an \(\N\)-indexed type family \(\ty{D}{\N \to \U}\) of discrete types, the \emph{discrete-sequence closeness function} \(\ty{c_{\Pi D}}{\closeness{\Pi D}}\) is defined at each point \(\ty{n}{\N}\) by case analysis of \(\ty{d(\alpha,\beta,n+1)}{\mathsf{decidable}(\alpha \sim^n \beta)}\) for arguments \(\ty{\alpha,\beta}{\Pi D}\):
\begin{itemize}
\item If \(\alpha \sim^{n+1} \beta\) then \(c_{\Pi D}(\alpha,\beta)_n := 1\),
\item If \(\neg (\alpha \sim^{n+1} \beta)\) then \(c_{\Pi D}(\alpha,\beta)_n := 0\).
\end{itemize}
\end{definition}

\noindent
This is decreasing, because \(\neg (\alpha \sim^n \beta)\) implies \(\neg (\alpha \sim^m \beta)\) for all \(\ty{n,m}{\N}\) such that \(m > n\).

\begin{lemma}
\label{lem:disseq-cspace}
\thesislit{3}{ClosenessSpaces-Examples}{\urlPia{}D-ClosenessSpace}
The product of an \(\N\)-indexed type family of discrete types is a closeness space by the discrete-sequence closeness function.
\end{lemma}
\begin{proof}[Proof (Sketch).] \axioms{f}
We prove each of the three conditions of \cref{def:cspace}:
\begin{enumerate}
\item In the direction where we have \(c_{\Pi D}(\alpha,\beta) = \infty\), we use the decidability of \(\alpha \sim^{n+1} \beta\) (\cref{lem:sim-decidable}) to ascertain that \(\alpha \sim^{n+1} \beta\) for all \(\ty{n}{\N}\) (because, if this were not the case, \(c_{\Pi D}(\alpha,\beta)_n\) would equal \(0\)). Therefore, we immediately have that \(\alpha \sim \beta\) and (by function extensionality) \(\alpha = \beta\). The other direction is trivial.

\item Given any \(\ty{n}{\N}\), in the case where \(\alpha \sim^{n+1} \beta\), by symmetry of prefix equality (\cref{lem:sim-eq}) we also have \(\beta \sim^{n+1} \alpha\), and therefore \(c(\alpha,\beta)_n = c(\beta,\alpha)_n\). The same is true in the case where \(\neg(\alpha \sim^{n+1} \beta)\).

\item Given any \(\ty{n}{\N}\) we need to show that if \(\mathsf{min}(c(\alpha,\beta)_n,c(\beta,\zeta)_n) = 1\) then \(c(\alpha,\zeta)_n = 1\). The assumption means that both \(c(\alpha,\beta)_n\) and \(c(\beta,\zeta)_n\) are \(1\); hence, \(\alpha \sim^{n+1} \beta\) and \(\beta \sim^{n+1} \zeta\) must hold.
The result then follows by transitivity of prefix equality (\cref{lem:sim-eq}).
\end{enumerate}
\end{proof}

\begin{lemma}
\label{lem:disseq-totallybounded}
\thesislit{3}{ClosenessSpaces-Examples}{\urlPia{}F-totally-bounded}
Given an \(\N\)-indexed type family \(\ty{F}{\N \to \U}\) of pointed finite linearly ordered types, the discrete-sequence closeness space \(\Pi F\) is totally bounded.
\end{lemma}
\begin{proof}
For any \(\ty{\varepsilon}{\N}\), the closeness function considers only the \(\varepsilon\)-prefix of the sequences.
Therefore, the \(\varepsilon\)-net is the type of \(\varepsilon\)-sized dependent vectors \(\mathsf{Vec}(\varepsilon,F_{0,...,\varepsilon-1})\), where \(F_{0,...,\varepsilon-1}\) is the finite-indexed type family that is the \(\varepsilon\)-prefix of \(F\).
We define the two maps for the net:  \(\ty{g}{\mathsf{Vec}(\varepsilon,\ F_{0,...,\varepsilon-1}) \to \Pi F}\) is the function that outputs the sequence that is the \(\varepsilon\)-sized vector followed by repeating arbitrary elements of \(F_{\varepsilon,...}\) (which we attain from the pointedness of all types in \(F\)), while \(\ty{h}{\Pi F \to \mathsf{Vec}(\varepsilon,F_{0,...,\varepsilon-1})}\) is the function that gives the \(\varepsilon\)-prefix of the sequence.
\end{proof}

\begin{lemma}
\label{lem:disseq-C-equiv}
\thesislit{3}{ClosenessSpaces-Examples}{\urlsim\urlsuperscriptn-to-C}
\thesislit{3}{ClosenessSpaces-Examples}{C-to-\urlsim\urlsuperscriptn}
Given an \(\N\)-indexed type family \(\ty{D}{\N \to \U}\) of discrete types and two sequences \(\ty{\alpha,\beta}{\Pi D}\), the types \(C_n(\alpha,\beta)\) (derived from \cref{lem:disseq-cspace}) and \(\alpha \sim^n \beta\) are propositionally equivalent for all \(\ty{n}{\N}\).
\end{lemma}
\begin{proof}[Proof (Sketch).]
By induction on \(\ty{n}{\N}\), though both \(C_0(\alpha,\beta)\) and \(\alpha \sim^0 \beta\) are trivially true, and so we only need to consider the case where \(n := n' + 1\) for some \(\ty{n'}{\N}\).

If \(C_{n+1}(\alpha,\beta) := \underline {n+1} \preceq c_{\Pi D}(\alpha,\beta)\) then clearly \(c_{\Pi D}(\alpha,\beta)_n = 1\); therefore by definition of the discrete-sequence closeness function (\cref{def:disseq-closeness}) and the decidability of \(\sim^{n+1}\) (\cref{lem:sim-decidable}), \(\alpha \sim^n \beta\) must hold.
The opposite direction follows similarly, but by the decidability of \(C_{n+1}\) (\cref{lem:C-dec}).
\end{proof}

Using the above, we define the non-dependent case; i.e.\ closeness spaces on types of sequences on a single discrete type.

\begin{corollary}
\label{cor:disseq-cspace}
\thesislit{3}{ClosenessSpaces-Examples}{\urlN\urlto{}D-ClosenessSpace}
\thesislit{3}{ClosenessSpaces-Examples}{\urlN\urlto{}F-totally-bounded}
The type of sequences on any finite linearly ordered type is a totally bounded closeness space.
\end{corollary}
\begin{proof} \axioms{f}
By \cref{lem:disseq-cspace,lem:disseq-totallybounded}.
\end{proof}

\begin{corollary}
\label{cor:disseq-C-equiv}
\thesislit{3}{ClosenessSpaces-Examples}{\urlsim\urlsuperscriptn-to-C}
\thesislit{3}{ClosenessSpaces-Examples}{C-to-\urlsim\urlsuperscriptn}
Given two sequences \(\ty{\alpha,\beta}{\N \to X}\) on a discrete type \(X\), the types \(C_n(\alpha,\beta)\) (derived from \cref{cor:disseq-cspace}) and \(\alpha \sim^n \beta\) are propositionally equivalent for all \(\ty{n}{\N}\).
\end{corollary}
\begin{proof}
By \cref{lem:disseq-C-equiv}.
\end{proof}

\subsubsection{Countable product closeness spaces.}

We have already shown that the product of an \(\N\)-indexed type family of discrete types is a closeness space (by \cref{lem:disseq-cspace}).
We will now generalise this by showing we can define a closeness function on the product of an \(\N\)-indexed type family of \emph{closeness spaces}.

In order to do this, we will employ a diagonal argument to the countably-many extended natural values of the closeness functions within the type family \(\ty{T}{\N \to \U}\) of closeness spaces (when applied to any given arguments \(\ty{\alpha,\beta}{\Pi T}\)).
To illustrate this idea, consider \cref{fig:diagonal} which gives potential values for \(\ty{c_{T_i}(\alpha_i,\beta_i)_j}{\Ni}\) where \(i \in \{0,...,3\}\) and \(j \in \{0,...,5\}\).
The diagonal argument that we employ to define \(\ty{c^{cs}_{\Pi T}(\alpha,\beta)}{\Ni}\) is such that the \(n\)\textsuperscript{th} digit should be \(1\) if the \(n\)\textsuperscript{th}  diagonal is made up only of \(1\)s. For the example, this means that \(c_{\Pi T}(\alpha,\beta) := 111000...\). After the third digit, no matter the values of the other extended naturals, \(c_2\) will always contribute a \(0\) to the diagonal --- this shows how the result is indeed decreasing.

\begin{figure}[h]
\centering
\begin{tabular}{ccccccccc}
 &  & \multicolumn{7}{c}{\textbf{Index}} \\
 &  & $0$ & $1$ & $2$ & $3$ & $4$ & $5$ & ... \\
 & $c_0$ & \cellcolor[HTML]{FFCCC9}1 & \cellcolor[HTML]{96FFFB}1 & \cellcolor[HTML]{9AFF99}1 & \cellcolor[HTML]{CBCEFB}1 & 1 & 0 & ... \\
 & $c_1$ & \cellcolor[HTML]{96FFFB}1 & \cellcolor[HTML]{9AFF99}1 & \cellcolor[HTML]{CBCEFB}1 & 1 & 1 & 1 & ... \\
 & $c_2$ & \cellcolor[HTML]{9AFF99}1 & \cellcolor[HTML]{CBCEFB}0 & 0 & 0 & 0 & 0 & ... \\
 & $c_3$ & \cellcolor[HTML]{CBCEFB}1 & 1 & 1 & 0 & 0 & 0 & ... \\
\multirow{-7}{*}{\rotatebox[origin=c]{90}{\textbf{Extended natural}}} & ... & ... & ... & ... & ... & ... & ... & ...
\end{tabular}\caption{Example indices of countably-many extended naturals \(c_i := c_{T_i}(\alpha_i,\beta_i)\) from countably-many closeness functions \(\ty{c_{T_i}}{T_i \to T_i \to \Ni}\) using \(\ty{\alpha,\beta}{\Pi T}\). The first four diagonals are coloured differently for emphasis.}
\label{fig:diagonal}
\end{figure}

This diagonalisation process can be inductively defined in our framework as follows:

\begin{definition}
\label{def:pi-closeness}
\thesislit{3}{ClosenessSpaces-Examples}{\urlPi-clofun}
Given an \(\N\)-indexed type family \(\ty{T}{\N \to \U}\) of closeness spaces (i.e.\ there is a family of closeness functions \(\ty{cs}{\pity{n}{\N}{\closeness{T_n}}}\)), the \emph{countable product closeness function} is defined by,
\begin{alignat*}{3}
c^{cs}_{\Pi T} &: \Pi T \to \Pi T &&\to \Ni, \\
c^{cs}_{\Pi T} &(\alpha,\beta)_0 &&:= cs_0(\alpha_0,\beta_0)_0, \\
c^{cs}_{\Pi T} &(\alpha,\beta)_{\suc{n}} &&:= \mathsf{min}(cs_0(\alpha_0,\beta_0)_{\suc{n}},c^{\mathsf{tail} \ cs}_{\Pi (\mathsf{tail} \ T)}(\mathsf{tail}\ \alpha,\mathsf{tail}\ \beta)_n).
\end{alignat*}
\end{definition}
\begin{lemma}
\label{lem:pi-cspace}
\thesislit{3}{ClosenessSpaces-Examples}{\urlPi-ClosenessSpace}
The product of an \(\N\)-indexed type family of closeness spaces is a closeness space by the countable product closeness function.
\end{lemma}

\noindent
The proof of the above is similar to \cref{lem:disseq-cspace}, but requires us to use the fact that the underlying types of the type family are \emph{themselves} closeness spaces.
We recommend viewing the formalisation for more technical details on this relatively straightforward sizable proof.

\begin{lemma}
\label{lem:pi-totallybounded}
\thesislit{3}{ClosenessSpaces-Examples}{\urlPi-totally-bounded}
Given an \(\N\)-indexed type family \(\ty{T}{\N \to \U}\) of totally bounded closeness spaces, the countable product closeness space \(\Pi T\) is totally bounded.
\end{lemma}
\begin{proof}
The base case is vacuous (by \cref{remark:0-close}), and so we only consider the inductive case.
When \(\varepsilon := \varepsilon' + 1\), the \(\varepsilon\)-net is \(t_0' \x t_s'\), where \(t_0'\) is the \(\varepsilon\)-net of \(T_0\) and \(t_s'\) is the \(\varepsilon'\)-net of \(\mathsf{tail} \ T\) by the inductive hypothesis.
\end{proof}

Countable product closeness functions truly generalise discrete sequence closeness functions.
Given a family of discrete types, one can use \cref{lem:discrete-cspace} to transform this into a family of closeness spaces.
After doing this for a particular type family, the resulting discrete-sequence and countable product closeness functions will be identical.

\begin{lemma}\thesislit{3}{ClosenessSpaces-Examples}{\urlPi-clofuns-id}
Given an \(\N\)-indexed type family \(\ty{T}{\N \to \U}\) of discrete types the discrete-sequence closeness function \(c_{\Pi T}\) (\cref{def:disseq-closeness}) and the countable product closeness function \(c^{cs}_{\Pi T}\) (\cref{def:pi-closeness}), where \(\ty{cs}{\Pity{n}{\N}{\closeness{T_n}}}\) is the collection of discrete closeness functions (\cref{def:discrete-closeness}) on each \(T_n\), are identical.
\end{lemma}
\begin{proof}[Proof (Sketch).] \axioms{f}
We prove that the two functions are pointwise-equal --- the result then follows immediately by function extensionality.
To prove that, for all \(\ty{\alpha,\beta}{\Pi T}\) and \(\ty{n}{\N}\) we have \(c_{\Pi T}(\alpha,\beta)_n = c^{cs}_{\Pi T}(\alpha,\beta)_n\), we proceed by induction on the given \(n\).

\vspace{1em}
In the base case where \(n := 0\), we want to show that \(c_{\Pi T}(\alpha,\beta)_0 = c_{T_0}(\alpha_0,\beta_0)_0\) where \(c_{T_0} := cs_0\), i.e. the discrete closeness function on \(T_0\).
Using the discreteness of \(T_0\), we ask whether \(\alpha_0 = \beta_0\). If it does then we have \(\alpha \sim^1 \beta\) and therefore by definition of the discrete-sequence closeness function (\cref{def:disseq-closeness}) we have \(c_{\Pi T}(\alpha,\beta)_0 = 1\). We also have \(c_{T_0}(\alpha_0,\beta_0) = \infty\) by definition of the discrete closeness function, and therefore clearly also \(1 = c_{T_0}(\alpha_0,\beta_0)_0\). The proof technique is similar in the case where \(\neg (\alpha_0 = \beta_0)\).

\vspace{1em}
In the inductive case where \(n := n' + 1\) for some \(\ty{n'}{\N}\), we want to show that \(c_{\Pi T}(\alpha,\beta)_{n} = \mathsf{min}(c_{T_0}(\alpha_0,\beta_0)_{n},c_{\Pi (\mathsf{tail} \ T)}(\mathsf{tail}\ \alpha,\mathsf{tail}\ \beta)_{n'})\).
Using the discreteness of each type in \(T\), we ask whether \(\alpha \sim^{n + 1} \beta\). If it does, then we have \(c_{\Pi T}(\alpha,\beta)_n = 1\) and therefore only need to show that (i) \(c_{T_0}(\alpha_0,\beta_0)_{n} = 1\) and (ii) \(c_{\Pi (\mathsf{tail} \ T)}(\mathsf{tail}\ \alpha,\mathsf{tail}\ \beta)_{n'} = 1\).
\begin{enumerate}[(i)]
\item holds by definition of the discrete closeness function, because we have shown that \(\alpha_0 = \beta_0\),
\item holds by the inductive hypothesis once we prove that \(c_{\Pi T}(\mathsf{tail}\ \alpha,\mathsf{tail}\ \beta)_{n'} = 1\); this is by definition of the discrete-sequence closeness function and the fact that \(\alpha \sim^{n+1} \beta\) implies \(\mathsf{tail}\ \alpha \sim^{n} \mathsf{tail}\ \beta\).
\end{enumerate}
The proof technique is similar in the case where \(\neg (\alpha \sim^{n+1} \beta)\).
\end{proof}

We close this subsection by noting the following two lemmas concerning countable product closeness spaces; note that these can of course be specialised for discrete-sequence closeness spaces.
The first states that if the tails of two sequences have closeness \(\ty{\delta}{\N}\) and the head elements have closeness \(\delta + 1\), then the sequences themselves also have closeness \(\delta + 1\).

\begin{lemma}
\label{lem:Pi-C-combine}
\thesislit{3}{ClosenessSpaces-Examples}{\urlPi-C-combine}
Given an \(\N\)-indexed type family \(\ty{T}{\N \to \U}\) of closeness spaces, two head elements \(\ty{x_0,y_0}{T_0}\) and two tail sequences \(\ty{\alpha,\beta}{\Pi T}\) and some precision \(\ty{\delta}{\N}\), if \(C_{\delta+1}(\alpha_0,\beta_0)\) and \(C_\delta(\mathsf{tail} \ \alpha,\mathsf{tail} \ \beta)\) then \(C_{\delta+1}(\alpha,\beta)\).
\end{lemma}
\begin{proof}
When \(\delta := 0\), the result is vacuous by \cref{remark:0-close}, therefore we only consider the \(\delta := \delta' + 1\) case where we need to show \(\underline{\delta' + 1} \preceq c^*_{\Pi T}(\alpha,\beta)\).
By the definition of the partial order on \(\Ni\) (\cref{def:Ni-order}), this is reduced to needing to show \(c^*_{\Pi T}(\alpha,\beta)_{\delta'+1} = 1\).

\vspace{1em}
By the definition of the closeness function on countable products (\cref{def:pi-closeness}), this means we need to show \[\mathsf{min}(cs_0(\alpha_0,\beta_0)_{\delta'+1},c_{\Pi(\mathsf{tail} \ T)}(\mathsf{tail} \ \alpha,\mathsf{tail} \ \beta)_{\delta'}) = 1, \] which holds because, by our assumptions, both arguments to the \(\mathsf{min}\) function are \(1\).
\end{proof}

The second states that any sequence is infinitely close to the composition of its own head and tail --- the point of this lemma is to avoid invoking function extensionality in this case.

\begin{lemma}
\label{lem:Pi-head-tail-eta}
\thesislit{3}{ClosenessSpaces-Examples}{\urlPi-C-eta}
Given an \(\N\)-indexed type family \(\ty{T}{\N \to \U}\) of closeness spaces, every sequence \(\ty{\alpha}{\Pi T}\) is such that \(C_\delta(\alpha,\mathsf{head} \ \alpha :: \mathsf{tail} \ \alpha)\), for every precision \(\ty{\delta}{\N}\).
\end{lemma}
\begin{proof}[Proof (Sketch).]
Proceeding by induction on the given \(\delta\), the base case is trivial and the inductive case is immediate by \cref{lem:Pi-C-combine}.
\end{proof}

\subsection{Pseudocloseness spaces}

We have found that closeness spaces are a convenient structure to reason about the closeness of elements of a wide variety of types in our framework.
Sometimes, such as for parametric regression (see \cref{sec:regression}), we find that we wish to use a more general structure so that we can reason about a wider variety of types.

We borrow the terminology of `pseudometric spaces' and define \emph{pseudo closeness spaces} below, in which we relax the first condition of \cref{def:cspace} such that we no longer require elements with (pseudo)closeness \(\infty\) to be identical.

\begin{definition}
\label{def:pcspace}
\thesislit{3}{ClosenessSpaces}{PseudoClosenessSpace}
A \emph{pseudocloseness space} is a type \(X\) equipped with a \emph{pseudocloseness function} \(\ty{c}{\closeness X}\) such that,
\begin{enumerate}
\item \({x = y} \to c(x,y) = \infty\),
\item \(c(x,y) = c(y,x)\),
\item \(\mathsf{min}(c(x,y),c(y,z)) \preceq c(x,z)\).
\end{enumerate}
\end{definition}

\noindent
The altered definitions of closeness relation (which remain equivalence relations) and uniform continuity for pseudocloseness spaces follow naturally from those on closeness spaces (\cref{def:C,def:clos-ucont,def:clos-ucont-pred}), and so we leave them out to avoid repetition.

The structure of a pseudocloseness space differs from a closeness space in that non-equal elements can have pseudocloseness \(\infty\).
This is required for parametric regression, for which we here define psuedocloseness spaces for function spaces, allowing functions to be compared at a finite number of given points.

\begin{definition}
\label{def:least-closeness}
\thesislit{3}{ClosenessSpace-Examples}{Least-PseudoClosenessSpace}
Given a type \(X\), closeness space \(Y\) and \((\ty{n}{\N})\)-size vector \(\ty{xs}{\F \ n \to X}\) of elements of \(X\), we define the \emph{least-closeness pseudocloseness function} as,
\begin{alignat*}{3}
{c'_{X \to Y}}^{(n,xs)} &: \closeness{(X \to Y)}  \\
{c'_{X \to Y}}^{(n,xs)} &(f,g) := c_{Y^n}(\mathsf{map}(f,xs),\mathsf{map}(g,xs)),
\end{alignat*}
where \(\ty{c_{Y^n}}{\closeness{(\F \ n \to X)}}\) is the closeness function derived from \cref{cor:vec-cspace}.
\end{definition}

Least-closeness pseudocloseness functions compare two functions \(f\) and \(g\) at a finite number of given points \(\{xs_0,...,xs_{n-1}\}\), and return the minimum closeness found between these points; i.e. at each point \(x_i\) (where \(i \in \{0,...,n-1\}\)) the closeness of \(f(xs_i)\) and \(g(xs_i)\) is computed as \(c_i := c_Y(f(xs_i),g(xs_i))\), and then the minimum of these values is returned as the least-closeness pseudocloseness \({c'_{X \to Y}}^{(n,xs)}(f,g) := \mathsf{min}(c_0,...,c_{n-1})\). This idea is inspired by least-squares approach to regression and pseudometrics on function spaces used in regression analysis~\cite{YanBook}.

\section{Searching infinite types}
\label{sec:search-infinite}

We return, equipped with closeness spaces, to the problem of searching  infinite types.

\subsection{Uniformly continuously searchable closeness spaces}

Recall from the close of \cref{sec:search-infinite-canwe} that we aim to restrict search to only those decidable predicates that have a constructive witness of their uniform continuity.
These moduli of uniform continuity will be provided by closeness spaces (\cref{def:clos-ucont-pred}).

\begin{definition}
\thesislit{3}{SearchableTypes}{decidable-uc-predicate}
The type of \emph{uniformly continuous and decidable predicates} on a closeness space \(K\) is defined by, \[\mathsf{decidable{\hy}uc{\hy}predicate}(K) := \Sigmatye{p}{\mathsf{decidable{\hy}predicate \ K}}{\mathsf{p{\hy}ucontinuous \ (p)}} .\]
\end{definition}

We now formally define the restricted definition of a searchable type (\cref{def:searchable}). We call a closeness space whose uniformly continuous and decidable predicates we can search a \emph{uniformly continuously searchable} closeness space.

\begin{definition}
\label{def:c-searcher}
\thesislit{3}{SearchableTypes}{csearchable\urlmathcalE}
A function \(\ty{\mathcal{E}_K}{\mathsf{decidable{\hy}uc{\hy}predicate}(K) \to K}\) is a \emph{uniformly continuous searcher} on a given closeness space \(K\) if, for all \(\ty{p}{\mathsf{decidable{\hy}uc{\hy}predicate}(K)}\), it is the case that \(p(\mathcal{E}_K(p))\) holds if there is some element \(\ty{k}{K}\) such that \(p(k)\) holds:
\[ \mathsf{is{\hy}uc{\hy}searcher}(\mathcal{E}) := \pitye{p}{\mathsf{decidable{\hy}uc{\hy}predicate}(K)}{\left( \sigmatye{k}{K}{p(k)} \right) \to p(\mathcal{E}(p))} .\]
\end{definition}

\begin{definition}
\label{def:c-searchable}
\thesislit{3}{SearchableTypes}{csearchable}
A closeness space \(K\) is \emph{uniformly continuously searchable} if we can define a uniformly continuous searcher on that closeness space:
\[ \mathsf{uc{\hy}searchable^\mathcal{E}}(K) := \sigmatye{\mathcal{E}_K}{\mathsf{decidable{\hy}uc{\hy}predicate} \ K \to K}{\mathsf{is{\hy}uc{\hy}searcher}(\mathcal{E})} .\]
We often use the following equivalent definition, which is more convenient:
\[ \mathsf{uc{\hy}searchable}(K) := \pitye{p}{\mathsf{decidable{\hy}uc{\hy}predicate}(K)}{ \sigmatye{k_0}{K}{\left( \sigmatye{k}{K}{p(k)} \right) \to p(k_0)}} .\]
\end{definition}

\begin{remark}
\label{remark:search-csearch}
\thesislit{3}{SearchableTypes}{csearchable\urlto{}csearchable}
Every searchable closeness space is automatically uniformly continuously searchable by discarding the continuity information.
\end{remark}

Much like searchable types, every uniformly continuously searchable closeness space is pointed.

\begin{lemma}
\label{lem:csearchable-pointed}
\thesislit{3}{SearchableTypes}{csearchable-pointed}
Every uniformly continuously searchable closeness space is pointed.
\end{lemma}
\begin{proof}
For the given uniformly continuously searchable space \(K\), define the constant predicate \(p^\top(k) := \top\), which every element satisfies and therefore has modulus of uniform continuity \(0\).
We can then introduce the element \(\ty{\mathcal{E}_K(p^\top)}{K}\).
\end{proof}

When searching a closeness space \(X\), a consequence of knowing that \(\ty{\delta}{\N}\) is a modulus of uniform continuity for the predicate \(\ty{p}{X \to \Omega}\) is that instead of checking each individual candidate \(\ty{x}{X}\), we can instead check each \(\delta\)-neighbourhood of \(X\) collectively by a single representing element.
If the representative satisfies the predicate, then we can simply return it; if not, we can discard the entire \(\delta\)-neighbourhood in which it lives from the search.

A corollary to this is that if the closeness space has a finite \(\delta\)-net then it has a finite number of \(\delta\)-neighbourhoods and an answer to the predicate can be searched for.
If the closeness space is totally bounded therefore, no matter the modulus of uniform continuity \(\ty{\delta}{\N}\) of the predicate, it can be searched.

\begin{theorem}
\label{thm:tb-csearch}
\thesislit{3}{SearchableTypes}{totally-bounded-csearchable}
If a closeness space is totally bounded and pointed, then it is uniformly continuously searchable.
\end{theorem}
\begin{proof}
Given any uniformly continuous predicate \(\ty{p}{K \to \Omega}\), where \(K\) is the totally bounded closeness space to search which is pointed with \(\ty{k^*}{K}\), we take \(K'\) to be the \(\delta\)-net of \(K\), where \(\ty{\delta}{\N}\) is the modulus of uniform continuity of \(p\).
By definition of nets (\cref{def:cover}), \(K'\) is finite and there are functions \(\ty{g}{K' \to K}\) and \(\ty{h}{K \to K'}\) such that for all \(\ty{k}{K}\) we have \(C_\delta(k,g(h(k)))\).

\vspace{1em}
As \(K'\) is finite and pointed (by \(\ty{h(k^*)}{K}\)), it is searchable (by \cref{lem:fin-searchable}).
We therefore search it for an answer to the predicate \(\ty{\left(p \circ g\right)}{K' \to \Omega}\), which we label \(\ty{k'_0}{K'}\). 
By the search condition (\cref{def:searchable}), \(k'_0\) is such that, if there is some \(\ty{k'}{K}\) such that \(p(g(k'))\), then \(p(g(k'_0))\).

\vspace{1em}
We therefore take \(\ty{g(k'_0)}{K}\) as the answer to \(p\), and must show that it satisfies the search condition; i.e.\ given some \(\ty{k}{K}\) such that \(p(k)\), it is the case that \(p(g(k'_0))\).
Of course, \(k\) is such that \(C_\delta(x,g(h(k)))\), and therefore --- by the uniform continuity of \(p\) (\cref{def:clos-ucont-pred}) --- \(p(g(h(k)))\).
Hence, because \((p \circ g)\) is satisfied, it is the case that \(p(g(k'_0))\).
\end{proof}

This theorem can be used to search a wide variety of infinite types (examples of which are given in the next subsection), but \emph{cannot} be used to give a version of the Tychonoff theorem (as discussed in \cref{sec:search-infinite-canwe}) in our framework --- we will come back to this in \cref{sec:tychonoff}.

\subsection{Examples of uniformly continuously searchable closeness spaces}
\label{sec:csearch-examples}

In this subsection, we give a variety of examples of uniformly continuously searchable types. Some of these are finite or totally bounded closeness spaces, while the others are preservation properties of continuous searchability that match up to those on the original definition of searchability.

\subsubsection{Finite uniformly continuously searchable spaces}

\begin{lemma}
\thesislit{3}{SearchableTypes-Examples}{finite-csearchable}
Every pointed, finite linearly ordered closeness space is uniformly continuously searchable.
\end{lemma}
\begin{proof}
By \cref{lem:fin-searchable,remark:search-csearch}.
\end{proof}

\subsubsection{Disjoint union of uniformly continuously searchable spaces}

\begin{lemma}
\thesislit{3}{SearchableTypes-Examples}{\urlplus-csearchable}
Given uniformly continuously searchable closeness spaces \(K\) and \(J\), the disjoint union closeness space \(K + J\) is uniformly continuously searchable.
\end{lemma}
\begin{proof}[Proof (Sketch).]
Given the predicate \(\ty{p}{\mathsf{decidable{\hy}predicate}(K + J)}\), we follow the same technique as \cref{lem:plus-searchable}, except we also have to show that the predicates\(p_K := p \circ \mathsf{inl}\) and \(p_J := p \circ \mathsf{inr}\) are uniformly continuous by the respective closeness functions on \(K\) and \(J\). Both of these proofs are immediate from the uniform continuity of \(p\) by the disjoint union closeness function (\cref{def:+-closeness}), as well as the uniform continuity of \(\mathsf{inl}\) and \(\mathsf{inr}\) respectively, and \cref{lem:f-p-ucont}.
\end{proof}

\subsubsection{Finite product of uniformly continuously searchable spaces}

\begin{lemma}
\label{lem:prod-csearchable}
\thesislit{3}{SearchableTypes-Examples}{\urlx-csearchable}
Given uniformly continuously searchable closeness spaces \(K\) and \(J\), the binary product closeness space \(K \x J\) is uniformly continuously searchable.
\end{lemma}
\begin{proof}[Proof (Sketch.)] \axioms{fp}
We follow the same technique as \cref{lem:prod-searchable}, except we also have to prove that --- assuming the given predicate is uniformly continuous by the binary product closeness function (\cref{def:prod-closeness}) --- each predicate in the family \(\ty{p_J}{K \to \mathsf{decidable{\hy}predicate}(J)}\) and the predicate \(\ty{p_K}{\mathsf{decidable{\hy}predicate}(K)}\) are uniformly continuous by the respective closeness functions on \(J\) and \(K\).

The former is straightforward by \cref{lem:x-C}, while the latter uses propositional extensionality in a similar way to the later \cref{lem:head-pred-ucont} but using \cref{lem:x-C} instead of \cref{lem:Pi-C-combine}. 
\end{proof}

\begin{corollary}
Given an \(\ty{n}{\N}\) and an \((n+1)\)-size vector \(\ty{Y}{\F (n+1) \to \U}\) of uniformly continuously searchable closeness spaces, the finite product closeness space of \((n+1)\)-size dependent vectors \(\F (n+1) \to Y_n\) is uniformly continuously searchable.
\end{corollary}
\begin{proof} \axioms{fp}
By induction and \cref{lem:prod-csearchable}.
\end{proof}

\subsubsection{Equivalent uniformly continuously searchable spaces}

\begin{lemma}
\thesislit{3}{SearchableTypes-Examples}{\urlsimeq-csearchable}
Given a closeness space \(K\) and a uniformly continuously searchable closeness space \(J\) such that \(K \simeq J\), the equivalent closeness space \(K\) is uniformly continuously searchable.
\end{lemma}
\begin{proof}
Given the predicate \(\ty{p}{\mathsf{decidable{\hy}predicate}(K)}\), we follow the same technique as \cref{lem:equiv-searchable}, except we also have to show that the predicate \(p' := p \circ g\) (where \(\ty{g}{J \to K}\) is derived from the proof of \(K \simeq J\), see \cref{def:simeq}) is uniformly continuous by the closeness function \(\ty{c_J}{\closeness{J}}\). 

We first assume that \(\ty{\delta}{\N}\) is the modulus of uniform continuity of \(p\), which is uniformly continuous by the subtype closeness function (\cref{def:subtype-closeness}) \(c_K := c_J \circ f\) (where \(\ty{f}{X \to Y}\) is the equivalence derived from the proof of \(K \simeq J\) which is such that \(f \circ g \sim \mathsf{id}_J\), again see \cref{def:simeq}). Recall that this means given any \(\ty{k_1,k_2}{K}\) such that \(C_\delta(k_1,k_2) := \underline \delta \preceq c_K(k_1,k_2) := \underline \delta \preceq c_J(f(k_1),f(k_2))\) then \(p(k_1) \Leftrightarrow p(k_2)\).

We will show that \(\delta\) is also the modulus of uniform continuity for \(\left( p \circ g \right)\); i.e.\ given \(\ty{j_1,j_2}{J}\) such that \(C_\delta(j_1,j_2) := \underline \delta \preceq c_J(j_1,j_2)\) then \(p(g(j_1)) \Leftrightarrow p(g(j_2))\). Because for \(i \in \{1,2\}\) we have \(j_i = f(g(j_i))\) then we also have \(C_n(j_i,f(g(j_i)))\) for all \(\ty{n}{\N}\). Therefore, by transitivity of the closeness relation (\cref{lem:C-eq}) we have \(C_\delta(f(g(j_1)),f(g(j_2))) := \underline \delta \preceq c_J(f(g(j_1)),f(g(j_2)))\), and thus the result follows by the uniform continuity of \(p\). 
\end{proof}

\subsubsection{Finite-sequence uniformly continuously searchable spaces}

One proof that discrete-sequence closeness spaces are uniformly continuously searchable comes from the fact that all such closeness spaces are totally bounded.

\begin{corollary}
\label{cor:disseq-csearch}
\thesislit{3}{SearchableTypes-Examples}{dep-discrete-finite-seq-csearchable}
Given an \(\N\)-indexed type family \(\ty{F}{\N \to \U}\) of finite linearly ordered types, the discrete-sequence closeness space \(\Pi F\) is uniformly continuously searchable.
\end{corollary}
\begin{proof}
By \cref{lem:disseq-totallybounded,thm:tb-csearch}.
\end{proof}

\begin{corollary}
\label{lem:disseq-csearch-1}
\thesislit{3}{SearchableTypes-Examples}{discrete-finite-seq-csearchable}
The type of sequences on any finite linearly ordered type is a uniformly continuously searchable closeness space.
\end{corollary}
\begin{proof}
By \cref{cor:disseq-csearch}.
\end{proof}

However, when extracting a search algorithm from this corollary, the \(\delta\)-net (where \(\ty{\delta}{\N}\) is the modulus of uniform continuity of the predicate being searched) must be fully computed before search can begin. This can cause efficiency issues, that are mildly improved by instead extracting a proof from the Tychonoff theorem (\cref{thm:tychonoff}) for uniformly continuously searchable types instead\footnote{In actuality, we rewrite the Tychonoff theorem specifically for finite-sequence spaces (as can be seen in the \textsc{Agda} formalisation), but the proof method is similar, and simpler, and so we leave it out to avoid repetition.}.

\subsection{Tychonoff theorem for uniformly continuously searchable spaces}
\label{sec:tychonoff}

\cref{thm:tb-csearch} allows us to prove that the Cantor space and most of the types we wish to search in \cref{chap:exact-real-search} are indeed uniformly continuously searchable.
However, in much the same way that searchability does not imply finiteness, continuous searchability does not imply totally boundedness.
This means that we cannot combine the theorem and \cref{lem:pi-totallybounded} to prove that continuous searchability preserves countable products.
In this subsection, we use a different proof technique, inspired by Escard\'o's (in \cite{Escardo08}) but which is not general recursive, to prove the Tychonoff theorem for uniformly continuously searchable spaces.

\begin{theorem}[Tychonoff theorem]
\label{thm:tychonoff}
\thesislit{3}{SearchableTypes-Examples}{tychonoff}
Given an \(\N\)-indexed type family \(\ty{T}{\N \to \U}\) of uniformly continuously searchable closeness spaces, the countable product closeness space \(\Pi T\) (see \cref{def:pi-closeness}) is uniformly continuously searchable.
\end{theorem}

The proof is by induction on the searched predicate's modulus of uniform continuity \(\ty{\delta}{\N}\). When \(\delta := 0\), then the result is trivial.
Otherwise, the idea of the technique is that we recursively construct finitely-many uniformly continuous and decidable predicates that test sequences \(\ty{xs}{\Pi T}\) with fixed prefixes of elements.
Each time the fixed prefix increases towards \(\delta\), the modulus of uniform continuity decreases towards \(0\); thus, the result will follow by \(\delta\)-many applications of the inductive hypothesis.

We first define the following family of `tail predicates'.

\begin{definition}
\thesislit{3}{SearchableTypes-Examples}{tail-predicate-tych}
Given an \(\N\)-indexed type family \(\ty{T}{\N \to \U}\), a decidable predicate \(\ty{p}{\mathsf{decidable{\hy}predicate}(\Pi T)}\) and a fixed head element \(\ty{x}{T_0}\), we define the decidable \emph{tail predicate} as follows:
\begin{alignat*}{3}
p_t &: \mathsf{decidable{\hy}predicate}(\Pi T) \to T_0 \to \mathsf{decidable{\hy}predicate}(\Pi (\mathsf{tail} \ T)) ,\\
p_t &(p,x) := \lambda xs.p(x :: xs).
\end{alignat*}
\end{definition}

\begin{lemma}
\label{lem:tail-pred-ucont}
\thesislit{3}{SearchableTypes-Examples}{tail-predicate-tych}
Given an \(\N\)-indexed type family \(\ty{T}{\N \to \U}\) of uniformly continuously searchable closeness spaces and a uniformly continuous and decidable predicate \(\ty{p}{\mathsf{decidable{\hy}uc{\hy}predicate}(\Pi T)}\) with modulus of uniform continuity \(\ty{\delta+1}{\N}\), the tail predicate \(p_t(p,x)\) for any fixed head element \(\ty{x}{T_0}\) is uniformly continuous with modulus of uniform continuity \(\ty{\delta}{\N}\).
\end{lemma}
\begin{proof}
This follows immediately from the uniform continuity of \(p\) and \cref{lem:Pi-C-combine}.
\end{proof}

\noindent
This family of `tail predicates' is matched by a family of `head predicates', which are defined mutually recursively with the proof of \cref{thm:tychonoff}.

\begin{definition}
\thesislit{3}{SearchableTypes-Examples}{head-predicate-tych}
Given an \(\N\)-indexed type family \(\ty{T}{\N \to \U}\) of uniformly continuously searchable closeness spaces and a uniformly continuous and decidable predicate \(\ty{p}{\mathsf{decidable{\hy}uc{\hy}predicate}(\Pi T)}\) with modulus of uniform continuity \(\ty{\delta+1}{\N}\), we define the decidable \emph{head predicate} as follows:
\begin{alignat*}{3}
p_h &: \mathsf{decidable{\hy}uc{\hy}predicate}(\Pi T) \to \mathsf{decidable{\hy}predicate}(T_0) ,\\
p_h &(p) := \lambda x.p(x :: \mathcal{E}_{\Pi (\mathsf{tail} \ T)}(p_t(p,x))),
\end{alignat*}
where \(\ty{\mathcal{E}_{\Pi (\mathsf{tail} \ T)}}{\mathsf{decidable{\hy}predicate}(\Pi (\mathsf{tail} \ T)) \to \mathsf{tail} \ T}\) is the uniformly continuous searcher (see \cref{def:c-searcher}) derived from the inductive hypothesis of \cref{thm:tychonoff}.
\end{definition}

\begin{lemma}
\label{lem:head-pred-ucont}
\thesislit{3}{SearchableTypes-Examples}{head-predicate-tych}
Given an \(\N\)-indexed type family \(\ty{T}{\N \to \U}\) of uniformly continuously searchable closeness spaces and a uniformly continuous and decidable predicate \(\ty{p}{\mathsf{decidable{\hy}uc{\hy}predicate}(\Pi T)}\) with modulus of uniform continuity \(\ty{\delta+1}{\N}\), the head predicate \(p_h(p,x)\) is uniformly continuous with modulus of uniform continuity \(\ty{\delta+1}{\N}\).
\end{lemma}
\begin{proof} \axioms{fp}
Given \(\ty{x,y}{T_0}\), we need to show that if \(C_{\delta+1}(x,y)\) then \(p(x :: \mathcal{E}_{\Pi (\mathsf{tail} \ T)}(p_t(p,x)))\) implies \(p(y :: \mathcal{E}_{\Pi (\mathsf{tail} \ T)}(p_t(p,y)))\).
Because \(p\) is uniformly continuous with modulus of uniform continuity \(\delta + 1\), this means showing that \(C_\delta(x :: \mathcal{E}_{\Pi (\mathsf{tail} \ T)}(p_t(p,x)),y :: \mathcal{E}_{\Pi (\mathsf{tail} \ T)}(p_t(p,y)))\). 

\vspace{1em}
Because we have \(C_{\delta+1}(x,y)\), we use \cref{lem:Pi-C-combine} to reduce this to needing to show \(C_\delta(\mathcal{E}_{\Pi (\mathsf{tail} \ T)}(p_t(p,x)),\mathcal{E}_{\Pi (\mathsf{tail} \ T)}(p_t(p,y)))\).
Using propositional extensionality (\cref{def:prop-equiv}), this follows from the fact that \(p_t(p,x)(xs) \Leftrightarrow p_t(p,y)(xs)\) for all \(\ty{xs}{\Pi X}\), which we prove below.

\vspace{1em}
Because \(C_{\delta+1}(x,y)\) and \(xs\) is infinitely close to itself (\cref{def:cspace}), then by \cref{lem:Pi-C-combine} we have \(C_{\delta+1}(x :: xs, y :: xs)\) and (by symmetry) \(C_{\delta+1}(y :: xs, x :: xs)\).
Therefore, using the uniform continuity of \(p\), we have \(p_t(p,x)(xs) \Leftrightarrow p_t(p,y)(xs)\).
\end{proof}

To reiterate, the idea is that, given any uniformly continuous and decidable predicate \(p\) with modulus of uniform continuity \(\delta+1\), the head of the sequence \(x\) is computed by finding an answer to the head predicate \(p_h(p)\), which requires finding an answer to the tail predicate \(p_t(p,x)\).
As the tail predicate has modulus of uniform continuity lower than the original predicate, it is recursively valid (i.e.\ it does not break \textsc{Agda}'s termination checker) to search for such an answer to the tail predicate.
The process then continues until the final tail predicate has modulus of uniform continuity \(0\).

\begin{proof}[Proof of \cref{thm:tychonoff}] \axioms{fp}
By induction on the modulus of uniform continuity \(\ty{\delta}{\N}\) on the predicate \(\ty{p}{\mathsf{decidable{\hy}uc{\hy}predicate}(\Pi T)}\) to be searched.

\ \\
When \(\delta := 0\), then by
\cref{remark:0-close} any element of \(\Pi T\) will satisfy the predicate. Therefore, recalling \cref{lem:csearchable-pointed}, we return the element
\[\left( \lambda n.\mathcal{E}_{T_n}(p^\top)\right) ,\] where \(\ty{\mathcal{E}_{T_n}}{\mathsf{decidable{\hy}uc{\hy}predicate}(T_n) \to T_n}\) is the uniformly continuous searcher on \(T_n\).

\ \\
When \(\delta := \delta' + 1\), then by \cref{lem:tail-pred-ucont} and the inductive hypothesis we construct the head predicate \[p_h(p) := \lambda x.p(x :: \mathcal{E}_{\Pi (\mathsf{tail} \ T)}(p_t(p,x))) ,\] which in turn constructs a tail predicate.
By \cref{lem:head-pred-ucont}, this head predicate is uniformly continuous and can thus be searched because \(T_0\) is uniformly continuously searchable --- therefore, we define \[x_0 := \mathcal{E}_{T_0}(p_h(p))).\] We then, by \cref{lem:tail-pred-ucont} and the inductive hypothesis, also define \[xs_0 := \mathcal{E}_{\Pi (\mathsf{tail} \ T)}(p_t(p,x)).\]
We return \(\ty{\left( x_0 :: xs_0 \right)}{\Pi T}\) as our answer to the predicate.

\ \\
We now need to show that, if there is some \(\ty{\alpha}{\Pi T}\) satisfying \(p(\alpha)\), then indeed \(p\left( x_0 :: xs_0 \right)\).
By \cref{lem:tail-pred-ucont} and the inductive hypothesis, for any \(\ty{x}{X}\) the tail predicate \(p_t(p,x)\) is such that if there is some \(\ty{xs}{\Pi (\mathsf{tail} \ T)}\) satisfying \(p_t(p,x)(xs)\) then \(p_t(p,x)(\mathcal{E}_{\Pi (\mathsf{tail} \ T)}(p_t(p,x)))\).
Therefore, because \[p_t(p,\alpha_0)(\mathsf{tail} \ \alpha) := p(\alpha_0 :: \mathsf{tail} \ \alpha),\] (by \cref{lem:Pi-head-tail-eta}), we have \[p_t(p,\alpha_0)(\mathcal{E}_{\Pi (\mathsf{tail} \ T)}(p_t(p,\alpha_0))) := p(\alpha_0 :: \mathcal{E}_{\Pi (\mathsf{tail} \ T)}(p_t(p,\alpha_0))) =: p_h(p,\alpha_0).\]
Of course, \(\ty{x_0}{T_0}\) is such that if there is some \(\ty{x}{T_0}\) satisfying \(p_h(p)(x)\) then \(p_h(p)(x_0)\).
Therefore, because we have shown \(\alpha_0\) satisfies \(p_h(p)(\alpha_0)\), we have \[p_h(p)(x_0) := p(x_0 :: \mathcal{E}_{\Pi (\mathsf{tail} \ T)}(p_t(p,x_0))) =: p(x_0 :: xs_0),\]
which is what we wanted.
\end{proof}
\chapter{Generalised Optimisation and Regression}
\label{chap:generalised}

We are interested in how the general view of search on infinite structures yielded by uniformly continuously searchable types can be applied to the purposes of performing optimisation and regression.

These foundational concepts of analysis are usually defined on \emph{real-valued} functions of \emph{multiple real variables}.
Informally, given a type \(\R\) of real numbers, the central goal of \emph{optimisation} is to compute, with mathematical guarantees, an argument which gives an approximately optimal value (e.g.\ a local/global minimum or maximum) of some given objective function \(\ty{f}{\R^d \to \R}\) in the compact intervals \([a_0,b_0],...,[a_{d-1},b_{d-1}]\), subject to some given constraints and degree of approximation.
The applications are numerous and obvious, in all areas of computational sciences, but in particular, this problem subsumes many aspects of \emph{parametric regression} if the objective function is a loss function, representing a kind of distance between a parameterised model and a reference model (or just some sampled data), that we wish to minimise~\cite{LossFunction}.

In this chapter, we provide a methodological contribution by describing type-theoretic variants of optimisation and regression to the general class of types (i.e.\ totally bounded and uniformly continuous searchable closeness spaces) we introduced in \cref{chap:searchable}.

\section{Global Optimisation}
\label{sec:global-opt}

Much of the recent literature in optimisation theory has focused on \emph{local optimisation}\footnote{Optimisation often concerns itself with finding maxima, rather than minima --- however, as these two problems are equivalent (we can find a maxima of \(f\) by finding a minima of \(-f\)), we will only focus on finding minima.} and methods for computing it efficiently, such as \emph{gradient descent} and supporting techniques such as \emph{automatic differentiation}~\cite{Baydin}.
Local optimisation is attractive because it can be computed efficiently, even for functions with high dimensionality.

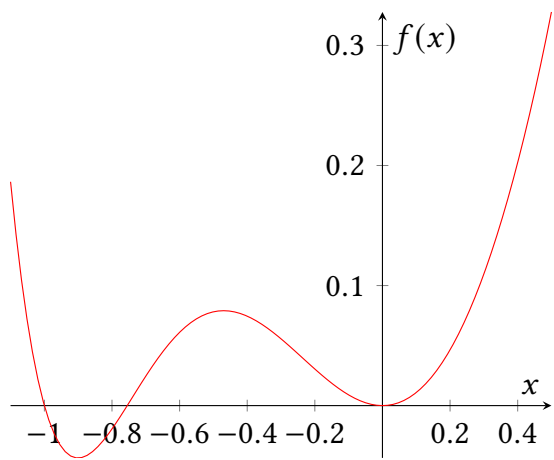
\begin{figure}
\centering
\resizebox{\columnwidth/2}{!}{
\begin{tikzpicture}
\begin{axis}[
    axis lines = center,
    xlabel = \(x\),
    ylabel = {\(f(x)\)},
]
\addplot [
    domain=-1.1:0.5,
    samples=100,
    color=red,
]
{x^6 - x^4 + x^3 + x^2};
\end{axis}
\end{tikzpicture}
}
\caption{Graph of the function \(f(x) = x^6 - x^4 + x^3 + x^2\) with \(-1.1 \leq x \leq 1\).}
\label{fig:graph}
\end{figure}

However, there exists a mathematically attractive alternative: computing a \emph{global minimum argument} of \(f\) in \([a,b]\), which can, for many problems, produce much stronger correctness guarantees than a local minimum~\cite{NewPerspectives,}.
A global minimum argument is an argument choice \(\ty{xs}{\R^d}\) such that \(f(xs)\) is a \emph{global minimum value} of \(f\) --- i.e., for any other choice \(\ty{ys}{\R^d}\), we have \(f(xs) \leq f(ys)\).

\begin{remark}
\lit
Given endpoints \(\ty{a,b}{\R}\) such that \(a \leq b\), we write \(x \in [a,b]\) to mean any real \(\ty{x}{\R}\) such that \(a \leq x \leq b\).
\end{remark}

\begin{definition}
\lit
Given a function \(\ty{f}{\R \to \R}\) and endpoints \(\ty{a,b}{\R}\), a given point \(x_0 \in [a,b]\) is a \emph{global minimum argument} of \(f\) in the compact interval \([a,b]\) if for all \(x \in [a,b]\) we have \(f(x_0) \leq f(x)\).
\end{definition}

For example, in \cref{fig:graph}, a global minimum argument of \(f(x)\) in the compact interval \([-1.1,1]\) is \(x \approx -0.90047\) (which yields a global minimum value \(f(x) \approx -0.04366\)), while a local minimum is \(x = 0\) (which yields a local minimum value \(f(x) = 0\)).

\clearpage

In constructive mathematics, computing a global minimum \emph{value} of a function on a compact interval is possible. For example, Simpson achieves this using the ternary signed-digit encoding of real numbers in the compact interval \([-1,1]\) (explored in \cref{sec:signed-digits})~\cite{Simpson}.
Nevertheless, the constructive existence of a global minimum \emph{argument} --- an argument to the function which gives a minimum value --- is a consequence of the extreme value theorem which is not valid in constructive mathematics, meaning that they cannot in general be computed~\cite{Troelstra}.
The fact that a global minimum argument is non-computable on the reals relies on the lack of a linear ordering: in constructive mathematics we cannot decide, for \(\ty{x,y}{\R}\), whether \((x \leq y) + (y \leq x)\)~\cite{Katz,Troelstra}. Indeed, this is sometimes referred to as the analytic \emph{lesser limited principle of omniscience} (LLPO)~\cite{ShulmanLLPO}.

Despite the non-computability of a global minimum argument, it is the case that, subject to continuity and compactness constraints (which are reasonable for many problem domains), an \emph{approximate} global minimum argument \(x\) of \(f\) is computable such that \(f(x)\) is a global minimum value \emph{up to a given precision} \(\varepsilon\).

\begin{definition}
\label{def:real-min}
\lit
Given a function \(\ty{f}{\R \to \R}\), endpoints \(\ty{a,b}{\R}\) and any precision \(\ty{\varepsilon}{\R}\), a given point \(x_0 \in [a,b]\) is an \emph{\(\varepsilon\)-global minimum argument} of \(f\) in the compact interval \([a,b]\) if for all \(x \in [a,b]\) such that \(\neg C_\varepsilon(f(x_0),f(x))\) we have \(f(x_0) \leq f(x)\).
\end{definition}

\noindent
Any algorithm which computes an \(\varepsilon\)-global minimum argument of a given function \(\ty{f}{\R \to \R}\), for any precision \(\ty{\varepsilon}{\R}\), is a \emph{global optimisation algorithm}\footnote{Global optimisation algorithms in general are defined for functions of multiple real variables; however, as our generalised perspective will eliminate the difference between the two, we stated only the single argument version for ease of presentation.}.

We are interested in constructing, for our wide class of types, global optimisation algorithms which are \emph{general} (few assumptions regarding a particular shape of the function graph) and \emph{complete} (guaranteed to find the solution within a specified margin of error).
In particular, we focus on algorithms in the style of Piyavskii, which apply to continuous functions and which discretise the domain of the function similar to a branch-and-bound-style optimisation algorithm~\cite{piyavskii1972algorithm,BnB}.
The continuity property is used for such algorithms to guarantee that the granularity of the domain can be used to bound the imprecision of the value of the objective function.

In this section, we outline how within our framework we can describe the convergent global optimisation of functions \(\ty{f}{X \to Y}\), where the types \(X\) and \(Y\) are kept as general as possible (and include those types that in \cref{chap:reals} we introduce to our type theory in order to encode the real numbers).

\subsection{Orders and approximate orders}
\label{sec:orders}

In order to reason about the minima of functions \(\ty{f}{X \to Y}\), our theory must have a concept of an order on those value types \(Y\).
We remain general in our idea of exactly which types these are, as we wish to establish a theory of \emph{generalised} optimisation and regression.

We first recall some basic notions of orders on a given type.

\begin{definition}
\label{def:preorder}
\thesislit{4}{ApproxOrder}{is-preorder}
For any type \(X\), a binary relation \(\ty{\leq}{X \to X \to \Omega}\) is a \emph{preorder on \(X\)} if it is reflexive and transitive:
\begin{enumerate}[(i)]
\item \(x \leq x\),
\item \(x \leq y \to y \leq z \to x \leq z\).
\end{enumerate}
\end{definition}

\begin{definition}
\label{def:linear-preorder}
\thesislit{4}{ApproxOrder}{is-linear-preorder}
For any type \(X\), a binary relation \(\ty{\leq}{X \to X \to \Omega}\) is a \emph{linear preorder on \(X\)} if it is a preorder and if, for all \(\ty{x,y}{X}\), it is the case that \((x \leq y) + (y \leq x)\) holds.
\end{definition}

\begin{definition}
\label{def:antisym-preorder}
\thesislit{4}{ApproxOrder}{is-partial-order}
For any type \(X\), a binary relation \(\ty{\leq}{X \to X \to \Omega}\) is a \emph{partial order on \(X\)} if it is an antisymmetric preorder; i.e. if for any \(\ty{x,y}{X}\) such that \(x \leq y\) and \(y \leq x\) it is the case that \(x = y\).
\end{definition}

\begin{definition}
\label{def:linear-order}
\thesislit{4}{ApproxOrder}{is-linear-order}
For any type \(X\), a binary relation \(\ty{\leq}{X \to X \to \Omega}\) is a \emph{linear order on \(X\)} if it is a preorder that is both antisymmetric and linear.
\end{definition}

The easiest example of a class of types with linear orders are the finite linearly ordered types.

\begin{remark}
\thesislit{4}{ApproxOrder-Examples}{finite-order}
Every finite linearly ordered type \(F\) has a linear order \(\ty{\leq_F}{F \to F \to \Omega}\) inherited from \(\F(n)\) as discussed in \cref{sec:finite-linear-ordered}.
\end{remark}

Recall that in the preamble to this section, we noted that we cannot compute a global minimum argument of a function \(\ty{f}{\R \to \R}\) because we cannot prove constructively that the expected preorder on the reals is linear.
We illustrate this further, without using the reals, by the following motivating example of a class of infinite types \(\seq{D}\) (where \(D\) is a discrete set) which has a partial order \(\ty{\leq_D}{D \to D \to \Omega}\) whose linearity --- similarly to that of the reals --- amounts (in non-trivial cases) to a constructive taboo.
This example is not just chosen for its intelligibility; this class of types is crucial for our later purposes of representing computable real numbers.

We first define the usual \emph{lexicographic order} on the discrete-sequence type \(\seq D\) --- recall from the literature that this order says \(\alpha \leq_{\seq D} \beta\) if at every point \(\ty{n}{\N}\) that \(\alpha\) and \(\beta\) agree prior to, we have \(\alpha_n \leq_{seq D}\beta_n\). For example, setting \(D := \2\), 
\vspace{-0.25em}
\[\{0,0,0,1,0,...\} \leq_{seq D}\{0,0,0,1,1,...\} \leq_{seq D}\{0,1,0,1,1,...\} \leq_{seq D}\{1,1,1,1,1,...\} .\]

\begin{definition}
\label{def:seq-order}
\thesislit{4}{ApproxOrder-Examples}{discrete-lexicorder}
Given a preorder \(\ty{\leq_D}{D \to D \to \Omega}\) on a discrete set \(D\), the \emph{lexicographic order} \(\ty{\leq_{\seq D}}{\seq D \to \seq D \to \Omega}\) on \(\seq D\) is defined:
\begin{alignat*}{3}
& \leq_{\seq D} &&: \seq{D} \to \seq{D} \to \Omega,\\
\alpha &\leq_{\seq D} \beta &&:= \Pity{n}{\N}{\alpha \sim^n \beta \to \alpha_n \leq_D \beta_n}.
\end{alignat*}
\end{definition}

\begin{lemma}
\label{lem:lexicographic-preorder}
\thesislit{4}{ApproxOrder-Examples}{discrete-lexicorder-is-preorder}
Given a partial order \(\ty{\leq_D}{D \to D \to \Omega}\) on a discrete set \(D\), the lexicographic order \(\ty{\leq_{\seq D}}{\seq D \to \seq D \to \Omega}\) is a preorder on \(\seq D\).
\end{lemma}
\begin{proof} \axioms{f}
We prove both reflexivity and transitivity of the lexicographic order:
\begin{enumerate}[(i)]
\item Given \(\ty{\alpha}{\seq D}\), then \(\alpha \leq_{\seq D} \alpha\) because, by reflexivity of the partial order on \(D\), \(\alpha_n \leq_D \alpha_n\) for all \(\ty{n}{\N}\).
\item Given \(\ty{\alpha,\beta,\zeta}{\seq D}\), such that \(\alpha \leq_{\seq D} \beta\) and \(\beta \leq_{\seq D} \zeta\) we need to show that, given any \(\ty{n}{\N}\) such that \(\alpha \sim^n \zeta\) we have \(\alpha_n \leq_D \zeta_n\). We continue by induction on the given \(n\). 
    \begin{itemize}
    \item[] In the case where \(n := 0\) then we have \(\alpha_0 \leq_D \beta_0\) and \(\beta_0 \leq_D \zeta_0\), because \(\alpha \sim^0 \beta\) and \(\beta \sim^0 \zeta\) trivially hold; therefore \(\alpha_0 \leq_D \zeta_0\) holds by transitivity of the partial order on \(D\).
    \item[] In the case where \(n := n' + 1\) then we have \(\alpha \sim^{n'+1} \zeta\) (and hence \(\mathsf{tail}\ \alpha \sim^{n'} \mathsf{tail}\ \zeta\)) and want to show that \((\mathsf{tail}\ \alpha)_{n'} \leq_D (\mathsf{tail}\ \zeta)_{n'}\). We do this by invoking the inductive hypothesis after first proving \(\mathsf{tail}\ \alpha \leq_{\seq D} \mathsf{tail}\ \beta\) and \(\mathsf{tail}\ \beta \leq_{\seq D} \mathsf{tail}\ \zeta\). 
        \begin{itemize}
        \item[] In order to prove that \(\mathsf{tail}\ \alpha \leq_{\seq D} \mathsf{tail}\ \beta\), we need to show that, given any \(\ty{i}{\N}\) such that \(\mathsf{tail}\ \alpha \sim^i \mathsf{tail}\ \beta\) we have \((\mathsf{tail}\\alpha)_i \leq_D (\mathsf{tail}\ \beta)_i\). This follows from the fact that \(\alpha \leq_{\seq D} \beta\) once we show that \(\alpha \sim^{i+1} \beta\); i.e. for all \(\ty{j}{\N}\) such that \(j < i+1\) we have \(\alpha_j = \beta_j\). This is shown by induction on \(j\).
            \begin{itemize}
            \item[] In the case where \(j := 0\), we have both \(\alpha_0 \leq_D \beta_0\) (trivially from \(\alpha \leq_{\seq D} \beta\)) and \(\beta_0 \leq_D \alpha_0\) (by \(\alpha \sim^{n'+1} \zeta\) and \(\beta \leq_{\seq D} \zeta\)). Therefore \(\alpha_0 = \beta_0\) follows by antisymmetry of the partial order on \(D\).
            \item[] In the case where \(j := j' + 1\) then we have \(j' < i\); therefore \(\alpha_{j'+1} = \beta_{j'+1}\) follows immediately from \(\mathsf{tail}\ \alpha \sim^i \mathsf{tail}\ \beta\).
            \end{itemize}
        \item[] The proof that \(\mathsf{tail}\ \beta \leq_{\seq D} \mathsf{tail}\ \zeta\) is by the same reasoning.
        \end{itemize}
    \end{itemize}
\end{enumerate}
\end{proof}

Assuming this order is linear allows us to prove a constructive taboo, the \emph{weak principle of omniscience} (WLPO) --- therefore, the linearity of this order is constructively invalid.

\begin{remark}
\label{remark:WLPO}
WLPO states that, given a binary sequence, we can decide whether or not all points of the sequence are \(1\). Another way of saying this is that, given any extended natural number, we can decide whether or not it is equal to infinity.
\begin{alignat*}{2}
\mathsf{WLPO} &: \U,\\
\mathsf{WLPO} &:= \Pity{u}{\Ni}{\mathsf{is{\hy}decidable} (u = \infty)}.
\end{alignat*}
\noindent
Note that WLPO is implied by LPO (\cref{remark:LPO}). \MartinEscardo~has previously shown that WLPO is equivalent to the existence of a function \(\ty{f}{\Ni \to \2}\) that is \emph{discontinuous} in the sense that, given some \(\ty{u}{\Ni}\), it is the case that \(f(u) = 0\) if \(u = \overline n\) (for some \(\ty{n}{\N}\)) but \(f(u) = 1\) if \(u = \infty\)~\cite{EscardoWLPO}.
\begin{alignat*}{2}
\mathsf{WLPO}' &: \U,\\
\mathsf{WLPO}' &:= \\ & \Sigmatye{f}{\Ni \to \2}{\Pity{u}{\Ni}{\left( \Sigmaty{n}{\N}{u = \overline n} \to f(u) = 0 \right) \x \left( u = \infty \to f(u) = 1 \right)}}.
\end{alignat*}
\end{remark}

\begin{lemma}
\typetop{Taboos}{BasicDiscontinuity}{\urlN\urlinfty-linearity-taboo}
\thesislit{4}{ApproxOrder-Examples}{linear-finite-lexicorder-implies-WLPO}
The linearity of the lexicographic order \(\ty{\leq_{\seq F}}{\seq F \to \seq F \to \Omega}\) on \(\seq F\), where \(F\) is a finite linearly ordered type with more than one element, is logically equivalent to WLPO.
\end{lemma}
\begin{proof}
The first, rather straightforward step, is showing that the linearity of \(\ty{\leq_{\seq F}}{\relation{\seq F}}\) implies the linearity of \(\ty{\preceq}{\relation{\Ni}}\) (as defined in \cref{def:Ni-order}).
The idea here is that, because \(F\) is finite linearly ordered and has more than one element, we define an order-preserving function \(\ty{\rho}{\2 \to F}\) --- i.e. for \(\ty{a,b}{\2}\) if \(a \leq_\2 b\) then \(\rho(a) \leq_F \rho(b)\) --- that we use pointwise to define an order-preserving function \(\ty{\mathsf{map}(\rho)}{\seq \2 \to \seq F}\). 
This means that, given any \(\ty{u,v}{\Ni}\), it is the case that \(u \preceq v\) holds if and only if \(\mathsf{map}(\rho,\mathsf{fst}(u)) \leq_{\seq F} \mathsf{map}(\rho,\mathsf{fst}(v))\) holds.
Using this, the linearity of \(\leq_{\seq F}\) implies the linearity of \(\preceq\).

We can now proceed with the core of the proof\footnote{This interesting proof was written and shown to me by \MartinEscardo~while I was completing my thesis corrections, who gave me permission to reproduce it here.}, i.e. that the linearity of \(\preceq\) implies \(\mathsf{WLPO}'\) (as given in \cref{remark:WLPO}).
Using the linearity of \(\preceq\), we define a function \(\ty{\mathsf{linearity{\hy}decider}}{\Ni \to \Ni \to \2}\) that on input of two extended naturals \(\ty{u,v}{\Ni}\) checks which side of \((u \preceq v) + (v \preceq u)\) holds and returns \(0\) if the left-hand side holds and \(1\) if the right-hand side holds. Of course, given two identical items either \(0\) or \(1\) can be returned. We make the following three observations for all \(\ty{n}{\N}\):
\begin{enumerate}
\item \(\mathsf{linearity{\hy}decider}(\overline n,\infty) = 0\),
\item \(\mathsf{linearity{\hy}decider}(\infty,\overline n) = 1\),
\item \(\mathsf{linearity{\hy}decider}(\infty,\infty)\) could be either \(0\) or \(1\).
\end{enumerate}
Using these observations and the function, we define a discontinuous function \(\ty{f}{\Ni \to \2}\) that returns \(0\) on a finite input and \(1\) on an infinite input. We do this by comparing \(\infty\) to \emph{itself} and defining \(f\) based on what is output. 
\begin{itemize}
\item[] If \(\mathsf{linearity{\hy}decider}(\infty,\infty) = 1\) then we define \(f(x) := \mathsf{linearity{\hy}decider}(x,\infty)\). On finite input (when \(x = \overline n\) for some \(\ty{n}{\N}\)), it is the case that \(f(x) := \mathsf{linearity{\hy}decider}(\overline n,\infty) = 0\) (by the first observation above). On infinite input (when \(x = \infty\)) then it is the case that \(f(x) := \mathsf{linearity{\hy}decider}(\infty,\infty) = 1\).\vspace{0.5em}
\item[] If \(\mathsf{linearity{\hy}decider}(\infty,\infty) = 0\) then we define \(f(x) := \mathsf{flip}_\2 (\mathsf{linearity{\hy}decider}(\infty,x))\). On finite input, it is the case that \(f(x) := \mathsf{flip}_\2(\mathsf{linearity{\hy}decider}(\infty,\overline n)) = \mathsf{flip}_\2(1) = 0\) (by the second observation above). On infinite input (when \(x = \infty\)) then it is the case that \(f(x) := \mathsf{flip}_\2(\mathsf{linearity{\hy}decider}(\infty,\infty)) = \mathsf{flip}_\2(0) = 1\).
\end{itemize}
In both cases we have defined a discontinuous function, and therefore the linearity of \(\preceq\) allows us to define \(\mathsf{WLPO}'\) and (by \cref{remark:WLPO}) \(\mathsf{WLPO}\) itself.
\end{proof}

A global optimisation algorithm must compute an approximate global minimum argument; in order to do this, we require some form of approximate linear preordering that relates to the underlying pre-order we are unable to prove the linearity of.
We therefore use our closeness spaces to define this concept of approximate linear preorders.

\begin{definition}
\label{def:approx-order}
\thesislit{4}{ApproxOrder}{is-approx-order}
For any closeness space \(X\), an \(\N\)-indexed family of binary relations \(\ty{\leq^-}{\N \to X \to X \to \Omega}\) is an \emph{approximate linear preorder} if, for any \(\ty{\varepsilon}{\N}\), the relation \(\ty{\leq^\varepsilon}{X \to X \to \Omega}\) is a linear preorder satisfying,
\begin{enumerate}[(i)]
\item \(\mathsf{decidable}(x \leq^\varepsilon y)\),
\item \(C_\varepsilon(x,y) \to x \leq^\varepsilon y\),
\end{enumerate}
\end{definition}

\noindent
The idea of the above definition is that two elements of a closeness space with an approximate linear preorder can be ordered in a decidable way and that \(\varepsilon\)-close elements must be considered indistinguishable from the point of view of the approximate order.

In the following definition, we state what it means for an approximate linear preorder to relate to the \emph{genuine} preorder.

\begin{definition}
\label{def:approx-order-relates}
\thesislit{4}{ApproxOrder}{ApproxOrder-Relates.approx-order-relates}
For any closeness space \(X\), an approximate linear preorder \(\ty{\leq^-}{\N \to X \to X \to \Omega}\) \emph{relates to} a preorder \(\ty{\leq}{X \to X \to \Omega}\) if the following hold:
\begin{enumerate}[(i)]
\item \(x \leq y \to \existstye{n}{\N}{\Pity{\varepsilon}{\N}{n < \varepsilon \to x \leq^\varepsilon y}}\),
\item \(\Pity{n}{\N}{x \leq^n y} \to x \leq y\)
\end{enumerate}
\end{definition}

\noindent
This first half of the above relationship says that if the two elements are genuinely such that \(x \leq y\) then \emph{eventually} the approximate order recognises this; the second half says that if the two elements are \emph{always} approximately ordered such that \(x \leq^n y\) then they genuinely have that order. 

We briefly note here that approximate linear preorders yield uniformly continuous and decidable predicates.

\begin{lemma}
\label{lem:approx-order-ucd-pred}
\thesislit{4}{ApproxOrder}{approx-order-uc-predicate-l}
\thesislit{4}{ApproxOrder}{approx-order-uc-predicate-r}
Given an approximate linear preorder \(\ty{\leq^-}{\N \to X \to X \to \Omega}\) on a closeness space \(X\), the predicates \[ p_l^{y,\varepsilon}(x) := x \leq^\varepsilon y \text{ and } p_r^{y,\varepsilon}(x) := y \leq^\varepsilon x \] are uniformly continuous and decidable for any \(\ty{y}{X}\).
\end{lemma}
\begin{proof}
The decidability of the predicates is immediate from \cref{def:approx-order}.(i).
The modulus of uniform continuity for both predicates is \(\ty{\varepsilon}{\N}\): we show this for \(p_l^{y,\varepsilon}\); the proof for \(p_r^{y,\varepsilon}\) is identical.

\vspace{0.25cm}
Given \(\ty{x_1,x_2}{X}\) such that \(C_\varepsilon(x_1,x_2)\), we want to show that \(x_1 \leq^\varepsilon y\) implies \(x_2 \leq^\varepsilon y\).
By the closeness of \(x_1\) and \(x_2\), we trivially have \(x_2 \leq^\varepsilon x_1\) from \cref{def:approx-order}.(ii).
The result then follows by the approximate linear preorder's transitivity.
\end{proof}

Returning to our motivating example, we show that sequences of finite linearly ordered types have an approximate linear preorder.

\begin{definition}
\thesislit{4}{ApproxOrder-Examples}{discrete-approx-lexicorder}
Given a preorder \(\ty{\leq_D}{D \to D \to \Omega}\) on a discrete set \(D\), the \emph{approximate lexicographic order} \(\ty{\leq^-_{\seq D}}{\N \to \seq D \to \seq D \to \Omega}\) on \(\seq D\) is defined:
\begin{alignat*}{3}
& \leq^\varepsilon_{\seq D} &&: \N \to \seq{D} \to \seq{D} \to \Omega,\\
\alpha &\leq^\varepsilon_{\seq D} \beta &&:= \Pity{n}{\N}{n < \varepsilon \to \alpha \sim^n \beta \to \alpha_n \leq_D \beta_n}.
\end{alignat*}
\end{definition}

\begin{lemma}
\label{lem:lexicographic-approx}
\thesislit{4}{ApproxOrder-Examples}{discrete-approx-lexicorder-is-approx-order}
Given a linear order \(\ty{\leq_D}{D \to D \to \Omega}\) on a discrete set \(D\), the approximate lexicographic order \(\ty{\leq^-_{\seq D}}{\N \to \seq D \to \seq D \to \Omega}\) is an approximate linear preorder.
\end{lemma}
\begin{proof} \axioms{f}
The proof that \(\alpha \leq^\varepsilon_{\seq D} \beta\) is a preorder is almost identical to \cref{lem:lexicographic-preorder}, and so we avoid repeating it here. In order to prove this preorder is linear, we proceed by induction on \(\ty{\varepsilon}{\N}\).
    \begin{itemize}
    \item[] In the case where \(\varepsilon := 0\) then both \(\alpha \leq^0_{\seq D} \beta\) and \(\beta \leq^0_{\seq D} \zeta\) are trivially proved. We arbitrarily choose the former and show that for all \(\ty{i}{\N}\) such that \(i < 0\) and \(\alpha \sim^i \beta\) we have \(\alpha_i \leq \beta_i\) vacuously holds because \(i < 0\) is empty.
    \item[] In the case where \(\varepsilon := \varepsilon' + 1\), we proceed by the linearity of the linear order on \(D\) -- i.e. we check which side of \(\left( \alpha_0 \leq_D \beta_0 \right) + \left( \beta_0 \leq_D \alpha_0 \right)\) and the inductive hypothesis -- i.e. we check which side of \(\left( \mathsf{tail}\ \alpha \leq^{\varepsilon'}_{\seq D} \mathsf{tail}\ \beta \right) + \left( \mathsf{tail}\ \beta \leq^{\varepsilon'}_{\seq D} \mathsf{tail}\ \alpha \right)\) holds.
        \begin{itemize}[]
        \item[] In the case where \(\alpha_0 \leq_D \beta_0\) and \(\mathsf{tail}\ \alpha \leq^{\varepsilon'}_{\seq D} \mathsf{tail}\ \beta\) then it is straightforward to show that \(\alpha \leq^\varepsilon_{\seq D} \beta\). The case where \(\beta_0 \leq_D \alpha_0\) and \(\mathsf{tail}\ \beta \leq^{\varepsilon'}_{\seq D} \mathsf{tail}\ \alpha\) is similarly straightforward.
        \item[] In the case where \(\alpha_0 \leq_D \beta_0\) and \(\mathsf{tail}\ \beta \leq^{\varepsilon'}_{\seq D} \mathsf{tail}\ \alpha\), we use the discreteness of \(D\) to check whether or not \(\alpha_0 = \beta_0\) holds. If it does, then also \(\beta_0 \leq \alpha_0\) by reflexivity of the linear order on \(D\), and hence \(\beta \leq^\varepsilon_{\seq D} \alpha\). If it does not, then we show that \(\alpha \leq^\varepsilon_{\seq D} \beta\) by induction on the \(\ty{i}{\N}\) that is such that \(i < \varepsilon\) and \(\alpha \sim^{i} \beta\) in order to show that \(\alpha_{i} \leq \beta_{i}\).
            \begin{itemize}
            \item[] The case where \(i := 0\) is trivial as we already have \(\alpha_0 \leq_D \beta_0\).
            \item[] The case where \(i := i' + 1\) is vacuous as there is a contradiction between \(\alpha \sim^{i'+1} \beta\) and \(\alpha_0 \neq \beta_0\).
            \end{itemize}
        \item[] The case where \(\beta_0 \leq_D \alpha_0\) and \(\mathsf{tail}\ \alpha \leq^{\varepsilon'}_{\seq D} \mathsf{tail}\ \beta\) uses the same technique as the previous case.
        \end{itemize}
    \end{itemize}
Following this, we prove the numbered conditions of \cref{def:approx-order}:
\begin{enumerate}[(i)]
\item In order to prove that the linear preorder is decidable, we again proceed by induction on \(\ty{\varepsilon}{\N}\). Recall from the above that the base case is trivial; in the case where \(\varepsilon := \varepsilon' + 1\) we use the induction hypothesis to check whether or not \(\alpha \leq^{\varepsilon'}_{\seq D} \beta\) holds. If it does not, then \(\alpha \leq^{\varepsilon}_{\seq D} \beta\) clearly does not hold. If it does, we next check whether or not \(\alpha \sim^{\varepsilon'} \beta\) holds (by \cref{lem:sim-decidable}).
    \begin{itemize}
    \item[] In the case where \(\alpha \sim^{\varepsilon'} \beta\) holds, we further check whether or not \(\alpha_{\varepsilon'} \leq \beta_{\varepsilon'}\) (by the fact that any linear order on a discrete type is decidable). If it does not, then \(\alpha \leq^{\varepsilon}_{\seq D} \beta\) clearly does not hold. If it does, then \(\alpha \leq^{\varepsilon}_{\seq D} \beta\) holds; i.e. if there is an \(\ty{i}{\N}\) that is such that \(i < \varepsilon\) and \(\alpha \sim^{i} \beta\) then \(\alpha_{i} \leq \beta_{i}\). If \(i < \varepsilon'\) then \(\alpha_{i} \leq \beta_{i}\) immediately follows by \(\alpha \leq^{\varepsilon'}_{\seq D} \beta\) and \(\alpha \sim^{\varepsilon'} \beta\). If \(i = \varepsilon'\), then the result is immediate as we have \(\alpha_{\varepsilon'} \leq \beta_{\varepsilon'}\).
    \item[] In the case where \(\alpha \sim^{\varepsilon'} \beta\) does not hold, then \(\alpha \leq^{\varepsilon}_{\seq D} \beta\) holds; i.e. if there is an \(\ty{i}{\N}\) that is such that \(i < \varepsilon\) and \(\alpha \sim^{i} \beta\) then \(\alpha_{i} \leq \beta_{i}\). If \(i < \varepsilon'\) then then, as above, the result is immediate. If \(i = \varepsilon'\), then result is vacuous as there is a contradiction between the assumptions that both \(\neg \alpha \sim^{\varepsilon'} \beta\) and \(\alpha \sim^{i} \beta\).
    \end{itemize}
\item By \cref{lem:disseq-C-equiv}, the assumption that \(C_\varepsilon(\alpha,\beta)\) is propositionally equivalent to \(\alpha \sim^\varepsilon \beta\). Therefore we assume the latter statement and want to show that if there is an \(\ty{i}{\N}\) that is such that \(i < \varepsilon\) and \(\alpha \sim^{i} \beta\) then \(\alpha_{i} \leq \beta_{i}\). Beecause \(i < \varepsilon\), the assumption tells us that \(\alpha_i = \beta_i\), and therefore the result follows by reflexivity of the linear order on \(D\).
\end{enumerate}
\end{proof}

\begin{lemma}
\label{lem:lexicographic-approx-relates}
\thesislit{4}{ApproxOrder-Examples}{LexicographicOrder-Relates.discrete-approx-lexicorder-relates}
Given a discrete set \(D\), the approximate lexicographic order \(\ty{\leq^-_{\seq D}}{\N \to \seq D \to \seq D \to \Omega}\) relates to the lexicographic order \(\ty{\leq_{\seq D}}{\seq D \to \seq D \to \Omega}\).
\end{lemma}
\begin{proof} \axioms{t}
We prove the numbered conditions of \cref{def:approx-order-relates}:
\begin{itemize}[(i)]
\item Assuming \(\alpha \leq_{\seq D} \beta\), then for all \(\ty{i}{\N}\) such that \(\alpha \sim^i \beta\) we have \(\alpha_i \leq \beta_i\). This means that \(\alpha \leq^\varepsilon_{\seq D} \beta\) for \emph{all} \(\ty{\varepsilon}{\N}\), and so by setting \(n := 0\) we have \( \sigmatye{n}{\N}{\pity{\varepsilon}{\N}{n < \varepsilon \to x \leq^\varepsilon y}}\). The result then follows by truncating this final proof term.
\item Assuming \(\alpha \leq^n_{\seq D} \beta\) holds for all \(\ty{n}{\N}\), we need to show that given any \(\ty{i}{\N}\) such that \(\alpha \sim^i \beta\) we have \(\alpha_i \leq \beta_i\). This is easy: for the given \(i\) we simply use the proof that \(\alpha \leq^{i+1}_{\seq D} \beta\).
\end{itemize}
\end{proof}

\noindent
Therefore, there are indeed infinite types that we can perform global optimisation on, so long as we use closeness spaces (or metric spaces, as in the case of the real numbers) to utilise approximate linear preorders.
We illuminate this in \cref{sec:gen-global-opt}.

Before doing that --- for the later purposes of regression in \cref{sec:regression} wherein we will optimise functions with co-domain \(\Ni\) --- we show that the extended naturals have an approximate lexicographic order that coincides with the established order on them (as defined in \cref{def:Ni-order}).

\begin{definition}
\thesislit{4}{ApproxOrder-Examples}{\urlSigma-order}
\thesislit{4}{ApproxOrder-Examples}{\urlSigma-approx-order}
Given a preorder \(\ty{\leq}{\relation{X}}\) and an approximate linear preorder \(\ty{\leq^-}{\Nrelation{X}}\) on a closeness space \(X\), and a type family \(\ty{P}{X \to \V}\), we define the \emph{inclusion order} on the type \(\Sigma P\) and the \emph{inclusion approximate order} on the corresponding closeness space (defined in \cref{cor:sigma-cspace}) by the following:
\begin{alignat*}{3}
(x , px) &\leq (y , py) &&:= x \leq y, \\
(x , px) &\leq^\varepsilon (y , py) &&:= x \leq^\varepsilon y, \\
\end{alignat*}
\end{definition}

\begin{lemma}
\label{lem:sigma-order}
\thesislit{4}{ApproxOrder-Examples}{\urlSigma-approx-order-is-approx-order}
\thesislit{4}{ApproxOrder-Examples}{\urlSigmaa{}Order-Relates.\urlSigma-approx-order-relates}
Given a preorder \(\ty{\leq}{\relation{X}}\) and an approximate linear preorder \(\ty{\leq^-}{\Nrelation{X}}\) on a closeness space \(X\), and a truth-valued type family \(\ty{P}{X \to \Omega}\), the inclusion approximate order on the closeness space \(\Sigma P\) is indeed an approximate linear preorder which relates to the inclusion preorder on \(\Sigma P\).
\end{lemma}
\begin{proof}[Proof (Sketch).] \axioms{t}
Because the definition of the inclusion orders simply compares the first arguments by the original orders --- and discards the second arguments --- each property we are required to prove for the inclusion orders is immediate from the fact the original orders satisfy those same properties.
\end{proof}

\begin{corollary}
\label{cor:Ni-approx-order}
\thesislit{4}{ApproxOrder-Examples}{\urlN\urlinfty-approx-lexicorder}
There is an approximate lexicographic order on \(\Ni\) which relates to the standard lexicographic preorder (as defined in \cref{def:Ni-order}).
\end{corollary}
\begin{proof} \axioms{t}
By \cref{lem:lexicographic-approx,lem:lexicographic-approx-relates,lem:sigma-order}.
\end{proof}

\subsection{Generalised global optimisation}
\label{sec:gen-global-opt}

The definition of a global minimum argument of a function whose domain is linearly ordered is intuitive, and we can always compute a global minimum argument of such a function if the codomain is pointed and finite linearly ordered.

\begin{definition}
\label{def:gen-opt-global-min}
\thesislit{4}{GlobalOptimisation}{is-global-minimal}
For any type \(X\) and preorder \(\ty{\leq}{Y \to Y \to \Omega}\) on a type \(Y\), an element \(\ty{x_0}{X}\) is a \emph{global minimum argument} of a function \(\ty{f}{X \to Y}\) if \(f(x_0) \leq f(x)\) for all \(\ty{x}{X}\).
\end{definition}

\begin{lemma}
\label{lem:fin-has-global-min}
\thesislit{4}{GlobalOptimisation}{finite-global-minimal}
For any pointed finite linearly ordered type \(X\) and linear preorder \(\ty{\leq}{Y \to Y \to \Omega}\) on a type \(Y\), every function \(\ty{f}{X \to Y}\) has a global minimum argument.
\end{lemma}
\begin{proof}
Recall from \cref{def:fin} of finite linearly ordered types that \(X \simeq \F(n)\) for some \(\ty{n}{\N}\), and therefore that (by \cref{def:simeq}) there is some equivalence \(\ty{g}{X \to \F(n)}\).

\vspace{1em}
We proceed by induction on \(n\). 
When \(n := 0\) then \(X \simeq \0\), which is not pointed and therefore the result follows vacuously.
When \(n : = 1\) then \(X \simeq \0 + \1\), and thus \(X\) has only one element \(\ty{g(\inr \ \star)}{X}\) which must, by reflexivity of the linear preorder, be the global minimum argument of \(f\).

\vspace{1em}
When \(n := n' + 1\) then \(X \simeq \F(n') + \1\). We use the inductive hypothesis to find the global minimum argument \(\ty{x}{\F(n')}\) of the function \(\ty{\left( f \circ g \circ \inl \right)}{\F(n') \to Y}\). 
This means that only \(\ty{g(\inl \ x)}{X}\) or \(\ty{g(\inr \ \star)}{X}\) can be the global minimum argument to \(f\). We therefore use the linearity of the linear preorder to find out which side of \(\left( f(g(\inl \ x)) \leq f(g(\inr \ \star) \right) + \left( f(g(\inr \ \star)) \leq f(g(\inl \ x)) \right)\) holds ---
if the left-hand side holds then \(f(g(\inl \ x)) \leq f(x)\) for all \(\ty{x}{X}\), while if the right-hand side holds then \(f(g(\inr \ \star)) \leq f(x)\) for all \(\ty{x}{X}\).
\end{proof}

\noindent
Replacing the order with an approximate linear preorder, we now state the general version of the global optimisation problem using closeness spaces.
Similarly to the above, it is the case that any function on pointed finite linearly ordered codomain and whose domain has an approximate linear preorder has a computable global minimum up to any degree of precision.

\begin{definition}
\thesislit{4}{GlobalOptimisation}{is_global-minimal}
For any type \(X\) and approximate linear preorder \(\ty{\leq^-}{\N \to Y \to Y \to \Omega}\) on a closeness space \(Y\), an element \(\ty{x_0}{X}\) is an \emph{\(\varepsilon\)-global minimum argument} of a function \(\ty{f}{X \to Y}\), given any precision \(\ty{\varepsilon}{\N}\), if \(f(x_0) \leq^\varepsilon f(x)\) for all \(\ty{x}{X}\).
\end{definition}

\begin{lemma}
\label{lem:fin-has-eps-global-min}
\thesislit{4}{GlobalOptimisation}{F-\urlepsilon-global-minimal}
For any pointed finite linearly ordered type \(X\) and approximate linear preorder \(\ty{\leq^-}{\N \to Y \to Y \to \Omega}\) on a closeness space \(Y\), every function \(\ty{f}{X \to Y}\) has an \(\varepsilon\)-global minimum argument given any precision \(\ty{\varepsilon}{\N}\).
\end{lemma}
\begin{proof}
By definition of approximate linear preorders (\cref{def:approx-order}), the relation \(\ty{\leq^\epsilon}{\relation X}\) is a linear preorder.
Therefore, the result immediately follows by \cref{lem:fin-has-global-min}.
\end{proof}

As we saw in \cref{def:real-min}, real global optimisation is performed on \emph{compact} intervals of the real numbers.
Recall that a metric space is compact if it is totally bounded and complete, and that there is a relationship between compact spaces and searchable sets.
As we are only minimising uniformly continuous functions, we do not require completeness\footnote{Informally, we do not require completeness due to the fact that uniformly continuous functions on metric spaces extend uniquely to the completion of that space and, further, the completion of a totally bounded metric space is compact~\cite{Troelstra}.}.
Furthermore, uniformly continuous searchability (see \cref{def:c-searchable}) is inappropriate for global optimisation. This is because, in order to prove the search condition when searching for a global minimum argument, we would have to prove that such an argument exists anyway.
Therefore, our generalised global optimisation theorem takes place on totally bounded closeness spaces.

By using the totally bounded (see \cref{def:totallybounded}) and uniformly continuous (see \cref{def:clos-ucont}) properties, we reduce the problem of searching the potentially infinite space \(X\) for an \(\varepsilon\)-global minimum into that of searching a finite \(\varepsilon\)-net (see \cref{def:cover}) for an element which represents a \(\varepsilon\)-global minimum.
This approach is intuitive, and reflects approaches to arbitrary-precision global optimisation in the literature~\cite{RecentAdvances}.
Furthermore, it is exactly the same process as what was performed in \cref{sec:search-infinite}, wherein we used uniform continuity to search a potentially-infinite space.

\begin{theorem}
\label{th:min}
\thesislit{4}{GlobalOptimisation}{global-opt}
Given a pointed totally bounded closeness space \(X\) and approximate linear preorder \(\ty{\leq^-}{\N \to Y \to Y \to \Omega}\) on a closeness space \(Y\), any uniformly continuous function \(\ty{f}{X \to Y}\) has an \(\varepsilon\)-global minimum given any precision \(\ty{\varepsilon}{\N}\).
\end{theorem}
\begin{proof}
Given the requested precision \(\ty{\varepsilon}{\N}\), by total boundedness we obtain a \(\delta\)-net \(X'\) of \(X\), where \(\ty{\delta}{\N}\) is the modulus of uniform continuity of \(f\) for \(\varepsilon\).
Recall from \cref{def:cover} that \(X'\) is such that there are \(\ty{g}{X' \to X}\) and \(\ty{h}{X \to X}\) such that for all \(\ty{x}{X}\) we have \(C_\delta(x,g(h(x)))\).

\vspace{0.25cm}
By \cref{lem:fin-has-eps-global-min}, we can compute an \(\varepsilon\)-global minimum of the function \(\ty{f \circ g}{X' \to Y}\), i.e.\ we have \(\ty{x'_0}{X'}\) such that \(f(g(x'_0)) \leq^\varepsilon f(g(x'))\) for all \(\ty{x'}{X'}\).

\vspace{0.25cm}
Then, given any \(\ty{x}{X}\), it is the case that \(C_\delta(x,g(h(x))\) and thus, by \cref{def:approx-order}.(i), that \(f(g(h(x))) \leq^\varepsilon f(x)\).

\vspace{0.5cm}
Therefore, given any \(\ty{x}{X}\), we have \(f(g(x'_0)) \leq^\varepsilon f(g(h(x))) \leq^\varepsilon f(x)\) and thus, by transitivity of the approximate linear preorder, \(\ty{g(x'_0)}{X}\) is an \(\varepsilon\)-global minimum argument of~\(f\).
\end{proof}

Finally, we apply our theorem to our motivating example, and show that we can optimise functions on sequences of finite linearly ordered types via their lexicographic orders.

\begin{corollary}
Given finite linearly ordered types \(F\) and \(G\), with \(F\) pointed, any uniformly continuous function \(\ty{f}{\seq F \to \seq G}\) has, for any requested precision \(\ty{\varepsilon}{\N}\), an \(\varepsilon\)-global minimum via discrete-sequence closeness spaces (see \cref{def:disseq-closeness}) and the approximate lexicographic ordering on \(\seq G\).
\end{corollary}
\begin{proof}
Recall that, by \cref{cor:disseq-cspace}, the discrete-sequence closeness space of a finite linearly ordered type \(F\) is totally bounded.
The result then follows from \cref{th:min} by \cref{lem:lexicographic-approx}.
\end{proof}

\section{Parametric Regression}
\label{sec:regression}

The work of this section was previously published as part of a joint paper with Dan R. Ghica at the \emph{Logic in Computer Science (LICS) 2021} conference~\cite{Todd21}.

Parametric regression analysis is a set of algorithms for estimating the \emph{relationship} between a dependent variable \(y\) (outcome) and several independent variables \(\{x_0,...,x_{n-1}\}\) (predictors), where the outcome is a function of the observations that has potentially, and indeterminably, been \emph{distorted}.
A parameterised function is proposed as a \emph{model} for this function.
The value of the parameters is then computed such that, given finitely-many predictor-outcome observations, a \emph{loss function} between the observed outcomes and those estimated by the model on input of the predictors is minimised~\cite{YanBook}.

Therefore, parametric regression is the problem of finding --- given finitely-many predictor-outcome observations --- (approximations of) best choice parameters for a given parameterised model function using a given loss function.
As with global optimisation in \cref{sec:global-opt}, we first state this problem using real numbers.

\begin{definition}
\label{def:real-pareg}
\lit
Given \(\ty{n,i,d}{\N}\), some loss function \(\ty{L}{\R \to \R \to \R_{\geq 0}}\), some parameterised model function \(\ty{M}{\R^i \to (\R^d \to \R)}\) and finitely-many predictor-outcome observations \(\ty{\{(xs_0,y_0),...,(xs_{n-1},y_{n-1})\}}{(\R^d \x \R)^n}\), the parameters \(\ty{ps^*}{\R^i}\) are \emph{best choice} if they minimise the loss between the outcomes estimated by the regressed function \(\ty{M_{ps}}{\R^d \to \R}\) and the observed outcomes; i.e.\ if they are a global minimum of the function \(\ty{\left(\lambda (\ty{ps}{\R^i}).\sum_{j = 0}^{n} L(M_{ps}(xs_i),y_i)\right)}{\R^i \to \R_{\geq 0}}\).
\end{definition}

\noindent
The choice of a particular loss function --- i.e.\ a pseudometric on the function space \((\R^d \to \R)\) --- is a field in its own right, but a common example is the least-squares method \(L(x,y) := d_R(x,y)^2\).
As an example of a particular model function, consider \emph{linear regression} where the model \(\ty{M}{\R^2 \to (\R \to \R)}\) is defined \(M_{(\alpha,\beta)}(x) := \alpha x + \beta\).

Based on the above rationalisation of the parametric regression problem, we view it as an instance of the global optimisation problem explored in the previous section.
As with global optimisation, therefore, we cannot in general compute a best choice set of parameters --- but we can compute parameters that are, for some degree of precision, \emph{approximately best choice}.

\begin{definition}
\label{def:approx-real-pareg}
\lit
Given \(\ty{n,i,d}{\N}\), some loss function \(\ty{L}{\R \to \R \to \R_{\geq 0}}\), some parameterised model function \(\ty{M}{\R^i \to (\R^d \to \R)}\), finitely-many predictor-outcome observations \(\ty{\{(xs_0,y_0),...,(xs_{n-1},y_{n-1})\}}{(\R^d \x \R)^n}\) and any precision \(\ty{\varepsilon}{\R_{\geq 0}}\), the parameters \(\ty{ps^*}{\R^i}\) are \emph{\(\varepsilon\)-best choice} if they minimise the loss between the outcomes estimated by the regressed function \(\ty{M_{ps}}{\R^d \to \R}\) and the observed outcomes up-to-\(\varepsilon\); i.e.\ if they are an \(\varepsilon\)-global minimum of the function \(\ty{\left(\lambda (\ty{ps}{\R^i}).\sum_{j = 0}^{n} L(M_{ps}(xs_i),y_i)\right)}{\R^i \to \R_{\geq 0}}\).
\end{definition}

\subsection{Generalised parametric regression}

By generalising \cref{def:approx-real-pareg}, we eliminate the need to specify the arity of our parameter/predictor spaces.
We replace the notion of an arbitrary loss function such as least-squares by the least-closeness pseudocloseness function (defined in \cref{def:least-closeness}); this means that rather than minimising loss we seek to maximise the least-closeness.
Furthermore, we replace the idea of predictor-outcome observations \(\ty{\{(xs_0,y_0),...,(xs_{n-1},y_{n-1})\}}{(X \x Y)^n}\) by a combination of predictor outcomes \(\ty{\{x_0,...,x_{n-1}\}}{X^n}\) and a function \(\ty{\OO}{X \to Y}\) which can be thought of as the \emph{oracle} of the outcome observations --- the (potentially distorted) function from which they arise on input of the predictors; i.e.\ \(\OO(xs_i) = y_i\).

\begin{definition}
\thesislit{4}{ParametricRegression}{is_global-maximal}
Given a function \(\ty{f}{X \to Y}\) where \(Y\) is a pre-ordered type, \(\ty{x_0}{X}\) is a \emph{global maximum} of \(f\) if for all \(\ty{x}{X}\) we have \(f(x_0) \geq f(x)\).
\end{definition}

\begin{definition}
\label{def:gen-reg-practical}
Given a type of predictors \(X\), closeness space of outcomes \(Y\), type of parameters \(P\), some parameterised model function \(\ty{M}{P \to (X \to Y)}\), finitely-many predictor observations \(\ty{\{x_0,...,x_{n-1}\}}{X^n}\) (for some \(\ty{n}{\N}\)), some oracle function \(\ty{\OO}{X \to Y}\) and any precision \(\ty{\varepsilon}{\N}\), the parameter \(\ty{p^*}{P}\) is \emph{\(\varepsilon\)-best choice} if it maximises the least-closeness pseudocloseness between the predictors' outcomes as estimated by the regressed function \(\ty{M_{p^*}}{X \to Y}\) and as observed from the oracle up-to-\(\varepsilon\); i.e.\ if \(p^*\) is an \(\varepsilon\)-global maximum of the function \(\ty{\left( \lambda (\ty{p}{P}).\mathsf{min}(c_Y(M_{p}(x_0),\OO(x_0)), ... , c_Y(M_{p}(x_{n-1}),\OO(x_{n-1})) \right)}{P \to \Ni}\).
\end{definition}

In this way, we have now re-imagined parametric regression as the process of approximating a black-box \emph{oracle} with a \emph{parameterised model} using a \emph{pseudocloseness function}.
Therefore, our generalised parametric regression not only (i) generalises from reals to closeness spaces, but also (ii) \emph{methodologically} generalises the problem of regression.

We can go one step further and remove explicit observations of the oracle altogether, allowing us to generalise the type of oracles themselves from function spaces to pseudocloseness spaces.
By doing this, instead of using the least-closeness pseudocloseness function, we use the pseudocloseness function that the pseudocloseness space of oracles is equipped with.

\begin{definition}
\label{def:gen-reg}
Given a type of parameters \(P\), a pseudocloseness space of oracles \(O\), some oracle \(\ty{\OO}{O}\) and some parameterised model function \(\ty{M}{P \to O}\), the parameter \(\ty{p^*}{P}\) is \emph{\(\varepsilon\)-best choice} if it maximises the closeness between the regressed function \(M_{p^*}\) and the oracle \(\OO\) up-to-\(\varepsilon\); i.e.\ if \(p^*\) is an \(\varepsilon\)-global maximum of the function \(\ty{\left( \lambda (\ty{p}{P}).c'_{O}(M_p,\OO) \right)}{P \to \Ni}\).
\end{definition}

We use this second generalisation of regression when stating and proving our convergence theorems for regression in the next subsection.
However, for the practical purposes of \cref{chap:exact-real-search}, we return to the less general \cref{def:gen-reg-practical}, which defines regression on function spaces by using the least-closeness pseudocloseness function.
We note here that, as \cref{def:gen-reg} is a generalisation of \cref{def:gen-reg-practical}, the convergence theorems proved on the former still hold for the latter.

\subsection{Convergence theorems for parametric regression}

Convergence properties of interpolation, another way of constructing models out of data, have been studied extensively by Weierstrass-style theorems~\cite{PINKUS20001}.
In this section, we make a methodological contribution by providing several convergence properties of our generalised variant of parametric regression (\cref{def:gen-reg}) using optimisation and search.
The contribution here is more methodological than technical, as the proofs follow naturally from the structures and theorems of \cref{chap:searchable,chap:generalised}.
The statements on the other hand are not obvious and require a different conceptual, not just mathematical, perspective on parametric regression analysis.

\subsubsection{Convergent parametric regression via global optimisation}
\label{sec:conv-opt}

Guarantees that can be made on computing an \(\varepsilon\)-global minimum of the loss function can be automatically considered as guarantees on the precision of the regressed model.

Recall that the convergence of interpolation guarantees that any oracle (subject to hygiene conditions) can be reconstituted to any desired precision if the number of samples is large enough~\cite{PINKUS20001}.
A similar theorem cannot hold for regression, for the simple reason that in regression we must commit to a model which may or may not be similar to the oracle function.
For example, if our oracle has quadratic behaviour, no amount of data will yield a precise linear approximation of it.
This commitment to a particular model must be taken into account in formulating convergence properties for regression.
If the model is completely wrong then, of course, convergence cannot be achieved.

However, even in the case where the model is incorrect, we can still employ generalised global optimisation (as long as our types permit that) in order to compute an \(\varepsilon\)-best choice parameter of the model up to any precision.

\begin{theorem}[Regression as minimisation]
\label{reg:min}
\thesislit{4}{ParametricRegression}{optimisation-convergence}
Given a totally bounded closeness space of parameters \(P\), pseudocloseness space of oracles \(O\) oracle \(\ty{\OO}{O}\), and any parameterised uniformly continuous model function \(\ty{M}{P \to O}\), we can compute an \(\varepsilon\)-best choice parameter for \(M\) given every precision \(\ty{\varepsilon}{\N}\).
\end{theorem}
\begin{proof}
Recall that \(\Ni\) is a closeness space (\cref{cor:Ni-cspace}) and has an approximate lexicographic order (\cref{cor:Ni-approx-order}). Therefore, we can perform \(\varepsilon\)-global maximisation\footnote{Maximisation can be achieved by the same theorem as minimisation by simply swapping the order in which elements are evaluated by the approximate order.} on it using \cref{th:min}.
We then compute an \(\varepsilon\)-global maximum of the function \(\ty{\left( c'_O(\OO,M(p)) \right)}{P \to \Ni}\), where \(\ty{c'_O}{O \to O \to \Ni}\) is the pseudocloseness function on \(O\).
The function that we maximise is uniformly continuous by the ultrametric property of the pseudocloseness space (\cref{def:pcspace}).
\end{proof}

\subsubsection{Convergent regression via uniformly continuous search}
\label{sec:conv-search}

Our first convergence theorem stated that we can compute a \(\varepsilon\)-best choice parameter given minimal conditions on the oracle and parameter types, as well as the model function.
We now consider computing parameters for regression that are not necessarily \(\varepsilon\)-best choice, but which maximise the pseudocloseness to an acceptable level parameterised by \(\varepsilon\).
These convergence theorems require additional conditions, but are more `practical' in the sense that we do not necessarily have to exhaust all possible candidate solutions in order to return an acceptable parameter for the given \(\varepsilon\) (as we see later in the examples of \cref{sec:K-examples}, such as \cref{ex:K-reg-1-search-imperfect}).

A general and absolute guarantee of precision can only be given if the model is the same as the oracle up to the value of the model parameters; i.e.\ the model is chosen correctly, and the oracle is not distorted.
If this is the case then, by using uniformly continuous search, we can indeed maximise the pseudocloseness between the regressed and true oracles and, further, we can make it arbitrarily large.

To express this property we introduce the concept of a \emph{synthetic oracle}, which is simply an oracle \(\OO\) ``synthesised'' from a model function \(M\) by applying it to an arbitrary and unknown parameter \(k\); i.e.\ \(\OO := M(k)\).

\begin{definition}
\label{def:synthesised}
Given type of parameters \(P\) and type of oracles \(O\), an oracle \(\ty{\OO}{O}\) is \emph{synthetically constructed from \(M\)} if there is some parameter choice \(\ty{p}{P}\) such that \(\OO = M_p\).
\end{definition}

\begin{definition}
\thesislit{4}{ParametricRegression}{p-regressor}
Given a uniformly continuously searchable type of parameters \(P\) and pseudocloseness space of oracles \(O\), we define the \emph{parametric regressor} function \(\ty{\mathsf{reg}}{\N \to (P \to O) \to O \to P}\) as,
\[ \mathsf{reg}(\varepsilon,M,\OO) := \mathcal{E}_P \left( \lambda (\ty{p}{P}).C_\varepsilon(M_p,\OO) \right) ,\]
where \(\ty{\mathcal{E}_P}{\mathsf{decidable{\hy}uc{\hy}predicate}(P) \to P}\) is the uniformly continuous searcher on \(P\) (as introduced in \cref{def:c-searcher}).
\end{definition}

\begin{theorem}[Convergence of distortion-free regression]
\label{th:perfect}
\thesislit{4}{ParametricRegression}{perfect-convergence}
Given a uniformly continuously searchable type of parameters \(P\), pseudocloseness space of oracles \(O\), parameterised uniformly continuous model function \(\ty{M}{P \to O}\) and oracle \(\ty{\OO}{O}\) synthetically constructed from \(M\), we can use the parametric regressor \(\ty{\mathsf{reg}}{\N \to (P \to O) \to O \to P}\) to build the regressed oracle \(\ty{\oo}{O}\) using only \(\varepsilon\), \(M\) and \(\OO\) (i.e.\ \(\oo := \mathsf{reg}(\varepsilon,M,\OO)\)) such that \(C_\varepsilon(\oo,\OO)\), for any precision \(\ty{\varepsilon}{\N}\).
\end{theorem}
\begin{proof}
The result follows from the later \cref{th:imp} by setting $\Psi := \mathsf{id}$; i.e.\ the distortion function is just the identity function, as the oracle we query is not distorted from the true oracle.
\end{proof}

\noindent
Note that a raw intuition of this theorem statement can be misleading: the regressor sees the synthetic oracle as a black box process, so it is not simply searching for the parameter of the synthetic oracle \(k\).
It is searching for \emph{any} parameter that makes the pseudocloseness between the true and regressed oracles \(\varepsilon\)-large.

The more traditional case is when the oracle \emph{is} prone to some distortion.
We model this case in our framework by using a \emph{distortion function} \(\ty{\Psi}{O \to O}\) which is applied to the \emph{true synthetic oracle} \(\ty{\OO}{O}\) to yield the \emph{distorted oracle} \(\ty{\OO_\Psi}{O}\).
It is this distorted oracle that the parametric regressor receives and can query for observations.

We cannot anymore expect to be able to maximise the pseudocloseness between the regressed and true oracles to any degree of precision.
However, we \emph{can} guarantee that the pseudocloseness between the regressed and true oracles is bounded by that between the regressed and distorted oracles.

\begin{theorem}[Convergence of distortion-prone regression]
\thesislit{4}{ParametricRegression}{s-imperfect-convergence}
Given a uniformly continuously searchable type of parameters \(P\), pseudocloseness space of oracles \(O\), parameterised uniformly continuous model function \(\ty{M}{P \to O}\), oracle \(\ty{\OO}{O}\) synthetically constructed from \(M\) and distortion function \(\ty{\Psi}{O \to O}\), we can use the parametric regressor \(\ty{\mathsf{reg}}{\N \to (P \to O) \to O \to P}\) to build the regressed oracle \(\ty{\oo}{O}\) using only \(\varepsilon\), \(M\) and \(\Psi\OO\) (i.e.\ \(\oo := \mathsf{reg}(\varepsilon,M,\Psi\OO)\)) such that if \(C_\varepsilon(\Psi\OO,\OO)\) then \(C_\varepsilon(\oo,\OO)\), for any precision \(\ty{\varepsilon}{\N}\).
\label{th:imp}
\end{theorem}
\begin{proof}
We just need to show that \(C_\varepsilon(\Psi\OO,\oo)\) as then the result will follow by transitivity of the closeness relation (\cref{lem:C-eq}) and the assumption that \(C_\varepsilon(\Psi\OO,\OO)\).

\vspace{1em}
Because \(\oo := \mathsf{reg}(\varepsilon,M,\Psi\OO) := \mathcal{E}_P(\lambda (\ty{p}{P}).C_\varepsilon(M_p,\OO))\), the result follows if there is some \(\ty{p'}{P}\) such that \(C_\varepsilon(M_{p'},\OO)\). Because the model is synthetically constructed, we can set \(p'\) as the parameter from which it was constructed, and the result follows immediately.
\end{proof}

\chapter{Real Numbers}
\label{chap:reals}

We introduced our type-theoretic framework using \textsc{Agda} in \cref{chap:mltt}, the core concepts of searchability and continuity in \cref{chap:searchable}, and our generalised variants of global optimisation and parametric regression — algorithms usually defined explicitly on real numbers — in \cref{chap:generalised}.
We now wish to take our work full circle: to program within our framework instantiations of these processes that operate on representations of (compact intervals of) the real numbers.

Constructive approaches to representing the real numbers are well-studied: the Cauchy reals and Dedekind reals are representations that have been previously defined and used in dependent type theory for analysis~\cite{HoTTBook,Auke}.
In practice, however, the Dedekind reals are inconvenient for computation, owing to the fact every real number is represented uniquely by exactly one Dedekind real~\cite{Geuvers}.
The Cauchy reals do not have such a uniqueness property, and as such have been found to be more convenient for computation: different variations of Cauchy (i.e., convergent) sequences of rational numbers have been used as representations of real numbers for the purpose of performing exact real computation.
In this chapter, we introduce two such convenient representations of (Cauchy) real numbers from the literature: ternary signed-digit encodings and ternary Boehm encodings~\cite{Gianantonio93,BoehmAPI}.

We formalise the structure and some of the algorithms of the signed-digit encodings of the compact interval \([-1,1]\) within our \textsc{Agda} framework and — going beyond this — prove the correctness of these definitions.
In order to achieve this latter goal of verification, we further formalise the \Escardo-Simpson interval object, an axiomatic specification of the real numbers which supports constructive mathematics by design~\cite{EscardoSimpson}.

In classical mathematics, the real numbers are axiomatised as the unique complete Archimedean field (see e.g.~\cite{ClassicalReals}).
This approach does not work in constructive mathematics -- for example, as we discussed in \cref{sec:global-opt}, the axiom that such reals have a linear order is the analytic LLPO~\cite{ShulmanLLPO}.
The alternative use of the interval object in this chapter is therefore appropriate; furthermore, it has the advantage of having fewer operations and axioms and so is easier to work with in practice.

For the Boehm encodings, we formalise and prove the correctness of Boehm’s definition; though the correctness of his arithmetic operations is relegated to further work (see \cref{fw:boehm-functions}).
For the sake of variation, we use a different approach to verifying the structure of the Boehm encodings, and instead use the Dedekind reals~\cite{DedekindReals}.
We seek to show that every ternary Boehm encoding of a real number gives a Dedekind real encoding of that real number.

\section{\Escardo-Simpson interval object}
\label{sec:interval-object}

In 2001, \MartinEscardo~and Alex Simpson proposed a categorical specification of closed real intervals which supports constructive mathematics by design~\cite{EscardoSimpson}.
The basic structure is that of bipointed \emph{midpoint algebras}, on which we give a universal property that is a variation of the completeness axiom, which serves as a computation principle for these real numbers.

In their paper, \Escardo~and Simpson worked in the generality of a category with finite products, but wrote that their specification ``applies to a variety of computational settings ... such as intuitionistic type theory"~\cite{EscardoSimpson}.
In this section, we contribute to this line of work by formalising the work in constructive type theory using \textsc{Agda}.
We will use this type in \cref{sec:signed-digits} in order to verify the signed-digit encodings.

This work was previously published as part of a joint paper with Dan R. Ghica at the \emph{Logic in Computer Science (LICS) 2021} conference~\cite{Todd21}, and given as a talk at the \emph{Workshop on Homotopy Type Theory/Univalent Foundations} (HoTT/UF)~\cite{Todd20,ToddIO}.

\subsection{Cancellative midpoint algebras}

We start off by defining the type of midpoint algebras.

\begin{definition}
\mbox{}
\label{def:magma}
\lit
A \emph{magma} is a set \(A\) equipped with a binary function to itself,
\[ \mathsf{Magma}_\U := \sigmaty{A}{\U}{\mathsf{is{\hy}set}(A) \x (A \to A \to A)} .\]
\noindent
For a given magma, we write \(\ty{(A,\oplus)}{\mathsf{Magma}}\) with the proof terms implicit.
\end{definition}

\begin{definition}
\label{def:mpa}
\thesislit{5}{IntervalObject}{Midpoint-algebra}
A magma \((A,\oplus)\) is a \emph{midpoint algebra} if it is,
\begin{enumerate}[(i)]
\item idempotent, \(\pity{a}{A}{a \oplus a = a}\),
\item commutative, \(\pity{a,b}{A}{a \oplus b = b \oplus a}\),
\item transpositional, \(\pity{a,b,c,d}{A}{(a \oplus b) \oplus (c \oplus d) = (a \oplus c) \oplus (b \oplus d)}\).
\end{enumerate}
For such a structure, we write \(\ty{(A,\oplus)}{\mathsf{Midpoint{\hy}algebra}}\) with proof terms (i)-(iii) implicit.
\end{definition}

Functions between midpoint algebras that preserve the structure are called midpoint homomorphisms.

\begin{definition}
\thesislit{5}{IntervalObject}{is-\urlmidpoint-homomorphism}
Given two midpoint algebras \((A,\oplus_A)\) and \((B,\oplus_B)\), a function \(\ty{h}{A \to B}\) is a \emph{midpoint homomorphism} if it preserves the midpoint operation:
\[ \mathsf{is{\hy}midpoint{\hy}hom}((A,\oplus_A),(B,\oplus_B),h) := \pity{a,b}{A}{h(a \oplus_A b) = h(a) \oplus_B h(b)} .\]
If the midpoint algebras are the same, we write \(\mathsf{is{\hy}midpoint{\hy}hom}((A,\oplus_A),h)\) as shorthand.
\end{definition}

\begin{lemma}
\label{lem:midhom-comp}
\thesislit{5}{IntervalObject}{id-is-\urlmidpoint-homomorphism}
\thesislit{5}{IntervalObject}{\urlmidpoint-hom-composition}
The identity function is a midpoint homomorphism and composition preserves homomorphisms.
\end{lemma}
\begin{proof}
For the identity function on a midpoint algebra \((A,\oplus_A)\), we need to show that for all \(\ty{a}{A}\) we have \(\mathsf{id}_A(a \oplus_A a) = \mathsf{id}_A(a) \oplus_A \mathsf{id}_A(a)\); this is clearly the case by reflexivity.
For composition of functions \(\ty{f}{A \to B}\) and \(\ty{g}{B \to C}\) between midpoint algebras \((A,\oplus_A)\), \((B,\oplus_B)\) and \((C,\oplus_C)\), we want to show that \(g(f(a_1 \oplus_A a_2)) = g(f(a_1)) \oplus_C g(f(a_2))\) for all \(\ty{a_1,a_2}{A}\). This is the case because \(g(f(a_1 \oplus_A a_2)) = g(f(a_1) \oplus_B f(a_2))\) (because \(f\) is a midpoint homomorphism) and then \(g(f(a_1) \oplus_B f(a_2)) = g(f(a_1)) \oplus_C g(f(a_2))\) (because \(g\) is a midpoint homomorphism).
\end{proof}

The midpoint algebras we utilise for the interval object satisfy an additional property called cancellation.

\begin{definition}
\label{def:cancellation}
\thesislit{5}{IntervalObject}{cancellative}
A magma \((A,\oplus)\) is \emph{cancellative} if for all \(\ty{a,b,c}{A}\) it is the case that \(a \oplus c = b \oplus c\) implies \(a = b\).
\end{definition}

\(\R^n\) is a cancellative midpoint algebra closed under the binary midpoint function \(\lambda (\ty{x,y}{\R^n}).\frac{1}{2}(x + y)\), as are various subsets of \(\R\), such as the rationals.

\subsection{Iteration property}

Starting from \(0\) and \(1\), the midpoint function can be used to generate every dyadic rational point in \([-1,1]\). In order to generate all rational and irrational numbers, the interval object requires its own version of the classical \emph{completeness} axiom; recall that this informally states that the real line has no ``gaps'' or ``missing points''~\cite{ClassicalReals}.
This property, which is called \emph{iteration}, states that there is an operator \(\ty{M}{\seq{A} \to A}\) that gives the `infinitely iterated' midpoint of a stream of points of \(A\). Formally, this operator is defined by two sub-properties.

\begin{definition}
\label{def:iterative1}
\thesislit{5}{IntervalObject}{iterative}
Given a magma \((A,\oplus)\) and function \(\ty{M}{\seq{A} \to A}\), the \emph{first iteration sub-property} is defined by:
\[ \mathsf{iterative}_1 (A,\oplus,M) :=  \pity{\alpha}{\seq{A}}{ M(\alpha) = \alpha_0 \oplus M(\mathsf{tail} \ \alpha)} .\]
\end{definition}

\begin{definition}
\label{def:iterative2}
\thesislit{5}{IntervalObject}{iterative}
Given a magma \((A,\oplus)\) and function \(\ty{M}{\seq{A} \to A}\), the \emph{second iteration sub-property} is defined by:
\[ \mathsf{iterative}_2 (A,\oplus,M) := \pitye{\alpha,\beta}{\seq{A}}{\left( \pity{i}{\N}{\beta_i = \alpha_i \oplus \beta_{\suc i}} \right) \to \beta_0 = M(\alpha)} .\]
\end{definition}

\begin{definition}
\thesislit{5}{IntervalObject}{iterative}
A magma \((A,\oplus)\) is \emph{iterative} if there is a function \(\ty{M}{\seq{A} \to A}\) such that both iteration sub-properties are satisfied:
\[ \mathsf{iterative}(A,\oplus) := \sigmaty{M}{\seq{A} \to A}{\mathsf{iterative}_1 (A,\oplus,M) \x \mathsf{iterative}_2 (A,\oplus,M)}\]
\end{definition}

The first sub-property characterises the iteration operator, while the second gives a computation rule for it with respect to a second stream which corresponds to the iteration on the first.
Both of these sub-properties are indeed \emph{properties} rather than additional \emph{structure} on the magma; i.e.\ given any magma \((A,\oplus)\) and function \(\ty{M}{\seq A \to A}\), the types \(\mathsf{iterative}_1 (A,\oplus,M)\) and \(\mathsf{iterative}_2 (A,\oplus,M)\) are propositions (\cref{def:proposition}), by the fact that \(A\) is a set.
This further (by \cref{lem:sigma-prop}) means that the composition of these two properties is a property for any given magma and function.
We now prove that the type \(\mathsf{iterative}(A,\oplus)\) itself is a property for any magma \((A,\oplus)\); the proof comes from the fact that for any given magma, any function satisfying both iteration sub-properties is unique.

\begin{lemma}
\label{lem:iterative-unique}
\thesislit{5}{IntervalObject}{iterative-uniqueness\urlcdot}
Given a magma \((A,\oplus)\) any two functions that satisfy both iteration sub-properties are pointwise-equal.
\end{lemma}
\begin{proof}
Given two functions \(\ty{M_1,M_2}{\seq\I \to \I}\) that satisfy the two iteration sub-properties, we want to show that for any \(\ty{\alpha}{\seq\I}\), it is the case that \(M_1(\alpha) = M_2(\alpha)\).

\vspace{1em}
To do this, we define a sequence \(\ty{\beta}{\seq\I}\) which gives the behaviour of \(M_1(\alpha)\) as a sequence: i.e.\ \(\beta_i := M_1(\lambda n.\alpha_{n+i})\).
By the first iteration sub-property on \(M_1\), for any \(\ty{i}{\N}\) we have that \(\beta_i = \alpha_i \oplus \beta_{i+1}\).
Therefore by the second iteration sub-property on \(M_2\), we have that \(\beta_0 = M_2(\alpha)\); i.e.\ \(M_1(\alpha) = M_2(\alpha)\).
\end{proof}

\begin{corollary}
\thesislit{5}{IntervalObject}{iterative-uniqueness}
Given a magma \((A,\oplus)\) the type \(\mathsf{iterative}(M,\oplus)\) is a proposition.
\end{corollary}
\begin{proof} \axioms{f}
By \cref{lem:iterative-unique} and function extensionality, any two functions that satisfy both iteration sub-properties are equal. Hence, because both sub-properties are subsingletons, the result follows by \cref{lem:sigma-prop}.
\end{proof}

From these sub-properties, \Escardo~and Simpson prove some expected properties about our iterative midpoint operator.

\begin{lemma}
\label{lem:M-idem}
\thesislit{5}{IntervalObject}{basic-interval-object-development.M-idem}
Given an idempotent magma \((A,\oplus)\) and \(\ty{M}{\seq{A} \to A}\) which satisfies the second iteration sub-property, \(M\) is itself idempotent,
\[ \pity{a}{A}{M (\lambda (\ty{-}{\N}).a) = a} .\]
\end{lemma}
\begin{proof}
By the second iteration sub-property (\cref{def:iterative2}), if we set \(\alpha,\beta := \lambda -.a\) then once we to prove the antecedent --- that for all \(\ty{i}{\N}\) we have \((\lambda -.a)_i = (\lambda -.a)_i \oplus (\lambda -.a)_i\) --- the result will follow. The antecedent just asks us to show that \(a = a \oplus a\) holds, which is by the idempotency of \(\oplus\) (\cref{def:mpa}).
\end{proof}

\begin{lemma}
\label{lem:M-hom-prop}
\thesislit{5}{IntervalObject}{basic-interval-object-development.M-hom}
Given a transpositional magma \((A,\oplus)\) that is iterative by \(\ty{M}{\seq{A} \to A}\), it is the case that \(M\) satisfies the following homomorphic property,
\[ \pity{\theta,\zeta}{\seq{A}}{M(\theta) \oplus M(\zeta) = M(\lambda(\ty{i}{\N}).\theta_i \oplus \zeta_i))} .\]
\end{lemma}
\begin{proof}[Proof (Sketch).]
By the second iteration sub-property (\cref{def:iterative2}), if we set \(\alpha := \lambda i.\theta_i \oplus \zeta_i\) and \(\beta := \lambda i.M(\lambda n.\theta_{n+i}) \oplus M(\lambda n.\zeta_{n+i})\) then once we prove the antecedent --- that for all \(\ty{i}{\N}\) we have \(M(\lambda n.\theta_{n+i}) \oplus M(\lambda n.\zeta_{n+i}) = (\theta_i \oplus \zeta_i) \oplus (M(\lambda n.\theta_{n+i+1}) \oplus M(\lambda n.\zeta_{n+i+1}))\) --- the result will follow.
The antecedent follows by the first iteration sub-property (\cref{def:iterative1}) and the transpositionality of \(\oplus\) (\cref{def:mpa}).
\end{proof}

\begin{lemma}
\thesislit{5}{IntervalObject}{basic-interval-object-development.M-symm}
Given a transpositional magma \((A,\oplus)\) that is iterative by \(\ty{M}{\seq{A} \to A}\), it is the case that \(M\) is symmetric,
\[ \pity{\theta}{\seq{(\seq{A})}}{M(\lambda i. M(\lambda j. \theta_{i,j})) = M(\lambda i. M(\lambda j. \theta_{j,i})) } .\]
\end{lemma}
\begin{proof}[Proof (Sketch).] \axioms{f}
By the second iteration sub-property (\cref{def:iterative2}), if we set \(\alpha := \lambda n.M(\lambda j.\theta_{j,n})\) and \(\beta := \lambda n.M(\lambda i.M(\lambda j.\theta_{i,j+n}))\) then once we prove the antecedent --- that for all \(\ty{n}{\N}\) we have \(M(\lambda i.M(\lambda j.\theta_{i,j+n})) = M (\lambda j.\theta_{j,n}) \oplus M(\lambda i.M(\lambda j.\theta_{i,j+n+1}))\) --- the result will follow.
The antecedent follows by the first iteration sub-property (\cref{def:iterative1}) and the above homomorphic property (\cref{lem:M-hom-prop}).
\end{proof}

\noindent
For the specific details of the latter two proofs, we invite the interested reader to view the formalisation.

The iteration property also gives us the notion of an iterated midpoint homomorphism (or `\(M\)-homomorphism').
Any midpoint homomorphism is automatically an \(M\)-homomorphism.

\begin{definition}
\thesislit{5}{IntervalObject}{basic-interval-object-development.\urlmidpoint-homs-are-M-homs}
Given midpoint algebras \((A,\oplus_A)\) and \((B,\oplus_B)\) that are iterative by functions \(\ty{M_A}{\seq{A} \to A}\) and \(\ty{M_B}{\seq{B} \to B}\), a function \(\ty{h}{A \to B}\) is an \emph{iterated midpoint homomorphism} if,
\[ \mathsf{is{\hy}M{\hy}hom}((A,\oplus_A),(B,\oplus_B),M_A,M_B,h) := \pity{\alpha}{\seq A}{h(M_A(\alpha)) = M_B(\lambda n.h(\alpha_n)} .\]
If the midpoint algebras are the same, we write \(\mathsf{is{\hy}M{\hy}hom}((A,\oplus_A),M,f)\) as shorthand.
\end{definition}

\begin{lemma}
\label{lem:mp-hom-M-hom}
\thesislit{5}{IntervalObject}{basic-interval-object-development.\urlmidpoint-homs-are-M-homs}
Given midpoint algebras \(\ty{(A,\oplus),(B,\oplus)}{\mathsf{Midpoint{\hy}algebra}}\) that are iterative by functions \(\ty{M_A}{\seq{A} \to A}\) and \(\ty{M_B}{\seq{B} \to B}\), if a function \(\ty{h}{A \to B}\) is a midpoint homomorphism then it is also an iterated midpoint homomorphism.
\end{lemma}
\begin{proof}
In order to show that \(h(M_A(\alpha)) = M_B(\lambda n.h(\alpha_n))\) for any \(\ty{\alpha}{\seq A}\) we first define a sequence \(\ty{\beta}{\seq B}\) which gives the behaviour of \(\ty{h(M_A(\alpha))}{B}\) as a sequence: i.e.\ \(\beta_i := h(M_A(\lambda n.\alpha_{n+1}))\).
We next show for all \(\ty{i}{\N}\) we have \(\beta_i = h(\alpha_i) \oplus \beta_{i + 1}\). Once, we have done this, our result will follow by the second iteration sub-property (\cref{def:iterative2}).
\begin{align*}
 && \beta_i& &\\
:= && h(M_A(\lambda n.\alpha_{n+i}))& &\\
= && h(\alpha_i \oplus M_A (\lambda n.\alpha_{n+i+1}))& &\text{by the first iteration sub-property ((\cref{def:iterative1})},\\
= && h(\alpha_i) \oplus h(M_A (\lambda n.\alpha_{n+i+1}))& &\text{by the fact \(h\) is an \(\oplus\)-homomorphism},\\
:= && h(\alpha_i) \oplus \beta_{i+1}& &\text{by (i)}.
\end{align*}
\end{proof}

\subsection{Finite approximations}

We formalise one final structure for iterative midpoint algebras: \emph{finite approximations}, the existence of which is equivalent to having the cancellation property (\cref{def:cancellation})~\cite{EscardoSimpson}.

\begin{definition}\mbox{}
\label{def:finiteapprox}
\thesislit{5}{IntervalObjectApproximation}{n-approx}
Given a midpoint algebra \((A,\oplus)\), two sequences \(\ty{\alpha,\beta}{\seq A}\) are \emph{\(n\)-approximately equal}, for a given \(\ty{n}{\N}\), if,
\[ \sigmaty{w,z}{A}{\alpha_0 \oplus (\alpha_1 \oplus ... (\alpha_{n-1} \oplus w)) = \beta_0 \oplus (\beta_1 \oplus ... (\beta_{n-1} \oplus z))} .\]
\end{definition}

\begin{definition}
\thesislit{5}{IntervalObjectApproximation}{approximation}
A midpoint algebra \((A,\oplus)\) that is iterative by \(\ty{M}{\seq{A} \to A}\) has \emph{finite approximations} if given two sequences \(\ty{\alpha,\beta}{\seq A}\) that are \(n\)-approximately equal for all \(\ty{n}{\N}\), then \(M(\alpha) = M(\beta)\).
\end{definition}

We formalise the more straightforward direction first: that having finite approximations implies the cancellation property.

\begin{lemma}
\thesislit{5}{IntervalObjectApproximation}{cancellation-holds}
A midpoint algebra \(\ty{(A,\oplus)}{\mathsf{Midpoint{\hy}algebra}}\) that is iterative by \(\ty{M}{\seq{A} \to A}\) and has finite approximations is cancellative.
\end{lemma}
\begin{proof}
Given \(\ty{a,b,c}{A}\) such that \((a \oplus c) = (b \oplus c)\), we wish to show that \(a = b\). This follows from the idempotency of \(M\) (\cref{lem:M-idem}) once we show that \(M(\lambda -.a) = M (\lambda -.b)\), which we will show using the finite approximation property \(\cref{def:finiteapprox}\). Therefore, we simply need to show that for any \(\ty{n}{\N}\), we have \((\lambda -.a)_0 \oplus ((\lambda -.a)_1 \oplus ... ((\lambda -.a)_{n-1} \oplus c)) = (\lambda -.b)_0 \oplus ((\lambda -.b)_1 \oplus ... ((\lambda -.b)_{n-1} \oplus c))\). We proceed by induction on \(n\). The first base case where \(n := 0\) is trivial, as we only need to show that \(c = c\); the second base case where \(n := 1\) requires us to show that \((a \oplus c) = (b \oplus c)\), which we have already assumed.

\vspace{1em}
The inductive case where \(n := n' + 1\) for some \(\ty{n'}{\N}\) requiers us to show that \(a \oplus (a \oplus wa) = b \oplus (b \oplus wb)\) where \(wa := (\lambda -.a)_0 \oplus ((\lambda -.a)_1 \oplus ... ((\lambda -.a)_{n'-1} \oplus c))\) and \(wb := (\lambda -.b)_0 \oplus ((\lambda -.b)_1 \oplus ... ((\lambda -.b)_{n'-1} \oplus c))\). By the inductive hypothesis, both \(a \oplus wa = b \oplus wb\) and \(wa = wb\); the result then follows by the below equational reasoning:
\begin{align*}
\ &a \oplus  (a \oplus wa)& &\\
&= (a \oplus a) \oplus (a \oplus wa),& &\text{by idempotency of \(\oplus\)},\\
&= (a \oplus a) \oplus (a \oplus wa),& &\text{by idempotency of \(\oplus\)},\\
&= (a \oplus a) \oplus (b \oplus wb),& &\text{by the inductive hypothesis},\\
&= (a \oplus b) \oplus (a \oplus wb),& &\text{by transpositionality of \(\oplus\)},\\
&= (a \oplus b) \oplus (a \oplus wa),& &\text{by the inductive hypothesis},\\
&= (a \oplus b) \oplus (b \oplus wb),& &\text{by the inductive hypothesis},\\
&= (b \oplus a) \oplus (b \oplus wb),& &\text{by commutativity of \(\oplus\)},\\
&= (b \oplus b) \oplus (a \oplus wb),& &\text{by transpositionality of \(\oplus\)},\\
&= (b \oplus b) \oplus (a \oplus wa),& &\text{by the inductive hypothesis},\\
&= (b \oplus b) \oplus (b \oplus wb),& &\text{by the inductive hypothesis},\\
&= b \oplus (b \oplus wb),& &\text{by idempotency of \(\oplus\)}.
\end{align*}
\end{proof}

The converse to the above is rather more complicated; we give the idea below.

\begin{theorem}
\label{thm:finite-approx}
\thesislit{5}{IntervalObjectApproximation}{approx-holds}
A cancellative midpoint algebra \(\ty{(A,\oplus)}{\mathsf{Midpoint{\hy}algebra}}\) that is iterative by \(\ty{M}{\seq{A} \to A}\) has finite approximations.
\end{theorem}
\begin{proof}[Proof (Sketch).] \axioms{f}
We first prove the \emph{finite-cancellation property}: that for any \(\ty{w,z}{A}\) and \(\ty{\alpha}{\seq A}\), if for any \(\ty{n}{\N}\) we have \(\alpha_0 \oplus (\alpha_1 \oplus ... (\alpha_{n-1} \oplus w)) = \alpha_0 \oplus (\alpha_1 \oplus ... (\alpha_{n-1} \oplus z))\), then \(w = z\). The proof of this is straightforward by the fact \(\oplus\) is cancellative and induction on \(\ty{n}{\N}\).

\vspace{1em}
We next prove \emph{one-sided approximation}: that for any \(\ty{a}{A}\) and \(\ty{\theta}{\seq A}\), if we have a sequence \(\ty{\zeta}{\seq A}\) such that \(a = \theta_0 \oplus (\theta_1 \oplus ... (\theta_{n-1} \oplus \zeta_n))\) for all \(\ty{n}{\N}\), then \(a = M(\theta)\). The proof of this is by the second iteration sub-property (\cref{def:iterative2}), as if we set \(\alpha := \theta\), \(\beta_0 := a\) and \(\beta_{i+1} = \zeta_{i+1}\) then once we prove the antecedent --- that for all \(\ty{i}{\N}\) we have \(\beta_i = \alpha_i \oplus \beta_{i+1}\) --- the result will follow. The antecedent follows by induction on \(i\). In the base case, where \(i := 0\), we give \(a = \theta_0 \oplus \zeta_1\) by the above equation concerning \(\alpha\) and \(\theta\) when \(n := 1\). In the inductive case, where \(i := i' + 1\) for some \(\ty{i'}{\N}\), we give \(\zeta_i = \theta_i \oplus \zeta_{i+1}\) by the equation when \(n := i + 1\) and the finite-cancellation property.

\vspace{1em}
Finally, we prove that we have \emph{finite approximations}: that for all \(\ty{\alpha,\beta}{\seq A}\), if we have sequences \(\ty{\theta,\zeta}{\seq A}\) such that \(\alpha_0 \oplus (\alpha_1 \oplus ... (\alpha_{n-1} \oplus \theta_n)) = \beta_0 \oplus (\beta_1 \oplus ... (\beta_{n-1} \oplus \zeta_n))\) for all \(\ty{n}{\N}\), then \(M(\alpha) = M(\beta)\). This follows by cancellation once we prove that \(M(\alpha) \oplus M(\mathsf{tail} \ \theta) = M(\beta) \oplus M(\mathsf{tail} \ \theta)\), and in turn this follows by the homomorphic property (\cref{lem:M-hom-prop}) once we prove that \(M(\lambda i.\alpha_i \oplus \theta_{i+1}) = M(\lambda i.\beta_i \oplus \theta_{i+1})\). We prove this using one-sided approximation: by setting \(\beta' := \lambda i.\beta_i \oplus \theta_{i+1}\) the result is reduced to showing that there is a sequence \(\ty{\gamma}{\seq A}\) such that \(M(\lambda i.\alpha_i \oplus \theta_{i+1}) = \beta'_0 \oplus (\beta'_1 \oplus ... (\beta'_{n-1} \oplus \gamma_n))\) for all \(\ty{n}{\N}\). 

\vspace{1em}
This sequence is defined \(\gamma_n := M((\beta_n \oplus \zeta_{n+1}) :: (\lambda i.\alpha_{i+n+1} \oplus \theta_{i+n+2}))\). Given any \(\ty{n}{\N}\), we prove that \(M(\lambda i.\alpha_i \oplus \theta_{i+1}) = (\beta_0 \oplus \theta_1) \oplus ((\beta_1 \oplus \theta_2) \oplus ... ((\beta_{n-1} \oplus \theta_n) \oplus M((\beta_n \oplus \zeta_{n+1}) :: (\lambda i.\alpha_{i+n+1} \oplus \theta_{i+n+2}))))\) by various straightforward rearrangements\footnote{For specific details on these rearrangements, we invite the interested reader to view the formalisation.} of \(\oplus\) and \(M\), and by the assumption that \(\alpha_0 \oplus (\alpha_1 \oplus ... (\alpha_{n-1} \oplus \theta_n)) = \beta_0 \oplus (\beta_1 \oplus ... (\beta_{n-1} \oplus \zeta_n))\).
\end{proof}

\begin{corollary}
\label{cor:fg-approx}
\thesislit{5}{IntervalObjectApproximation}{fg-approx-holds}
Given a midpoint algebra \((A,\oplus)\), some type \(X\) and two functions \(\ty{f,g}{X \to \seq A}\),
if for all \(\ty{x}{X}\) and \(\ty{n}{\N}\) we have that \(f(x)\) and \(g(x)\) are \(n\)-approximately equal,
then \(\ty{\left(M \circ f \right)}{X \to A}\) and \(\ty{\left(M \circ g \right)}{X \to A}\) are pointwise-equal.
\end{corollary}
\begin{proof} \axioms{f}
By \cref{thm:finite-approx}.
\end{proof}

\subsection{Bipointed convex bodies}

Adding iteration to a cancellative midpoint algebra gives us the structure we call an \emph{abstract convex body}.

\begin{definition}
\thesislit{5}{IntervalObject}{Convex-body}
A \emph{convex body} is a cancellative midpoint algebra \((A,\oplus)\) that is iterative by \(\ty{M}{\seq{A} \to A}\).
For such a structure, we write \(\ty{(A,\oplus,M)}{\mathsf{Convex{\hy}body}}\) with the \(M\) operator explicit and the proof terms implicit.
\end{definition}

\noindent
Every line segment of \(\mathbb{R}^n\) is an abstract convex body; following this fashion, a closed line segment corresponds to a \emph{bipointed convex body}~\cite{EscardoSimpson}.

\begin{definition}
A \emph{bipointed convex body} is a convex body \((A,\oplus,M)\) with two distinguished points (the `\emph{endpoints}') \(\ty{u,v}{A}\),
\[ \mathsf{Bipointed{\hy}convex{\hy}body} := \sigmaty{(A,\oplus,M)}{\mathsf{Convex{\hy}body}}{A \times A} .\]
\noindent
Note that the above defines the \emph{type of} bipointed convex bodies.
\end{definition}

Finally, a closed and bounded line segment -- called an \emph{interval object} -- is defined as a bipointed convex body that satisfies the following universal property.

\begin{definition}
\label{def:universal}
\thesislit{5}{IntervalObject}{is-interval-object}
A bipointed convex body \((A,\oplus_A,M_A,u,v)\) satisfies the \emph{universal property of interval objects} if, given any bipointed convex body \((B,\oplus_B,M_B,s,t)\) there is a unique\footnote{Note we use here the \textsc{TypeTopology} notation for \(\Sigma\)-types that have a unique witness.
Given a type \(\ty{X}{\U}\) and \(\ty{Y}{X \to V}\), \emph{unique existence} is the type that proves \(\Sigma Y\) is a singleton: i.e.\
\(\existsutye{A}{X \to \V} \ := \sigmatye{p}{\Sigma A}{\pity{q}{\Sigma A}{p = q}}\).} function \(\ty{h}{A \to B}\) that maps the endpoints of \(A\) to their respective endpoints of \(B\) and is a midpoint homomorphism,
\begin{gather*}
\mathsf{is{\hy}interval{\hy}object}((A,\oplus_A,M_A,u,v)) :=
\pitye{(B,\oplus_B,M_B,s,t)}{\mathsf{Bipointed{\hy}convex{\hy}body}}
\\ \existsuty{h}{A \to B}{(h(u) = s) \x (h(v) = t) \x \mathsf{is{\hy}midpoint{\hy}hom}((A,\oplus_A),(B,\oplus_B),h)} .
\end{gather*}
\end{definition}

\noindent
The idea of the universal property is illustrated in \cref{fig:affine}.

\begin{figure}
    \centering
    \includegraphics[width=\textwidth]{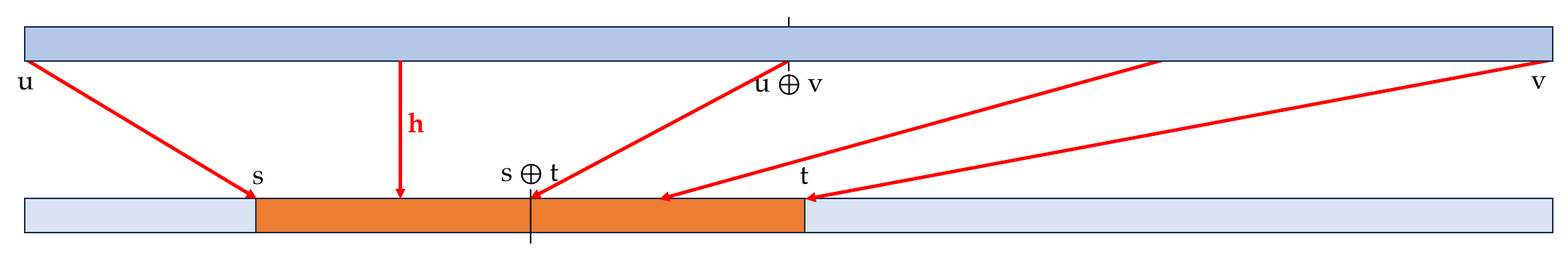}
    \caption{Illustration of the universal property \cref{def:universal}; the map \color{red}{\(\ty{h}{A \to B}\)} \color{black}{maps points on the bipointed convex body} \color{Periwinkle}{\((A,\oplus_A,M_A,u,v)\)} \color{black}{to the relative point on the bipointed convex body} \color{orange}{\((B,\oplus_B,M_B,s,t)\)}\color{black}{. The colours used here are as in the illustration.}}
    \label{fig:affine}
\end{figure}

\begin{definition}
\thesislit{5}{IntervalObject}{Interval-object}
An \emph{interval object} is a bipointed convex body that satisfies the universal property of interval objects,
\[ \mathsf{interval{\hy}object}_\U := \sigmatye{(A,\oplus,M,u,v)}{\mathsf{Bipointed{\hy}convex{\hy}body}}{\mathsf{is{\hy}interval{\hy}object}((A,\oplus,M,u,v))} .\]
For such a structure, we write \(\ty{(A,\oplus,M,u,v)}{\mathsf{interval{\hy}object}}\) with the proof term of the universal property implicit.
\end{definition}

\noindent
Note that the above defines the \emph{type of} interval objects.
This type is in fact a subsingleton in any univalent universe: we proved this using \Escardo's \textsc{TypeTopology} version of the \emph{structure identity principle}~\cite{EscardoSIP,AhrensSIP}.

\begin{theorem}
\label{thm:io-sip}
\typetop{TWA}{SIP-IntervalObject}{interval-object-prop}
For any universe \(\U\) such that \(\mathsf{is{\hy}univalent}(\U)\), it is the case that \(\mathsf{is{\hy}prop}(\mathsf{interval{\hy}object}_\U)\).
\end{theorem}
\begin{proof}[Proof (Sketch).] \axioms{fu}
By the universal property (\cref{def:universal}), any two interval objects are equivalent --- thus, by univalence, they are equal. For the full proof, see \cite{ToddSIP}.
\end{proof}

\noindent
Given this, we will only ever want to use the map derived from the universal property from an interval object to cast objects from and to the same underlying convex body.

We define, using the universal property map, the following \emph{affine maps} on an interval object, which cast objects from that interval object onto ``sub-intervals'' of that object.

\begin{definition}
\thesislit{5}{IntervalObject}{basic-interval-object-development.affine}
Given an interval object \((A,\oplus,M,u,v)\) we define the \emph{affine map} \(\ty{\mathsf{affine}}{A \to A \to A \to A}\) as,
\[ \mathsf{affine}(s,t,a) := h(a) ,\]
where \(\ty{h}{A \to A}\) is the map derived from the universal property, which maps elements from the interval object to the `sub-interval' bipointed convex body \((A,\oplus,M,s,t)\).
\end{definition}

\begin{lemma}
\label{lem:affine-ends}
\thesislit{5}{IntervalObject}{basic-interval-object-development.affine-equation-l}
\thesislit{5}{IntervalObject}{basic-interval-object-development.affine-equation-r}
Given an interval object \((A,\oplus,M,u,v)\), the affine map \(\ty{\mathsf{affine}(s,t)}{A \to A}\), for any points \(\ty{s,t}{A}\), correctly maps the endpoints; i.e.\ \(\mathsf{affine}(s,t,u) = s\) and \(\mathsf{affine}(s,t,v) = t\).
\end{lemma}
\begin{proof}
By the correct mapping of endpoints requirement in \cref{def:universal}.
\end{proof}

\begin{lemma}
\label{lem:affine-midhom}
\thesislit{5}{IntervalObject}{basic-interval-object-development.affine-is-\urlmidpoint-homomorphism}
Given an interval object \((A,\oplus,M,u,v)\), the affine map \(\ty{\mathsf{affine}(s,t)}{A \to A}\), for any points \(\ty{s,t}{A}\), is a midpoint homomorphism,
\[\mathsf{is{\hy}midpoint{\hy}hom}((A,\oplus),\mathsf{affine}(s,t)) .\]
\end{lemma}
\begin{proof}
By the midpoint homomorphism requirement in \cref{def:universal}.
\end{proof}

\begin{corollary}
\thesislit{5}{IntervalObject}{basic-interval-object-development.affine-M-hom}
Given an interval object \((A,\oplus,M,u,v)\), the affine map \(\ty{\mathsf{affine}(s,t)}{A \to A}\), for any points \(\ty{s,t}{A}\), is an \(M\)-homomorphism,
\[\mathsf{is{\hy}M{\hy}hom}((A,\oplus),M,\mathsf{affine}(s,t)) .\]
\end{corollary}
\begin{proof}
By \cref{lem:mp-hom-M-hom,lem:affine-midhom}.
\end{proof}

Many properties of our later-defined computational functions are proved by the following lemma, which gives the uniqueness of the affine map.

\begin{lemma}
\label{lem:affine-unique}
\thesislit{5}{IntervalObject}{basic-interval-object-development.affine-uniqueness}
Given an interval object \((A,\oplus,M,u,v)\), the affine map \(\ty{\mathsf{affine}(s,t)}{A \to A}\), for any points \(\ty{s,t}{A}\), is identical to any function \(\ty{f}{A \to A}\) that satisfies (i) \(f(u) = s\), (ii) \(f(v) = t\) and (iii) \(\mathsf{is{\hy}midpoint{\hy}hom}(f)\).
\end{lemma}
\begin{proof}
By the uniqueness requirement in \cref{def:universal}.
\end{proof}

The affine map functional is the computational seed of the interval object, forth from which springs the algorithms we will define on the interval, which represent algorithms on the real numbers.
Two basic examples are the identity function and any constant map.

\begin{lemma}
\label{lem:affine-id}
\thesislit{5}{IntervalObject}{basic-interval-object-development.affine-uv-involutive}
Given an interval object \((A,\oplus,M,u,v)\), the identity function \(\ty{\mathsf{id}}{A \to A}\) is given by the affine map \(\ty{\mathsf{affine}(u,v)}{A \to A}\).
\end{lemma}
\begin{proof}
The function \(\ty{\mathsf{id}}{A \to A}\) trivially satisfies (i) \(\mathsf{id}(u) = u\), (ii) \(\mathsf{id}(v) = v\) and (iii) \(\pity{a,b}{A}{\mathsf{id}(a \oplus b) = \mathsf{id}(a) \oplus \mathsf{id}(b)}\).
Therefore, by \cref{lem:affine-unique}, \(\mathsf{id} = \mathsf{affine}(u,v)\).
\end{proof}

\begin{lemma}
\label{lem:affine-const}
\thesislit{5}{IntervalObject}{basic-interval-object-development.affine-constant}
Given an interval object \((A,\oplus,M,u,v)\), the constant function \(\ty{\left( \lambda (\ty{-}{A}).x \right)}{A \to A}\), for any \(\ty{x}{A}\), is given by the affine map \(\ty{\mathsf{affine}(x,x)}{A \to A}\).
\end{lemma}
\begin{proof}
For any \(\ty{x}{A}\) the function \(\ty{\left( \lambda (\ty{-}{A}).x \right)}{A \to A}\) trivially maps both endpoints of \(A\) to the only endpoint \(x\) of the sub-interval, and is a midpoint homomorphism by idempotency of \(\oplus\) in \cref{def:mpa}.
Therefore, by \cref{lem:affine-unique}, \(\left( \lambda (\ty{-}{A}).x \right) = \mathsf{affine}(x,x)\).
\end{proof}

\subsection{Arithmetic on \texorpdfstring{\([-1,1]\)}{[-1,1]} by assuming an interval object}

For the rest of this section, we assume that an interval object \(\ty{\I}{\mathsf{interval{\hy}object}}_\U\) is given in any given type universe \(\U\) ---
we will \emph{not} try to construct the interval object in this thesis.

\begin{assumption}
\thesislit{5}{IntervalObject}{basic-interval-object-development}
We assume the existence of \(\I := (\I,\oplus,M,-1,+1)\) to specify the closed real interval \([-1,1]\).
\end{assumption}

\noindent
Note that we abuse notation and write \(\I\) both for the interval object itself and the set over which it operates.

Of course \(\ty{-1,+1}{\I}\) represent those same numbers. Using the midpoint, we can also represent other dyadic numbers in this interval.

\begin{definition}
\thesislit{5}{IntervalObject}{basic-interval-object-development.O}
The distinguished element \(\ty{0}{\I}\) is given by \(-1 \oplus +1\).
\end{definition}

Arithmetic functions such as negation and multiplication are not further axioms of this specification. They are instead defined, and their properties derived, from the existing --- rather small number --- of axioms in the interval object specification.
In particular, as we saw with the two trivial examples above, the affine map is used to define both negation and multiplication.

\subsubsection{Negation}

Negation is defined by using the affine map that `flips' the interval, casting elements of \(\I\) to their negated counterparts.

\begin{definition}
\thesislit{5}{IntervalObject}{basic-interval-object-development.\urlminus_}
\emph{Negation} is defined on the interval \(\I\) as the unique function \(\ty{-}{\I \to \I}\),
\[ -x := \mathsf{affine}(+1,-1,x) .\]
\end{definition}

\begin{lemma}
\label{lem:neg-props}
\thesislit{5}{IntervalObject}{basic-interval-object-development.\urlminus1-inverse}
\thesislit{5}{IntervalObject}{basic-interval-object-development.+1-inverse}
\thesislit{5}{IntervalObject}{basic-interval-object-development.\urlminus-is-\urlmidpoint-homomorphism}
Negation is the unique function on \(\I\) that negates every element of the object, i.e.,
\begin{enumerate}[(i)]
\item \(-(-1) = +1\),
\item \(-(+1) = -1\),
\item \(\pity{x,y}{\I}{- (x \oplus y) = - x \oplus - y}\).
\end{enumerate}
\end{lemma}
\begin{proof}
By \cref{lem:affine-ends,lem:affine-midhom,lem:affine-unique}.
\end{proof}

\noindent
We can further show that negation behaves as we expect; for example by showing it has a unit and is involutive.

\begin{corollary}
\label{lem:neg-zero}
\thesislit{5}{IntervalObject}{basic-interval-object-development.O-inverse}
Negating \(\ty{0}{\I}\) gives back \(0\).
\end{corollary}
\begin{proof}
By definition \(0 := -1 \oplus +1\), so we wish to show that \(- (-1 \oplus +1) = -1 \oplus +1\).
By \cref{lem:neg-props}, \(- (-1 \oplus +1) = -(-1) \oplus -(+1)\) because negation is a midpoint homomorphism; then, because negation negates the endpoints, we have \(-(-1) \oplus -(+1) = +1 \oplus -1\).
The final step \(+1 \oplus -1 = -1 \oplus +1\) follows by the commutativity of the midpoint operator in \cref{def:mpa}.
\end{proof}

\begin{lemma}
\label{lem:double-neg}
\thesislit{5}{IntervalObject}{basic-interval-object-development.\urlminus-involutive}
Negation on \(\I\) is an involution, i.e.\ \(- (-x) = x\) for all \(\ty{x}{\I}\).
\end{lemma}
\begin{proof}
We first use the fact that negation negates the endpoints (\cref{lem:neg-props}) to derive \(- - (-1) = - (+1) = -1\) and \(- - (+1) = - (-1) = +1\).
Also, because negation is a midpoint homomorphism (\cref{lem:neg-props}) and the composition of two homomorphisms is itself a homomorphism (\cref{lem:midhom-comp}), double negation is a midpoint homomorphism.
Therefore, because any affine map is unique (\cref{lem:affine-unique}) and double negation is a midpoint homomorphism that maps the endpoints of \(\I\) to themselves, it is the case that \(\left( \lambda(\ty{x}{\I}).-(-x) \right) = \mathsf{affine}(-1,+1)\).
The proof follows by the fact that, due to its identification with \(\mathsf{affine}(-1,+1)\), double negation gives the identity map (
\cref{lem:affine-id}) and therefore \(- (-x) = \mathsf{id}(x) = x\) for all \(\ty{x}{\I}\).
\end{proof}

\subsubsection{Multiplication}

Multiplication by some \(\ty{x}{\I}\) is defined by using the affine map that maps the interval to the sub-interval \([-x,x]\).

\begin{definition}
\thesislit{5}{IntervalObject}{basic-interval-object-development._*_}
\emph{Multiplication} is defined on the interval \(\I\) as the unique function \(\ty{*}{\I \to \I \to \I}\),
\[ x * y := \mathsf{affine}(-x,x,y) .\]
\end{definition}

\begin{lemma}
\label{lem:mul-props}
\thesislit{5}{IntervalObject}{basic-interval-object-development.*-gives-negation-l}
\thesislit{5}{IntervalObject}{basic-interval-object-development.*-gives-id-l}
\thesislit{5}{IntervalObject}{basic-interval-object-development.*-is-\urlmidpoint-homomorphism-l}
Multiplication of some \(\ty{x}{\I}\) is the unique function on \(\I\) that multiplies elements of the object by \(x\), i.e.,
\begin{enumerate}[(i)]
\item \(x * -1 = -x\),
\item \(x * +1 = x\),
\item \(\pity{x,y}{\I}{x * (y \oplus z) = (x * y) \oplus (x * z)}\).
\end{enumerate}
\end{lemma}
\begin{proof}
By \cref{lem:affine-ends,lem:affine-midhom,lem:affine-unique}.
\end{proof}

\noindent
Furthermore, multiplication also behaves as expected in a variety of ways.

\begin{lemma}
\label{lem:mul-io--1}
\thesislit{5}{IntervalObject}{basic-interval-object-development.*-gives-negation-r}
Multiplication on \(\I\) satisfies \(-1 * y = -y\).
\end{lemma}
\begin{proof}
The function \(\lambda y.-1 * y := \mathsf{affine}(-(-1),-1)\). By \cref{lem:neg-props}, this is pointwise-equal to \(\mathsf{affine}(+1,-1)\), which is the negation function.
\end{proof}

\begin{lemma}
\label{lem:mul-io-1}
\thesislit{5}{IntervalObject}{basic-interval-object-development.*-gives-id-r}
Multiplication on \(\I\) satisfies \(+1 * y = y\).
\end{lemma}
\begin{proof}
The function \(\lambda y.+1 * y := \mathsf{affine}(-(+1),+1)\). By \cref{lem:neg-props}, this is pointwise-equal to \(\mathsf{affine}(-1,+1)\), which by \cref{lem:affine-id} is the identity function.
\end{proof}

\begin{lemma}
\label{lem:mul-io-0}
\thesislit{5}{IntervalObject}{basic-interval-object-development.*-gives-zero-l}
\thesislit{5}{IntervalObject}{basic-interval-object-development.*-gives-zero-r}
Multiplication on \(\I\) satisfies \(x * 0 = 0\) and \(0 * y = 0\).
\end{lemma}
\begin{proof}
The former follows by \cref{lem:mul-props}.
The latter follows because the function \(\lambda y.0 * y := \mathsf{affine}(-0,0)\), and by \cref{lem:neg-zero} this is pointwise-equal to \(\mathsf{affine}(0,0)\); this, by \cref{lem:affine-const}, is the constant function that outputs \(0\).
\end{proof}

\begin{theorem}
\label{lem:mul-comm}
\thesislit{5}{IntervalObject}{basic-interval-object-development.*-commutative}
Multiplication on \(\I\) is commutative.
\end{theorem}
\begin{proof}
Given any \(\ty{x,y}{\I}\), we wish to show that \(x * y = y * x\); i.e.\ that \(\mathsf{affine}(-x,x,y) = y * x\).
This is proved by showing \(\mathsf{affine}(-x,x) = \lambda y.y * x)\) by \cref{lem:affine-unique}.
Therefore, we must show that (1) \(-1 * x = -x\), (2) \(+1 * x = x\) and (3) \(\mathsf{is{\hy}midpoint{\hy}hom}(\lambda y.y * x)\).

\vspace{1em}
The first two conditions are given by  \cref{lem:mul-io--1,lem:mul-io-1}.
The third is more complicated: given any \(\ty{y',z'}{\I}\) we need to show that \((y' \oplus z') * x = (y' * x) \oplus (z' * x)\).
This is given by another application of \cref{lem:affine-unique} wherein we show \(\mathsf{affine}(-(y' \oplus z'),y' \oplus z') = \lambda x.(y' * x) \oplus (z' * x)\).
Therefore, we must now show that (1) \((y' * -1) \oplus (z' * -1) = -(y' \oplus z')\), (2) \((y' * +1) \oplus (z' * +1) = (y' \oplus z')\) and (3) \(\mathsf{is{\hy}midpoint{\hy}hom}(\lambda x.(y' * x) \oplus (z' * x)\).

\vspace{1em}
For the first condition, by \cref{lem:mul-props}.(i) and (iii) we have that \((y' * -1) \oplus (z' * -1) = (-y' \oplus -z') = -(y' \oplus z')\).
For the second, by \cref{lem:mul-props}.(ii) we have that \((y' * +1) \oplus (z' * +1) = (y' \oplus z')\).
The third is again more complicated: given any \(\ty{a,b}{\I}\), the result is given by the following equational reasoning:
\begin{align*}
\ & (x * (a \oplus b)) \oplus (y * (a \oplus b)),& &\\
= & ((x * a) \oplus (x * b)) \oplus (y * (a \oplus b))& &\text{by \cref{lem:mul-props}.(iii)},\\
= & ((x * a) \oplus (x * b)) \oplus ((y * a) \oplus (y * b))& &\text{by \cref{lem:mul-props}.(iii)},\\
= & ((x * a) \oplus (y * a)) \oplus ((x * b) \oplus (y * b)) & &\text{by trans. of \(\oplus\) (\cref{def:mpa})}.
\end{align*}
\vspace{-1em}
\end{proof}

\begin{lemma}
\label{lem:mul-hom-r}
\thesislit{5}{IntervalObject}{basic-interval-object-development.*-is-\urlmidpoint-homomorphism-r}
Multiplication on \(\I\) satisfies \(((x \oplus y) * z) = (x * z) \oplus (y * z)\).
\end{lemma}
\begin{proof}
By \cref{lem:mul-props,lem:mul-comm}.
\end{proof}

\begin{theorem}
\thesislit{5}{IntervalObject}{basic-interval-object-development.*-assoc}
Multiplication on \(\I\) is associative.
\end{theorem}

\noindent
The proof that multiplication is associative has a similar proof technique as the proof of commutativity (\cref{lem:mul-comm}). For the full details, we invite the interested reader to view the formalisation.

\section{Verified ternary signed-digit encodings}
\label{sec:signed-digits}

\subsection{Background and definition in our type theory}
\label{sec:signed-digits-background}

Boehm et al.'s early paper on exact real arithmetic explores multiple representations of real numbers, one of which is an approach where a real number is represented by an infinitary sequence of digits~\cite{Boehm86}.

The idea is that a (base \(2\)) \emph{digit encoding} of a real number in the compact interval \([0,d]\) (for some \(\ty{d}{\N}\)) is an infinitary seqeuence of digits \(\ty{\alpha}{\{0,...,d\}^\N}\). The real number represented by such a digit encoding \(\ty{\llbracket \alpha \rrbracket}{[0,d]}\) is defined by,
\[ \llbracket \alpha \rrbracket := \sum_{n=0}^\infty \frac{\alpha_i}{2^{n+1}} .\]

\noindent
For example, with \(d := 2\), \(\llbracket \{0,2,0,2,0,...\} \rrbracket = 0.666...\) and \(\llbracket \{0,1,0,1,0,...\} \rrbracket = 0.333...\).
This approach is problematic, however, as even addition is not definable~\cite{Boehm86}. An illuminating example is the addition of the two above realisers of \(0.666...\) and \(0.333...\). Computing only the first digit of the output requires the function to look at an infinite number of digits of the input, because (for example)
\(\llbracket \{0,2,0,2,0,...\} \rrbracket + \llbracket \{0,1,0,1,0,...,0,0,0,...\} \rrbracket = 0.9...\), whereas \(\llbracket \{0,2,0,2,0,...\} \rrbracket + \llbracket \{0,1,0,1,0,...,0,2,0,...\} \rrbracket = 1.0...\).

However, the \emph{signed-digit encoding} representation does not have this problem.
A signed-digit encoding of a real number in a compact interval \([-d,d]\) (for some \(\ty{d}{\N}\)) is an infinitary sequence of digits \(\ty{\alpha}{\seq{\{-d,...,d\}}}\).
The real number represented by such a signed-digit encoding \(\ty{\llbracket \alpha \rrbracket}{\R}\) is defined by,
\[ \llbracket \alpha \rrbracket := \sum_{n = 0}^\infty \frac{\alpha_n}{2^{n+1}} ,\]
which can be infinitely unfolded using a midpoint operator \(x \oplus y := \frac{x + y}{2}\):
\[ \llbracket \alpha \rrbracket := \alpha_0 \oplus (\alpha_1 \oplus (\alpha_2 \oplus ...)) .\]

This representation of the reals uses extra `redundant' digits than is necessary for representing each number in the interval.
The redundant digits, however, are what enable addition (and multiplication) to be defined at the expense of having multiple representations for the same real number.
In the previous addition example, the second item of the sequence could have been set as \(2\) and then, if necessary, corrected by negative digits later~\cite{Plume98}~\cite{Escardo11fun}.

In this section, we recall the type of \emph{ternary signed-digit encodings} of real numbers in the compact interval \([-1,1]\), as extensively explored in the literature of exact real arithmetic~\cite{Gianantonio93,Plume98,BergerCoinductive}.
Every real number in \([-1,1]\) can be represented by a function of type \(\N \to \{-1,1\}\), but by using the redundant representation \(\N \to \{-1,0,1\}\) we can perform exact real arithmetic on representations of this compact interval.

Our formalisation is based on \Escardo's \textsc{Haskell} library for exact real computation on ternary signed-digits~\cite{Escardo11fun}.
We first convert his definitions of various arithmetic functions to \textsc{Agda}, which is non-trivial as we must convince \textsc{Agda}'s termination checker.
Following this, we verify the correctness of his algorithms using the interval object seen in \cref{sec:interval-object}, showing that they encode the correct operations on the real numbers.

\begin{definition}
    \thesislit{5}{SignedDigit}{\urlthree}
The type of \emph{ternary digits} \(\3\), equivalent to \(\F(3)\), is defined by its elements \(\ty{\overline{1},0,1}{\3}\).
\end{definition}

\begin{definition}
\thesislit{5}{SignedDigit}{\urlthree\urlsuperscriptN}
A \emph{ternary signed-digit encoding} of a real number in \([-1,1]\) is any function \(\ty{\alpha}{\seq \3}\).
\end{definition}

\subsection{Representation via the interval object}
\label{sec:representation-map}

A \emph{representation map} takes a representation of a real to the real that it represents.

\begin{definition}
\label{def:representation-map}
\thesislit{5}{SignedDigitIntervalObject}{_realises\urlsuperscriptone_}
\thesislit{5}{SignedDigitIntervalObject}{_realises\urlsuperscripttwo_}
\thesislit{5}{SignedDigitIntervalObject}{_realises\urlsuperscriptN_}
Given a type of real numbers \(\R\), a type for representing reals \(K\) and a representation map \(\ty{\llbracket - \rrbracket}{K \to \R}\), the \(n\)-ary function \(\ty{\overline f}{K^n \to K}\) (for \(\ty{n}{\N}\)) \emph{realises} a function \(\ty{f}{\R^n \to \R}\) if,
\[ \pitye{(x_0 ,..., x_{n-1})}{K^n} \llbracket \overline f \left(x_0 ,..., x_{n-1}\right) \rrbracket = f \left( \llbracket x_0 \rrbracket ,..., \llbracket x_{n-1} \rrbracket \right) .\]
\end{definition}

\noindent
The idea here is that in order to verify \(\ty{\overline f}{K^n \to K}\) realises \(\ty{f}{\R \to \R}\), using the representation map \(\ty{\llbracket - \rrbracket}{K \to \R}\), we show that the diagram in \cref{fig:commute} commutes.

\begin{figure}[ht]
\centering
\tikzset{ampersand replacement=\&}
    \hspace{-1.5cm}
    \begin{tikzcd}[sep=large]
	{K^n} \& {K} \\
	{\R^n} \& {\R}
	\arrow["{\overline f}", from=1-1, to=1-2]
	\arrow["{f}"', from=2-1, to=2-2]
	\arrow["{\llbracket - \rrbracket \ \x \ ... \ \x \ \llbracket - \rrbracket}"', from=1-1, to=2-1]
	\arrow["{\llbracket - \rrbracket}", from=1-2, to=2-2]
    \end{tikzcd}
\caption{Commutative diagram illustrating \cref{def:representation-map}.}
\label{fig:commute}
\end{figure}

The identity map always has a realiser, and the composition of two realisers is a realiser of the composition of the two realised functions; i.e.\ the diagrams in \cref{fig:id-comp-commute} commute.

\begin{figure}[ht]
\centering
\tikzset{ampersand replacement=\&}
\begin{floatrow}
    \begin{tikzcd}[sep=large]
	{K} \& {K} \\
	{\R} \& {\R}
	\arrow["{\mathsf{id}_K}", from=1-1, to=1-2]
	\arrow["{\mathsf{id}_\R}"', from=2-1, to=2-2]
	\arrow["{\llbracket - \rrbracket}"', from=1-1, to=2-1]
	\arrow["{\llbracket - \rrbracket}", from=1-2, to=2-2]
    \end{tikzcd}
    \ \ \
	\begin{tikzcd}[sep=large]
	K \& K \& K \\
	\R \& \R \& \R
	\arrow["{\overline f}", from=1-1, to=1-2]
	\arrow["f"', from=2-1, to=2-2]
	\arrow["{\llbracket - \rrbracket}"', from=1-1, to=2-1]
	\arrow["{\llbracket - \rrbracket}"{description}, from=1-2, to=2-2]
	\arrow["g"', from=2-2, to=2-3]
	\arrow["{\overline g}", from=1-2, to=1-3]
	\arrow["{\llbracket - \rrbracket}", from=1-3, to=2-3]
\end{tikzcd}
\end{floatrow}
\caption{Commutative diagrams illustrating \cref{lem:realise-id,lem:realise-comp}.}
\label{fig:id-comp-commute}
\end{figure}

\begin{lemma}
\label{lem:realise-id}
\label{lem:realise-comp}
\thesislit{5}{SignedDigitIntervalObject}{id-realiser}
\thesislit{5}{SignedDigitIntervalObject}{\urlcirc-realiser}
The identity map is realised by the identity map, and composition preserves realisers.
\end{lemma}
\begin{proof}
For the identity map, we need to show that for all \(\ty{k}{K}\) we have \(\llbracket \mathsf{id}_K(k) \rrbracket = \mathsf{id}_{\R}\left( \llbracket k \rrbracket \right)\); this is clearly the case by reflexivity.
For composition of functions \(\ty{\overline f,\overline g}{K \to K}\), which realise \(\ty{f,g}{\R \to \R}\) respectively, we want to show that \(\llbracket \overline g(\overline f(k)) \rrbracket = g(f(\llbracket k \rrbracket))\). This is the case because \(\llbracket \overline g(\overline f(k)) \rrbracket = g (\llbracket \overline f (x) \rrbracket )\) (because \(\overline f\) realises \(f\)) and then \(g (\llbracket \overline f (x) \rrbracket ) = g(f(\llbracket k \rrbracket))\) (because \(\overline g\) realises \(g\)).
\end{proof}

We define the \emph{representation map} for ternary signed-digit encodings, which maps such encodings of reals in \([-1,1]\) to the numbers they represent on the interval object \(\I\), using the infinitary midpoint operator \(\ty{M}{\seq{\I} \to \I}\).

\begin{definition}
\label{def:repr-map-digits}
\thesislit{5}{SignedDigitIntervalObject}{\urllangle_\urlrangle}
We define the \emph{representation map} from ternary digits \(\3\) to the interval object \(\I\),
\begin{alignat*}{3}
\langle - \rangle &: \3 \to \I,\\
\langle \overline 1 \rangle &:= -1,\\
\langle           0 \rangle &:= 0,\\
\langle           1 \rangle &:= +1,
\end{alignat*}
\end{definition}

\begin{definition}
\label{def:repr-map-signed-digits}
\thesislit{5}{SignedDigitIntervalObject}{\urlllangle_\urlrrangle}
We define the \emph{representation map} from ternary signed-digit encodings \(\seq \3\) to the interval object \(\I\),
\begin{alignat*}{3}
\llangle - \rrangle &: \K \to \I,\\
\llangle \alpha \rrangle & := M(\mathsf{map}( \langle - \rangle,\alpha)).
\end{alignat*}
\end{definition}

In order to verify that an operation on ternary digits \(\ty{\overline{f'}}{\3^n \to \3}\) or one on ternary signed-digit encodings \(\ty{\overline f}{(\K)^n \to \K}\) correctly realises an operation on the interval object \(\ty{f}{\I^n \to \I}\), we will show that the diagrams in \cref{fig:ternary-commute} commute.

\begin{figure}[ht]
\centering
\tikzset{ampersand replacement=\&}
\begin{floatrow}
    \begin{tikzcd}[sep=large]
	{\3^n} \& {\3} \\
	{\I^n} \& {\I}
	\arrow["{\overline {f'}}", from=1-1, to=1-2]
	\arrow["{f}"', from=2-1, to=2-2]
	\arrow["{\langle - \rangle^n}"', from=1-1, to=2-1]
	\arrow["{\langle - \rangle}", from=1-2, to=2-2]
    \end{tikzcd}
    \ \ \
    \begin{tikzcd}[row sep=large, column sep=large]
	(\seq\3)^n \& \seq\3 \\
	\I^n \& \I
	\arrow["{\overline f}", from=1-1, to=1-2]
	\arrow["{f}"', from=2-1, to=2-2]
	\arrow["{\llangle - \rrangle^n}"', from=1-1, to=2-1]
	\arrow["{\llangle - \rrangle}", from=1-2, to=2-2]
    \end{tikzcd}
\end{floatrow}
\caption{Commutative diagrams illustrating \cref{def:representation-map} (which is itself illustrated in \cref{fig:commute}) using the representation maps defined in \cref{def:repr-map-digits} (left) and \cref{def:repr-map-signed-digits} (right).}
\label{fig:ternary-commute}
\end{figure}

\subsection{Exact real arithmetic}
\label{sec:signed-digits-era}

\subsubsection{Negation}

Negation on signed-digit encodings is straightforward to define --- we simply flip every digit of the sequence.

\begin{definition}
\label{def:neg-sd}
\thesislit{5}{SignedDigit}{flip}
\thesislit{5}{SignedDigit}{neg}
We first define the negation function \(\mathsf{flip}\) on ternary digits in the expected way:
\begin{alignat*}{3}
\mathsf{flip} &: \3 && \to \3, \\
\mathsf{flip} &(\overline 1) && :=           1 ,\\
\mathsf{flip} &(          0) && :=           0 ,\\
\mathsf{flip} &(          1) && := \overline 1 .
\end{alignat*}
Then, the negation function \(\mathsf{neg}\) on ternary signed-digit encodings is defined as follows:
\begin{alignat*}{3}
\mathsf{neg} &: \K && \to \K, \\
\mathsf{neg} &(x) &&:= \mathsf{map}(\mathsf{flip},x) .
\end{alignat*}
\end{definition}

We then verify these operations; i.e.\ we show that the following diagrams commute.

\begin{figure}[ht]
\centering
\tikzset{ampersand replacement=\&}
\begin{floatrow}
    \begin{tikzcd}[sep=large]
	{\3} \& {\3} \\
	{\I} \& {\I}
	\arrow["{\mathsf{flip}}", from=1-1, to=1-2]
	\arrow["{-}"', from=2-1, to=2-2]
	\arrow["{\langle - \rangle}"', from=1-1, to=2-1]
	\arrow["{\llangle - \rrangle}", from=1-2, to=2-2]
    \end{tikzcd}
    \ \ \
    \begin{tikzcd}[baseline=-0.05cm, row sep=large, column sep=large]	\seq\3 \& \seq\3 \\
	\I \& \I
	\arrow["{\mathsf{neg}}", from=1-1, to=1-2]
	\arrow["{\oplus}"', from=2-1, to=2-2]
	\arrow["{\llangle - \rrangle}"', from=1-1, to=2-1]
	\arrow["{\llangle - \rrangle}", from=1-2, to=2-2]
    \end{tikzcd}
\end{floatrow}
\caption{Commutative diagrams illustrating \cref{lem:neg-pw-realise} (left) and \cref{lem:neg-realise} (right).}
\label{fig:neg-commute}
\end{figure}

\begin{lemma}
\label{lem:neg-pw-realise}
\thesislit{5}{SignedDigitIntervalObject}{flip-realiser}
The negation function on ternary digits realises the negation function on the interval object.
\end{lemma}
\begin{proof}
We prove \(\langle \mathsf{flip}(t) \rangle = - \langle t \rangle\) by case splitting on the given \(\ty{t}{\3}\).
In the \(t := 0\) case, we show \(\langle 0 \rangle = - \langle 0 \rangle\) by \cref{lem:neg-zero}.
In the other two cases, the proof follows by the mapping of the endpoints in \cref{lem:neg-props}.
\end{proof}

In order to verify negation on signed-digit encodings, we use the fact that the \(\mathsf{map}\) function preserves a realiser; this is illustrated by the diagram in \cref{fig:map-commute}.

\begin{figure}[ht]
\tikzset{ampersand replacement=\&}
    \centering
\begin{floatrow}
    \begin{tikzcd}[sep=large]
	\3 \& \3 \\
	\I \& \I
	\arrow["{\langle - \rangle}"', from=1-1, to=2-1]
	\arrow["f"', from=2-1, to=2-2]
	\arrow["{\langle - \rangle}", from=1-2, to=2-2]
	\arrow["{\overline f}", from=1-1, to=1-2]
    \end{tikzcd}
    \(\Longrightarrow\)
    \begin{tikzcd}[sep=large]
	\seq\3 \& \seq\3 \\
	\I \& \I
	\arrow["{\llangle - \rrangle}"', from=1-1, to=2-1]
	\arrow["f"', from=2-1, to=2-2]
	\arrow["{\llangle - \rrangle}", from=1-2, to=2-2]
	\arrow["{\mathsf{map} \ \overline f}", from=1-1, to=1-2]
    \end{tikzcd}
\end{floatrow}
    \caption{Commutative diagram illustrating \cref{lem:map-realise}.}
    \label{fig:map-commute}
\end{figure}

\begin{lemma}
\label{lem:map-realise}
\thesislit{5}{SignedDigitIntervalObject}{map-realiser}
If a function \(\ty{\overline{f}}{\3 \to \3}\) on ternary digits realises a midpoint homomorphism \(\ty{f}{\I \to \I}\) on the interval object, then \(\ty{ \mathsf{map}(\overline f) }{\K \to \K}\) realises \(f\) on signed-digit encodings.
\end{lemma}
\begin{proof} \axioms{f}
We want to show that, for all \(\ty{\alpha}{\K}\), \(\llangle \mathsf{map}(\overline f,\alpha) \rrangle = f(\llangle \alpha \rrangle)\).
By function extensionality and \cref{def:repr-map-signed-digits}, the left-hand side of the equation becomes \(M(\lambda n.\langle \overline f (\alpha_n) \rangle)\).
By the fact that \(\overline f\) realises \(f\), it further becomes \(M(\lambda n. f(\langle \alpha_n \rangle))\).
By \cref{lem:mp-hom-M-hom}, \(f\) is an \(M\)-homomorphism, and therefore the equation becomes \(f(M(\lambda n. \langle \alpha_n \rangle))\); which is definitionally equal to the conclusion by \cref{def:repr-map-signed-digits}.
\end{proof}

\begin{theorem}
\label{lem:neg-realise}
\thesislit{5}{SignedDigitIntervalObject}{neg-realiser}
The negation function on signed-digit encodings realises the negation function on the interval object.
\end{theorem}
\begin{proof} \axioms{f}
By \cref{lem:map-realise,lem:neg-pw-realise,lem:neg-props}.
\end{proof}

\subsubsection{Binary midpoint}

The midpoint functions --- both binary and infinitary, which are intended to realise \(\oplus\) and \(M\), respectively --- are much more complicated to define on signed-digit encodings. We follow the definitions \Escardo~gave in \cite{Escardo11fun}, formalising them in \textsc{Agda} for our framework.

The binary midpoint defined here sums two encodings of type \(\seq \3\) (which encodes the interval \([-1,1]\)) to achieve an encoding of type \(\seq \5\) (which encodes the interval \([-2,2]\)) which is then divided by two to output an encoding once again in \(\seq \3\).

\begin{definition}
\thesislit{5}{SignedDigit}{\urlfive}
The type of \emph{quinary digits} \(\5\), equivalent to \(\F(5)\), is defined by its elements \(\ty{\overline{2},\overline{1},0,1,2}{\5}\).
\end{definition}

\begin{definition}
\thesislit{5}{SignedDigit}{_+\urlthree_}
We define the addition function on ternary digits \(\ty{\mathsf{add\3}}{\3 \to \3 \to \5}\) in the expected way by pattern matching (e.g.\ \(\mathsf{add\3}(\overline 1,\overline 1) := \overline 2)\), etc.).
\end{definition}

\noindent
Now we define the halving function \(\ty{\mathsf{div2}}{\seq \5 \to \K}\) corecursively (though, as we see in the below \cref{remark:div2-aux}, we define the function by induction in our \textsc{Agda} formalisation).
When the first element is \(\overline 2, 0\) or \(2\), the output is straightforward. Otherwise, the definition must `offset' the output using the redundant digits \(\overline 1\) and \(1\).

\begin{definition}
\label{def:div2}
\thesislit{5}{SignedDigit}{div2}
We define the halving function from quinary signed-digit encodings to ternary signed-digit encodings as,
\begin{alignat*}{6}
\mathsf{div2} &: \seq \5 &\to \K, & & & &\\
\mathsf{div2} & (\overline 2                &:: \alpha) & := \overline 1                &\ :: \mathsf{div2} &(              &\alpha),\\
\mathsf{div2} & (\overline 1 :: \overline 2 &:: \alpha) & := \overline 1            &\ :: \mathsf{div2} &(0 ::            &\alpha),\\
\mathsf{div2} & (\overline 1 :: \overline 1 &:: \alpha) & := \overline 1                &\ :: \mathsf{div2} &(          1 :: &\alpha),\\
\mathsf{div2} & (\overline 1 ::           0 &:: \alpha) & :=           0 &\ :: \mathsf{div2} &(\overline 2 :: &\alpha),\\
\mathsf{div2} & (\overline 1 ::           1 &:: \alpha) & :=           0                &\ :: \mathsf{div2} &(\overline 1 :: &\alpha),\\
\mathsf{div2} & (\overline 1 ::           2 &:: \alpha) & :=           0            &\ :: \mathsf{div2} &(   0 ::        &\alpha),\\
\mathsf{div2} & (          0                &:: \alpha) & :=           0                &\ :: \mathsf{div2} &(               &\alpha),\\
\mathsf{div2} & (          1 :: \overline 2 &:: \alpha) & :=           0 &\ :: \mathsf{div2} &(   0 ::        &\alpha),\\
\mathsf{div2} & (          1 :: \overline 1 &:: \alpha) & :=           0                &\ :: \mathsf{div2} &(          1 :: &\alpha),\\
\mathsf{div2} & (          1 ::           0 &:: \alpha) & :=           0             &\ :: \mathsf{div2} &(   2 ::        &\alpha),\\
\mathsf{div2} & (          1 ::           1 &:: \alpha) & :=           1                &\ :: \mathsf{div2} &(\overline 1 :: &\alpha),\\
\mathsf{div2} & (          1 ::           2 &:: \alpha) & :=           1 &\ :: \mathsf{div2} &(   0 ::        &\alpha),\\
\mathsf{div2} & (          2                &:: \alpha) & :=           1                &\ :: \mathsf{div2} &(               &\alpha).
\end{alignat*}
\end{definition}

\begin{remark}
\label{remark:div2-aux}
\thesislit{5}{SignedDigit}{div2-aux}
We actually define this function in \textsc{Agda} using an auxiliary function \[\ty{\mathsf{div2'}}{\5 \x \5 \to \3 \x \5}\] such that we define
\begin{alignat*}{3}
\mathsf{div2} &(\alpha)_0 &&:= \mathsf{pr_1}(\mathsf{div2'} (\alpha_0,\alpha_1)),\\
\mathsf{div2} &(\alpha)_{n+1} &&:= \mathsf{div2}(\mathsf{pr_2}(\mathsf{div2'}(\alpha_0,\alpha_1) :: \mathsf{tail}(\mathsf{tail} \ \alpha))).
\end{alignat*}
The values of \(\mathsf{div2'}\) can be seen in \cref{def:div2}. For example \(\mathsf{div2'}(\overline{1},\overline{2}) := (\overline{1},0)\).
\end{remark}

\noindent
We use this function to define the binary midpoint function.

\begin{definition}
\label{def:mid-sd}
\thesislit{5}{SignedDigit}{mid}
The binary midpoint function on ternary signed-digit encodings is defined by adding the two sequences and then halving the result:
\begin{alignat*}{3}
\mathsf{mid} &: \K \to &&\K \to \K, \\
\mathsf{mid} &(\alpha,\beta) && := \mathsf{div2}(\mathsf{zipWith}(\mathsf{add\3},\alpha,\beta)).
\end{alignat*}
\end{definition}

\noindent
To verify that \(\mathsf{mid}\) correctly realises \(\oplus\), we utilise the following halving map \(\mathsf{half}\) and prove its relationship to \(\mathsf{div2}\)

\begin{definition}
\thesislit{5}{SignedDigitIntervalObject}{half}
We define the halving map from quinary digits to the interval object as,
\begin{alignat*}{3}
\mathsf{half} &: \5 &&\to \I, \\
\mathsf{half} &(\overline 2) &&:= -1, \\
\mathsf{half} &(\overline 1) &&:= -1 \oplus 0, \\
\mathsf{half} &(          0) &&:= 0,  \\
\mathsf{half} &(          1) &&:= +1 \oplus 0, \\
\mathsf{half} &(          2) &&:= +1.
\end{alignat*}
\end{definition}

\begin{figure}[ht]
\centering
\tikzset{ampersand replacement=\&}
\begin{floatrow}
    \begin{tikzcd}[sep=large]
	{\3 \x \3} \& \5 \\
	{\I \x \I} \& \I
	\arrow["{\mathsf{add}\3}", from=1-1, to=1-2]
	\arrow["{\oplus}"', from=2-1, to=2-2]
	\arrow["{\langle - \rangle \x \langle - \rangle}"', from=1-1, to=2-1]
	\arrow["{\mathsf{half}}", from=1-2, to=2-2]
    \end{tikzcd}
    \ \ \
    \begin{tikzcd}[row sep=large, column sep = huge]
	{\seq\3 \x \seq\3} \& \seq\5 \& \seq\3 \\
	{\I \x \I} \& \I \& \I
	\arrow["{\mathsf{zipWith}(\mathsf{add}\3)}", from=1-1, to=1-2]
	\arrow["\oplus"', from=2-1, to=2-2]
	\arrow["{\llangle - \rrangle \x \llangle - \rrangle}"', from=1-1, to=2-1]
	\arrow["{M \circ \mathsf{map}(\mathsf{half})}"{description}, from=1-2, to=2-2]
	\arrow[Rightarrow, no head, from=2-2, to=2-3]
	\arrow["{\mathsf{div2}}", from=1-2, to=1-3]
	\arrow["{\llangle - \rrangle}", from=1-3, to=2-3]
	\arrow["{\mathsf{mid}}", curve={height=-30pt}, from=1-1, to=1-3]
\end{tikzcd}
\end{floatrow}
\caption{Commutative diagrams illustrating \cref{lem:half-add-pw} (left) and \cref{lem:half-add,lem:div2-realise,cor:mid-realise} (right).}
\label{fig:mid-commute}
\end{figure}

We first prove the following lemma about the auxiliary function \(\mathsf{div2'}\) (\cref{remark:div2-aux}).

\begin{lemma}
\label{lem:div2-aux}
\thesislit{5}{SignedDigitIntervalObject}{div2-aux-\urleq}
Given any \(\ty{x,y}{\5}\) and \(\ty{z}{\I}\), we have \[\langle a \rangle \oplus (\mathsf{half} \ b \oplus z) = (\mathsf{half} \ x \oplus (\mathsf{half} \ y \oplus z)),\]
where \(\ty{a,b}{\3 \x \5} := \mathsf{div2'}(\alpha_0,\alpha_1)\).
\end{lemma}
\begin{proof}
The cases where \(a := \{\overline{2},0,2\}\) are trivial.
The ten other cases all require slightly different approaches, but all are just the rearrangement of elements of the midpoint algebra by idempotency, commutativity and transpositionality (\cref{def:mpa}).
\end{proof}

\begin{lemma}
\label{lem:div2-realise}
\thesislit{5}{SignedDigitIntervalObject}{div2-realiser}
Given any quinary signed-digit encoding \(\ty{\alpha}{\seq \5}\), we have that \(\llangle \mathsf{div2}(\alpha) \rrangle = M(\mathsf{map}(\mathsf{half},\alpha))\).
\end{lemma}
\begin{proof} \axioms{f}
By using \cref{cor:fg-approx}, this is reduced to showing \(\mathsf{map}(\langle - \rangle,\mathsf{div2}(\alpha))\) and \(\mathsf{map}(\mathsf{half},\alpha)\), for any \(\ty{\alpha}{\seq\5}\), are \(n\)-approximately equal for any \(\ty{n}{\N}\). 
In the base case, when \(n := 0\), the proof is trivial --- we can set \(\ty{z,w}{\I}\) as any elements of \(\I\), and so we only need to consider the inductive case where \(n := n' + 1\).

\vspace{1em}
Recall from \cref{def:finiteapprox} that this means we must give some \(\ty{z,w}{\I}\) that satisfy \[\langle \mathsf{div2}(\alpha)_0 \rangle \oplus (\langle \mathsf{div2}(\alpha)_1 \rangle \oplus ... (\mathsf{div2}(\alpha)_n \oplus w)) = \mathsf{half}(x_0) \oplus (\mathsf{half}(\alpha_1) \oplus ... (\mathsf{half}(\alpha_n) \oplus z)).\]

By taking \(\ty{(a,b)}{\3\x\5} := \mathsf{div2'}(\alpha_0,\alpha_1)\) this reduces to
\[\langle a \rangle \oplus (\langle \mathsf{div2}(b :: \alpha')_1 \rangle \oplus ... (\mathsf{div2}(b :: \alpha')_n \oplus w)) = \mathsf{half}(\alpha_0) \oplus (\mathsf{half}(\alpha_1) \oplus ... (\mathsf{half}(\alpha_n) \oplus z)),\]
where \(\alpha' := \mathsf{tail}(\mathsf{tail}(\alpha))\). 
By using the inductive hypothesis on the sequence \(b :: \alpha'\), this becomes
\[\langle a \rangle \oplus (\mathsf{half}(b) \oplus ... (\mathsf{half}(\alpha_n) \oplus z)) = \mathsf{half}(\alpha_0) \oplus (\mathsf{half}(\alpha_1) \oplus ... (\mathsf{half}(\alpha_n) \oplus z)).\]
Therefore, the result follows by \cref{lem:div2-aux}.
\end{proof}

We next show that halving an added sequence in the representation space realises the midpoint operation.

\begin{lemma}
\label{lem:half-add-pw}
\thesislit{5}{SignedDigitIntervalObject}{half-add-realiser}
Given any ternary digits \(\ty{a,b}{\3}\), we have that \(\mathsf{half}(\mathsf{add\3}(a,b)) = \langle a \rangle \oplus \langle b \rangle\).
\end{lemma}
\begin{proof} \axioms{f}
By induction on \(\ty{a,b}{\3}\), as well as the idempotency and commutativity of \(\oplus\) by \cref{def:mpa}.
\end{proof}

\begin{lemma}
\label{lem:half-add}
\thesislit{5}{SignedDigitIntervalObject}{half-add-realiser}
Given any ternary signed-digit encodings \(\ty{\alpha,\beta}{\K}\), we have that \(M(\mathsf{map}(\mathsf{half},\mathsf{zipWith}(\mathsf{add\3},\alpha,\beta))) = \llangle \alpha \rrangle \oplus \llangle \beta \rrangle\).
\end{lemma}
\begin{proof} \axioms{f}
By function extensionality and \cref{lem:half-add-pw}, \(M(\mathsf{map}(\mathsf{half},\mathsf{zipWith}(\mathsf{add\3},\alpha,\beta))) = M(\lambda n.\langle \alpha \rangle \oplus \langle \beta \rangle)\).
The result then follows by \cref{lem:M-hom-prop}.
\end{proof}

Using the two lemmas we have proved, we can now complete our verification of the midpoint operation.

\begin{theorem}
\label{cor:mid-realise}
\thesislit{5}{SignedDigitIntervalObject}{mid-realiser}
The midpoint function on signed-digit encodings realises the midpoint operator on the interval object.
\end{theorem}
\begin{proof} \axioms{f}
By \cref{lem:div2-realise,lem:half-add}.
\end{proof}

\subsubsection{Infinitary midpoint}

The infinitary midpoint function is obviously more complicated, both to define and verify. \Escardo~defines the infinitary midpoint on signed-digit encodings by using another intermediary signed-digit representation (this time of the interval \([-4,4]\)).

\begin{definition}
\thesislit{5}{SignedDigit}{\urlnine}
The type of \emph{nonary digits} \(\9\), equivalent to \(\F(9)\), is defined by its elements \(\ty{\overline{4},\overline{3},\overline{2},\overline{1},0,1,2,3,4}{\9}\).
\end{definition}

\begin{definition}
\thesislit{5}{SignedDigit}{_+\urlfive_}
We define the addition function on quinary digits \(\ty{\mathsf{add\5}}{\5 \to \5 \to \9}\) in the expected way by pattern matching.
\end{definition}

\noindent
We also define the quartering function \(\ty{\mathsf{div4}}{\seq \9 \to \K}\) by induction, and the related function \(\ty{\mathsf{quarter}}{\9 \to \I}\) on nonary digits.

\begin{definition}
\thesislit{5}{SignedDigit}{div4}
We define the quartering function \(\ty{\mathsf{div4}}{\seq \9 \to \K}\) from nonary signed-digit encodings to ternary signed-digit encodings similarly to how we defined \(\ty{\mathsf{div2}}{\seq \5 \to \K}\).
\end{definition}

\begin{definition}
\thesislit{5}{SignedDigitIntervalObject}{quarter}
We define the quartering map \(\ty{\mathsf{quarter}}{\9 \to \I}\) from nonary digits to the interval object similarly to how we defined \(\ty{\mathsf{half}}{\5 \to \I}\).
\end{definition}

\begin{lemma}
\label{lem:bigMid-finite-approx-2}
\thesislit{5}{SignedDigitIntervalObject}{quarter-realiser}
Given any ternary signed-digit encoding \(\ty{\alpha}{\K}\), we have that \(\llangle \mathsf{div4}(\alpha) \rrangle = M(\mathsf{map}(\mathsf{quarter},\alpha))\).
\end{lemma}
\begin{proof} \axioms{f}
Similar to \cref{lem:div2-realise}.
\end{proof}

We now define Adam Scriven's infinitary midpoint function on ternary signed-digit encodings in our framework. This definition was first given by Scriven in his MSc thesis, and reimplemented by \Escardo in his \textsc{Haskell} library~\cite{Scriven,Escardo11fun}.

\begin{definition}
\label{def:bigMid-sd}
\thesislit{5}{SignedDigit}{bigMid}
The infinitary midpoint function on ternary signed-digit encodings is defined by translating Scriven's definition to \textsc{Agda}:
\begin{alignat*}{3}
\mathsf{bigMid}' &: \seq{(\K)} \to \seq \9, && \\
\mathsf{bigMid}' &(((a :: b :: x) :: (c :: y) :: \zeta))_0 && := \mathsf{add\5}(\mathsf{add\3}(a,a),\mathsf{add\3}(b,c)), \\
\mathsf{bigMid}' &(((a :: b :: x) :: (c :: y) :: \zeta))_{\suc n} && := \mathsf{bigMid}'(\mathsf{mid}(x,y) :: \zeta)_n, \\
\ \\
\mathsf{bigMid} &: \seq{(\K)} \to K, && \\
\mathsf{bigMid} & && := \mathsf{div4} \circ \mathsf{bigMid}'.
\end{alignat*}
\end{definition}

\begin{figure}[ht]
    \centering
    \[\begin{tikzcd}[sep=large]
	{(\seq\3)^\N} & \seq\9 & \seq\3 \\
	\seq\I & \I & \I
	\arrow["{\mathsf{bigMid'}}", from=1-1, to=1-2]
	\arrow["M"', from=2-1, to=2-2]
	\arrow["{\mathsf{map} \ \llangle - \rrangle}"', from=1-1, to=2-1]
	\arrow["{M \circ \mathsf{map}(\mathsf{quarter})}"{description}, from=1-2, to=2-2]
	\arrow[Rightarrow, no head, from=2-2, to=2-3]
	\arrow["{\llangle - \rrangle}", from=1-3, to=2-3]
	\arrow["{\mathsf{div4}}", from=1-2, to=1-3]
	\arrow["{\mathsf{bigMid}}", curve={height=-30pt}, from=1-1, to=1-3]
    \end{tikzcd}\]
    \caption{Commutative diagram illustrating \cref{thm:bigmid-realiser}.}
    \label{fig:bigMid-commute}
\end{figure}

This operation does indeed realise the iteration operator \(\ty{M}{\seq \I \to \I}\) on the interval object; i.e.\ the diagram in \cref{fig:bigMid-commute} commutes.
We prove this by using the following lemmas.

\begin{lemma}
\label{lem:quarter-add-pw}
\thesislit{5}{SignedDigitIntervalObject}{\urlninea{}s-conv-\urleq}
Given any ternary digits \(\ty{a,b,c}{\3}\), we have that \(\mathsf{quarter}(\mathsf{add\5}(\mathsf{add\3}(a,a),\mathsf{add\3}(b,c))) = \langle a \rangle \oplus (\langle b \rangle \oplus \langle c \rangle)\).
\end{lemma}
\begin{proof}
By induction on \(\ty{a,b,c}{\3}\), as well as the idempotency, commutativity and transpositionality of \(\oplus\) by \cref{def:mpa}.
\end{proof}

\begin{lemma}
\label{lem:mid-first-three}
\thesislit{5}{SignedDigitIntervalObject}{M-bigMid'-\urleq}
Given any ternary signed-digit encodings \(\ty{\alpha,\beta}{\K}\) and real number \(\ty{z}{\I}\) in the interval object, we have, \[\llangle \alpha \rrangle \oplus (\llangle  \beta \rrangle \oplus z) = (\langle a \rangle \oplus (\langle b \rangle \oplus \langle c \rangle)) \oplus (\llangle \mathsf{mid}(\alpha',\beta') \rrangle \oplus z) ,\]
where \((a :: b :: \alpha') := \alpha\) and \((c :: \beta') := \beta\).
\end{lemma}
\begin{proof} \axioms{f}
By the following equational reasoning, (i) idempotency of \(M\) (\cref{lem:M-idem}), (ii) commutativity and (iii) transpositionality of \(\oplus\) (\cref{def:mpa}), and (iv) the correctness of the binary midpoint operation (\cref{cor:mid-realise}),
\begin{align*}
&\llangle \alpha \rrangle &\oplus \  &(\llangle \beta \rrangle \oplus z),& &\\
:= &\llangle (a :: (b :: \alpha')) \rrangle &\oplus \  &(\llangle (c :: \beta') \rrangle \oplus z),& &\\
= &\langle a \rangle \oplus \llangle (b :: \alpha') \rrangle &\oplus \  &(\llangle (c :: \beta') \rrangle \oplus z),& &\text{by (i)},\\
= &(\langle a \rangle \oplus ( \langle b \rangle \oplus \llangle \alpha' \rrangle )) &\oplus \  &(\llangle (c :: \beta') \rrangle \oplus z),& &\text{by (i)},\\
= &(\langle a \rangle \oplus ( \langle b \rangle \oplus \llangle \alpha' \rrangle )) &\oplus \  &((\langle c \rangle \oplus \llangle \beta' \rrangle) \oplus z),& &\text{by (i)},\\
= &( (\langle b \rangle \oplus \llangle \alpha' \rrangle) \oplus \langle a \rangle) &\oplus \  &((\langle c \rangle \oplus \llangle \beta' \rrangle) \oplus z),& &\text{by (ii)},\\
= &( (\langle b \rangle \oplus \llangle \alpha' \rrangle) \oplus (\langle c \rangle \oplus \llangle \beta' \rrangle)) &\oplus \  &(\langle a \rangle \oplus z),& &\text{by (iii)},\\
= &( (\langle b \rangle \oplus \langle c \rangle) \oplus (\llangle \alpha' \rrangle \oplus \llangle \beta' \rrangle)) &\oplus \  &(\langle a \rangle \oplus z),& &\text{by (iii)},\\
= &( (\langle b \rangle \oplus \langle c \rangle) \oplus (\langle a \rangle) &\oplus \ &(\llangle \alpha' \rrangle \oplus \llangle \beta' \rrangle) \oplus z),& &\text{by (iii)},\\
= &(\langle a \rangle \oplus (\langle b \rangle \oplus \langle c \rangle) &\oplus \ &(\llangle \alpha' \rrangle \oplus \llangle \beta' \rrangle) \oplus z),& &\text{by (iii)},\\
= &(\langle a \rangle \oplus (\langle b \rangle \oplus \langle c \rangle) &\oplus \ &(\mathsf{mid}(\alpha',\beta') \oplus z)),& &\text{by (iv)}.
\end{align*}
\end{proof}

\begin{corollary}
\label{cor:mid-first-three}
\thesislit{5}{SignedDigitIntervalObject}{bigMid'-approx}
Given any ternary signed-digit encodings \(\ty{\alpha,\beta}{\K}\) and real number \(\ty{z}{\I}\) in the interval object, we have, \[\llangle \alpha \rrangle \oplus (\llangle  \beta \rrangle \oplus z) = \mathsf{quarter}(\mathsf{add\5}(\mathsf{add\3}(a,a),\mathsf{add\3}(b,c))) \oplus (\llangle \mathsf{mid}(\alpha',\beta') \rrangle \oplus z) ,\]
where \((a :: b :: \alpha') := \alpha\) and \((c :: \beta') := \beta\).
\end{corollary}
\begin{proof} \axioms{f}
By \cref{lem:quarter-add-pw,lem:mid-first-three}.
\end{proof}

\noindent
Thus, it turns out that the \(\mathsf{bigMid}\) operation provides an \(n\)-approximation of \(M\) for every \(\ty{n}{\N}\).

\begin{lemma}
\label{lem:bigMid-finite-approx-1}
\thesislit{5}{SignedDigitIntervalObject}{bigMid'-approx}
The functions \(\mathsf{map}(\llangle - \rrangle)\) and \(\mathsf{map}(\mathsf{quarter}) \circ \mathsf{bigMid'}\) --- both of type \(\seq{(\K)} \to \seq \I\) --- are \(n\)-approximately equal for all \(\ty{n}{\N}\).
\end{lemma}
\begin{proof} \axioms{f}
By induction and \cref{cor:mid-first-three}.
\end{proof}

\begin{theorem}
\label{thm:bigmid-realiser}
\thesislit{5}{SignedDigitIntervalObject}{M-realiser}
The infinitary midpoint function on ternary signed-digit encodings realises the iteration operator on the interval object; i.e.\ given any sequence of ternary signed-digit encodings \(\ty{\zeta}{\seq{(\K)}}\), we have \(\llangle \mathsf{bigMid}(\zeta) \rrangle = M(\mathsf{map}(\llangle - \rrangle,\zeta))\).
\end{theorem}
\begin{proof} \axioms{f}
By \cref{thm:finite-approx,lem:bigMid-finite-approx-2,lem:bigMid-finite-approx-1}.
\end{proof}

\subsubsection{Multiplication}

Finally, we define and verify multiplication.
In the Haskell library, \Escardo~defined a variety of multiplication functions on signed-digit encodings; we choose to define and verify \texttt{mul\_version0}, which uses the \(\mathsf{bigMid}\) function~\cite{Escardo11fun}.

\begin{definition}
\thesislit{5}{SignedDigit}{digitMul}
The auxiliary multiplication function, which multiplies a ternary signed-digit encoding by a ternary digit, is defined as follows:
\begin{alignat*}{3}
\mathsf{digitMul} &: \3 \to &&\K \to \K, \\
\mathsf{digitMul} &(\overline 1,\beta) && := \mathsf{neg}(\beta), \\
\mathsf{digitMul} &(          0,\beta) && := \lambda n.0, \\
\mathsf{digitMul} &(          1,\beta) && := \beta.
\end{alignat*}
\end{definition}

\begin{definition}
\label{def:mul-sd}
\thesislit{5}{SignedDigitInterval}{mul}
The multiplication function on ternary signed-digit encodings is defined by multiplying the first argument against each individual digit of the second, and then taking the infinitary midpoint:
\begin{alignat*}{3}
\mathsf{mul} &: \K \to &&\K \to \K, \\
\mathsf{mul} &(\alpha,\beta) && := \mathsf{bigMid}(\mathsf{zipWith}(\mathsf{digitMul},\alpha,\mathsf{repeat} \ \beta)).
\end{alignat*}
\end{definition}

\noindent
This function realises the multiplication function on the interval object; i.e.\ the diagrams in \cref{fig:mul-commute} commute.

\begin{figure}[ht]
\centering
\tikzset{ampersand replacement=\&}
\begin{floatrow}
    \begin{tikzcd}[sep=large]
	{\3 \x \seq\3} \& K \\
	{\I \x \I} \& \I
	\arrow["{\mathsf{digitMul}}", from=1-1, to=1-2]
	\arrow["{*}"', from=2-1, to=2-2]
	\arrow["{\langle - \rangle \x \llangle - \rrangle}"', from=1-1, to=2-1]
	\arrow["{\llangle - \rrangle}", from=1-2, to=2-2]
    \end{tikzcd}
    \ \ \
    \begin{tikzcd}[column sep=huge,row sep=large]
	{\seq\3 \x \seq\3} \& {(\seq\3)^\N} \& \seq\3 \\
	{\I \x \I} \& {\seq\I} \& \I
	\arrow["{\mathsf{repeat} \ \circ \ *}"', from=2-1, to=2-2]
	\arrow["{\mathsf{map} \ \llangle - \rrangle}"{description}, from=1-2, to=2-2]
	\arrow["{\mathsf{digitMul'}}", from=1-1, to=1-2]
	\arrow["{\llangle - \rrangle \x \llangle - \rrangle}"', from=1-1, to=2-1]
	\arrow["M"', from=2-2, to=2-3]
	\arrow["{\mathsf{bigMid}}", from=1-2, to=1-3]
	\arrow["{\mathsf{mul}}", curve={height=-30pt}, from=1-1, to=1-3]
	\arrow["{\llangle - \rrangle}", from=1-3, to=2-3]
	\arrow["{*}"', curve={height=30pt}, from=2-1, to=2-3]
    \end{tikzcd}
\end{floatrow}
\caption{Commutative diagrams illustrating \cref{lem:digitMul-realiser} (left) and \cref{thm:mul-realiser} (right). In the diagram, \(\mathsf{digitMul'} := \lambda \alpha\beta.\mathsf{zipWith}(\mathsf{digitMul},\alpha,\mathsf{repeat}\ \beta)\).}
\label{fig:mul-commute}
\end{figure}

\begin{lemma}
\label{lem:digitMul-realiser}
\thesislit{5}{SignedDigitIntervalObject}{digitMul-realiser}
The auxiliary multiplication operator realises multiplication on the interval object; i.e.\ given any ternary digit \(\ty{t}{\3}\) and a ternary signed-digit encodings \(\ty{\beta}{\K}\), we have \(\llangle \mathsf{digitMul}(t,\beta) \rrangle = \langle t \rangle * \llangle \beta \rrangle\).
\end{lemma}
\begin{proof} \axioms{f}
By induction on \(\ty{t}{\3}\).

In the case where \(t := \overline 1\), we must show that \(\llangle \mathsf{neg}(\beta) \rrangle = -1 * \llangle \beta \rrangle\). By the negation realiser (\cref{lem:neg-realise}), \(\llangle \mathsf{neg}(\beta) \rrangle = - \llangle \beta \rrangle\). The result follows by \cref{lem:mul-io--1}.

In the case where \(t :=           0\), we must show that \(\llangle \lambda n.0 \rrangle = 0 * \llangle \beta \rrangle\). By the idempotency of \(M\) (\cref{lem:M-idem}), \(\llangle \lambda n.0 \rrangle := M(\mathsf{map}(\langle - \rangle,\lambda n.0)) = \langle 0 \rangle = 0\). The result follows by \cref{lem:mul-io-0}.

In the case where \(t :=           1\), we must show that \(\llangle \beta \rrangle = 1 * \llangle \beta \rrangle\). The result follows by \cref{lem:mul-io-1}.
\end{proof}

\begin{theorem}
\label{thm:mul-realiser}
\thesislit{5}{SignedDigitIntervalObject}{mul-realiser}
The multiplication function on ternary signed-digit encodings realises the multiplication operation on the interval object.
\end{theorem}
\begin{proof} \axioms{f}
By the following equational reasoning, (i) the correctness of the infinitary midpoint operation (\cref{thm:bigmid-realiser}), (ii) the correctness of the auxiliary multiplication operation (\cref{lem:digitMul-realiser}), (iii) the fact right-multiplication --- i.e.\ \(\lambda y.x * y\) for a given \(\ty{x}{\I}\) --- is a midpoint homomorphism (\cref{lem:mul-hom-r}), and (iv) the fact that midpoint homomorphisms are \(M\)-homomorphisms (\cref{lem:mp-hom-M-hom}),
\begin{align*}
&\llangle \mathsf{mul}(\alpha,\beta) \rrangle, & &\\
:= &\llangle \mathsf{bigMid}(\mathsf{zipWith}(\mathsf{digitMul},\alpha,\mathsf{repeat} \ \beta)) \rrangle, & &\\
=  & M(\mathsf{map}(\llangle - \rrangle,\mathsf{zipWith}(\mathsf{digitMul},\alpha,\mathsf{repeat} \ \beta))) & &\text{by (i)},\\
:= & M(\lambda n. \llangle \mathsf{digitMul}(\alpha_n,\beta) \rrangle), & &\\
=  & M(\lambda n. \langle \alpha_n \rangle * \llangle \beta \rrangle) & &\text{by (ii)},\\
=  & M(\lambda n. \langle \alpha_n \rangle) * \llangle \beta \rrangle & &\text{by (iii) and (iv)},\\
:=  & \llangle \alpha \rrangle * \llangle \beta \rrangle. & &
\end{align*}
\end{proof}

\section{Ternary Boehm encodings}
\label{sec:boehm}

We shall see in \cref{chap:exact-real-search} that the ternary signed-digit encodings are an ideal type for formalising the convergence of what we call `exact real search'.
However, in that section it will also become clear that the efficiency of the algorithms extracted from the convergence theorems of search, optimisation and regression on that type leave much to be desired.
For this reason, we introduce another type --- ternary Boehm encodings --- that yield more efficient algorithms at the expense of being more difficult to formalise in our \textsc{Agda} framework\footnote{Hence, we have not verified the ternary Boehm encodings to the same extent as the ternary signed-digit encodings.}.

In the 1990s, Hans-J. Boehm produced a practical \textsc{Java} library for exact real arithmetic~\cite{Boehm90s}.
Since then, Boehm's library has been developed to become the underpinning of Google's bespoke Android calculator mobile phone application~\cite{Boehm17,BoehmAPI}.
The star of the library is the class \texttt{CR}, objects of which Boehm calls \emph{constructive reals}.
Objects \(x\) of \texttt{CR} have a single private method \texttt{approximate}, which takes a (small) integer and outputs an (effectively unbounded) integer.
The input integer \(n\) denotes the requested \emph{precision-level} of the constructive real, while the output integer \(x_n\) denotes an integer approximation --- scaled relative to the given precision-level --- of the real number \(\ty{\llbracket x \rrbracket}{\R}\) which \(x\) encodes.

Every object of \texttt{CR} is therefore a bi-infinite sequence of integer approximations of a real number, which is constructed and manipulated to ensure that the following condition about its relationship to the real number it encodes always holds:
\[ d_\R(\llbracket x \rrbracket,2^{n}x_n) < 2^{n} .\]

In this section, we will show how we re-rationalise this class --- adapted slightly for the purposes of exact real search --- within our \textsc{Agda} framework as the type \(\T\).
We will then verify the structure of the type by defining the representation map \(\ty{\llbracket - \rrbracket}{\T \to \R}\); for this purpose, we utilise the Dedekind reals rather than the interval object, both for the sake of variety and because elements of \(\T\) represent reals without reference to a particular compact interval.
However, following this, we will show how we use subtypes of \(\T\) to represent only reals in particular compact intervals, as this is required for search.
Finally, we will discuss how Boehm defines exact real arithmetic on \texttt{CR} by using interval arithmetic.

\subsection{Definition in our type theory}

In order to define \texttt{CR} in our type theory, we start off by representing the class as functions \(\ty{x}{\Z \to \Z}\).
It is helpful to think of the integer approximation \(\ty{x_n}{\Z}\), at any precision-level \(\ty{n}{\Z}\), as an \emph{interval approximation} of \(\llbracket x \rrbracket\) --- namely, \(x_n\) represents to the open interval \(\left( 2^n(x_n-1),2^n(x_n+1) \right),\) which clearly contains \(\llbracket x \rrbracket\).

As we wish to use these encodings for search, the open interval structure is not ideal. Therefore, we alter \texttt{CR} by changing the order used in the relationship above from \emph{strict} to \emph{non-strict}. This is done because we require the use of compact intervals for search, though it makes little difference to the way \texttt{CR} objects are manipulated algorithmically.
Furthermore, for stylistic reasons we make two further changes. Firstly, with the ternary signed-digits we used higher precision-levels \(n\) to mean \emph{more} precise integer approximations, rather than less, and so we alter \texttt{CR} so that we can do the same. Furthermore, we slightly reposition the interval codes to avoid unnecessary subtraction in our formalisation.

By making the two changes above, the functions \(\ty{x}{\Z \to \Z}\) we define for \texttt{CR} in our type theory must satisfy \[ d_\R(\llbracket x \rrbracket,2^{-n}(x_n+1)) \leq 2^{-n} ,\] meaning that they are an infinitary sequence of integers, with each input-output pair \(\ty{(x_n,n)}{\Z^2}\) representing the compact interval with dyadic rational endpoints \[\left[\frac{x_n}{2^n},\frac{x_n+2}{2^n}\right].\]

We illustrate the structure of these encodings in \cref{fig:brick-pattern}, which shows some possible interval approximations of real numbers in \([-3,3]\) at precision-levels \(0\), \(1\) and \(2\) --- note that the intervals halve in size as the precision-level increases further down the illustration.
We refer to this as the \emph{ternary structure}, because each interval represented by \(\ty{(n,m)}{\Z^2}\) in the structure is perfectly trisected into three representations \((2n,m+1)\), \((2n+1,m+1)\) and \((2n+2,m+1)\) on the next precision-level.

\begin{figure}
\centering
\includegraphics[width=\textwidth]{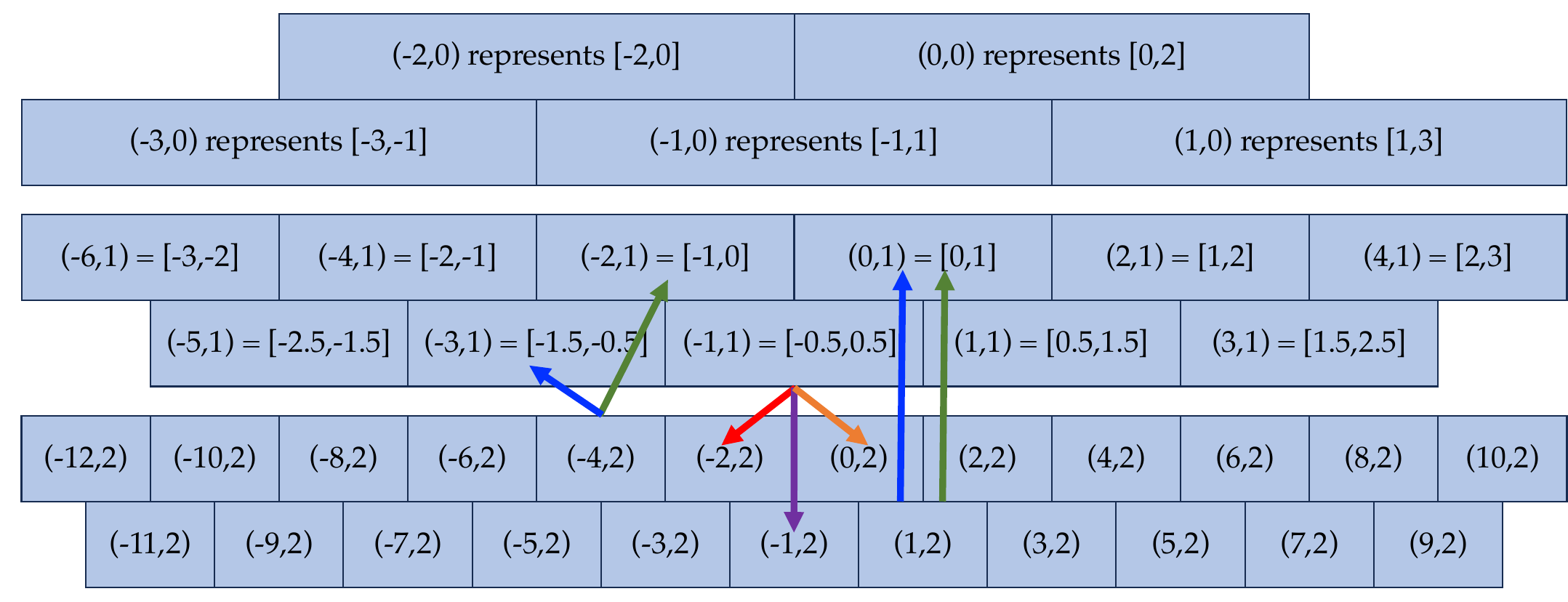}
\caption{Ternary structure underlying the Boehm encodings. Each interval approximation is represented by an integer pair. The coloured arrows illustrate how the structural operations \(\color{red}{\mathsf{downLeft}}\), \(\color{violet}{\mathsf{downMid}}\), \(\color{orange}{\mathsf{downRight}}\), \(\color{blue}{\mathsf{upLeft}}\) and \(\color{OliveGreen}{\mathsf{upRight}}\) work.}
\label{fig:brick-pattern}
\end{figure}

This structure can be `navigated' by the following operations on integer approximations.
The operations \(\color{red}{\mathsf{downLeft}}\), \(\color{violet}{\mathsf{downMid}}\) and \(\color{orange}{\mathsf{downRight}}\) take an integer that refers to an interval on an arbitrary precision-level and return one which refers, respectively, to one of the three possible trisections on the next precision-level. Meanwhile, \(\color{blue}{\mathsf{upLeft}}\) and \(\color{OliveGreen}{\mathsf{upRight}}\) do the reverse.
The names here are colour coded with respect to the colours in \cref{fig:brick-pattern}, which illustrate how the operations work for particular intervals.
\begin{alignat*}{3}
\mathsf{downLeft},& \mathsf{downMid}, \mathsf{downRight}, \mathsf{upLeft}, \mathsf{upRight} : \Z \to \Z, \\
\mathsf{downLeft}(k) & := 2k, \\
\mathsf{downMid}(k) & := 2k+1, \\
\mathsf{downRight}(k) & := 2k+2, \\
\mathsf{upRight}(\pos k) & := \pos (k/2), \\
\mathsf{upRight}(\negsuc k) & := \negsuc (k/2), \\
\mathsf{upLeft}(k) & := \mathsf{upRight}(k-1).
\end{alignat*}

\begin{definition}
\thesislit{5}{BoehmVerification}{_below_}
\thesislit{5}{BelowAndAbove}{_below'_}
Given two integers \(\ty{n,m}{\Z}\), we say that \(n\) is \emph{below} \(m\) if the interval represented by \(\ty{(n,i+1)}{\Z^2}\) is a strict subinterval of that which is represented by \(\ty{(m,i)}{\Z^2}\) for any \(\ty{i}{\Z}\):
\begin{alignat*}{3}
& \mathsf{below} &&: \Z \to \Z \to \Omega,\\
n \ &\mathsf{below} \ m &&:= \mathsf{downLeft} \ m \leq n \leq \mathsf{downRight} \ m.
\end{alignat*}
Equivalently, this can be defined as follows:
\begin{alignat*}{3}
& \mathsf{below}' &&: \Z \to \Z \to \Omega,\\
n \ &\mathsf{below}' \ m &&:= (n = \mathsf{downLeft} \ m) + (n = \mathsf{downMid} \ m) + (n = \mathsf{downRight} \ m).
\end{alignat*}
\end{definition}

To further illustrate how ternary Boehm encodings work, let us now give an example encoding \(\ty{x}{\Z \to \Z}\) which is intended to represent the real number \(\ty{2}{\R}\); \(x\) could be defined in a variety of ways:
\begin{itemize}
\item \(x_0 = 1\), \(x_1 = 2\), ...
\item \(x_0 = 1\), \(x_1 = 3\), ...
\item \(x_0 = 1\), \(x_1 = 4\), ...
\item \(x_0 = 0\), \(x_1 = 2\), ...
\item \(x_0 = 0\), \(x_1 = 3\), ...
\end{itemize}

\noindent
All of these are valid representations of \(2\) as a ternary Boehm encoding because the two interval approximations \emph{intersect}. However, we wish in our type theory to restrict ourselves to only those encodings that follow the ternary structure --- the last item in the above list would not, therefore, be valid.

\begin{definition}
\label{def:ternary}
\thesislit{5}{BoehmVerification}{ternary}
A sequence of integers \(\ty{x}{\Z \to \Z}\) is \emph{ternary} if, \[\mathsf{ternary}(x) := \pity{n}{\Z}{x_{n+1} \ \mathsf{below} \ x_n} .\]
\end{definition}

\noindent
Restricting ourselves in this way will allow us to remove cases that are difficult to reason about, and to define continuity and search similarly to how we defined it for the ternary signed-digit encodings.

\begin{definition}
\label{def:ternary-boehm-encoding}
\thesislit{5}{BoehmVerification}{\urlT}
The type of \emph{ternary Boehm encodings} is defined as the collection of ternary sequences:
\[ \T := \sigmatye{x}{\Z \to \Z}{\mathsf{ternary}(x)} .\]
\end{definition}

Returning to our illustration, we show one possible map that takes any integer to a ternary Boehm encoding that represents it.

\begin{definition}
The function that maps any integer to a ternary Boehm encodings is defined as follows:
\begin{alignat*}{3}
\mathsf{\Z{\hy}to{\hy}\T} &: \Z \to \T, && \ \\
\mathsf{\Z{\hy}to{\hy}\T} &(n)(\pos \ 0) &&:= n, \\
\mathsf{\Z{\hy}to{\hy}\T} &(n)(\pos \ (i + 1)) &&:= \mathsf{downLeft}^{i+1}(n), \\
\mathsf{\Z{\hy}to{\hy}\T} &(n)(\negsuc \ i) &&:= \mathsf{upRight}^{i+1}(n).
\end{alignat*}
\end{definition}

\noindent
Before proceeding, we invite the reader to verify in their minds that this map is indeed correct.

\subsection{Verification of Boehm encodings via Dedekind reals}

In this subsection, we will formally verify that the above definition of ternary Boehm encodings does indeed represent real numbers. The application of this work is for the future verification of exact real arithmetic on the Boehm encodings (discussed in \cref{sec:boehm-era,fw:boehm-functions}).

We first quickly recap the Dedekind reals, which we use for the type of real numbers \(\R\) in this section.
The work of implementing and formalising the Dedekind reals (as well as the rationals and dyadic rationals) for the \textsc{Agda} library \textsc{TypeTopology} was performed by Andrew Sneap in his M.Sci. project~\cite{Sneap}.

The Dedekind reals can be defined over any dense subset of the real numbers, and are usually defined over the rationals \(\Q\).
We choose to instead define them over the \emph{dyadic rationals} \(\D\) --- i.e.\ those rational numbers of form \(\frac{z}{2^n}\) for \(\ty{z,n}{\Z}\) --- as this will be more convenient when verifying the structure of \(\T\), whose interval approximations have dyadic endpoints.
We use the dyadics informally in this thesis; simply noting they are a discrete set with a partial order, and that there is a function \(\ty{\iota}{\Z \x \Z \to \D}\) which map pairs of integers \(\ty{(z,n)}{\Z \x \Z}\) to the dyadic rational \(\frac{z}{2^n}\).

\begin{definition}
\label{def:dedekind}
\andrew{DyadicReals}{\urlR-d}
A \emph{Dedekind real number} \(\ty{(L,R)}{\R}\) consists of two predicates on dyadics \(\ty{L,R}{\D \to \Omega}\) (the \emph{left cut} and the \emph{right cut} respectively) that satisfy,
\begin{enumerate}
\item \(\existstye{p}{\D}{L(p)}\),
\item \(\existstye{q}{\D}{R(p)}\),
\item \(\pity{p}{\D}{L(p) \Leftrightarrow \existsty{p'}{\D}{L(p') \x p < p'}}\),
\item \(\pity{q}{\D}{R(q) \Leftrightarrow \existsty{q'}{\D}{R(q') \x q' < q}}\),
\item \(\pity{p,q}{\D}{L(p) \x R(q) \to p < q}\),
\item \(\pity{p,q}{\D}{p < q \to L(p) \vee R(q)}\).
\end{enumerate}
\end{definition}

Now recall the definitions of representation maps in \cref{sec:representation-map}.
In order to define the representation map \(\ty{\llbracket - \rrbracket}{\T \to \R}\) on ternary Boehm encodings, we first look at how one can represent real numbers as sequences of dyadic rational intervals that satisfy the \emph{nested} and \emph{positioned} properties (defined below).
We will then show that --- as we illustrated in the previous subsection --- every ternary Boehm real \(\ty{\chi}{\T}\) can be transformed into such a sequence of dyadic intervals, and therefore that we can indeed define \(\ty{\llbracket \chi \rrbracket}{\R}\).

\begin{definition}
\thesislit{5}{BoehmVerification}{\urlZ[1/2]\urlsuperscriptI}
The type of \emph{valid representations of dyadic intervals} \(\D^I\) is defined as the type of pairs of dyadics such that the first is smaller than the second:
\[ \D^I := \sigmaty{(l,r)}{\D \x \D}{l \leq r} .\]
For convenience, we simply refer to elements of this type as \emph{dyadic intervals}.
\end{definition}

\begin{definition}
\thesislit{5}{BoehmVerification}{_covers_}
One dyadic interval \emph{covers} another if,
\begin{alignat*}{3}
& \ \mathsf{covers} \ && : \relation{\D^I}, \\
(l_1 , r_1) &\ \mathsf{covers} \ &&(l_2,r_2) := (l_1 \leq l_2) \x (r_2 \leq r_1) .
\end{alignat*}
\end{definition}

\begin{definition}
\thesislit{5}{BoehmVerification}{nested}
A sequence of dyadic intervals \(\ty{\chi}{\Z \to \D^I}\) is \emph{nested} if every interval in the sequence covers the next:
\begin{alignat*}{3}
\mathsf{nested} &: (\Z \to \D^I) \to \Omega, \\
\mathsf{nested} &(\chi) := \Pity{n}{\N}{\chi_n \ \mathsf{covers} \ \chi_{\suc n}} .
\end{alignat*}
\end{definition}

\begin{definition}
\label{def:positioned}
\thesislit{5}{BoehmVerification}{positioned}
A sequence of dyadic intervals \(\ty{\chi}{\Z \to \D^I}\) is \emph{positioned} if it contains arbitrary small intervals,
\begin{alignat*}{3}
\mathsf{positioned} &: (\Z \to \D^I) \to \Omega, \\
\mathsf{positioned} &(\chi) := \Pitye{\varepsilon}{\D}{(\varepsilon > 0) \to \Sigmaty{n}{\N}{r_n - l_n \leq \varepsilon}} \ \ \emph{(where \((l_n,r_n) := \chi_n\))}.
\end{alignat*}
\end{definition}

\begin{lemma}
\label{lem:real-seq-dyad}
\thesislit{5}{BoehmVerification}{\urlllparenthesis_\urlrrparenthesis}\footnote{In \textsc{Agda}, the formalisation of this proof currently requires the use of a number of unformalised lemmas about dyadic numbers.}
If a sequence of dyadic intervals \(\ty{\chi}{\Z \to \D^I}\) is nested and positioned, there is a real number \(\ty{\llparenthesis \chi \rrparenthesis}{\R}\) which it represents.
\end{lemma}
\begin{proof} \axioms{t}
The proof is due to Andrew Sneap.
The Dedekind cuts \(\ty{L,R}{\D \to \Omega}\) are defined as follows:
\begin{alignat*}{3}
L&(p) &&:= \existsty{n}{\Z}{p < \mathsf{pr_1}(\chi_n)},\\
R&(q) &&:= \existsty{n}{\Z}{\mathsf{pr_2}(\chi_n) < q}.
\end{alignat*}
The idea of \(L\) is that a rational number is less than the represented number if there is an interval in the sequence that it is smaller than the lower endpoint of. The idea of \(R\) is the opposite of this.
The first four properties of the cuts (\cref{def:dedekind}) are proved by properties of the dyadics, while the fifth is by the nested property and the sixth is by the positioned property.
\end{proof}

Any interval that has dyadic endpoints with the same denominator can be represented as a triple of integers.

\begin{definition}
\label{def:dyadic-interval-code}
\thesislit{5}{BoehmVerification}{\urlZ\urlsuperscriptthree}
We call a triple of integers \(\ty{(k,c,p)}{\Z^3}\) such that \(k \leq c\) (the proof term of which we leave implicit) a \emph{dyadic interval code} when we use it for the purpose of representing the dyadic interval \(\ty{\Eta(k,c,p) := (\iota(k,p),\iota(c,p))}{\D^I}\) (which has width \(\frac{c-k}{2^p}\)).
\end{definition}

\noindent
Given a sequence of dyadic interval codes \(\ty{\chi}{\Z \to \Z^3}\), we overload terminology by saying \(\chi\) is nested/positioned to mean that \(\mathsf{map}(\Eta,\chi)\) is nested/positioned.

\begin{corollary}
\label{cor:real-seq-dyad-code}
\thesislit{5}{BoehmVerification}{\urlllparenthesis_\urlrrparenthesis'}
If a sequence of dyadic interval codes \(\ty{\chi}{\Z \to \Z^3}\) is nested and positioned, it represents the real number \(\ty{\llparenthesis \chi \rrparenthesis' := \llparenthesis \mathsf{map}(\Eta,\chi) \rrparenthesis}{\R}\).
\end{corollary}
\begin{proof} \axioms{t}
By \cref{lem:real-seq-dyad}.
\end{proof}

The dyadic intervals underlying the structure of \(\T\) (illustrated in \cref{fig:brick-pattern}) can be represented by pairs of integers.

\begin{definition}
\label{def:ternary-interval-code}
\thesislit{5}{BoehmVerification}{\urlZ\urlsuperscripttwo}
We call a pair of integers \(\ty{(k,p)}{\Z^2}\) a \emph{ternary interval code} when we use it for the purpose of representing the interval \(\ty{\Rho(k,p) := (\iota(k,p),\iota(k+2,p))}{\D^I}\) (which has width \(\frac{1}{2^{p-1}}\)).
\end{definition}

\noindent
Furthermore, every ternary interval code \((k,p)\) gives a dyadic interval code \[\Mu(k,p) := (k,k+2,p), \] such that the below diagram commutes.

\[
  \begin{tikzcd}[sep=large]
    \Z^2 \arrow{r}{\Mu} \arrow[swap]{dr}{\Rho} & \Z^3 \arrow{d}{\Eta} \\
     & \D^I
  \end{tikzcd}
\]

Given a sequence of ternary interval codes \(\ty{\chi}{\Z \to \Z^2}\), we overload terminology by saying \(\chi\) is nested/positioned to mean that \(\mathsf{map}(\Rho,\chi)\) is nested/positioned.

\begin{corollary}
\label{cor:real-seq-ternary-code}
\thesislit{5}{BoehmVerification}{\urlllparenthesis_\urlrrparenthesis''}
If sequence of ternary interval codes \(\ty{\chi}{\Z \to \Z^2}\)  is nested and positioned, it represents the real number \(\ty{\llparenthesis \chi \rrparenthesis'' := \llparenthesis \mathsf{map}(\Mu,\chi) \rrparenthesis'}{\R}\).
\end{corollary}
\begin{proof} \axioms{t}
By \cref{cor:real-seq-dyad-code}.
\end{proof}

For sequences of ternary interval codes, being positioned is implied by a stronger property we call the \emph{normalised} property.
A sequence of such interval codes is normalised if the precision-level of the output interval is always identical to that which was requested.

\begin{definition}
\label{def:normalised}
\thesislit{5}{BoehmVerification}{normalised}
A sequence of ternary interval codes is \emph{normalised} if the precision-level of the code at every point \(\ty{n}{\N}\) is always exactly \(n\).
\begin{alignat*}{3}
\mathsf{normalised} &: (\Z \to \Z^2) \to \Omega ,\\
\mathsf{normalised} &(\chi) :=  \Pity{n}{\N}{\mathsf{pr_2}(\chi_n) = n}.
\end{alignat*}
\end{definition}

\begin{lemma}
\thesislit{5}{BoehmVerification}{normalised-positioned}\footnote{In \textsc{Agda}, the formalisation of this proof currently requires the use of a number of unformalised lemmas about dyadic numbers.}
If a sequence of ternary interval codes is normalised, then it is positioned.
\end{lemma}
\begin{proof}
Recall from \cref{def:positioned} that, given some \(\ty{\varepsilon}{\D} > 0\) we wish to find an interval in the sequence whose width is at most \(\varepsilon\).

We define \(\ty{q}{\Z}\) to be any integer such that \(\frac{2}{2^q} < \varepsilon\).
By the normalised property, \(q = pr_2(\chi_q)\) and, by definition of ternary interval codes, \(\chi_q\) has width \(\frac{2}{2^q}\), which is smaller than \(\varepsilon\).
\end{proof}

If a sequence of ternary interval codes \(\ty{\chi}{\Z \to \Z^2}\) is normalised, we can effectively discard the precision-level information in the output and simply consider it as an integer sequence \(\ty{\mathsf{map}(\mathsf{pr_1},\chi)}{\Z \to \Z}\).
Further, if \(\chi\) is nested, then \(\mathsf{map}(\mathsf{pr_1},\chi)\) is ternary (\cref{def:ternary}).

\begin{lemma}
\label{lem:ternary-nested}
\thesislit{5}{BoehmVerification}{ternary-nested}\footnote{In \textsc{Agda}, the formalisation of the proof of this lemma currently requires the use of a number of unformalised lemmas about dyadic numbers.}
Given a normalised sequence of ternary interval codes, it is nested if and only if it is ternary.
\end{lemma}

\noindent
As the resulting integer sequence is ternary, it is clearly a ternary Boehm encoding (\cref{def:ternary-boehm-encoding}).
Indeed, the type of ternary interval codes that are nested and positioned are \emph{equivalent} to the type of ternary Boehm encodings \(\T\).

\begin{definition}
\thesislit{5}{BoehmVerification}{to-interval-seq}
We define the inclusion map from ternary Boehm encodings to sequences of ternary interval codes by the following:
\begin{alignat*}{3}
\mathsf{to{\hy}interval{\hy}seq} &: \T \to (\Z \to \Z^2) ,\\
\mathsf{to{\hy}interval{\hy}seq} &(\chi)_n :=  (\chi_n,n).
\end{alignat*}
\end{definition}

\begin{lemma}
\label{lem:T-equiv}
\thesislit{5}{BoehmVerification}{ternary-normalised\urlsimeq\urlT}
The type of nested and normalised sequences of ternary interval codes that are equivalent to the type of ternary Boehm encodings:
\[ \sigmaty{\chi}{\Z \to \Z^2}{ \mathsf{nested}(\mathsf{map}(\mathsf{pr_1},\chi)) \x \mathsf{normalised}(\chi)} \simeq \T. \]
\end{lemma}
\begin{proof} \axioms{f}
Converting from \(\T \to (\Z \to \Z^2)\) is done by \(\mathsf{to{\hy}interval{\hy}seq}\), while converting in the other direction is done by \(\mathsf{map}(\mathsf{pr_1})\); in both cases, we prove the necessary properties of the constructed representation using \cref{lem:ternary-nested}.
It is then trivial to see that these two functions yield identities.
\end{proof}

\begin{corollary}
\label{def:real-T}
\thesislit{5}{BoehmVerification}{\urlllbracket_\urlrrbracket}
Any given ternary Boehm encoding \(\ty{x}{\T}\) represents the real number
\[\ty{\llbracket x \rrbracket := \llparenthesis \lambda n.(x_n,n) \rrparenthesis''}{\R}.\]
\end{corollary}
\begin{proof} \axioms{ft}
By \cref{lem:T-equiv,cor:real-seq-ternary-code}.
\end{proof}

\subsection{Exact real arithmetic}
\label{sec:boehm-era}

Boehm defines arithmetic operations on \texttt{CR} by looking to the interval arithmetic on the dyadic rational intervals~\cite{BoehmAPI}.
In this subsection we follow this same process in order to informally define negation, addition and multiplication on \(\T\).
The definitions here differ to those on \texttt{CR} because elements of \(\T\) are different to those of \texttt{CR} (for example we must satisfy the ternary property, \cref{def:ternary}); furthermore, we give simpler, more convenient definitions than Boehm's, though this is at the expense of some efficiency.

Note that the definitions of this subsection are informal --- we do not define these operations in our formal library and have not verified their correctness; indeed, this is ongoing work as discussed in \cref{fw:boehm-functions}.
We will, however, use these definitions in our own proof-of-concept \textsc{Java} library for exact real search on ternary Boehm encodings in \cref{chap:exact-real-search}.

\subsubsection{Negation}

Negation on dyadic intervals is defined by
\[-\left[\frac{k}{2^{p}},\frac{c}{2^{p}}\right] := \left[\frac{-c}{2^{p}},\frac{-k}{2^{p}}\right].\]
Hence, negation on dyadic interval codes \(\Z^3\) is defined by \(-(k,c,p) := (-c,-k,p)\).
The width of the output interval is the same as that of the input interval, meaning that if the input was in fact a ternary interval code \(\Z^2\), then the output also would be a ternary interval code on the same precision-level.
Therefore, we can easily define negation on ternary Boehm encodings.

\begin{definition}
\emph{Negation} on ternary Boehm encodings is defined by the following:
\begin{alignat*}{3}
- &: \T \to \T, \\
- &x := \lambda n.- x_n - 2.
\end{alignat*}
\end{definition}

\subsubsection{Addition}

Addition on dyadic intervals is defined by
\begin{multline*}
\left[\frac{k_1}{2^{p_1}},\frac{c_1}{2^{p_1}}\right] + \left[\frac{k_2}{2^{p_2}},\frac{c_2}{2^{p_2}}\right] := \\
\left[\frac{2^{p_2-\mathsf{min}(p_1,p_2)}k_1 + 2^{p_1-\mathsf{min}(p_1,p_2)}k_2}{2^{\mathsf{max}(p_1,p_2)}},\frac{2^{p_2-\mathsf{min}(p_1,p_2)}c_1 + 2^{p_1-\mathsf{min}(p_1,p_2)}c_2}{2^{\mathsf{max}(p_1,p_2)}}\right].
\end{multline*}
As we are looking to define this operation on \(\T\), whereon we can approximate the elements to any precision-level we like, we can consider the simpler case where the precision-levels of the inputs are identical.
Hence, addition on dyadic interval codes \(\Z^3\) in the case where the precision-levels of the inputs is identical is defined by \((k_1,c_1,p) + (k_2,c_2,p) :=
(k_1 + k_2, c_1 + c_2, p)\).
The width of the output interval in this case is
\[ \frac{(c_1 + c_2) - (k_1 + k_2)}{2^p}, \]
meaning that given two ternary interval codes on precision-level \(p\) the width would be
\[ \frac{(k_1 + 2 + k_2 + 2) - (k_1 + k_2)}{2^p} = \frac{1}{2^{p-2}} .\]
In order to achieve an output interval approximation of a ternary Boehm encoding at a requested precision-level \(\ty{p}{\Z}\), therefore, the input ternary Boehm encodings must be evaluated to precision-level \(p+2\).

\begin{definition}
\emph{Addition} on ternary Boehm encodings is defined by the following:
\begin{alignat*}{3}
\ &+ &&: \T \to \T \to \T, \\
x &+ y &&:= \lambda n.(x_{n+2} + y_{n+2}).
\end{alignat*}
\end{definition}

\subsubsection{Multiplication}

Multiplication on dyadic intervals is defined by
\[\left[\frac{k_1}{2^{p_1}},\frac{c_1}{2^{p_1}}\right] * \left[\frac{k_2}{2^{p_2}},\frac{c_2}{2^{p_2}}\right] := \left[\frac{\mathsf{min}(k_1k_2,k_1c_2,c_1k_2,k_2c_2)}{2^{p_1+p_2}},\frac{\mathsf{max}(k_1k_2,k_1c_2,c_1k_2,k_2c_2)}{2^{p_1+p_2}}\right].\]
As with addition, we can consider the simpler case where the precision-levels of the inputs are identical.
Hence, multiplication on dyadic interval codes \(\Z^3\) in the case where the precision-levels of the inputs is identical is defined by \((k_1,c_1,p) * (k_2,c_2,p) :=
(\mathsf{min}(k_1k_2,k_1c_2,c_1k_2,k_2c_2),\mathsf{max}(k_1k_2,k_1c_2,c_1k_2,k_2c_2),2p)\).
The width of the output interval varies based on the inputs --- this is related to the fact that multiplication is continuous but not uniformly continuous.
In the case where both intervals are positive, for example, the width of the output interval is
\[ \frac{((k_1+2)(k_2+2))-k_1k_2}{2^{2p}} = \frac{k1+k2+2}{2^{2p-1}} = \frac{1}{2^{2p-1-\mathsf{log2}(k1+k2+2)}}. \]
This case in fact supercedes all others, and so in order to achieve an output interval approximation of a ternary Boehm encoding at a requested precision-level \(\ty{p}{\Z}\), the input ternary Boehm encodings must be evaluated to precision-level \((p + \mathsf{log2}(\mathsf{abs}(k_1)+\mathsf{abs}(k_2)+2) + 1) / 2\).

\begin{definition}
\emph{Multiplication} on ternary Boehm encodings is defined by the following:
\begin{alignat*}{3}
\ &* &&: \T \to \T \to \T, \\
x &* y &&:= \lambda n.(x_{(n + \mathsf{log2}(\mathsf{abs}(x_n)+\mathsf{abs}(y_n)+2) + 1) / 2} + y_{(n + \mathsf{log2}(\mathsf{abs}(x_n)+\mathsf{abs}(y_n)+2) + 1) / 2}).
\end{alignat*}
\end{definition}

\subsection{Representing compact intervals}

A key difference between ternary Boehm encodings \(\T\) and ternary signed-digit encodings \(\K\) is that with the former we can represent real numbers across the real line, whereas with the latter we can only represent reals in a given compact interval (in this thesis, we use the example of \([-1,1]\)).
However, \(\T\) is not a uniformly continuous searchable type --- in order to search the ternary Boehm encodings, we will have to restrict ourselves to particular compact intervals.
In this subsection, we define three subtypes of \(\T\) that can be used to represent real numbers that are in the compact interval \([\frac{k}{2^i},\frac{k+2}{2^i}]\), which is represented by the ternary interval code \(\ty{(k,i)}{\Z^2}\).

The first subtype \(\T(k,i)_1\) is straightforward to define: it is simply the type of ternary Boehm encodings which `pass through' the interval encoded by \((k,i)\) at precision-level \(i\).

\begin{definition}
\thesislit{5}{BoehmVerification}{CompactInterval}
For every ternary interval code \(\ty{(k,i)}{\Z^2}\) there is a subtype \(\T(k,i)_1\) of \emph{ternary Boehm encodings in the compact interval \([\frac{k}{2^i},\frac{k+2}{2^i}]\)} defined as the collection of ternary Boehm encodings that feature the interval approximation \((k,i)\):
\[ \T(k,i)_1 := \sigmaty{x}{\T}{x_i = k} .\]
\end{definition}

\noindent
This subtype is useful because it is easy to map back to elements of \(\T\) by simply taking the first projection.

The second subtype \(\T(k,i)_2\) is a little less convenient to work with directly, but is more convenient for the purpose of defining a closeness function (which we will do in \cref{chap:exact-real-search}).
Elements \(\ty{x}{\T(k,i)_2}\) are \(\N\)-indexed sequences that only give the interval approximations of the represented real number \(\llbracket x \rrbracket\) for precision-levels greater than \(i\) --- i.e.\ for any \(\ty{n}{\N}\), an interval approximation of \(\llbracket x \rrbracket\) is represented by the interval code \(\ty{(x_n,i+1+n)}{\Z^2}\).

\begin{definition}
\label{def:compact-2}
\thesislit{5}{BoehmVerification}{CompactInterval2}
For every ternary interval code \(\ty{(k,i)}{\Z^2}\) there is a subtype \(\T(k,i)_2\) of \emph{ternary Boehm encodings in the compact interval \([\frac{k}{2^i},\frac{k+2}{2^i}]\)} defined as the collection of sequences of type \(\seq\Z\) that locate, after the interval approximation \((k,i)\), any ternary Boehm encoding:
\[ \T(k,i)_2 := \sigmaty{\chi}{\seq\Z}{\below{\chi_0}{k} \x \Pity{n}{\N}{\below{\chi_{n+1}}{\chi_n}}} .\]
\end{definition}

We can convert from each of these definitions to the other.

\begin{definition}
\thesislit{5}{BoehmVerification}{CompactInterval-1-to-2}
Given any ternary interval code \(\ty{(k,i)}{\Z^2}\), the function \[\T(k,i)_1\mathsf{{\hy}to{\hy}}\T(k,i)_2 : \T(k,i)_1 \to \T(k,i)_2\] is defined for every \(\ty{((x,b),e)}{\T(k,i)_1}\) where
\begin{itemize}
\item \(\ty{x}{\Z \to \Z}\),
\item \(\ty{b}{\mathsf{ternary}(x)}\),
\item \(\ty{e}{x_i = k}\).
\end{itemize}

We first define the sequence
\begin{alignat*}{3}
\chi &: \seq\Z, \\
\chi&_n := x_{i+1+n}.
\end{alignat*}
The proof that \((\below{\chi_0}{k})\) is then by \(\ty{b_i}{(\below{x_{i+1}}{x_i})}\) and \(e\); the proof that \(\Pity{n}{\N}{\below{\chi_{n+1}}{\chi_n}}\) is by \(b\).
\end{definition}

\begin{definition}
\label{def:compact-2-to-1}
\thesislit{5}{BoehmVerification}{CompactInterval-2-to-1}
Given any ternary interval code \(\ty{(k,i)}{\Z^2}\), the function \[\T(k,i)_2\mathsf{{\hy}to{\hy}}\T(k,i)_1 : \T(k,i)_2 \to \T(k,i)_1\] is defined for every \(\ty{(\chi,b_0,b_s)}{\T(k,i)_2}\) where
\begin{itemize}
\item \(\ty{\chi}{\seq\Z}\),
\item \(\ty{b_0}{\below{\chi_0}{k}}\),
\item \(\ty{b_s}{\Pity{n}{\N}{\below{\chi_{n+1}}{\chi_n}}}\).
\end{itemize}

\vspace{0.25cm}
We first define the sequence
\begin{alignat*}{3}
x &: \Z \to \Z,\\
x&_n := k \ \ & \text{\emph{(when \(n = i\))}}, \\
x&_n := \chi_{n-1-i} \ \ & \text{\emph{(when \(n > i\))}}, \\
x&_n := \mathsf{upRight}^{}(k) \ \ & \text{\emph{(when \(n < i\))}}.
\end{alignat*}

The proof that \((\mathsf{ternary}(x))\) is then by \(b_0\) and \(b_s\); the proof that \(x_n = k\) is immediate.
\end{definition}

Despite the fact that, given a particular interval code \(\ty{(k,i)}{\Z^2}\), the types \(\T(k,i)_1\) and \(\T(k,i)_2\) can represent the same real numbers, the types are not equivalent.
This is because \(\T(k,i)_1\) contains more information about locating the represented real number than \(\T(k,i)_2\), and this information cannot be recovered when converting from the latter to the former --- indeed, in \cref{def:compact-2-to-1} we used the \(\mathsf{upRight}\) function to prepend arbitrary interval approximations of the real number for precision-levels higher than \(i\).
Despite this, it \emph{is} the case that the composition of the two functions defined above preserves the represented real number.
The informal idea of this is that from precision-level \(\ty{i}{\Z}\) the ternary Boehm encoding (and thus, the underlying dyadic interval approximations) are identical --- hence, the limit of the two sequences is the same real number.

We now connect the two representations of real numbers in compact intervals that we use in this thesis, and show that the second subtype defined in this subsection is equivalent to the type of ternary Boehm encodings.

\begin{lemma}
\label{thm:Compact-2-simeq-K}
\thesislit{5}{BoehmVerification}{CompactInterval2-ternary}
Given any ternary interval code \(\ty{(k,i)}{\Z^2}\), it is the case that \(\T(k,i)_2 \simeq \K\).
\end{lemma}
\begin{proof}[Proof (Sketch).] \axioms{f}
Although the formalisation is rather involved, the intuition here is clear. 
We convert any ternary Boehm encoding in a compact interval \(\ty{\chi}{\T(k,i)_2}\) into a ternary signed-digit encoding \(\ty{\alpha}{\K}\) by using the values of \(\ty{b_0}{\below{\chi_0}{k}}\) and \(\ty{b_s}{\pitye{n}{\N}\below{\chi_{n+1}}{\chi_n}}\).
If, at any point \(\ty{n}{\N}\) in the sequence, \(\chi_n\) goes \(\mathsf{downLeft}\) then \(\alpha_n = \mone\); if \(\chi_n\) goes \(\mathsf{downMid}\) then \(\alpha_n = 0\); and else if \(\chi_n\) goes \(\mathsf{downRight}\) then \(\alpha_n = 1\).
Following the opposite method, we can also convert in the other direction; and clearly either composition of these conversions gives an identity.
\end{proof}

Finally, we give the third subtype for representing compact intervals using \(\T\).
We can remove the redundant interval approximations from the structure of \(\T(k,i)_2\), which reflects the ability to represent the same real number as any ternary signed-digit encoding \(\K\) as a function \(\N \to \2\) (discussed in \cref{sec:signed-digits-background}).

\begin{definition}
\thesislit{5}{BoehmVerification}{_split-below_}
Given two integers \(\ty{n,m}{\Z}\), we say that \(n\) is \emph{split-below} \(m\) if \(n\) is either \(\mathsf{downLeft} \ m\) or \(\mathsf{downRight} \ m\):
\begin{alignat*}{3}
& \mathsf{split{\hy}below} &&: \Z \to \Z \to \Omega,\\
n \ &\mathsf{split{\hy}below} \ m &&:= (n = \mathsf{downLeft} \ m) + (n = \mathsf{downRight} \ m).
\end{alignat*}
\end{definition}

\begin{definition}
\label{def:compact-3}
\thesislit{5}{BoehmVerification}{CompactInterval3}
For every ternary Boehm interval code \(\ty{(k,i)}{\Z^2}\) there is a subtype \(\T(k,i)_3\) of \emph{ternary Boehm encodings in the compact interval \([\frac{k}{2^i},\frac{k+2}{2^i}]\)}  defined as the collection of sequences of type \(\seq\Z\) that locate, after the interval approximation \((k,i)\), only those ternary Boehm encodings that never use \(\mathsf{downMid}\):
\[ \T(k,i)_3 := \sigmaty{\chi}{\N \to \Z}{\belowsplit{\chi_0}{k} \x \Pity{n}{\N}{\belowsplit{\chi_{n+1}}{\chi_n}}} .\]
\end{definition}

Therefore, this subtype defined is equivalent to the Cantor type \(\seq \2\), as at each interval approximation there are two choices for the next.

\begin{lemma}
\label{thm:Compact-3-simeq-Cantor}
\thesislit{5}{BoehmVerification}{CompactInterval3-cantor}
Given any ternary Boehm interval code \(\ty{(k,i)}{\Z^2}\), it is the case that \(\T(k,i)_3 \simeq \seq\2\).
\end{lemma}
\begin{proof} \axioms{f}
Similar to \cref{thm:Compact-2-simeq-K}.
\end{proof}

\noindent
Split-belowness trivially implies belowness, and hence every element of \(\T(k,i)_3\) can be mapped to one in \(\T(k,i)_2\) --- and, hence, one in \(\T(k,i)_3\) --- all of which represent the same real number.

As discussed in \cref{sec:signed-digits-background}, there is a problem when using the Cantor type \(\seq\2\) for representing real numbers: namely, it is unsuitable for computation due to its lack of redundant digits. This problem extends to the equivalent subtype \(\T(k,i)_3\), and so it is important to understand that we do not (and indeed cannot) use this type for computation in our framework. The reason for its introduction here is instead for the purposes of \emph{search} --- in \cref{sec:boehm-suitable}, we will see that it is this subtype which is most suitable for searching the type of ternary Boehm encodings in the given compact interval \(\left[\frac{k}{2^i},\frac{k+2}{2^i}\right]\) (for any \(\ty{(k,i)}{\Z^2}\)). Indeed, this further means that the type \(\seq\2\) is suitable for searching ternary signed-digit encodings, as we explore in \cref{sec:K-suitable}.
\chapter{Exact Real Search}
\label{chap:exact-real-search}

Our generalised framework for search, optimisation and regression introduced in \cref{chap:searchable,chap:generalised} has been shown to operate on a wide class of types.
Optimisation takes place on totally bounded (\cref{def:totallybounded}) closeness spaces (\cref{def:cspace}) equipped with a preorder (\cref{def:preorder}) and approximate linear preorder (\cref{def:approx-order}); whereas search takes place on uniformly continuously searchable (\cref{def:c-searchable}) closeness spaces --- further, recall that totally boundedness implies uniformly continuous searchability (\cref{thm:tb-csearch}).
Regression can take place on either of the above classes of closeness space, depending on whether we perform it via optimisation (\cref{reg:min}) or search (\cref{th:perfect,th:imp}).

In this thesis' final chapter, we bring this generalised framework full circle by instantiating it on the two representations of real numbers we explored and verified in \cref{chap:reals}: ternary signed-digits and ternary Boehm encodings.
In order to achieve this, we formalise that each type's representations of real numbers in compact intervals --- i.e.\ \(\K\) itself for the former and the subtype \(\T(k,i)_3\) for any given \(\ty{(k,i)}{\Z^2}\) (\cref{def:compact-2}) for the latter --- are members of the aforementioned class of types, in that we can perform search, optimisation and regression on them.
We then formalise, only for \(\K\), that the exact real arithmetic functions we wish to search, optimise and regress are indeed uniformly continuous functions --- proofs that are required for the algorithms to run.
Finally, we give a variety of toy, proof-of-concept examples of search, optimisation and regression using these functions.

On signed-digit encodings, our example \textsc{Haskell} algorithms are compiled directly from instantiations of the formal \textsc{Agda} proofs of convergence (we explain how in the final paragraph of \cref{appendix:agda}). This means that, while the extracted computation is relatively inefficient, we have formalised the fact that it will compute a correct answer.
Meanwhile, our examples of the framework working on ternary Boehm encodings are written in a \textsc{Java} library which informally reflects the ideas of the formal framework; we sacrifice direct proofs of convergence in order to attain more efficient proof-of-concept algorithms.

\section{Exact Real Search using signed-digit encodings}

Some of the work of this section was previously published as part of a joint paper with Dan R. Ghica at the \emph{Logic in Computer Science (LICS) 2021} conference~\cite{Todd21}.

\subsection{Suitability for search, optimisation and regression}
\label{sec:K-suitable}

In this subsection, we show that the type of ternary signed digit encodings \(\K\) --- explored in \cref{sec:signed-digits} --- is a member of our class of types in that we can perform search, optimisation and regression upon it.
We show, therefore, that \(\K\) is a totally bounded and uniformly continuously searchable closeness space with a preorder and approximate linear preorder.

\subsubsection{\(\K\) is a continuously searchable closeness space}

The closeness space on \(\K\) is a discrete-sequence closeness space (\cref{def:disseq-closeness}), as the type \(\3\) is finite and, therefore, discrete.
Further, by the finiteness of \(\3\), this closeness space is totally bounded.

\begin{corollary}
\label{cor:signed-digit-tb}
\thesislit{6}{SignedDigitSearch}{\urlthree\urlsuperscriptN-totally-bounded}
The type of ternary signed-digit encodings \(\K\) is a totally bounded closeness space.
\end{corollary}
\begin{proof} \axioms{f}
Because \(\3\) is finite and discrete, the result follows from \cref{cor:disseq-cspace}.~\qedhere
\end{proof}

There are two proofs of uniformly continuous searchability we can employ for \(\K\).
The first is directly from the above fact that \(\K\) is totally bounded: this yields the \emph{totally bounded uniformly continuous searcher}.

\begin{corollary}
\thesislit{6}{SignedDigitSearch}{\urlthree\urlsuperscriptN-csearchable-tb}
\thesislit{6}{SignedDigitSearch}{\urlthree\urlsuperscriptN-csearchable}
\label{cor:signed-digit-ucsearch}
The closeness space of ternary signed-digit encodings \(\K\) is uniformly continuously searchable.
\end{corollary}
\begin{proof}[Proof 1]
By \cref{cor:signed-digit-tb,thm:tb-csearch}.
\end{proof}

\noindent
The second proof is inspired instead by Berger's searcher on the Cantor space and Escard\'o's on countable product spaces (which we formalised in \cref{sec:tychonoff}): this yields the \emph{decreasing-modulus uniformly continuous searcher}.

\begin{proof}[Proof 2]
By \cref{cor:disseq-csearch}.
\end{proof}

\noindent
In the examples given in \cref{sec:K-examples}, it will become apparent that the totally bounded searcher is almost always more efficient than the decreasing-modulus searcher; specifically, the decreasing-modulus searcher is only better in \cref{ex:K-reg-1-search-imperfect}.

By combining the proofs of \cref{cor:signed-digit-ucsearch}
with the examples in \cref{sec:cspace-examples,sec:csearch-examples,sec:tychonoff}, we further have (among other things) that arbitrary products of \(\K\) are uniformly continuously searchable and totally bounded closeness spaces.
However, we do not usually wish to search \(\K\) \emph{directly}, because (as discussed in \cref{sec:signed-digits}) there are infinitely-many redundant representations of any number that features a \(0\), and so a direct searcher will have inherent inefficiencies.
Instead, we search \(\K\) \emph{indirectly} using the Cantor type \(\seq\2\).
Recall that every real number in \([-1,1]\) can be represented as a sequence \(\N \to \{-1,1\}\) and, hence, by using \(0\) for \(-1\), as a sequence \(\seq\2\).
We can therefore trivially define the following inclusion map from \(\seq\2\) to \(\K\):

\begin{definition}
\thesislit{6}{SignedDigitSearch}{_\urluparrow}
\label{def:uparrow-map}
We define the inclusion map \(\ty{-^\uparrow}{\seq\2 \to \K}\) that converts representations of \([-1,1]\) as Cantor sequences \(\seq\2\) to ternary signed-digit encodings of the same number by mapping every \(0\) to \(\mone\).
\end{definition}

\noindent
Using this map, which is trivially uniformly continuous, we can convert a uniformly continuous function \(\ty{f}{\K \to X}\), for any closeness space \(X\), into a uniformly continuous function \(\ty{f^\uparrow}{\seq\2 \to X}\) with the same modulus of uniform continuity. 
For our examples, therefore, we usually convert the uniformly continuous predicates and functions on \(\K\) that we wish to search, optimise and regress into ones on \(\seq\2\) that we can search, optimise and regress indirectly using uniformly continuous searchers on \(\seq\2\).

\begin{corollary}
\label{cor:cantor-tb}
\thesislit{6}{SignedDigitSearch}{\urltwo\urlsuperscriptN-totally-bounded}
The Cantor type \(\seq\2\) is a totally bounded closeness space.
\end{corollary}
\begin{proof} \axioms{f}
Because \(\2\) is finite and discrete, the result follows from \cref{cor:disseq-cspace}.~\qedhere
\end{proof}

\begin{corollary}
\label{cor:cantor-ucsearch}
\thesislit{6}{SignedDigitSearch}{\urltwo\urlsuperscriptN-csearchable-tb}
\thesislit{6}{SignedDigitSearch}{\urltwo\urlsuperscriptN-csearchable}
The Cantor closeness space \(\seq\2\) is uniformly continuously searchable.
\end{corollary}
\begin{proof}[Proof 1]
By \cref{cor:cantor-tb,thm:tb-csearch}.
\end{proof}
\begin{proof}[Proof 2]
By \cref{cor:disseq-csearch}.
\end{proof}

\subsubsection{\(\K\) has an approximate linear preorder}

A preorder on \(\K\) is given by the lexicographic order (\cref{lem:lexicographic-preorder}), whereas an approximate linear preorder is given by the approximate lexicographic order (\cref{lem:lexicographic-approx}).
Recall that these are defined by the following:
\begin{alignat*}{3}
\alpha &\leq_{\seq \3} \beta &&:= \Pity{n}{\N}{\alpha \sim^n \beta \to \alpha_n \leq_D \beta_n}, \\
\alpha &\leq^\varepsilon_{\seq \3} \beta &&:= \Pity{n}{\N}{n < \varepsilon \to \alpha \sim^n \beta \to \alpha_n \leq_D \beta_n}.
\end{alignat*}

\noindent
However, it is not appropriate to use these orders for optimisation and regression on ternary signed-digit encodings, because the lexicographic ordering of the representations of type \(\K\) does not amount to the numerical ordering of the real numbers of type \(\I\).
As an example, consider the following different representations \(\ty{z\alpha,z\beta}{\K}\) of \(\ty{0}{\I}\):
\[ z\alpha := 01\mone\mone\mone\mone\mone\mone\ldots \text{ and } z\beta := 1\mone\mone\mone\mone\mone\mone\mone\ldots .\]
By the above lexicographic order on \(\K\), it is the case that \(z\alpha \leq z\beta\) but \(\neg (z\beta \leq z\alpha)\). This issue also propagates to the approximate ordering.
Escard\'o has shown that the lexicographic ordering on signed-digit encodings and the ordering on the encoded numbers can indeed coincide by introducing a normalisation operator \(\ty{\mathsf{onorm}}{\K \x \K \to \K \x \K}\)~\cite{Escardo98}.
The idea is that the normalised pair represents the same numbers (i.e.\ \(\llangle \alpha \rrangle = \llangle \mathsf{pr_1}(\mathsf{onorm}(\alpha,\beta)) \rrangle\) and \(\llangle \beta \rrangle = \llangle \mathsf{pr_2}(\mathsf{onorm}(\alpha,\beta)) \rrangle\)) but that, on the normalised pair, the lexicographic and numerical orderings coincide.

We take a simpler approach to the problem, and utilise the equivalence between \(\K\) and the type \(\T(-1,0)_2\) of ternary Boehm encodings used for representing the interval \([-1,1]\). This equivalence, given in \cref{thm:Compact-2-simeq-K}, yields a `normalisation' map \[\ty{\mathsf{integer{\hy}approx}}{\K \to (\N \to \Z)} ,\] where the \(n\)\textsuperscript{th} integer \(k\) of the sequence \(\mathsf{integer{\hy}approx}(\alpha)\), for any \(\ty{\alpha}{\K}\), refers to the \(n\)\textsuperscript{th} ternary interval code (\cref{def:ternary-interval-code}) approximation of \(\alpha\) --- i.e., the interval \(\left[\frac{k}{2^n},\frac{k+2}{2^n}\right]\), in which \(\llangle \alpha \rrangle\) lies.
Any two signed-digit encodings whose \(n\)-prefixes evaluate to the same interval approximations will therefore have the same value of \(k\). 

\begin{figure}
    \centering
    \includegraphics[width=6cm]{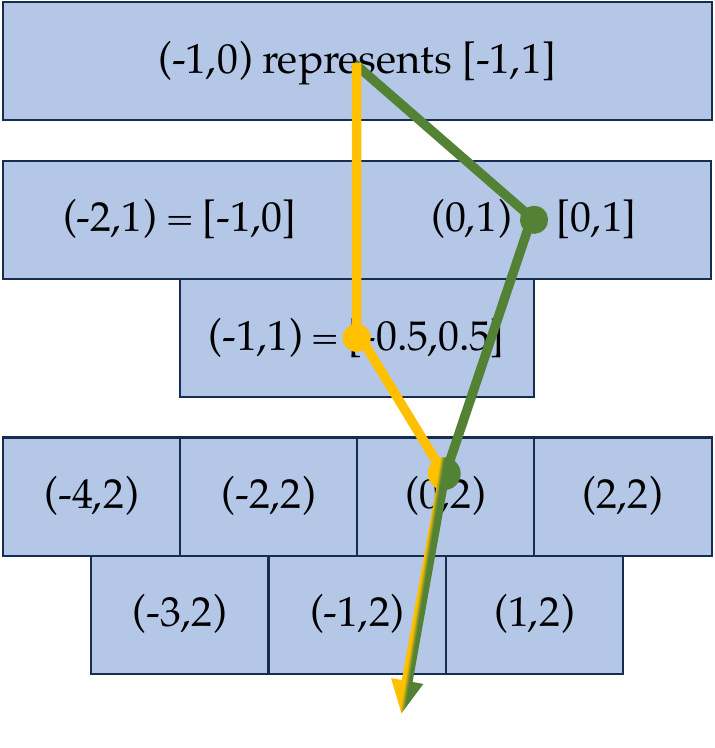}
    \caption{Illustration of the interval approximation evaluations of {\color{orange} \(z\alpha := 01\mone\mone\mone\mone\mone\mone\ldots\)} and {\color{ForestGreen} \(z\beta := 1\mone\mone\mone\mone\mone\mone\mone\ldots\)} at the first three precision-levels.
    }
    \label{fig:norm}
\end{figure}

The formal definition of the normalisation map is given in \cref{def:integer-approx}, and its use is informally visualised in \cref{fig:norm}, which shows the relationship between the evaluations of \(z\alpha\) and \(z\beta\) and their normalised interval approximations.
From this illustration, we can read off that
\[ \mathsf{integer{\hy}approx}(z\alpha) := -1,-1,0,0,\ldots \text{ and } \mathsf{integer{\hy}approx}(z\beta) := -1,0,0,0,\ldots .\]

\begin{definition}
\label{def:integer-approx}
\thesislit{6}{SignedDigitOrder}{integer-approx}
The \emph{normalisation map} \(\ty{\mathsf{integer{\hy}approx}}{\K \to (\N \to \Z)}\), which maps every ternary signed-digit encoding \(\ty{\alpha}{\K}\) to a sequence of integers such that, for any \(\ty{\alpha}{\K}\), \(\mathsf{integer{\hy}approx}(\alpha)_n\) refers to the \(n\)\textsuperscript{th} ternary interval code (\cref{def:ternary-interval-code}) approximation of \(\alpha\) is defined as follows:
\begin{alignat*}{3}
\mathsf{\3{\hy}to{\hy}down} &: \3 &&\to (\Z \to \Z),\\
\mathsf{\3{\hy}to{\hy}down} &\ \overline{1} &&:= \mathsf{downLeft}, \\
\mathsf{\3{\hy}to{\hy}down} &\ 0 &&:= \mathsf{downMid}, \\
\mathsf{\3{\hy}to{\hy}down} &\ 1 &&:= \mathsf{downRight},
\end{alignat*}
\vspace{-1cm}
\begin{alignat*}{3}
\mathsf{integer{\hy}approx}' &: \Z \to \K \to &&\ \N \to \Z,\\
\mathsf{integer{\hy}approx}' &(k,\alpha,0) && := k,\\
\mathsf{integer{\hy}approx}' &(k,\alpha,n + 1) &&:= \mathsf{integer{\hy}approx}' (\mathsf{\3{\hy}to{\hy}down} (\mathsf{head} \ \alpha, k),\mathsf{tail} \ \alpha,n),
\end{alignat*}
\vspace{-1cm}
\begin{alignat*}{3}
\mathsf{integer{\hy}approx} &: &&\K \to \N \to \Z,\\
\mathsf{integer{\hy}approx} &\ \alpha && := \mathsf{integer{\hy}approx}' (\alpha,-1).
\end{alignat*}
\end{definition}

Using the normalisation map, we define the following preorder and approximate linear preorder on ternary signed-digit encodings \(\K\), which \emph{are} correct with respect to the numerical ordering on the interval object \(\I\). However, the order on \(\I\) is not well established, and we have not formalised its notion or verified the correctness of the preorder --- we discuss this as further work in \cref{fw:io}.

\begin{definition}
\label{def:real-order-preorder}
\thesislit{6}{SignedDigitOrder}{RealPresOrder._\urlleq\urlthree\urlsuperscriptN_}
We define the \emph{real-order preserving linear preorder} on \(\K\) by using the normalisation map \(\ty{\mathsf{integer{\hy}approx}}{\K \to (\N \to \Z)}\):
\begin{alignat*}{3}
\ & \leq_\K && : \relation{\K}, \\
\alpha & \leq_\K \beta && := \existsty{n}{\N}{\Pity{i}{\N}{n \leq i \to \mathsf{integer{\hy}approx}(\alpha)_i \leq \mathsf{integer{\hy}approx}(\beta)_i}}.
\end{alignat*}
\end{definition}

\begin{lemma}
\thesislit{6}{SignedDigitOrder}{RealPresOrder.\urlleq\urlthree\urlsuperscriptN-is-preorder}
The real-order preserving preorder is indeed a preorder (\cref{def:preorder}).
\end{lemma}
\begin{proof} \axioms{t}
Using propositional truncation (recall this from \cref{sec:prop-trunc}), we prove both reflexivity and transitivity of the real-order preserving preorder:
\begin{enumerate}[(i)]
\item For any \(\ty{\alpha}{\K}\), we have that \(\mathsf{integer{\hy}approx}(\alpha)_i \leq \mathsf{integer{\hy}approx}(\beta)_i\) by reflexivity of the order on the integers. Therefore, by setting \(n := 0\) we have a proof term \(\ty{(0,p)}{\Sigmaty{n}{\N}{\Pity{i}{\N}{n \leq i \to \mathsf{integer{\hy}approx}(\alpha)_i \leq \mathsf{integer{\hy}approx}(\beta)_i}}}\). The result then follows by truncating this proof term; i.e. \(\ty{| (0,p) |}{\existsty{n}{\N}{\Pity{i}{\N}{n \leq i \to \mathsf{integer{\hy}approx}(\alpha)_i \leq \mathsf{integer{\hy}approx}(\beta)_i}}}\).
\item Given \(\ty{\alpha,\beta,\zeta}{\K}\) such that \(\alpha \leq_\K \beta\) and \(\beta \leq_\K \zeta\), we wish to show that \(\alpha \leq_\K \zeta\). To do this, we define a function \(\ty{f}{\Sigmatye{(n,m)}{\N \x \N}{} {(\Pity{i}{\N}{n \leq i \to \mathsf{integer{\hy}approx}(\alpha)_i \leq \mathsf{integer{\hy}approx}(\beta)_i}}}\) \(\x (m \leq i \to \mathsf{integer{\hy}approx}(\beta)_i \leq \mathsf{integer{\hy}approx}(\zeta)_i )) \to \existstye{k}{\N}{}(\Pitye{i}{\N}{(k \leq i \to}\) \({\mathsf{integer{\hy}approx}(\alpha)_i \leq \mathsf{integer{\hy}approx}(\zeta)_i}))\), and the conclusion is given by the truncation of \(f\).
To define this function, we simply set \(k := \mathsf{max}(n,m)\), and the proof that \(\mathsf{integer{\hy}approx}(\alpha)_i \leq \mathsf{integer{\hy}approx}(\zeta)_i\) follows by transitivity of the order on the integers.
\end{enumerate}
\end{proof}

\begin{definition}
\label{def:real-order-approx-order}
\thesislit{6}{SignedDigitOrder}{_\urlleq\urlsuperscriptn\urlthree\urlsuperscriptN_}
We define the \emph{real-order preserving approximate linear preorder} on \(\K\) by using the normalisation map \(\ty{\mathsf{integer{\hy}approx}}{\K \to (\N \to \Z)}\):
\begin{alignat*}{3}
\ & \leq^-_\K && : \Nrelation{\K}, \\
\alpha & \leq^n_\K \beta && := \mathsf{integer{\hy}approx}(\alpha)_n \leq \mathsf{integer{\hy}approx}(\beta)_n.
\end{alignat*}
\end{definition}

\begin{lemma}
\thesislit{6}{SignedDigitOrder}{\urlleq\urlsuperscriptn\urlthree\urlsuperscriptN-is-approx-order}
The real-order preserving approximate linear preorder is indeed an approximate linear preorder (\cref{def:approx-order}).
\end{lemma}
\begin{proof}
The proof that, given any \(\ty{\varepsilon}{\N}\), the relation \(\ty{\leq^\varepsilon_\K}{\relation \K}\) is decidable and a linear preorder is immediate from the fact that the order on the integers is a linear preorder.

\vspace{1em}
To prove that if \(C_\varepsilon(\alpha,\beta)\) then \(\alpha \leq^\varepsilon_\K \beta\) for all \(\ty{\alpha,\beta}{\K}\), we need to show that \(\mathsf{integer{\hy}approx}(\alpha)_\varepsilon \leq \mathsf{integer{\hy}approx}(\beta)_\varepsilon\). We first prove the following lemma:
\[ \alpha \sim^\varepsilon \beta \to \Pity{k}{\Z}{\mathsf{integer{\hy}approx}'(k,\alpha,\varepsilon) = \mathsf{integer{\hy}approx}'(k,\beta,\varepsilon)} .\]

We prove this lemma by induction on the given \(\varepsilon\). In the base case where \(\varepsilon := 0\), then by definition of \(\mathsf{integer{\hy}approx}'\) (see \cref{def:integer-approx}) we simply have to show that \(k = k\).
In the inductive case where \(\varepsilon := \varepsilon' + 1\) for some \(\ty{\varepsilon'}{\N}\), then from \(\alpha \sim^{\varepsilon' + 1} \beta\) we have that \(\mathsf{head}\ \alpha = \mathsf{head}\ \beta\) and therefore \(\3\mathsf{{\hy}to{\hy}down}(\mathsf{head}\ \alpha,k) = \3\mathsf{{\hy}to{\hy}down}(\mathsf{head}\ \beta,k)\); the result then follows by the inductive hypothesis.

\vspace{1em}
Now that we have proved the lemma, the result follows from it, \cref{lem:disseq-C-equiv} and the reflexivity of the order on the integers.
\end{proof}

\begin{lemma}
\thesislit{6}{SignedDigitOrder}{RealPresOrder-Relates.\urlleq\urlsuperscriptn\urlthree\urlsuperscriptN-relates}
The real-order preserving approximate linear preorder relates (by \cref{def:approx-order-relates}) to the real-order preserving preorder.
\end{lemma}
\begin{proof} \axioms{t}
The first condition, that if \(\alpha \leq_\K \beta\) then \(\existstye{n}{\N}{\Pitye{\varepsilon}{\N}{(n < \varepsilon \to \alpha \leq^\varepsilon_\K \beta)}}\), is immediate, because the types are identical.
The second condition, that if \(x \leq^n_\K y\) holds for all \(\ty{n}{\N}\) then \(x \leq_\K y\) follows by truncating the proof term of type \(\Sigmatye{n}{\N}{\Pitye{\varepsilon}{\N}{(n < i \to \alpha \leq^i_\K \beta})}\), which is constructed by setting \(n := 0\) and using the proof \(\ty{f(i)}{x \leq^i_\K y}\).
\end{proof}

\subsection{Uniformly continuous exact real arithmetic}
\label{sec:K-ucont}

In order to perform exact real search on ternary signed-digit encodings, we must prove the uniform continuity of the exact real arithmetic operations.
In this subsection, we prove that negation, binary midpoint, infinitary midpoint and multiplication functions (all defined in \cref{sec:signed-digits-era}) are uniformly continuous via the closeness space \(\K\).

\subsubsection{Sequence uniform continuity}

Recall that discrete-sequence closeness spaces satisfy, for all \(\ty{n}{\N}\), \(C_n(\alpha,\beta) \Leftrightarrow \alpha \sim^n \beta\) (\cref{cor:disseq-C-equiv}); a fact which we use to give the following bespoke definitions of uniform continuity on sequence functions and predicates on discrete types.

\begin{definition}
\thesislit{6}{SequenceContinuity}{seq-f-ucontinuous\urlsuperscriptone}
\label{def:dis-seq-ucont}
Given discrete types \(X\) and \(Y\), a unary sequence function \(\ty{f}{\seq X \to \seq Y}\) is \emph{uniformly continuous} if for all \(\ty{\varepsilon}{\N}\) there is some \(\ty{\delta}{\N}\) such that sequences that agree in their \(\delta\)-prefixes map by \(f\) to sequences that agree in their \(\varepsilon\)-prefixes:
\begin{multline*}
\mathsf{seq{\hy}f{\hy}ucontinuous}^1(f) := \\
\Pitye{\varepsilon}{\N} \Sigmatye{\delta}{\N}{ \Pitye{x_1,x_2}{\seq X} ( x_1 \sim^\delta x_2 ) \to  ( f(x_1) \sim^\varepsilon f(x_2) ) } .
\end{multline*}
\end{definition}

\begin{lemma}
\thesislit{6}{SequenceContinuity}{seq-f-ucontinuous\urlsuperscriptone-\urlleftrightarrow-closeness}
Given discrete types \(X\) and \(Y\), a unary sequence function \(\ty{f}{\seq X \to \seq Y}\) is uniformly continuous by \cref{def:dis-seq-ucont} if and only if it is uniformly continuous by \cref{def:clos-ucont} from the discrete-sequence closeness space yielded by \(X\) to the discrete-sequence closeness space yielded by \(Y\).
\end{lemma}
\begin{proof}
By \cref{cor:disseq-C-equiv}.
\end{proof}

We next define the binary version of discrete-sequence uniform continuity.

\begin{definition}
\label{def:dis-seq-ucont-bin}
\thesislit{6}{SequenceContinuity}{seq-f-ucontinuous\urlsuperscripttwo}
Given discrete types \(X\), \(Y\) and \(Z\), a binary sequence function \(\ty{f}{\seq X \to \seq Y \to \seq Z}\) is \emph{uniformly continuous} if for all \(\ty{\varepsilon}{\N}\) there are some \(\ty{\delta_1,\delta_2}{\N}\) such that, given two pairs of sequence arguments, if the first pair agree in their \(\delta_1\)-prefixes and the second pair agree in their \(\delta_2\)-prefixes then together they map by \(f\) to sequences that agree in their \(\varepsilon\)-prefixes:
\begin{multline*}
\mathsf{seq{\hy}f{\hy}ucontinuous}^2(f) :=
\Pitye{\varepsilon}{\N} \Sigmatye{(\delta_1,\delta_2)}{\N \x \N} \Pitye{x_1,x_2}{\seq X} \Pitye{y_1,y_2}{\seq Y} \\ 
\hspace{2cm} ( x_1 \sim^{\delta_1} x_2 ) \to  ( y_1 \sim^{\delta_2} y_2 ) \to ( f(x_1,y_1) \sim^\varepsilon f(x_2,y_2) ) .
\end{multline*}
\end{definition}

\noindent
Note that although the proof of uniform continuity of binary sequence functions by the above definition is logically equivalent to that using the definition of uniform continuity on the binary product (\cref{def:prod-closeness}) of two discrete-sequence closeness functions, the former has two moduli of continuity (one for each argument) whereas the latter has only one (i.e.\ the maximum of the two). This means that the computational content of the proofs differ.

\begin{lemma}
\thesislit{6}{SequenceContinuity}{seq-f-ucontinuous\urlsuperscripttwo-\urlleftrightarrow-closeness}
Given discrete types \(X\), \(Y\) and \(Z\), a binary sequence function \(\ty{f}{\seq X \to \seq Y \to \seq Z}\) is uniformly continuous by \cref{def:dis-seq-ucont-bin} if and only if it is uniformly continuous by \cref{def:clos-ucont} from the binary product of the discrete-sequence closeness spaces yielded by \(X\) and \(Y\) to the discrete-sequence closeness space yielded by \(Z\).
\end{lemma}
\begin{proof}
Once again we use the equivalence of \(C_n(\alpha,\beta)\) and \(\alpha \sim^n \beta\) for all \(\ty{\alpha,\beta}{\seq D}\) (where \(D\) is discrete) and \(\ty{n}{\N}\); i.e.\ \cref{cor:disseq-C-equiv}.

\vspace{0.25cm}
We first show sequence uniform continuity implies closeness uniform continuity. 
By the former, given \(\ty{\varepsilon}{\N}\), \(\ty{x_1,x_2}{\seq X}\) and \(\ty{y_1,y_2}{\seq Y}\), we have that there are \(\ty{\delta_1,\delta_2}{\N}\) such that if \(C_{\delta_1}(x_1,x_2)\) and \(C_{\delta_2}(y_1,y_2)\) then \(C_\varepsilon(f(x_1,y_1),f(x_2,y_2))\).
Therefore we take the modulus of uniform continuity for the closeness space definition to be \(\ty{\delta := \mathsf{max}(\delta_1,\delta_2)}{\N}\), because (by \cref{lem:x-C,cor:C-mono}) \(C_\delta((x_1,y_1),(x_2,y_2))\) implies \(C_{\delta_1}(x_1,x_2)\) and \(C_{\delta_2}(y_1,y_2)\) and, hence, \(C_\varepsilon(f(x_1,y_1),f(x_2,y_2))\) as desired.

\vspace{0.25cm}
The other direction is more straightforward. Given \(\ty{\varepsilon}{\N}\), \(\ty{x_1,x_2}{\seq X}\) and \(\ty{y_1,y_2}{\seq Y}\), we have that there is \(\ty{\delta}{\N}\) such that if \(C_{\delta}((x_1,y_1),(x_2,y_2))\) then \(C_\varepsilon(f(x_1,y_1),f(x_2,y_2))\).
Therefore we take the moduli of uniform continuity for the sequence definition to be \(\ty{(\delta,\delta)}{\N \x \N}\), because (by \cref{lem:x-C}) \(C_\delta((x_1,y_1),(x_2,y_2))\) implies \(C_\delta(x_1,x_2)\) and \(C_\delta(y_1,y_2)\) and, hence, \(C_\varepsilon(f(x_1,y_1),f(x_2,y_2))\) as desired.
\end{proof}

Lastly, we define a specialisation of uniform continuity on certain infinitary sequence functions that on discrete types is logically equivalent to uniform continuity on the dependent product (\cref{def:prod-closeness}) of a discrete-sequence closeness type.

\begin{definition}
\label{def:dis-seq-ucont-inf}
\thesislit{6}{SequenceContinuity}{seq-f-ucontinuous\urlsuperscriptN}
Given discrete types \(X\) and \(Y\), an infinitary sequence function \(\ty{f}{\seq{(\seq X)} \to \seq Y}\) is \emph{uniformly continuous} if
for all \(\ty{\varepsilon}{\N}\) there are \(\ty{d,\delta}{\N}\) such that infinitary sequences whose \(d\)-prefixes agree in their respective \(\delta\)-prefixes map by \(f\) to sequences that agree in their \(\varepsilon\)-prefixes:
\begin{multline*}
\mathsf{seq{\hy}f{\hy}ucontinuous}^\N(f) := \Pitye{\varepsilon}{\N} \Sigmatye{(d,\delta)}{\N \x \N} \\
 \big( (d \leq \delta) \x (\Pitye{\alpha,\beta}{\seq{(\seq X)}} \Pitye{n}{\N} ( n < d ) \to  ( \alpha_n \sim^\delta \beta_n ) \to ( f(\alpha) \sim^\varepsilon f(\beta) )) \big) .
\end{multline*}
\end{definition}

\begin{lemma}
\thesislit{6}{SequenceContinuity}{seq-f-ucontinuous\urlsuperscriptN-\urlleftrightarrow-closeness}
Given discrete types \(X\) and \(Y\), an infinitary sequence function \(\ty{f}{\seq{(\seq X)} \to \seq Y}\) is uniformly continuous by \cref{def:dis-seq-ucont-inf} if and only if it is uniformly continuous by \cref{def:clos-ucont} from the countable product of the discrete-sequence closeness space yielded by \(X\) to the discrete-sequence closeness space yielded by \(Y\).
\end{lemma}
\begin{proof}
We first require a relationship similar to \cref{cor:disseq-C-equiv} but for sequences \(\ty{\alpha,\beta}{\seq{(\seq X)}}\).
Recall from the definition of countable product closeness functions (\cref{def:pi-closeness}) and the accompanying visualisation of the diagonalisation argument (\cref{fig:diagonal}) that for all \(\ty{\delta}{\N}\) having \(C_\delta(\alpha,\beta)\) means we have \(\alpha_0 \sim^{\delta} \beta_0\), \(\alpha_1 \sim^{\delta-1} \beta_1\), \(\alpha_2 \sim^{\delta-2} \beta_2\), etc.; i.e., for all \(\ty{n}{\N}\) such that \(n < \delta\) we have \(\alpha_n \sim^{\delta-n} \beta_n\).
Therefore, we conclude the following relationship:
\begin{enumerate}[(i)]
\item \(C_{2\delta}(\alpha,\beta) \to \Pity{n}{\N}{n < \delta \to \alpha_n \sim^\delta \beta_n}\),
\item \(\Pity{n}{\N}{n < \delta \to \alpha_n \sim^\delta \beta_n} \to C_\delta(\alpha,\beta)\).
\end{enumerate}
From this relationship, we prove the logical equivalence between the two definitions of uniform continuity.

\vspace{1em}
We first show sequence uniform continuity implies closeness uniform continuity.
By the former, given \(\ty{\varepsilon}{\N}\), \(\ty{\alpha,\beta}{\seq{(\seq X)}}\) we have that there are \(\ty{d,\delta'}{\N}\) such that \(d \leq \delta'\) and such that if \(\alpha_n \sim^{\delta'} \beta_n\) for all \(n < d\) then (by \cref{cor:disseq-C-equiv}) \(C_\varepsilon(f(\alpha),f(\beta))\). 
To show that closeness uniform continuity is satisfied, we need to give some \(\ty{\delta}{\N}\) such that if \(C_\delta(\alpha,\beta)\) then \(C_\varepsilon(f(\alpha),f(\beta))\). 
We set \(\delta := 2{\delta'}\), as by (i) this means we have \(\alpha_n \sim^{\delta'} \beta_n\) for all \(n < \delta'\), and therefore for all \(n < d \leq \delta'\). By sequence uniform continuity as described, we hence have \(C_\varepsilon(f(\alpha),f(\beta))\).

\vspace{1em}
The other direction is again more straightforward. Given \(\ty{\alpha,\beta}{\seq{(\seq X)}}\) we have by closeness uniform continuity that there is some \(\ty{\delta'}{\N}\) such that if \(C_\delta(\alpha,\beta)\) then \(C_\epsilon(f(\alpha),f(\beta))\). To show that sequence uniform continuity is satisfied, we need to give some \(\ty{d,\delta}{\N}\) such that \(d \leq \delta\) and such that if \(\alpha_n \sim^\delta \beta_n\) for all \(n < d\) then \(C_\varepsilon(f(\alpha),f(\beta))\).
We set both \(d,\delta := \delta'\), as by (ii) this means we have \(C_\delta(\alpha,\beta)\). By closeness continuity as described, we hence have \(C_\varepsilon(f(\alpha),f(\beta))\).
\end{proof}

Now we have specialised definitions of uniform continuity for discrete-sequence types which will aid us in proving our operations on \(\K\) are uniformly continuous. Before doing this, we note some expected facts about discrete-sequence uniform continuity.

\begin{lemma}
\label{lem:id-comp-ucont-seq}
The identity function on discrete-sequences is uniformly continuous and composition of discrete-sequence functions preserves discrete-sequence uniform continuity.
\end{lemma}
\begin{proof}[Proof (Sketch).]
Similar to \cref{lem:id-comp-ucont}.
\end{proof}

\begin{lemma}
\label{lem:map-ucont}
\thesislit{6}{SequenceContinuity}{map-ucontinuous'}
For any discrete types \(X\) and \(Y\), given any function \(\ty{f}{X \to Y}\), the function \(\ty{\mathsf{map}(f)}{\seq X \to \seq Y}\) is discrete-sequence uniformly continuous.
\end{lemma}
\begin{proof}
For every \(\ty{\varepsilon}{\N}\) the modulus of uniform continuity is \(\varepsilon\) itself, because \(\alpha_\varepsilon = \beta_\varepsilon\) gives \(\left( \mathsf{map}(f,\alpha)_\varepsilon = \mathsf{map}(f,\beta)_\varepsilon \right) := \left( f(\alpha_\varepsilon) = f(\beta_\varepsilon) \right)\) using \(\mathsf{ap}\).
\end{proof}

\begin{lemma}
\label{lem:zip-ucont}
\thesislit{6}{SequenceContinuity}{zipWith-ucontinuous'}
For any discrete types \(X\), \(Y\) and \(Z\), given any function \(\ty{f}{X \to Y \to Z}\), the function \(\ty{\mathsf{zipWith}(f)}{\seq X \to \seq Y \to \seq Z}\) is discrete-sequence uniformly continuous.
\end{lemma}
\begin{proof}
Similar to \cref{lem:map-ucont}.
\end{proof}

We now explicitly prove the uniform continuity, via the above discrete-sequence definitions (and hence, by the logical equivalences set out above and the discreteness of \(\3\), via the necessary closeness space definitions) of the functions \(\mathsf{neg}\), \(\mathsf{mid}\), \(\mathsf{bigMid}\) and \(\mathsf{mul}\).
This task takes up the rest of this subsection.

\subsubsection{Negation}

First, we prove the uniform continuity of negation.

\begin{corollary}
\label{lem:neg-ucont}
\thesislit{6}{SignedDigitContinuity}{neg-ucontinuous'}
Negation on signed-digits encodings is a uniformly continuous function.
\end{corollary}
\begin{proof}
Recall from \cref{def:neg-sd} that negation is defined \(\mathsf{neg} := \mathsf{map}(\mathsf{flip})\); therefore this is immediately uniformly continuous by \cref{lem:map-ucont}.
\end{proof}

\subsubsection{Binary midpoint}

Next, the uniform continuity of binary midpoint is slightly more complex.

\begin{lemma}
\label{lem:div2-ucont}
\thesislit{6}{SignedDigitContinuity}{div2-ucontinuous'}
The function \(\ty{\mathsf{div2}}{\seq \5 \to \K}\) is uniformly continuous.
\end{lemma}
\begin{proof}
By its recursive definition (\cref{def:div2}), it can be seen that determining the value of \(\mathsf{div2}(\alpha)_\varepsilon\) relies only on \(\alpha_\varepsilon\) and \(\alpha_{\suc \varepsilon}\). Therefore, \(\suc \varepsilon\) is the modulus of uniform continuity for a given \(\varepsilon\).
\end{proof}

\begin{corollary}
\label{lem:mid-ucont}
\thesislit{6}{SignedDigitContinuity}{mid-ucontinuous'}
Binary midpoint on signed-digit encodings is a uniformly continuous function.
\end{corollary}
\begin{proof}
Recall from \cref{def:mid-sd} that binary midpoint is defined \(\mathsf{mid}(\alpha,\beta) := \mathsf{div2}(\mathsf{zipWith}(\mathsf{add\3}(\alpha,\beta)))\). The \(\mathsf{zipWith}\) function is uniformly continuous by \cref{lem:zip-ucont} and the \(\mathsf{div2}\) function is uniformly continuous by \cref{lem:div2-ucont}. Therefore, the composed function \(\mathsf{mid}\) is uniformly continuous by \cref{lem:id-comp-ucont-seq}.
\end{proof}

\subsubsection{Infinitary midpoint}

Before we come to multiplication, we must first prove the infinitary midpoint function is uniformly continuous.

\begin{lemma}
\thesislit{6}{SignedDigitContinuity}{bigMid'-ucontinuous'}
\label{lem:bigMid-prime-ucont}
The function \(\ty{\mathsf{bigMid'}}{\seq{(\K)} \to \seq \9}\) is uniformly continuous.
\end{lemma}
\begin{proof}
By induction on the requested precision \(\ty{\varepsilon}{\N}\). In the base case, recall from \cref{def:bigMid-sd} that \(\mathsf{bigMid}'(\zeta)_0\) only uses \((\zeta_0)_0\), \((\zeta_0)_1\) and \((\zeta_1)_0\); therefore \(d,\delta := 2\).
In the inductive case, recall from \cref{def:bigMid-sd} that \(\mathsf{bigMid}'(\zeta)_{\suc \varepsilon} := \mathsf{bigMid}'(\mathsf{mid}(\mathsf{tail} (\mathsf{tail} \ \zeta_0),\mathsf{tail} \ \zeta_1) :: (\mathsf{tail} (\mathsf{tail} \ \zeta)))\). The argument is the composition of the successor, composition and binary midpoint functions, all of which are uniformly continuous (for binary midpoint, see \cref{lem:mid-ucont}); therefore, the function is uniformly continuous by the inductive hypothesis and \cref{lem:id-comp-ucont-seq}.
\end{proof}

\begin{lemma}
\label{lem:div4-ucont}
\thesislit{6}{SignedDigitContinuity}{div4-ucontinuous'}
The function \(\ty{\mathsf{div4}}{\seq \9 \to \K}\) is uniformly continuous.
\end{lemma}
\begin{proof}
Similar argument to \cref{lem:div2-ucont}.
\end{proof}

\begin{corollary}
\label{cor:bigMid-ucont}
\thesislit{6}{SignedDigitContinuity}{bigMid-ucontinuous'}
Infinitary midpoint on signed-digit encodings is a uniformly continuous function.
\end{corollary}
\begin{proof}
Recall from \cref{def:bigMid-sd} that infinitary midpoint is defined \(\mathsf{bigMid} := \mathsf{div4} \circ \mathsf{bigMid}'\). The result follows by \cref{lem:bigMid-prime-ucont,lem:div4-ucont}.
\end{proof}

\subsubsection{Multiplication}

Finally, the uniform continuity of multiplication follows from the above.

\begin{corollary}
\label{lem:mul-ucont}
\thesislit{6}{SignedDigitContinuity}{mul-ucontinuous'}
Multiplication on signed-digit encodings is a uniformly continuous function.
\end{corollary}
\begin{proof}
Recall from \cref{def:mul-sd} that multiplication is defined \(\mathsf{mul}(\alpha,\beta) := \mathsf{bigMid}(\mathsf{zipWith}(\mathsf{digitMul},\alpha,\lambda n.\beta))\). This is uniformly continuous by \cref{lem:zip-ucont,cor:bigMid-ucont,lem:id-comp-ucont-seq}.
\end{proof}

\subsection{\textsc{Agda}-extracted examples}
\label{sec:K-examples}

We have developed, within the \textsc{TypeTopology} library, a large framework of \textsc{Agda} proofs concerning searchable types, generalised optimisation and regression and ternary signed-digit encodings.
The true test of this framework is in extracting some proof-of-concept computational algorithms of our generalised framework on \(\K\).
The \textsc{Agda} code is compiled into \textsc{Haskell} as described in \cref{appendix:agda}, which allows it to run faster than if we ran it directly in \textsc{Agda} --- though the extracted algorithms are still slow.
The reader can try these examples themselves by following the instructions in \cref{appendix:agda}.

Whether we are searching for an answer to a predicate or optimising/regressing a function up to a given precision \(\ty{\varepsilon}{\N}\), our algorithms will --- using the witness of uniform continuity --- effectively compute a finite prefix of a sequence such that any sequence with that prefix will be a correct answer, approximate minimum or satisfactory parameter of the model function.

For each example, we give a table which notes the answer \(\ty{x}{\K}\) computed for the requested precision-level \(\ty{\varepsilon}{\N}\).
In the few cases where we search \(\K\) \emph{directly}\footnote{Recall the discussion on direct and indirect search from \cref{sec:K-suitable}} each answer is given as a finite prefix and then an ellipses `\(\ldots\)' to notate the infinitely many \(0\)s our algorithm repeats after the computed prefix.
In the usual case where we search \(\K\) \emph{indirectly}, and hence \(x\) is mapped from a sequence of type \(\seq\2\) using the map defined in \cref{def:uparrow-map}), the ellipses `\(\ldots\)' instead denotes the infinitely-many \(\mone\)s that the algorithm repeats after the computed prefix.
To aid illustration, we also note which real \(\ty{\llangle x \rrangle}{\R}\) it is that \(x\) represents and we often give further such information, such as the represented value \(\llangle f(x) \rrangle\) when we are minimising a particular function \(f\).
We note any times above one second taken to compute this answer\footnote{For reference, all of our examples are computed using a MacBook Air M1 laptop.}, which grows quickly due to the inefficiency of our underlying arithmetic and the exhaustive nature of our search (further discussions on this folllow in \cref{sec:exact-real-search-boehm}).

Note that, for ease of reading, we abuse notation and often write functions and representations on \(\K\) as if they are those numbers they represent on the reals \(\I\). For example, we write \(\frac{1}{4}\) instead of \(\mone :: (1 :: (\mone :: \mathsf{repeat} \ 1))\) --- but recall that all of these algorithms operate on the representations of the reals.

\subsubsection{Uniformly continuous search}

Search on ternary signed-digit encodings has been performed previously, for example by Escard\'o in \textsc{Haskell}~\cite{Escardo11fun}.
However, we still provide four examples of uniformly continuous search as proof-of-concept examples of our explicit-continuity assumptions.

\begin{example}
\thesislit{6}{SignedDigitExamples}{Search-Example1}
\label{ex:K-search-1a}
We search for a ternary signed-digit encoding \(\ty{x}{\K}\) that satisfies \[p(x) := \frac{-x}{2} \leq^\varepsilon \frac{1}{4},\] for a variety of requested precision values \(\ty{\varepsilon}{\N}\) indirectly using the uniformly continuous searcher derived from the totally boundedness of \(\seq\2\) (i.e.\ the first proof of \cref{cor:cantor-ucsearch}).

\vspace{0.25cm}
The decidability and uniform continuity of the predicate \(\ty{p}{\K \to \Omega}\) is ensured by that of the approximate linear preorder (\cref{lem:approx-order-ucd-pred}), the uniform continuity of \(\mathsf{mid}\) (\cref{lem:neg-ucont,lem:mid-ucont}) and the composition of these (\cref{lem:f-p-ucont}).

\begin{center}

\begin{tabular}{ L L L L l }
\hline
    \varepsilon & x & \llangle x \rrangle & \frac{-\llangle x \rrangle}{2} & Time (s) \\ \hline\hline
    5 & 1\ldots & 0 & 0 & \\ \hline
    10 & 1\ldots & 0 & 0 & \\ \hline
    15 & 1\ldots & 0 & 0 & 1.09\\ \hline
    20 & 1\ldots & 0 & 0 & 24.88\\ \hline
\end{tabular}
\end{center}

Not long after this, the searcher causes a stack overflow. This was not unexpected: recall that the totally bounded searcher in must compute a \(2^\delta\)-sized \(\delta\)-net of \(\seq\2\) in advance of the search.

\vspace{0.25cm}
We next try using the indirect uniformly continuous searcher derived from the decreasing-modulus uniformly continuous searcher of \(\seq\2\) (i.e.\ the second proof of \cref{cor:cantor-ucsearch}).

\begin{center}
\begin{tabular}{ L L L L l }
\hline
    \varepsilon & x & \llangle x \rrangle & \frac{-\llangle x \rrangle}{2} & Time (s) \\ \hline\hline
    3 & \mone1\ldots & -0.5 & 0.25 & \\ \hline
    6 & \mone\mone1\ldots & -0.515625 & 0.2578125 & \\ \hline
    9 & \mone\mone11\ldots & -0.501953125 & 0.25097656 & 35.68 \\ \hline
\end{tabular}
\end{center}

This searcher doesn't have the overflow problem, but is in this case less efficient than the totally bounded searcher (likely because the proof of the decreasing-modulus searcher is more computationally expensive).
Note also that the two searchers computed different answers; this is because the order in which they evaluate candidate solutions differs --- though of course, both are correct up to the requested precision.
\end{example}

The totally bounded searcher in these examples is usually more efficient, but sometimes (due to the difference in search strategy) the decreasing-modulus searcher is better.
As in this section we wish to illustrate the correctness of our algorithms, and not their efficiency, from now on we use whichever searcher allows us to produce better results for the given example.

\begin{example}
\thesislit{6}{SignedDigitExamples}{Search-Example2}
\label{ex:K-search-2a}
We search for a ternary signed-digit encoding \(\ty{x}{\K}\) that satisfies \[p(x) := C_\varepsilon(\mathsf{mul}(x,x),\frac{1}{2}),\] for a variety of requested precision values \(\ty{\varepsilon}{\N}\) indirectly using the totally bounded uniformly continuous searcher on \(\seq\2\).

\vspace{0.25cm}
The decidability and uniform continuity of the predicate \(\ty{p}{\K \to \Omega}\) is ensured by that of the closeness relation (\cref{lem:closeness-ucd-pred}), the uniform continuity of \(\mathsf{mul}\) (\cref{lem:mul-ucont}) and the composition of these (\cref{lem:f-p-ucont}).

\begin{center}
\begin{tabular}{ L L L L l }
\hline
    \varepsilon & x & \llangle x \rrangle & \llangle \mathsf{mul}(x,x) \rrangle & Time (s) \\ \hline\hline
    1 & \ldots & -1 & 1 & \\ \hline
    2 & \mone\mone1\ldots & -0.75 & 0.5625 & \\ \hline
    3 & \mone\mone1\ldots & -0.75 & 0.5625 & \\ \hline
    4 & 11\mone11\ldots & 0.6875 & 0.47265625 & \\ \hline
    5 & \mone\mone1\mone\mone1\ldots & -0.71875 & 0.516601563 & 3.32 \\ \hline
    6 & 11\mone11\mone1\ldots & 0.703125 & 0.494384766 & 24.19 \\ \hline
    \end{tabular}
\end{center}
The answer correctly converges towards \(\pm\sqrt{0.5} = \pm0.707106781187\ldots\); indeed, it flips between approximations of the two answers for different levels of precision.

\vspace{0.25cm}
At \(\varepsilon := 7\), there was a stack overflow. But the decreasing-modulus searcher did not produce an answer in two minutes for \(n := 4\).
Therefore, this search is much less efficient than that in \cref{ex:K-search-1a}; this is because multiplication requires much higher degrees of input precision than negation and binary midpoint.
\end{example}

\begin{example}
\label{ex:K-search-3}
\thesislit{6}{SignedDigitExamples}{Search-Example3}
We search for a pair of ternary signed-digit encodings \(\ty{(x,y)}{\K \x \K}\) that satisfy \[p(x,y) := C_\varepsilon(\mathsf{mid}(x,y),0),\] for a variety of requested precision values \(\ty{\varepsilon}{\N}\) indirectly using the totally bounded uniformly continuous searcher on \(\seq\2\x\seq\2\).

\vspace{0.25cm}
The decidability and uniform continuity of the predicate \(\ty{p}{\K \to \Omega}\) is ensured by that of the closeness relation (\cref{lem:closeness-ucd-pred}), the uniform continuity of \(\mathsf{mid}\) (\cref{lem:mid-ucont}) and the composition of these (\cref{lem:f-p-ucont}).

\begin{center}
\begin{tabular}{ L L L L l }
\hline
    \varepsilon & x , y & \llangle x \rrangle , \llangle y \rrangle & \llangle \mathsf{mid}(x,y) \rrangle & Time (s) \\ \hline\hline
    5 & 11111\ldots, & 0.9375, & -0.03125 & 9.84 \\ & \mone\mone\mone\mone\mone\ldots & -1 & &  \\ \hline
    10 & 1111111111\ldots, & 0.998046875, & -0.000976563 & \\ & \mone\mone\mone\mone\mone\mone\mone\mone\mone\mone\ldots & -1 & & 75.5  ~ \\ \hline
\end{tabular}
\end{center}

We have shown we can search for two answers in parallel; although this predicate was particularly well-suited to the search strategy of our exhaustive searcher.
\end{example}

\subsubsection{Global optimisation}

Using the optimisation algorithm (\cref{th:min}), we can optimise any function of ternary signed-digit encodings composed from \(\mathsf{neg}\), \(\mathsf{mid}\), \(\mathsf{bigMid}\) and \(\mathsf{mul}\) --- including multivariable and stream functions --- to any degree of precision.
We give two examples of this.

As global optimisation must evaluate a \(\delta\)-net of candidates (for a required degree of input precision \(\ty{\delta}{\N}\) for the requested output precision) in advance of the optimisation process \emph{and} must check each one of these candidates, we find that it very quickly becomes inefficient.

\begin{example}
\label{ex:K-opt-1}
\thesislit{6}{SignedDigitExamples}{Optimisation-Example1}
We compute an \(\varepsilon\)-global minimum of the function \[ f(x) := \mathsf{neg}(x) \] for a variety of requested precision values \(\ty{\varepsilon}{\N}\) indirectly using the totally bounded property of \(\seq\2\).

\vspace{0.25cm}
The continuity of the function \(f\) is by \cref{lem:neg-ucont}.

\begin{center}
\begin{tabular}{ L L L L l }
\hline
    \varepsilon & x & \llangle x \rrangle & \llangle f(x) \rrangle & Time (s) \\ \hline
    \hline
    10 & 1111111111\ldots & 0.998046875 & -0.999511719 & 7.68 \\  \hline
    11 & 11111111111\ldots & 0.999023438 & -0.999511719 & 30.23\\ \hline
    12 & 111111111111\ldots & 0.999511719 & -0.999511719 & 117.18\\ \hline\hline
\end{tabular}
\end{center}
\end{example}

\begin{example}
\label{ex:K-opt-2}
\thesislit{6}{SignedDigitExamples}{Optimisation-Example2}
We compute an \(\varepsilon\)-global minimum of the function \[ f(x) := \mathsf{mul}(x,x) \] for a variety of requested precision values \(\ty{\varepsilon}{\N}\) indirectly using the totally bounded property of \(\seq\2\).

\vspace{0.25cm}
The continuity of the function \(f\) is by \cref{lem:mul-ucont}.

\begin{center}
\begin{tabular}{ L L L L l }
\hline
    \varepsilon & x & \llangle x \rrangle & \llangle f(x) \rrangle & Time (s) \\ \hline
    \hline
    1 & \mone\ldots & -1 & 1 & \\ \hline
    2 & 1\mone\ldots & 0 & 0 & 1.54 \\ \hline
    3 & 1\mone\ldots & 0 & 0 & 81.33 \\ \hline
\end{tabular}
\end{center}

Although the same answer is returned each time, the larger \(\varepsilon\) values means, especially due to the modulus of uniform continuity of exponentiation, an exponentially larger \(\delta\)-net to exhaust; hence the large jump between the time taken for \(\varepsilon := 2\) and \(\varepsilon := 3\).
\end{example}

\subsubsection{Parametric regression}

For regression on ternary signed-digits, we follow the model outlined in \cref{def:gen-reg-practical}.
This means that we will be performing regression where the oracle \(\OO\) is a function and the loss function used is the least-closeness pseudocloseness (\cref{def:least-closeness}) function defined from a given vector \(\ty{v}{(\K)^n}\) of \(n\)-many predictor observations.
Recall that this means the algorithm will only have access to the oracle at the outcomes of the given observations.

Using the regression-as-\emph{optimisation} algorithm (\cref{reg:min}), we can find an \(\varepsilon\)-best choice parameter for any uniformly continuous function.

\begin{example}
\thesislit{6}{SignedDigitExamples}{Regression-Example1a-Optimisation}
\label{ex:K-reg-1-opt}
By fixing the predictor observations \(v := \ty{\{-1,0,1\}}{(\K)^3}\), we define the least-closeness pseudocloseness function \(\ty{L_v}{\closeness{(X \to Y)}}\) between functions, which compares their values at the points in \(v\).
For the oracle function \[ \OO(x) := \mathsf{mid}(\frac{1}{3},x), \] we compute an \(\varepsilon\)-best choice parameter \(\ty{p}{\K}\) of the parameterised model function \[M(p,x) := \mathsf{mid}(\mathsf{neg}(p),x)\] for a variety of requested precision values \(\ty{\varepsilon}{\N}\) indirectly by maximising the function \(\ty{\left(\lambda p.L_v(\OO,M(p^\uparrow))\right)}{\seq\2 \to \Ni}\).The shape of the model function matches the oracle function exactly, except for the fact that the parameter is negated -- we therefore expect the computed parameter to be close to \(-\frac{1}{3}\).

\vspace{0.25cm}
The continuity of the model function is by \cref{lem:mid-ucont,lem:neg-ucont,lem:id-comp-ucont-seq}.

\begin{center}
\begin{tabular}{ L L L l }
\hline
    \varepsilon & p & \llangle p \rrangle & Time (s) \\ \hline
    \hline
    2 & \mone1\mone\ldots & -0.25 & \\ \hline
    4 & \mone1\mone1\mone\ldots & -0.3125 & \\ \hline
    6 & \mone1\mone1\mone1\mone\ldots & -0.328125 & \\ \hline
    8 & \mone1\mone1\mone1\mone1\mone\ldots & -0.33203125 & 1.26 \\ \hline
    10 & \mone1\mone1\mone1\mone1\mone1\mone\ldots & -0.333007813 & 13.79 \\ \hline
    \end{tabular}
\end{center}

The optimisation has returned the value that best connects the observations; with this parameter, the regressed function clearly matches the true oracle up to the requested precision.
\end{example}

It is often more practical to use \emph{search} for regression; i.e.\ to use the algorithms derived from \cref{th:perfect,th:imp}, depending on whether or not there is distortion present in the oracle function.

\begin{example}
\label{ex:K-reg-1-search-perfect}
\thesislit{6}{SignedDigitExamples}{Regression-Example1a-SearchDistortionFree}
For the same predictor observations \(v\), oracle \(\OO\) and parameterised model function \(M\) as in \cref{ex:K-reg-1-opt}, we search for a parameter \(\ty{p}{\K}\) such that \[\underline \varepsilon \preceq L_v(\OO,M(p)),\] for a variety of requested precision values \(\ty{\varepsilon}{\N}\) indirectly using the totally bounded uniformly continuous searcher on \(\seq\2\).

\begin{center}
\begin{tabular}{ L L L l }
\hline
    \varepsilon & p & \llangle p \rrangle & Time (s) \\ \hline
    \hline
    4 & \mone1\mone\ldots & -0.25 & \\ \hline
    8 & \mone1\mone1\mone1\mone\ldots & -0.328125 & 2.20 \\ \hline
    12 & \mone1\mone1\mone1\mone1\mone1\mone\ldots & -0.333007813 & 9.07 \\ \hline
    16 & \mone1\mone1\mone1\mone1\mone1\mone1\mone1\mone\ldots & -0.333343506 & 37.18 \\ \hline
    \end{tabular}
\end{center}

Although the parameter is not necessarily \(\varepsilon\)-best choice, compared to \cref{ex:K-reg-1-opt} the search routine is quicker and (due to lack of distortion in the oracle) the regressed function still matches the true oracle up to the requested precision.
\end{example}

We continue this example by exploring what happens when we receive outcome observations that are distorted from the true oracle.

\begin{example}
\label{ex:K-reg-1-search-imperfect}
\thesislit{6}{SignedDigitExamples}{Regression-Example1a-SearchDistortionProne}
For the same predictor observations \(v\), oracle \(\OO\) and parameterised model function \(M\) as in \cref{ex:K-reg-1-opt}, we search for a parameter \(\ty{p}{\K}\) such that \[\underline \varepsilon \preceq L_v(\Psi(\OO),M(p)),\] directly using the decreasing-modulus uniformly continuous searcher on \(\K\),
for a variety of requested precision values \(\ty{\varepsilon}{\N}\). \(\ty{\Psi}{(\K \to \K) \to (\K \to \K)}\) is a distortion function that distorts the oracle like so:
\[ \Psi(\OO) := \lambda x.\OO(\mathsf{mid}(x,\frac{1}{4})).\]

The graph below shows the true oracle function \(\OO\) (in \textcolor{red}{red}) plotted against the distorted oracle function \(\Psi(\OO)\) (in \textcolor{brown}{brown}), the latter of which the searcher can query at points \(-1\), \(0\) and \(1\).

\begin{center}
\resizebox{\columnwidth/2}{!}{
\begin{tikzpicture}
\begin{axis}[
    axis lines = center,
    xlabel = \(x\),
    ylabel = {\(f(x)\)},
]
\addplot [
    domain=-1:1, 
    samples=100, 
    color=red,
]
{(x + -1/3)/2)};
\addplot [
    domain=-1:1, 
    samples=100, 
    color=brown,
]
{(((5 * x / 8) + -1/3)/2)};
\end{axis}
\end{tikzpicture}
}
\end{center}

\begin{center}
\begin{tabular}{ L L L l }
\hline
    \varepsilon & p & \llangle p \rrangle \\ \hline
    \hline
    1 & 0\mone\ldots & -0.25 \\ \hline
    2 & 0\mone\ldots & -0.25 \\ \hline
    3 & 1111\ldots & 0.9375 \\ \hline
    \end{tabular}
\end{center}

After precision-level \(\varepsilon := 2\), the searcher cannot find a parameter \(p\) that allows the least-close points (either \(M(p,-1)\) and \(\Psi(\OO(-1))\), \(M(p,0)\) and \(\Psi(\OO(0))\), or \(M(p,1)\) and \(\Psi(\OO(1))\)) to become \(\varepsilon\)-close.
\end{example}

However, we could instead simply use optimisation to find the \(\varepsilon\)-best choice parameter.

\begin{example}
\label{ex:K-reg-1-search-imperfect-2}
\thesislit{6}{SignedDigitExamples}{Regression-Example1a-OptimisationDistortionProne}
For the same predictor observations \(v\), oracle \(\OO\), parameterised model function \(M\) and distortion function \(\Psi\) as in \cref{ex:K-reg-1-search-imperfect}, we compute an \(\varepsilon\)-best choice parameter \(\ty{p}{\K}\) of \(M\)
for a variety of requested precision values \(\ty{\varepsilon}{\N}\) directly using the totally bounded property of \(\K\).

\begin{center}
\begin{tabular}{ L L L l }
\hline
    \varepsilon & p & \llangle p \rrangle & Time (s) \\ \hline
    \hline
    1 & \mone1\ldots & -0.25 & \\ \hline
    2 & \mone11\ldots & -0.125 & \\ \hline
    3 & \mone111\ldots & -0.0625 & \\ \hline
    7 & \mone1111111\ldots & -0.00390625 & 138.13 \\ \hline
    \end{tabular}
\end{center}

The graph below shows the true oracle function \(\OO\) (in {\textcolor{red}{red}}) plotted against the distorted oracle function \(\Psi(\OO)\) (in \textcolor{brown}{brown}) and the function \(M(p^7)\) (in {\textcolor{blue}{blue}}) where \(p^7\) is the \(p\) computed when \(\varepsilon := 7\), for the interval \([-1,1]\)
\begin{center}
\resizebox{\columnwidth/2}{!}{
\begin{tikzpicture}
\begin{axis}[
    axis lines = center,
    xlabel = \(x\),
    ylabel = {\(f(x)\)},
]
\addplot [
    domain=-1:1, 
    samples=100, 
    color=red,
]
{(x + -1/3)/2)};
\addplot [
    domain=-1:1, 
    samples=100, 
    color=brown,
]
{(((5 * x / 8) + -1/3)/2)};
\addplot [
    domain=-1:1, 
    samples=100, 
    color=blue,
]
{(x + -0.00390625)/2)};
\end{axis}
\end{tikzpicture}
}
\end{center}
\end{example}

Our final example for this section is inspired by linear regression.

\begin{example}
\label{ex:K-reg-2-search}
\thesislit{6}{SignedDigitExamples}{Regression-Example2-SearchDistortionFree}
By fixing the predictor observations \(v := \ty{\{-\frac{1}{2},\frac{1}{2}\}}{(\K)^2}\), we define the least-closeness pseudocloseness function \(\ty{L_v}{\closeness{(X \to Y)}}\) between functions, which compares their values at the points in \(v\).
We employ the parameterised model function \[ M((p_1,p_2),x) := \mathsf{mid}(p_1,\mathsf{mul}(p_2,x)) \] to search for a parameters \(\ty{p_1,p_2}{\K}\) such that \[\underline \varepsilon \preceq L_v(\OO,M(p_1,p_2)),\]
for a variety of precision values \(\ty{\varepsilon}{\N}\) where \(\OO\) is the synthetically-constructed oracle function (\cref{def:synthesised}) \[ \OO := M(\frac{1}{3},-1) .\] 

\vspace{0.25cm}
The continuity of the model function is by \cref{lem:mid-ucont,lem:mul-ucont,lem:id-comp-ucont-seq}.

\begin{center}
\begin{tabular}{ L L L L L l }
\hline
    \varepsilon & p_1 & p_2 & \llangle p_1 \rrangle & \llangle p_2 \rrangle & Time (s) \\ \hline
    \hline
    3 & 1\mone\ldots & 11\ldots & 0.5 & 0.5 & \\ \hline
    4 & 1\mone1\mone\ldots & 11\ldots & 0.375 & 0.5 & 2.30 \\ \hline
    5 & 1\mone1\mone1\mone\ldots & 11\ldots & 0.34375 & 0.5 & 15.65 \\ \hline
    \end{tabular}
\end{center}

The graph below shows the oracle function \(\OO\) (in {\textcolor{red}{red}}) plotted against the function \(M((p_i)^5)\) (in {\textcolor{blue}{blue}}) where \((p_i)^5\) are the \(p_i\) computed when \(\varepsilon := 5\), for the interval \([-1,1]\)

\begin{center}
\resizebox{\columnwidth/2}{!}{
\begin{tikzpicture}
\begin{axis}[
    axis lines = center,
    xlabel = \(x\),
    ylabel = {\(f(x)\)},
]
\addplot [
    domain=-1:1, 
    samples=100, 
    color=red,
]
{(1/3 + 0.5 * x)/2};
\addplot [
    domain=-1:1, 
    samples=100, 
    color=blue,
]
{(0.34375 + 0.5 * x)/2};
\end{axis}
\end{tikzpicture}
}
\end{center}
\end{example}

Thus we have concluded that we can indeed perform uniformly continuous search and generalised global optimisation and parametric regression on ternary signed-digits.
The correctness of the algorithms are immediate, because they are extracted from our formal \textsc{Agda} framework, but we enjoyed illustrating this fact.
Unfortunately, but not unexpectedly, the efficiency of the algorithms leaves a lot to be desired.

\section{Exact Real Search using ternary Boehm encodings}
\label{sec:exact-real-search-boehm}

The signed-digits have provided a proof-of-concept for applying our generalised perspective of optimisation and regression to types for representing real numbers. 
Furthermore, we have verified the signed-digits, so that we are genuinely optimising/regressing representations of functions on the compact interval we are searching.

There are, however, clear problems with using signed-digits for potential practical applications of exact real search, which can be summed up in two points.
Firstly, the arithmetic defined on the signed-digits is not user-friendly in the same way that, say, arithmetic on the floating-point numbers are --- we cannot even perform addition using signed-digit numbers without utilising multiple representations of different intervals.
Secondly, arithmetic and search on the signed-digits is inefficient; in particular, determining the \(n\)-approximation of a represented number requires the evaluation of the whole \(n\)-prefix of the sequence.

The ternary Boehm encodings address both of these problems. 
For the former, the Boehm encodings represent reals across the real line, meaning we can immediately perform arithmetic operations such as addition that are not suitable for compact intervals.
For the latter, although the arithmetic still remains inefficient, it is much more eficient than that on the ternary signed-digits, and leads us towards much more practical algorithms. Further, recall that the structure of the ternary Boehms means that the \(n\)\textsuperscript{th} interval approximation of a represented number can be determined by evaluating \emph{exactly} the \(n\)\textsuperscript{th} integer approximation of the sequence.

We turn our attention in this final section, therefore, towards the potential of practical search, optimisation and regression using the ternary Boehm encodings.
From this point on we depart from absolute guarantees of correctness in favour of investigating whether our framework can lead us towards a practical implementation of exact real search.
Although we find a positive answer, there is much work to be done to actually deliver on this desire --- and, further still, to tie it in with guarantees of correctness, which we discuss as further work in \cref{fw:practical}.

\subsection{Suitability for search, optimisation and regression}
\label{sec:boehm-suitable}

Recall that, in \cref{sec:boehm}, we re-rationalised Boehm's encodings in our formal framework as the type \(\T\) and prepared them for search by defining the subtypes \(\T(k,i)_i\) (for \(i \in \{1,2,3\}\) for representing real numbers in compact intervals \(\left[ \frac{k}{2^i},\frac{k+2}{2^i} \right]\) encoded as pairs \(\ty{(k,i)}{\Z^2}\).

\subsubsection{\(\T\) yields continuously searchable closeness spaces}

We can easily show that any type \(\T(k,i)_2\) is suitable for search, optimisation and regression because we have already proved its equivalence with \(\K\) (\cref{thm:Compact-2-simeq-K}), a type we recently proved is suitable for these processes (in \cref{sec:K-suitable}).
Therefore, we can search \(\T\) in the same way as we directly search \(\K\).

\begin{corollary}
\(\T(k,i)_2\) is a totally bounded, uniformly continuously searchable closeness space.
\end{corollary}
\begin{proof} \axioms{f}
By \cref{cor:signed-digit-tb}, via the equivalence with \(\K\) (\cref{thm:Compact-2-simeq-K}).
\end{proof}

\noindent
We can instead search \(\T(k,i)_3\) in the same way as we \emph{indirectly} search \(\K\) using \(\seq\2\). 
This is because --- recalling the structural operations from \cref{sec:boehm} --- we only need to consider those interval approximations that are recursively \(\mathsf{downLeft}\) or \(\mathsf{downRight}\) from \(\ty{(k,i)}{\Z^2}\) in order to to determine an answer for any of our algorithms (i.e.\ we do not need to use \(\mathsf{downMid}\)).
In our \textsc{Java} library, we effectively search elements of \(\T(k,i)_3\) as described in \cref{sec:B-examples}.

\subsubsection{\(\T\) yields approximate linear preorders}

Recall that in order to determine the \(n\)\textsuperscript{th} interval approximation of some \(\ty{x}{\T}\), rather than evaluating the whole \(n\)-prefix, it is enough to evaluate \(\ty{x_n}{\Z}\).
This means that we can immediately compare these interval approximations in a more convenient way than the order and closeness functions are used for searching \(\K\).
For example, instead of a predicate asking whether \(C_\varepsilon(x,y)\), requiring us to evaluate \(x\) and \(y\) up to \(\varepsilon\) in order to find out (by \cref{remark:closeness-metric}) whether or not \(d_\R(\llbracket x \rrbracket, \llbracket y \rrbracket) < 2^{-\varepsilon}\), we can instead ask whether \(| x_{\varepsilon} - y_{\varepsilon} | \leq 1\). This further means our representations do not have to match at every point up to \(\varepsilon\), only at \(\varepsilon\) itself.

Another consequence of this direct evaluation is for comparing the order of elements of \(\T\). 
Recall that, for \(\K\), we had to introduce a new approximate linear preorder (\cref{def:real-order-approx-order}) which evaluates prefixes of \(\K\) and converts them to ternary interval codes so that they can be compared.
Elements of \(\T\) are trivial to convert to ternary interval codes at any point, and therefore we can directly compare the interval approximations without any need for exhaustive evaluation.

\subsection{\textsc{Java}-implemented examples}
\label{sec:B-examples}

We have written, in \textsc{Java}, an implementation of the ternary Boehm encodings that allows us to perform search, optimisation and regression.
The implementation is based on the re-rationalisation of ternary Boehm encodings in our \textsc{Agda} library (described in \cref{sec:boehm}) as well as on our informal discussions in the previous section. 
We define arithmetic operations on \(\T\) by completing approximations of them defined on dyadic interval codes \(\Z^2\) (\cref{def:dyadic-interval-code}). A base operation's modulus of uniform continuity is hard-coded into the function object, and composing functions builds a new modulus of uniform continuity from its constitutent functions'.

The \textsc{Java} implementation is outlined in \cref{appendix:java}; here we just give a brief idea of the algorithms, which are effectively the same as the indirect algorithms on \(\seq\2\) that we used to search \(\K\).
For the compact interval represented by \(\ty{(k,i)}{\Z^2}\), whether we are searching for an answer to a predicate or optimising/regressing a function up to a given precision \(\ty{\varepsilon}{\Z}\), our implementation will --- using the witness of uniform continuity \(\ty{\delta}{\Z}\) --- search the finitely-many search candidates that are recursively \(\mathsf{downLeft}\) or \(\mathsf{downRight}\) of \((k,i)\) on precision-level \(\delta\).
Each candidate \((c,\delta)\) is then cast to a ternary Boehm encoding \(\ty{x}{\T}\) such that \(x_\delta = c\), which can be tested against the predicate (or passed to the function being optimised).
In this way, the algorithms compute an interval approximation of a ternary Boehm encoding such that any element of \(\T(k,i)_3\) that features that interval approximation will be a correct answer, approximate minimum or satisfactory parameter of the model function.

For each example, we give a table which notes the computed interval approximation \(\ty{(x_\delta,\delta)}{\Z^2}\) of the answer \(\ty{x}{\T(k,i)_3}\) for the requested output precision-level \(\varepsilon\), required input-precision level \(\delta\) and searched compact interval \((k,i)\).
To aid illustration, we also note the dyadic \(\frac{k+1}{2^i}\) at the center of the interval that \((k,i)\) represents, and we often give further such information, such as the represented value \(\frac{f(\varepsilon)+1}{2^\varepsilon}\) when we are minimising a particular function \(f\).

We note any times above one second taken to compute this answer, and note that we are broadly more efficient than on the ternary signed-digits.
However, as the nature of the search is still exhaustive, the improvements are not seismic.
We discuss the ability to perform branch-and-bound style optimisation (and search) techniques, to further increase efficiency, in \cref{sec:exact-real-search-bnb}.

\subsubsection{Uniformly continuous search}

\begin{example}
\label{ex:B-search-1}
This example is based on \cref{ex:K-search-2a}.
We search for a ternary Boehm encoding \(\ty{x}{\T}\) in \([-1,1]\) that satisfies \[p(x) := \mathsf{abs} \left( \left( x^2 \right)_{\varepsilon+1} - \left(\frac{1}{2}\right)_{\varepsilon+1} \right) \leq 1,\] for a variety of requested precision values \(\ty{\varepsilon}{\Z}\).

\begin{center}
\begin{tabular}{ L L L L l }
\hline
    \varepsilon & (x_\delta,\delta) & \frac{x_\delta+1}{2^\delta} & \frac{(x^2)_\varepsilon+1}{2^\varepsilon} & Time (s) \\ \hline\hline
    5 & (-96 , 7) & -0.7421875 & 0.550842285 & ~ \\ \hline
    10 & (-2902 , 12) & -0.708251953 & 0.501620829 & ~ \\ \hline
    15 & (-92688 , 17) & -0.707145691 & 0.500055028 & ~ \\ \hline
    20 & (-2965826 , 22) & -0.707107782 & 0.500001416 & 1.03 \\ \hline
    25 & (-94906272 , 27) & -0.707106821 & 0.500000057 & 32.4 \\ \hline
\end{tabular}
\end{center}

The answer correctly converges towards the irrational number \(-\sqrt{0.5} = -0.707106781187\ldots\).
It could have alternatively converged towards \(\sqrt{0.5}\), but our searcher evaluates candidates in ascending integer approximation order.

\vspace{0.25cm}
Due to the efficiency gains of ternary Boehm encodings, along with the better (but not formally verified) modulus of uniform continuity on multiplication, we are able to compute the answer to a much higher precision-level than we could in \cref{ex:K-search-2a}. 
\end{example}

\begin{example}
\label{ex:B-search-2}
This example is based on \cref{ex:K-search-3}.
We search for ternary Boehm encodings \(\ty{x,y}{\T}\) in \([-1,1]\) that satisfy \[p(x) := \mathsf{abs} \left( \mathsf{mid}(x,y)_{\varepsilon+1} - 0_{\varepsilon+1} \right) \leq 1,\] for a variety of requested precision values \(\ty{\varepsilon}{\Z}\).

\begin{center}
\begin{tabular}{ L L L L L l }
\hline
    \varepsilon & (x_{\delta_1},\delta_1) & (y_{\delta_2},\delta_2) & \frac{x_{\delta_1}+1}{2^{\delta_1}} & \frac{y_{\delta_1}+2}{2^{\delta_2}} & Time (s) \\ \hline\hline
    5 & (-256 , 8) & (242, 8) & -0.99609375 & 0.94921875 & ~ \\ \hline
    10 & (-8192 , 13) & (8178, 13) & -0.99987793 & 0.998413086 & ~ \\ \hline
    15 & (-262144 , 18) & (262130 , 18) & -0.999996185 & 0.999950409 & \\ \hline
    20 & (-8388608 , 23) & (8388594 , 23) & -0.999999881 & 0.99999845 \\ \hline
    25 & (-268435456 , 28) & (268435442 , 28) & -0.999999996 & 0.999999952 & 356.72 \\ \hline
\end{tabular}
\end{center}

This answer computes quickly to a reasonably high degree of precision; although, as we noted in \cref{ex:K-search-3}, the predicate is particularly well-suited to the search strategy of our exhaustive searcher.
\end{example}

\begin{example}
\label{ex:B-search-3}
This example is of a predicate not well-suited to our exhaustive searcher, and in a different interval than \([-1,1]\).

\vspace{0.25cm}
We tried search for a ternary Boehm encoding \(\ty{x}{\T}\) in \([16,24]\) that satisfies \[p(x) := x^3 + 3x \geq^\varepsilon 9000 ,\] for a variety of requested precision values \(\ty{\varepsilon}{\Z}\).
Unfortunately, even for \(\varepsilon := 1\), the search process hanged for over two minutes. Clearly the number of search candidates at the level of input precision required is too great for an efficient result to be returned. We will revisit this example later, in \cref{ex:B-search-3-branch}.
\end{example}

\subsubsection{Global optimisation}

We define the optimisation algorithm on ternary Boehm encodings based on that arising from \cref{th:min}; it computes the \(\delta\)-net of interval approximations of \(\T(k,i)_3\), where \(\ty{\delta}{\Z}\) is the modulus of uniform continuity of the function being optimised and \((k,i)\) is the interval being searched for an \(\varepsilon\)-global minimum.

\begin{example}
\label{ex:B-opt-1}
This example is based on \cref{ex:K-opt-2}.
We compute an \(\varepsilon\)-global minimum of the function \[ f(x) := x * -1 \] in \([-1,1]\) for a variety of requested precision values \(\ty{\varepsilon}{\Z}\).

\begin{center}
\begin{tabular}{ L L L L l }
\hline
    \varepsilon & (x_\delta,\delta) & \frac{x_\delta+1}{2^\delta} & \frac{(-x)_\varepsilon+1}{2^\varepsilon} & Time (s) \\ \hline\hline
    1 & (14 , 4) & 0.9375 & -0.9375 & ~ \\ \hline
    2 & (62 , 6) & 0.984375 & -0.984375 & ~ \\ \hline
    3 & (254 , 8) & 0.99609375 & -0.99609375 & ~ \\ \hline
    4 & (1022 , 10) & 0.999023438 & -0.999023438 & ~ \\ \hline
    5 & (4094 , 12) & 0.999755859 & -0.999755859 & ~ \\ \hline
    6 & (16382 , 14) & 0.999938965 & -0.999938965 & ~ \\ \hline
    7 & (65534 , 16) & 0.999984741 & -0.999984741 & ~ \\ \hline
    8 & (262142 , 18) & 0.999996185 & -0.999996185 & 4.45 \\ \hline
    9 & (1048574 , 20) & 0.999999046 & -0.999999046 & 97.94 \\ \hline
\end{tabular}
\end{center}

This problem was designed specifically so that we had to nearly exhaust the net; hence, it takes a long time to compute.
Compared to \cref{ex:K-opt-2}, we are only able to compute slightly more precise approximations. This is because, although multiplication's modulus of uniform continuity and the structure of the ternary Boehm encodings admit more efficiency than those on ternary signed-digit encodings, the fact that we have to exhaust the \(\delta\)-net cannot be avoided. 
We will revisit this example later, in \cref{ex:B-opt-1-branch}.
\end{example}

\begin{example}
\label{ex:B-opt-2}
The efficiency issues seen in \cref{ex:B-opt-1} become worse with a more complicated function. 
We tried to compute an \(\varepsilon\)-global minimum of the function \[ f(x) := x^6 + x^5 - x^4 + x^2 \] in \([-2,2]\) (illustrated in \cref{fig:graph}) for a variety of requested precision values \(\ty{\varepsilon}{\Z}\).
Unfortunately, even for \(\varepsilon := 1\), the search process hanged for over two minutes. Clearly the number of candidates in the \(\delta\)-net is too great for an result to be returned in reasonable time. We will revisit this example later, in \cref{ex:B-opt-2-branch}.
\end{example}

\subsubsection{Parametric regression}

For regression on ternary Boehm encodings, we follow the rough idea of the model outlined in \cref{def:gen-reg-practical} and used in \cref{sec:K-examples}.
However, as we now have access to addition, we can tweak our loss function to be more practical.
Rather than returning the least-closeness value \(\mathsf{min}(c(M_p(x_0),\OO(x_0)),...,c(M_p(x_{n-1}),\OO(x_{n-1})))\), for model function \(M\), oracle \(\OO\), parameter choice \(p\) and observations \(x_0,...,x_{n-1}\), we instead simply sum the distances\footnote{This is similar to using the least-squares loss function, a common loss function for parametric regression.}:
\[ \sum_{i:=0}^n \mathsf{abs}(M_p(x_i) - \OO(x_i)) .\]

By implementing the regression-as-\emph{optimisation} algorithm (\cref{reg:min}), we can find an \(\varepsilon\)-best choice parameter for any uniformly continuous function.

\begin{example}
\label{ex:B-reg-1-opt}
This example is based on \cref{ex:K-reg-1-opt}.
By fixing the predictor observations \(v := \ty{\{-1,0,1\}}{(\T)^3}\), we define the absolute loss function \(\ty{L_v}{(\T \to \T) \to (\T \to \T) \to \T}\) between functions, which sums the difference of their values at the points in \(v\).
For the oracle function \[ \OO(x) := \mathsf{mid}(\frac{1}{3},x), \] we compute an \(\varepsilon\)-best choice parameter \(\ty{p}{\T}\) in \([-1,1]\) of the parameterised model function \[M(p,x) := \mathsf{mid}(-p,x)\] for a variety of requested precision values \(\ty{\varepsilon}{\Z}\) by minimising the function \(\ty{\left(\lambda p.L_v(\OO,M(p))\right)}{\T(-1,0)_3 \to \T}\) in the interval \([-1,1]\).

\begin{center}
\begin{tabular}{ L L L l }
\hline
    \varepsilon & (p_\delta,\delta) & \frac{p_\delta+1}{2^\delta} & Time (s) \\ \hline
    \hline
    3 & (-84,8) & -0.25 & \\ \hline
    6 & (-682,11) & -0.328125 & \\ \hline
    9 & (-5460,14) & -0.33203125 & \\ \hline
    12 & (-43690,17) & -0.333251953125 & 4.72 \\ \hline
    15 & (-349524,20) & -0.33331298828125 & 103.59 \\ \hline
    \end{tabular}
\end{center}

The graph below shows the oracle function \(\OO\) (in {\textcolor{red}{red}}) plotted against the function \(M(p^{15})\) (in {\textcolor{blue}{blue}}) where \(p^{15}\) is the \(p\) computed when \(\varepsilon := 15\), for the interval \([0.3,0.4]\). It is hard to tell the two lines apart due to the closeness of the true parameter and the regressed parameter.

\begin{center}
\resizebox{\columnwidth/2}{!}{
\begin{tikzpicture}
\begin{axis}[
    axis lines = center,
    xlabel = \(x\),
    ylabel = {\(f(x)\)},
]
\addplot [
    domain=0.3:0.4, 
    samples=100, 
    color=red,
]
{(x + -1/3)/2)};
\addplot [
    domain=0.3:0.4, 
    samples=100, 
    color=blue,
]
{(x + -0.33331298828125)/2)};
\end{axis}
\end{tikzpicture}
}
\end{center}

We can see that the output precisions are very close, and that the efficiency is much improved when using ternary Boehm, encodings.
We will revisit this example later, in \cref{ex:B-reg-1-opt-branch}.
\end{example}

As we have discussed, it is often more practical to use our regression algorithms that are derived from our searchers.

\begin{example}
\label{ex:B-reg-1-search-perfect}
This example is based on \cref{ex:K-reg-1-search-perfect}.
For the same predictor observations \(v\), oracle \(\OO\) and parameterised model function \(M\) as in \cref{ex:K-reg-1-opt}, we search for a parameter \(\ty{p}{\T}\) in \([-1,1]\) such that \[L_v(\OO,M(p))_\varepsilon \leq 1,\] for a variety of requested precision values \(\ty{\varepsilon}{\N}\).

\begin{center}
\begin{tabular}{ L L L l }
\hline
    \varepsilon & (p_\delta,\delta) & \frac{p_\delta+1}{2^\delta} & Time (s) \\ \hline
    \hline
    3 & (-218,9) & -0.42578125 & \\ \hline
    6 & (-1412,12) & -0.3447265625 & \\ \hline
    9 & (-10970,15) & -0.33477783203125 & 1.76 \\ \hline
    12 & (-87428,18) & -0.3335113525390625 & 20.2 \\ \hline
    15 & (-699098,21) & -0.3333559036254883 & 161.84 \\ \hline
    \end{tabular}
\end{center}

The graph below shows the oracle function \(\OO\) (in {\textcolor{red}{red}}) plotted against the function \(M(p^8)\) (in {\textcolor{blue}{blue}}) where \(p^8\) is the \(p\) computed when \(\varepsilon := 8\), for the interval \([0.3,0.4]\)

\begin{center}
\resizebox{\columnwidth/2}{!}{
\begin{tikzpicture}
\begin{axis}[
    axis lines = center,
    xlabel = \(x\),
    ylabel = {\(f(x)\)},
]
\addplot [
    domain=0.3:0.4, 
    samples=100, 
    color=red,
]
{(x + -1/3)/2)};
\addplot [
    domain=0.3:0.4, 
    samples=100, 
    color=blue,
]
{(x + -0.33333396911621094)/2)};
\end{axis}
\end{tikzpicture}
}
\end{center}

We found that with the ternary signed-digit encodings, the search version of this example (i.e.\ \cref{ex:K-reg-1-search-perfect}) was more efficient than the optimisation version (\cref{ex:K-reg-1-opt}). Interestingly, we find that this is reversed when using the ternary Boehm encodings (i.e.\ this example is less efficient than \cref{ex:B-reg-1-opt}).
\end{example}

\begin{example}
\label{ex:B-reg-1-search-imperfect}
This example is based on \cref{ex:K-reg-1-search-imperfect-2}.
For the same predictor observations \(v\), oracle \(\OO\) and parameterised model function \(M\) as in \cref{ex:B-reg-1-search-perfect}, we compute an \(\varepsilon\)-best choice parameter \(\ty{p}{\T}\) of \(M\) for a variety of requested precision values \(\ty{\varepsilon}{\Z}\) by minimising the function \(\ty{\left(\lambda p.L_v(\Psi(\OO),M(p))\right)}{\T(-1,0)_3 \to \T}\) in the interval \([-1,1]\).

\begin{center}
\begin{tabular}{ L L L l }
\hline
    \varepsilon & (p_\delta,\delta) & \frac{p_\delta+1}{2^\delta} \rrangle & Time (s) \\ \hline
    \hline
    3 & (-114,8) & -0.44140625 & \\ \hline
    6 & (-936,11) & -0.456542969 & \\ \hline
    9 & (-7506,14) & -0.458068848 &  \\ \hline
    12 & (-60072,17) & -0.458305359 & 5.42 \\ \hline
    15 & (-480594,21) & -0.458329201 & 110.76 \\ \hline
    \end{tabular}
\end{center}

The graph below shows the true oracle function \(\OO\) (in {\textcolor{red}{red}}) plotted against the distorted oracle function \(\Psi(\OO)\) (in \textcolor{brown}{brown}) and the function \(M(p^{15})\) (in {\textcolor{blue}{blue}}) where \(p^{15}\) is the \(p\) computed when \(\varepsilon := 15\), for the interval \([-1,1]\)

\begin{center}
\resizebox{\columnwidth/2}{!}{
\begin{tikzpicture}
\begin{axis}[
    axis lines = center,
    xlabel = \(x\),
    ylabel = {\(f(x)\)},
]
\addplot [
    domain=-1:1, 
    samples=100, 
    color=red,
]
{(x + -1/3)/2)};
\addplot [
    domain=-1:1, 
    samples=100, 
    color=brown,
]
{(((5 * x / 8) + -1/3)/2)};
\addplot [
    domain=-1:1, 
    samples=100, 
    color=blue,
]
{(x + -0.458329201)/2)};
\end{axis}
\end{tikzpicture}
}
\end{center}
\end{example}

Our final example for this subsection is an instantiation of linear regression.

\begin{example}
\label{ex:B-reg-2-search}
This example is based on \cref{ex:K-reg-2-search} (but with binary midpoint replaced with addition).
By fixing the predictor observations \(v := \ty{\{-\frac{1}{2},\frac{1}{2}\}}{(\T)^2}\), we define the absolute loss function \(\ty{L_v}{(\T \to \T) \to (\T \to \T) \to \T}\) between functions, which sums the difference of their values at the points in \(v\).
We employ the parameterised model function \[ M((p_1,p_2),x) := p_1 + p_2 * x \] to search for parameters \(\ty{p_1,p_2}{\T}\) in \([-1,1]\) such that \[L_v(\OO,M(p_1,p_2))_\varepsilon \leq 1,\] for a variety of precision values \(\ty{\varepsilon}{\Z}\) where \(\OO\) is the synthetically-constructed oracle function (\cref{def:synthesised}) \[ \OO := M(\frac{1}{3},-1) .\] 

\begin{center}
\begin{tabular}{ L L L L L l }
\hline
    \varepsilon & ((p_1)_{\delta_1},\delta_1) & ((p_2)_{\delta_2},\delta_2) & \frac{(p_1)_{\delta_1}+1}{\delta_1} & \frac{(p_2)_{\delta_2}+1}{\delta_2} & Time (s) \\ \hline
    \hline
    1 & (0,5) & (-2048,11) & 0.03125 & -0.999511719 & 1.3 \\ \hline
    2 & (10,6) & (-8192,13) & 0.171875 & -0.99987793 & 9.50 \\ \hline
    3 & (32,7) & (-32768,15) & 0.2578125 & -0.999969482 & 86.52 \\ \hline
    \end{tabular}
\end{center}

The graph below shows the oracle function \(\OO\) (in {\textcolor{red}{red}}) plotted against the function \(M(p^3)\) (in {\textcolor{blue}{blue}}) where \(p^3\) is the \(p\) computed when \(\varepsilon := 3\), for the interval \([-1,1]\)

\begin{center}
\resizebox{\columnwidth/2}{!}{
\begin{tikzpicture}
\begin{axis}[
    axis lines = center,
    xlabel = \(x\),
    ylabel = {\(f(x)\)},
]
\addplot [
    domain=-1:1, 
    samples=100, 
    color=red,
]
{(1/3 + -1 * x)};
\addplot [
    domain=-1:1, 
    samples=100, 
    color=blue,
]
{(0.2578125 + -0.999969482 * x)};
\end{axis}
\end{tikzpicture}
}
\end{center}

The computation when fitting a linear model to a correct oracle, even in a simple case, is still remarkably inefficient.
We will revisit this example later, in \cref{ex:B-reg-2-search-branch}.
\end{example}

\subsection{\textsc{Java}-implemented branch-and-bound examples}
\label{sec:exact-real-search-bnb}

The ternary Boehms are clearly a much more efficient data-type than the ternary signed-digits for performing exact real search, optimisation and regression on functions for exact real arithmetic.
However, the exhaustive nature of each algorithm means that the efficiency of the search relies heavily on how coincidentally well-suited the predicate or function is to the order in which the algorithms evaluates the candidate interval approximations.
In this final subsection, we discuss how \emph{branch-and-bound} techniques can be defined on the ternary Boehm encodings, and give examples of their use from our \textsc{Java} implementation.

\emph{Branch-and-bound} algorithms are a popular and well-studied technique in optimisation theory that can improve the efficiency of exhaustive search by discarding any candidates that outright cannot contain a solution~\cite{BnB2,BnB3}.
A (unary) branch-and-bound algorithm works as follows:
\begin{enumerate}
\item \emph{Initialise} the search area as some candidate interval,
\item \emph{Select} some candidate from the search area using heuristic criteria,
\item \emph{Branch} the candidate into multiple sub-intervals,
\item \emph{Bound} the sub-intervals by computing each of their lower and upper bounds of \(f\),
\item \emph{Discard} any candidates in the search area that cannot contain a global minimiser (i.e.\ those whose lower bound is strictly higher than another candidate's upper bound),
\item Repeat from step (2) until the width of the range of the search space is less than the desired precision; then return any remaining candidate interval.
\end{enumerate}

\noindent
In the next iteration, the potential solution will either be the same width or thinner than the current.
After each iteration, the remaining search area contains a solution to the global minimisation problem; furthermore, if the width of the remaining search area's output is less than the desired precision \(\varepsilon\), then it is a solution to the \(\varepsilon\)-global minimisation problem.

Branch-and-bound algorithms converge given (i) the function \(f\) is continuous, (ii) the \emph{branching} procedure ensures the width of the widest interval tends to 0, and (iii) the \emph{bounding} procedure ensures that the distance between the lower and upper-bound estimates for each interval also tends to 0~\cite{BnB}.

For ternary Boehm encodings, the branching procedure is the dissection of a ternary interval code \(\ty{(k,i)}{\Z^2}\) into the two intervals \(\ty{(\mathsf{downLeft} \ k,i+1),(\mathsf{downRight} \ k,i+1)}{\Z^2}\) directly below it.
The bounding procedure, meanwhile, is the use of the function's modulus of uniform continuity to bound its behaviour on interval codes.

We postulate that the ternary Boehm encodings can be used to satisfy all three of these conditions:
\begin{enumerate}
\item[(i)] The functions we define are indeed uniformly continuous on the intervals on which we search them,
\item[(i)] The branching procedure halves the width of the candidate interval approximations and --- as there are finitely-many candidate intervals that are only branched up to the precision-level given by the modulus of uniform continuity --- the widest interval will be divided in finite time,
\item[(i)] The bounding procedure of the straightforward functions we have defined (in \cref{sec:boehm-era,fw:boehm-functions}) will decrease the width of the output intervals as the width of the input intervals decreases.
\end{enumerate}

We aim in future work to define this class of algorithms formally in our library and verify their convergence (see \cref{fw:practical}); for now, we give informal \textsc{Java} implementations of their use.

\subsubsection{Uniformly continuous search}

\begin{example}
\label{ex:B-search-3-branch}
We revisit \cref{ex:B-search-3}, an example where we previously could not even compute an answer for output precision-level \(\varepsilon := 1\).
We search for a ternary Boehm encoding \(\ty{x}{\T}\) in \([16,24]\) that satisfies \[p(x) := x^3 + 3x \geq^\varepsilon 9000 ,\] for a variety of requested precision values \(\ty{\varepsilon}{\Z}\) using the branching searcher.

\begin{center}
\begin{tabular}{ L L L L l }
\hline
    \varepsilon & (x_\delta,\delta) & \frac{x_\delta+1}{2^\delta} & \frac{(x^3 + 3x)_\varepsilon+1}{2^\varepsilon} \\ \hline\hline
    50 & (42 , 1) & 21 & 9324 \\ \hline
    100 & (42 , 1) & 21 & 9324 \\ \hline
    150 & (42 , 1) & 21 & 9324 \\ \hline
    200 & (42 , 1) & 21 & 9324 \\ \hline
\end{tabular}
\end{center}

Our branching searcher appears almost to have cheated.
Previously, the number of search candidates at the level of input precision required was too great for an efficient result to be returned.
But, using the branching searcher, we find that the input precision required is in fact very low: the predicate is satisfied easily by any interval approximation that gives an output in the upper half of the search area.
\end{example}

\subsubsection{Global optimisation}

\begin{example}
\label{ex:B-opt-1-branch}
Revisiting \cref{ex:K-opt-2},
we compute an \(\varepsilon\)-global minimum of the function \[ f(x) := x * -1 \] in \([-1,1]\) for a variety of requested precision values \(\ty{\varepsilon}{\Z}\) using the branch-and-bound technique.

\vspace{0.25cm}
\begin{tabular}{ L L }\hline
    \varepsilon & (x_\delta,\delta) \\ \hline\hline
    50 & (1125899906842622,50) \\ \hline
    100 &  (1267650600228229401496703205374,100) \\ \hline
    150 & (1427247692705959881058285969449495136382746622,150) \\ \hline
    200 & (1606938044258990275541962092341162602522202993782792835301374,200) \\ \hline
\end{tabular}
\begin{tabular}{ L L L }
    \hline
    \varepsilon & \frac{x_\delta+1}{2^\delta} & \frac{(-x)_\varepsilon+1}{2^\varepsilon} \\ \hline\hline
    50 & 0.9999999999999982 & -0.9375 \\ \hline
    100 & \approx 1.0 & \approx -1.0 \\ \hline
    150 & \approx 1.0 & \approx -1.0 \\ \hline
    200 & \approx 1.0 & \approx -1.0 \\ \hline
\end{tabular}
\vspace{0.25cm}

Originally, this problem was designed specifically so that we had to nearly exhaust the net.
Using branch-and-bound, this no longer occurs; the simple linear shape of the function allows the algorithm to quickly (in less than \(10\)ms in all above cases) 'zoom in' on the solution.
\end{example}

\begin{example}
\label{ex:B-opt-2-branch}
We revisit \cref{ex:B-opt-2}, an example where we previously could not even compute an answer for output precision-level \(\varepsilon := 1\).
We compute an \(\varepsilon\)-global minimum of the function \[ f(x) := x^6 + x^5 - x^4 + x^2 \] in \([-2,2]\) (illustrated in \cref{fig:graph}) for a variety of requested precision values \(\ty{\varepsilon}{\Z}\) using the branch-and-bound technique.

\begin{center}
\begin{tabular}{ L L L L }
\hline
    \varepsilon & (x_\delta,\delta) & \frac{x_\delta+1}{2^\delta} & \frac{(f(x)_\varepsilon+1}{2^\varepsilon} \\ \hline\hline
    5 & (-236,8) & -0.90625 & -0.0625 \\ \hline
    10 &  (-7384,13) & -0.9013671875 & -0.044921875 \\ \hline
    15 & (-236058,18) & -0.900482177734375 & -0.043670654296875 \\ \hline
    20 & (-7553656,23) & -0.9004659652709961 & -0.04366016387939453 \\ \hline
    25 & (-241716810,28) & -0.9004652798175812 & -0.043659746646881104  \\ \hline
\end{tabular}
\end{center}

The issue with the exhaustive algorithm was that the \(\delta\)-net quickly became so granular, due to the complex shape of the function, that there were far too many candidates to search efficiently.
The branch-and-bound technique helps to tackle this problem: it quickly discards intervals that cannot contain the minimum, refining the search space so that we find an answer much more efficiently.
\end{example}

\subsubsection{Parametric regression}

\begin{example}
\label{ex:B-reg-1-opt-branch}
Revisiting \cref{ex:B-reg-1-opt}, we fix the predictor observations \(v := \ty{\{-1,0,1\}}{(\T)^3}\) and define the absolute loss function \(\ty{L_v}{(\T \to \T) \to (\T \to \T) \to \T}\) between functions, which sums the difference of their values at the points in \(v\).
For the oracle function \[ \OO(x) := \mathsf{mid}(\frac{1}{3},x), \] we compute an \(\varepsilon\)-best choice parameter \(\ty{p}{\T}\) in \([-1,1]\) of the parameterised model function \[M(p,x) := \mathsf{mid}(-p,x)\] for a variety of requested precision values \(\ty{\varepsilon}{\Z}\) by minimising the function \(\ty{\left(\lambda p.L_v(\OO,M(p))\right)}{\T(-1,0)_3 \to \T}\) in the interval \([-1,1]\) using the branch-and-bound technique.

\vspace{0.25cm}
\begin{tabular}{ L L}
\hline
    \varepsilon & (p_\delta,\delta) \\ \hline
    \hline
    50 & (-1501199875790164,52) \\ \hline
    100 & (-1690200800304305868662270940500,102) \\ \hline
    150 & (-1902996923607946508077714625932660181843662164,152) \\ \hline
    200 & (-2142584059011987034055949456454883470029603991710390447068500,202) \\ \hline
\end{tabular}
\begin{tabular}{ L L}
\hline
    \varepsilon & \frac{p_\delta+1}{2^\delta} \\ \hline
    \hline
    50 & -0.33333333333333304 \\ \hline
    100 & \approx -0.33333333333333333\ldots \\ \hline
    150 & \approx -0.33333333333333333\ldots \\ \hline
    200 & \approx -0.33333333333333333\ldots \\ \hline
    \end{tabular}
\vspace{0.25cm}

The graph below shows the oracle function \(\OO\) (in {\textcolor{red}{red}}) plotted against the function \(M(p^{200})\) (in {\textcolor{blue}{blue}}) where \(p^{200}\) is the \(p\) computed when \(\varepsilon := 200\), for the interval \([0.3,0.4]\)

\begin{center}
\resizebox{\columnwidth/2}{!}{
\begin{tikzpicture}
\begin{axis}[
    axis lines = center,
    xlabel = \(x\),
    ylabel = {\(f(x)\)},
]
\addplot [
    domain=0.3:0.4, 
    samples=100, 
    color=red,
]
{(x + -1/3)/2)};
\addplot [
    domain=0.3:0.4, 
    samples=100, 
    color=blue,
]
{(x + -0.333333333333333333333333333)/2)};
\end{axis}
\end{tikzpicture}
}
\end{center}

The efficiency is hugely improved by using the branch-and-bound technique.
\end{example}

\begin{example}
\label{ex:B-reg-2-search-branch}
Revisiting \cref{ex:B-reg-2-search}, we fix the predictor observations \(v := \ty{\{-\frac{1}{2},\frac{1}{2}\}}{(\T)^2}\) and define the absolute loss function \(\ty{L_v}{(\T \to \T) \to (\T \to \T) \to \T}\) between functions, which sums the difference of their values at the points in \(v\).
We employ the parameterised model function \[ M((p_1,p_2),x) := p_1 + p_2 * x \] to search for parameters \(\ty{p_1,p_2}{\T}\) in \([-1,1]\) such that \[L_v(\OO,M(p_1,p_2))_\varepsilon \leq 1,\] for a variety of precision values \(\ty{\varepsilon}{\Z}\), where \(\OO\) is the synthetically-constructed oracle function (\cref{def:synthesised}) \[ \OO := M(\frac{1}{3},-1) ,\] using the branching searcher.

\begin{center}
\begin{tabular}{ L L L L L l }
\hline
    \varepsilon & ((p_1)_{\delta_1},\delta_1) & ((p_2)_{\delta_2},\delta_2) & \frac{(p_1)_{\delta_1}+1}{\delta_1} & \frac{(p_2)_{\delta_2}+1}{\delta_2} & Time (s) \\ \hline
    \hline
    1 & (0,2) & (-2048,11) & 0.25 & -0.999511719 & \\ \hline
    2 & (2,3) & (-8192,13) & 0.375 & -0.99987793 & 2.38 \\ \hline
    3 & (2,3) & (-32768,15) & 0.375 & -0.999969482 & 9.67 \\ \hline
    4 & (10,5) & (-131072,17) & 0.34375 & -0.999998093 & 211.59 \\ \hline
    5 & (10,5) & (-524288,19) & 0.34375 & -0.999998093 & 871.70 \\ \hline
    \end{tabular}
\end{center}

The graph below shows the oracle function \(\OO\) (in {\textcolor{red}{red}}) plotted against the function \(M((p_1)^5,(p_2)^5)\) (in {\textcolor{blue}{blue}}) where \((p_i)^5\) are the \(p_i\) computed when \(\varepsilon := 3\), for the interval \([-1,1]\)

\begin{center}
\resizebox{\columnwidth/2}{!}{
\begin{tikzpicture}
\begin{axis}[
    axis lines = center,
    xlabel = \(x\),
    ylabel = {\(f(x)\)},
]
\addplot [
    domain=-1:1, 
    samples=100, 
    color=red,
]
{(1/3 + -1 * x)};
\addplot [
    domain=-1:1, 
    samples=100, 
    color=blue,
]
{(0.34375 + -0.999998093 * x)};
\end{axis}
\end{tikzpicture}
}
\end{center}

The efficiency is only mildly improved in this case. This example shows that we require better efficiency-minded approaches for search with multiple variables.
\end{example}
\chapter{Conclusion}
\label{chap:conclusion}

\section{Summary of contributions}

This thesis has developed, in a constructive and univalent \textsc{Agda} formalisation (\cref{chap:mltt}), a framework for performing search, optimisation and regression (\cref{chap:generalised}) on a wide class of types given by closeness spaces and uniformly continuously searchable types (\cref{chap:searchable}).
Furthermore, we formally proved that uniformly continuously searchable types are closed under countable products (\cref{thm:tychonoff}).

The \Escardo-Simpson interval object (\cref{sec:interval-object}) and the ternary signed-digit encodings (\cref{sec:signed-digits}) are formalised within our \textsc{Agda} library, and we verify the correctness of the operations on the latter using the former (\cref{sec:signed-digits-era}).

We extracted examples of our formal framework directly from the \textsc{Agda} proofs, showing that we can perform search, optimisation and regression on formally verified functions of the ternary signed-digit encodings (\cref{sec:K-examples}).

We formalised the structure of another type for exact real computation, the ternary Boehm encodings (\cref{sec:boehm}), and informally implemented search, optimisation and regression algorithms on this type in \textsc{Java}, in a way which reflects the formal \textsc{Agda} framework (\cref{sec:B-examples}).

Finally, we discussed the implementation of more efficient algorithms inspired by our formal approach, and gave some examples of their evaluation using \textsc{Java} (\cref{sec:exact-real-search-bnb}).

\section{Further work}

\subsection{Verification of the order and closeness relations on ternary-signed digit encodings}
\label{fw:io}

Although we have formalised the correctness of the functions that we search, optimise and regress on the ternary signed-digit encodings \(\K\) using the interval object \(\I\) (\cref{sec:signed-digits-era}), we have not formally proved that search, optimisation and regression on these representations amounts to those processes on the reals themselves.

In order to achieve this, we will need to verify that the \emph{real-order preserving orders} (\cref{def:real-order-preorder,def:real-order-approx-order}) do indeed preserve the numerical order on the interval object.
However, this task may be somewhat involved, as the notion of an order on \(\I\) is not well established. We therefore seek to establish this notion in our formalisation of the interval object, and then to verify our orders on \(\K\).

Following this, we will also need to verify that the discrete-sequence closeness relation on \(\K\) (\cref{cor:signed-digit-tb}) relates to a notion of a metric on the interval object \(\I\), in the way informally described in \cref{remark:closeness-metric}.

\subsection{Verification of arithmetic on ternary Boehm encodings}
\label{fw:boehm-functions}

We verified the arithmetic operations on the ternary signed-digit encodings, though we have not yet done this on the ternary Boehm encodings.
In order to achieve this, we have begun to develop machinery for completing continuous functions approximated via dyadic interval codes \(\Z^3\) into the equivalent function on ternary Boehm encodings \(\T\), which automatically verifies it relative to the Dedekind reals \(\R\).
Using this machinery, the idea is that we can define a wide variety of operations by following the same blueprint each time.

We already utilise this machinery in our informal \textsc{Java} implementation, though the formal work still relies on conjectures that are fairly open and which require more work to prove both informally and formally.
Although this this work is ongoing, we give the general idea in this section.

\begin{remark}
In this section, we use the notation \(\{x_i\} := \{x_0,...,x_{n-1}\}\) for an \(n\)-sized vector.
\end{remark}

\begin{definition}
A multivariable function \(\ty{f}{\R^n \to \R}\) is approximated by a \emph{dyadic interval approximator} \(\ty{A}{((\Z^3)^n \to \Z^3}\) if,
\begin{enumerate}
\item Given two \(n\)-dimensional vectors of dyadic interval codes \(\ty{\{(k_i,c_i,p_i)\},\{(j_i,b_i,q_i)\}}{(\Z^3)^n}\) and one of dyadic rationals \(\ty{\{w_i\}}{\D^n}\) such that \(\frac{k_i}{2^{p_i}} \leq \frac{j_i}{2^{q_i}} \leq w_i \leq \frac{b_i}{2^{q_i}} \leq \frac{c_i}{2^{p_i}}\), it is the case that \(\frac{Ak}{2^{Ap}} \leq \frac{Aj}{2^{Aq}} \leq f(w) \leq \frac{Ab}{2^{Aq}} \leq \frac{Ac}{2^{Ap}}\), where \((Ak,Ac,Ap) := A(\{(k_i,c_i,p_i)\})\) and \((Aj,Ab,Aq) := A(\{(j_i,b_i,q_i)\})\),
\item Given dyadic intervals \(\ty{\{(a_i,b_i)\}}{(\D^I)^n}\) and required distance \(\ty{\varepsilon}{\D}\), it is the case that there are distances \(\ty{\{\delta_i\}}{\D^n}\) such that for all vectors \(\ty{\{(k_i,c_i,p_i)\}}{(\Z^3)^n}\) satisfying \(a_i \leq \frac{k_i}{2^{p_i}}\), \(\frac{c_i}{2^{p_i}} \leq b_i\) and \(\frac{c_i-k_i}{2^{p_i}} \leq \delta_i\), we have \(\frac{Ac-Ak}{2^{Ap}} \leq \varepsilon\) where \((Ak,Ac,Ap) := A(\{(k_i,c_i,p_i)\})\).
\end{enumerate}
\end{definition}

\noindent
The first condition determines that by refining our inputs we also refine the output, whereas the second determines that the approximator is uniformly continuous on any compact interval.

\begin{example}[Addition's dyadic interval approximator]
Addition on dyadic-rational intervals is defined by:
\begin{multline*}
\left[\frac{k_1}{2^{p_1}},\frac{c_1}{2^{p_1}}\right] + \left[\frac{k_2}{2^{p_2}},\frac{c_2}{2^{p_2}}\right] := \\ 
\left[\frac{2^{p_2-\mathsf{min}(p_1,p_2)}k_1 + 2^{p_1-\mathsf{min}(p_1,p_2)}k_2}{2^{\mathsf{max}(p_1,p_2)}},\frac{2^{p_2-\mathsf{min}(p_1,p_2)}c_1 + 2^{p_1-\mathsf{min}(p_1,p_2)}c_2}{2^{\mathsf{max}(p_1,p_2)}}\right],    
\end{multline*}
Therefore, its dyadic interval approximator is defined:
\begin{multline*}
A((k_1,c_1,p_1),(k_2,c_2,p_2)) := \\ 
(2^{p_2-\mathsf{min}(p_1,p_2)}k_1 + 2^{p_1-\mathsf{min}(p_1,p_2)}k_2, 2^{p_2-\mathsf{min}(p_1,p_2)}c_1 + 2^{p_1-\mathsf{min}(p_1,p_2)}c_2,\mathsf{max}(p_1,p_2)).
\end{multline*}
\end{example}

We use the interval approximator to define the corresponding function on dyadic interval codes. 

\begin{definition}
Given a multivariable function \(\ty{f}{\R^n \to \R}\) which is approximated by a given dyadic interval approximator \(\ty{A}{(\Z^3)^n \to \Z^3}\), the \emph{corresponding function on dyadic interval codes} is defined as follows:
\begin{alignat*}{3}
f' &: (\Z \to \Z^3)^n \to (\Z \to \Z^3), \\
f' &(\{\chi_i\})_n := A(\{(\chi_i)_n\}) .
\end{alignat*}
\end{definition}

\noindent
In order to encode a real number, the output of this corresponding function must be nested and positioned (recall these properties from \cref{cor:real-seq-dyad-code}).

\begin{lemma}
\label{fw:f'-1}
Given a multivariable function \(\ty{f}{\R^n \to \R}\) which is approximated by a given dyadic interval approximator \(\ty{A}{(\Z^3)^n \to \Z^3}\), if the input sequences of dyadic interval codes \(\ty{\{\chi_n\}}{(\Z \to \Z^3)^n}\) are nested then the output of the corresponding function on dyadic interval codes applied at these arguments \(f'(\{\chi_i\})\) is nested.
\end{lemma}
\begin{proof}
By the first condition.
\end{proof}

\begin{conjecture}
\label{fw:f'-2}
Given a multivariable function \(\ty{f}{\R^n \to \R}\) which is approximated by a given dyadic interval approximator \(\ty{A}{(\Z^3)^n \to \Z^3}\), if the input sequences of dyadic interval codes \(\ty{\{\chi_n\}}{(\Z \to \Z^3)^n}\) are nested and positioned then the output of the corresponding function on dyadic interval codes applied at these arguments \(f'(\{\chi_i\})\) is positioned.
\end{conjecture}

\noindent
Using the above, we prove that the corresponding function on dyadic interval codes correctly realises the original function.

\begin{conjecture}
\label{fw:f'}
Given a multivariable function \(\ty{f}{\R^n \to \R}\) which is approximated by a given dyadic interval approximator \(\ty{A}{(\Z^3)^n \to \Z^3}\), the corresponding function on dyadic interval codes \(\ty{f'}{(\Z \to \Z^3)^n \to (\Z \to \Z^3)}\) is such that \[f(\{\llparenthesis \chi_i \rrparenthesis'\}) = \llparenthesis f'(\{ \chi_i \}) ,\rrparenthesis'\] for all \(\ty{\{\chi_i\}}{(\Z \to \Z^3)^n}\) that are nested and positioned.
\end{conjecture}
\begin{proof}
By \cref{lem:real-seq-dyad,fw:f'-1,fw:f'-2}.
\end{proof}

Once we have approximated a function on dyadic interval codes, we next convert it into the equivalent operation on ternary interval codes.
This is done using a special function \(\ty{\mathsf{join}}{(\Z \to \Z^3) \to (\Z \to \Z^2)}\).

\begin{conjecture}
\label{fw:join'}
There is a function \(\ty{\mathsf{join}'}{\Z^3 \to \Z^2}\), which takes as input a dyadic interval code and outputs the narrowest ternary interval code that covers it.
\end{conjecture}

\begin{definition}
The function \(\ty{\mathsf{join}}{(\Z \to \Z^3) \to (\Z \to \Z^2)}\), which converts a sequence of dyadic interval codes into a sequence of ternary interval codes, is defined as follows:
\[ \mathsf{join} := \mathsf{map}(\mathsf{join'}) .\]
\end{definition}

In order to encode a real number, this sequence of ternary interval codes must be nested and positioned.

\begin{conjecture}
\label{fw:join'-1}
Given a sequence of dyadic interval codes \(\ty{\chi}{\Z \to \Z^3}\), if \(\chi\) is nested and positioned then so is \(\mathsf{join}(\chi)\).
\end{conjecture}

\noindent
Furthermore, joining the sequence does not change the real that is encoded.

\begin{conjecture}
\label{fw:join'-2}
Given a sequence of dyadic interval codes \(\ty{\chi}{\Z \to \Z^3}\), \(\llparenthesis \chi \rrparenthesis' = \llparenthesis \mathsf{join}(\chi) \rrparenthesis'\).
\end{conjecture}

\begin{definition}
Given a multivariable function \(\ty{f}{\R^n \to \R}\) which is approximated by a given dyadic interval approximator \(\ty{A}{(\Z^3)^n \to \Z^3}\), the \emph{corresponding function on ternary interval codes} is defined as follows:
\begin{alignat*}{3}
f'' &: (\Z \to \Z^2)^n \to (\Z \to \Z^2), \\
f'' &(\{\chi_i\}) := \mathsf{join}(f'(\{ \mathsf{map}(\mathsf{to{\hy}dcode},\chi_i) \})) .
\end{alignat*}
\end{definition}

\begin{conjecture}
\label{fw:f''}
Given a multivariable function \(\ty{f}{\R^n \to \R}\) which is approximated by a given dyadic interval approximator \(\ty{A}{(\Z^3)^n \to \Z^3}\), the corresponding function on ternary interval codes \(\ty{f''}{(\Z \to \Z^2)^n \to (\Z \to \Z^2)}\)
is such that, \[f(\{\llparenthesis \chi_i \rrparenthesis''\}) = \llparenthesis f''(\{ \chi_i \}) \rrparenthesis''\] for all \(\ty{\{\chi_i\}}{(\Z \to \Z^2)^n}\) that are nested and positioned.
\end{conjecture}
\begin{proof}
By \cref{lem:real-seq-dyad,fw:f',fw:join',fw:join'-1,fw:join'-2}.
\end{proof}

The final step is to normalise the inputs and output of this function.

\begin{conjecture}
\label{fw:normalise}
There is a function \(\ty{\mathsf{normalise}}{(\Z \to \Z^2) \to (\Z \to \Z^2)}\), which takes a nested and positioned sequence of ternary interval codes and gives back a normalised sequence of ternary interval codes that represents the same real number.
\end{conjecture}

\begin{definition}
Given a multivariable function \(\ty{f}{\R^n \to \R}\) which is approximated by a given dyadic interval approximator \(\ty{A}{(\Z^3)^n \to \Z^3}\), the \emph{corresponding function on ternary Boehm encodings} is defined as follows:
\begin{alignat*}{3}
\overline{f} &: \T^n \to \T, \\
\overline{f} &(\{\chi_i\}) := \mathsf{normalise}(f''(\{ \mathsf{to{\hy}interval{\hy}seq} (\chi_i) \})) ,
\end{alignat*}
where \(\ty{\mathsf{to{\hy}interval{\hy}seq}}{\T \to (\Z \to \Z^2)}\) is the map resulting from the equivalence between the ternary Boehm encodings and normalised sequences of ternary interval codes (\cref{lem:T-equiv}).
\end{definition}

\begin{conjecture}
Given a multivariable function \(\ty{f}{\R^n \to \R}\) which is approximated by a given dyadic interval approximator \(\ty{A}{(\Z^3)^n \to \Z^3}\), the corresponding function on ternary Boehm encodings \(\ty{\overline{f}}{\T^n \to \T}\)
is such that, \[f(\{\llbracket \chi_i \rrbracket\}) = \llbracket \overline{f}(\{ \chi_i \}) \rrbracket\] for all \(\ty{\{\chi_i\}}{\T^n}\).
\end{conjecture}
\begin{proof}
By \cref{cor:real-seq-dyad-code,fw:f'',fw:normalise}
\end{proof}

This machinery amounts to showing that, for the above definitions, the following diagram commutes:
\[\begin{tikzcd}
	{\mathbb{T}^n} && {(\Z \to \Z^2)^n} && {(\Z\to\Z^3)^n} && {\R^n} \\
	\\
	{\mathbb{T}} && {(\Z \to \Z^2)} && {(\Z\to\Z^3)} && \R
	\arrow["{\overline{f}}"', dashed, from=1-1, to=3-1]
	\arrow[shift left=1, from=1-1, to=1-3]
	\arrow["{\mathsf{normalise}}", shift left=1, from=1-3, to=1-1]
	\arrow[shift left=1, from=1-3, to=1-5]
	\arrow["{\mathsf{join}}", shift left=1, from=1-5, to=1-3]
	\arrow[shift left=1, from=1-5, to=1-7]
	\arrow[shift left=1, from=3-1, to=3-3]
	\arrow[shift left=1, from=3-3, to=3-1]
	\arrow["{\overline{f}''}", shift right=1, dashed, from=1-3, to=3-3]
	\arrow["{\overline{f}'}", shift right=1, from=1-5, to=3-5]
	\arrow["f", from=1-7, to=3-7]
	\arrow[shift left=1, from=3-3, to=3-5]
	\arrow[shift left=1, from=3-5, to=3-3]
	\arrow[from=3-5, to=3-7]
	\arrow["{\llbracket-\rrbracket^n}", shift left=5, curve={height=-30pt}, dashed, from=1-1, to=1-7]
	\arrow["{\llbracket-\rrbracket}", shift right=5, curve={height=30pt}, dashed, from=3-1, to=3-7]
\end{tikzcd}\]

\subsection{Towards practical implementations of exact real search}
\label{fw:practical}

Our further work is largely concerned with further formalisation and verification results, but there exists a loftier goal of this work: more efficient practical implementations that do not sacrifice the verified correctness.

In our \textsc{Agda} formalisation, we have only used inefficient exhaustive methods for search, optimisation and regression.
In \cref{sec:exact-real-search-bnb}, we discussed more efficient exhaustive methods via branch-and-bound techniques, and gave examples of our implementation in our \textsc{Java} code.
The first step towards bridging the gap between correctness and efficiency in our work is the formalisation of the correctness of these methods --- i.e., the development of general convergence theorems for branch-and-bound techniques in our \textsc{Agda} framework.

Following this, we would like to look at the development of efficient methods that utilise heuristics, and the development of efficient methods for the search, optimisation and regression of multivariable functions.

\appendix
\chapter{Formal \textsc{Agda} Framework}
\label{appendix:agda}

This thesis' primary contribution is that the lion's share of its \emph{other} contributions are fully-formalised in the programming language and proof assistant \textsc{Agda}.

The formalisation is available for viewing \href{\agdarepo}{on this thesis' GitHub repository}.
We outline the formalisation's files (divided into seven directories) below.

\subsubsection{Chapter2}

\agdalink{Chapter2}{Finite.lagda.md} contains the functions we require for finite linearly ordered types.

\vspace{0.5em} \noindent
\agdalink{Chapter2}{Vectors.lagda.md} contains some additional functions we require for vectors.

\vspace{0.5em} \noindent
\agdalink{Chapter2}{Sequences.lagda.md} contains the functions we require for sequences.

\subsubsection{Chapter3}

\agdalink{Chapter3}{ClosenessSpaces.lagda.md} contains the formalisation of closeness spaces and their related lemmas, as described in \cref{sec:cspace}. 

\vspace{0.5em} \noindent
\agdalink{Chapter3}{ClosenessSpaces-Examples.lagda.md} contains the examples of closeness spaces, as described in \cref{sec:cspace-examples}.

\vspace{0.5em} \noindent
\agdalink{Chapter3}{SearchableTypes.lagda.md} contains the formalisation of searchable and uniformly continuously searchable types, as described in \cref{sec:search-finite,sec:search-infinite}.

\vspace{0.5em} \noindent
\agdalink{Chapter3}{SearchableTypes-Examples.lagda.md} contains the examples of uniformly continuously searchable types, as described in \cref{sec:csearch-examples}. The formalised Tychonoff theorem for uniformly continuously searchable types (\cref{thm:tychonoff}) is at the bottom of this file.

\vspace{0.5em} \noindent
\agdalink{Chapter3}{PredicateEquality.lagda.md} contains a small number of lemmas for proving the equality of uniformly continuous and decidable predicates. These lemmas are used in the formalisations of \cref{lem:prod-csearchable,thm:tychonoff}.

\subsubsection{Chapter4}

\agdalink{Chapter4}{ApproxOrder.lagda.md} contains the formalisation of approximate linear preorders, as described in \cref{sec:orders}.

\vspace{0.5em} \noindent
\agdalink{Chapter4}{ApproxOrder-Examples.lagda.md} contains the examples of approximate linear preorders, as described in \cref{sec:orders}.

\vspace{0.5em} \noindent
\agdalink{Chapter4}{GlobalOptimisation.lagda.md} contains the formalisation of our generalised, type-theoretic variant of global optimisation, as described in \cref{sec:gen-global-opt}. The global optimisation algorithm \cref{th:min} is at the bottom of this file.

\vspace{0.5em} \noindent
\agdalink{Chapter4}{ParametricRegression.lagda.md} contains the formalisation of our generalised, type-theoretic variant of parametric regression, as described in \cref{sec:regression}. The parametric regression convergence theorems \cref{reg:min,th:perfect,th:imp} are at the bottom of this file.

\subsubsection{Chapter5}

\agdalink{Chapter5}{IntervalObject.lagda.md} contains the formalisation of the \Escardo-Simpson interval object, as described in \cref{sec:interval-object}.

\vspace{0.5em} \noindent
\agdalink{Chapter5}{IntervalObjectApproximation.lagda.md} contains the formal verification of finite approximations for the interval object, as described in \cref{sec:interval-object}.

\vspace{0.5em} \noindent
\agdalink{Chapter5}{SignedDigit.lagda.md} contains the formalisation of the ternary signed-digit encodings and their arithmetic, as described in \cref{sec:signed-digits}.

\vspace{0.5em} \noindent
\agdalink{Chapter5}{SignedDigitIntervalObject.lagda.md} contains the formal verification of the ternary signed-digit encodings using the interval object, as described in \cref{sec:signed-digits}.

\vspace{0.5em} \noindent
\agdalink{Chapter5}{BoehmVerification.lagda.md} contains our current formalised work on the ternary Boehm encodings, as described in \cref{sec:boehm}.

\vspace{0.5em} \noindent
\agdalink{Chapter5}{BelowAndAbove.lagda.md} contains a variety of lemmas concerning the structure of the ternary Boehm encodings.

\subsubsection{Chapter6}

\agdalink{Chapter6}{SequenceContinuity.lagda.md} contains the definitions and proofs concerning the specialised form of uniform continuity for sequence functions, as seen in \cref{sec:K-suitable}.

\vspace{0.5em} \noindent
\agdalink{Chapter6}{SignedDigitSearch.lagda.md} contains the corollaries required to instantiate our framework for search, optimisation and regression on the ternary signed-digit encodings, as described in \cref{sec:K-suitable}.

\vspace{0.5em} \noindent
\agdalink{Chapter6}{SignedDigitOrder.lagda.md} contains the formalisation of the real-order preserving orders, which allows us to correctly order the reals using the ternary signed-digit encodings, as seen in \cref{sec:K-suitable}.

\vspace{0.5em} \noindent
\agdalink{Chapter6}{SignedDigitContinuity.lagda.md} contains the proofs that the functions we have defined for exact real arithmetic on the ternary signed-digit encodings are uniformly continuous, as described in \cref{sec:K-suitable}.

\vspace{0.5em} \noindent
\agdalink{Chapter6}{SignedDigitExamples.lagda.md} contains the examples of our formal framework for search, optimisation and regression applied to the ternary signed-digits, as described in \cref{sec:K-examples}.

\vspace{0.5em} \noindent
\agdalink{Chapter6}{Main.lagda.md} is a file that can be compiled into a \textsc{Haskell} file in order to run the examples, should the reader desire. Before compiling, ensure that the example being computed (from \verb|SignedDigitExamples.lagda.md|) is the one desired at the correct level of precision; then, run \verb|agda --compile TWA/Thesis/Chapter6/Main.lagda.md| in the \verb|source| folder of the branch. The code can then be ran by performing \\\verb|ghci MAlonzo/Code/TWA/Thesis/Chapter6/Main.hs|. Once \verb|ghci| has loaded, type \verb|main| and hit enter to run the example.
\chapter{\textsc{Java} Implementation of Ternary Boehm Encodings}
\label{appendix:java}

The examples of search, optimisation and regression performed on ternary Boehm encodings given in \cref{sec:exact-real-search-boehm} are implemented in a small \textsc{Java} library written by Andrew Sneap and Todd Waugh Ambridge.

Note that we do not utilise any of Boehm's code (from \cite{Boehm90s}), instead re-implementing both the representation --- due to our modifications (detailed in \cref{sec:boehm}) --- and the basic arithmetic functions --- in order to align with the machinery detailed in \cref{fw:boehm-functions}, which allows the function's modulus of continuity to be easily extracted.

The implementation is available for viewing \href{https://github.com/tnttodda/tnttodda.github.io/tree/master/thesis/java}{on this thesis' GitHub repository}.
We outline the implementation's files (divided into seven packages) below, using the type-theoretic parlance of the rest of the thesis.

\subsubsection{DyadicsAndIntervals}

\java{DyadicsAndIntervals}{Dyadic.java} implements dyadic rational numbers \(\D\) as pairs \(\Z \x \Z\), where \(\ty{(k,i)}{\Z \x \Z}\) represents the dyadic \(\frac{k}{2^i}\). The structural operations defined in \cref{sec:boehm} are extended here to the dyadics. Arithmetic and comparison operations on dyadics are also defined here.

\vspace{0.5em} \noindent
\java{DyadicsAndIntervals}{DyadicIntervalCode.java} implements dyadic interval codes \(\Z^3\) (\cref{def:dyadic-interval-code}), where \(\ty{(k,c,p)}{\Z^3}\) represents the dyadic interval \([\frac{k}{2^p},\frac{c}{2^p}]\). We can yield the dyadic endpoints of a dyadic interval code. Again, we extend the structural operations defined in \cref{sec:boehm} to these interval codes. Arithmetic and comparison operations on dyadic interval codes are also defined here.

\vspace{0.5em} \noindent
\java{DyadicsAndIntervals}{TernaryIntervalCode.java} implements ternary interval codes \(\Z^2\) (\cref{def:ternary-interval-code}), where \(\ty{(k,p)}{\Z^3}\) represents the dyadic interval \([\frac{k}{2^p},\frac{k+2}{2^p}]\). We can convert any ternary interval code into a dyadic interval code. Again, we extend the structural operations defined in \cref{sec:boehm} to these interval codes. Comparison operations on ternary interval codes are also defined here. Functions which `discretise' the interval (i.e. convert it into the intervals directly below it on a higher-precision level) are implemented here.

\subsubsection{TernaryBoehm}

\java{TernaryBoehm}{TBEncoding.java} implements ternary Boehm encodings \(\T\) (\cref{def:ternary-boehm-encoding}). We can convert dyadics, integers and ternary interval codes into ternary Boehm encodings. A ternary Boehm encoding can be converted into a sequence of dyadic interval codes or ternary interval codes. Arithmetic on the ternary Boehm encodings is defined by the application of particular \verb|CFunction| objects in this file.

\subsubsection{FunctionsAndPredicates}

\java{FunctionsAndPredicates}{CFunction.java} implements multivariable continuous functions \(\ty{f}{\T^n \to \T}\) on the ternary Boehm encodings. A function is constructed by giving its interval approximator \((\Z^3)^n \to \Z^3\) on ternary interval codes, as well as information about the continuity of that function. This interval approximated is `completed' to a function on the ternary Boehm encodings, in the manner described in \cref{fw:boehm-functions}. This class contains a large number of static methods that define functions such as negation, addition and multiplication. There are also a variety of methods for the composition of functions, which automatically computes the new interval approximators and continuity information.

\vspace{0.5em} \noindent
\java{FunctionsAndPredicates}{UCUnaryPredicate.java} implements uniformly continuous unary predicates \(\ty{p}{\T \to \Omega}\). A predicate is constructed by giving a function \verb|TBEncoding -> Boolean| and by giving the predicate's modulus of uniform continuity. There are also static methods for some simple predicates, such as \(p(x) := x \leq^\epsilon y\) for a given \(\ty{y}{\T}\). Importantly, there is a constructor for building predicates \(p(x) := p'(f(x))\) defined by functions \(f\) --- this automatically works out the modulus of uniform continuity of the resulting predicate by the moduli of uniform continuity of the underlying function and predicate.

\vspace{0.5em} \noindent
\java{FunctionsAndPredicates}{UCBinaryPredicate.java} implements uniformly continuous binary predicates.

\subsubsection{Search}

\java{Search}{SearchUnary.java} implements a uniformly continuous search algorithm for unary predicates.
The constructor takes a unary predicate and a ternary interval code to search for an answer, while the function \verb|search()| actually performs the algorithm. The algorithm  sets the bounds on the search candidates on the precision-level required (i.e. the level given by the modulus of uniform continuity) and then tests each candidate in numerical order. As soon as an answer is found, it is returned by the algorithm, and if no answer is found then the algorithm exhausts the space and states that no answer was found.

\java{Search}{SearchBinary.java} implements the uniformly continuous search algorithm for binary predicates.

\subsubsection{Optimisation}

\java{Optimisation}{Optimisation.java} implements the global optimisation procedure (\cref{th:min}) for unary functions \(\ty{f}{\T \to \T}\). The constructor takes a function, a requested output precision \(\ty{\epsilon}{\Z}\) and a ternary interval code to search for an \(\epsilon\)-global minimum. The input precision \(\delta\) required to achieve the requested output precision is computed and the \(\delta\)-net of search candidates is generated. Then, each candidate is checked in numerical order until the net is completely exhausted. The algorithm keeps track of the current \(\epsilon\)-minimum argument to the function. Once the net is exhausted, it returns this \(\epsilon\)-minimum argument.

\java{Optimisation}{OptimisationHeuristic.java} implements a branch-and-bound (\cref{sec:exact-real-search-bnb}) global optimisation procedure for unary functions \(\ty{f}{\T \to \T}\). This differs to the usual algorithm in that the a \(\delta\)-net is not computed in advance of the optimisation. Instead, the initial candidate ternary interval code is branched into the two ternary interval codes below it, and their output bounds for the function are computed using the interval approximator and continuity information held in the \verb|CFunction| object. This process repeats, and any candidate interval that has a lower output bound greater than another candidate interval's upper output bound is discarded, as this candidate can clearly not contain a minimum.

\subsubsection{Regression}

\java{Regression}{Regression.java} implements regression algorithms (\cref{sec:regression}) via the above optimisation and search algorithms. The construction of a regression algorithm requires an oracle function, model function and list of predictor observations. The loss function used is defined as \verb|averageModelOracleDistance()|, which sums the distance between the two functions' outcomes on the observations.

\subsubsection{Examples}

\java{ExamplesAndMain}{Examples.java} contains the examples of search, optimisation and regression that we described in \cref{sec:B-examples}.

\backmatter%
\printbibliography[heading=bibintoc]%
\printnomenclature%
\printindex

\end{document}